\documentclass{article}
\usepackage{arxiv}
\usepackage[utf8]{inputenc}
\usepackage{epsf,  amsmath, amssymb, graphicx}
\usepackage{mathtools}
\usepackage{amsthm}
\usepackage{amssymb}
\usepackage{bbding}
\usepackage[sort]{natbib}
\usepackage{a4wide}
\usepackage{subfigure}
\usepackage{tikz}
\usetikzlibrary{decorations.pathreplacing}
\usepackage{xcolor}
\usepackage{dsfont}
\usepackage{pdflscape}
\usepackage{hyperref}
\hypersetup{
  colorlinks = true, 
  urlcolor = blue, 
  linkcolor = blue,
  citecolor = blue 
}
\usepackage[T1]{fontenc}
\usepackage[justification=centering]{caption}
\usepackage{lmodern}
\usepackage{accents}
\usepackage{nomencl}
\usepackage{enumitem}
\usepackage{hhline}
\usepackage{booktabs}
\usepackage{multirow}
\usepackage{colortbl}
\usepackage{tablefootnote}
\usepackage{float}
\usepackage{mathtools}
\usepackage{caption}
\makenomenclature
\newcommand{\Levy}{L\'{e}vy\ }
\newcommand{\Levydriven}{L\'{e}vy-driven\ }
\newcommand{\LevyIto}{L\'{e}vy-It\^{o}\ }
\newcommand{\Ito}{It\^{o}\ }
\newcommand{\cadlag}{c\`{a}dl\`{a}g\ }
\newcommand{\caglad}{c\`{a}gl\`{a}d\ }

\newtheorem{theorem}{Theorem}[section]
\newtheorem{assumption}{Assumption}

\newtheorem{definition}{Definition}[section]
\newtheorem{example}{Example}[section]
\newtheorem{lemma}{Lemma}[section]
\newtheorem{proposition}{Proposition}[section]
\newtheorem{remark}{Remark}[section]

\newlength{\dhatheight}

\newcolumntype{L}[1]{>{\raggedright\let\newline\\\arraybackslash\hspace{0pt}}m{#1}}
\newcolumntype{C}[1]{>{\centering\let\newline\\\arraybackslash\hspace{0pt}}m{#1}}
\newcolumntype{R}[1]{>{\raggedleft\let\newline\\\arraybackslash\hspace{0pt}}m{#1}}

\title{Estimation and Inference for Multivariate Continuous-time Autoregressive Processes}

\renewcommand{\headeright}{}
\renewcommand{\undertitle}{}
\renewcommand{\shorttitle}{Estimation and Inference for MCAR Processes}

\hypersetup{
pdftitle={Estimation and Inference for Multivariate Continuous-time Autoregressive Processes},
pdfauthor={Lorenzo Lucchese}
}

\author{Lorenzo Lucchese\thanks{This research has been supported by the EPSRC Centre for Doctoral Training in Mathematics of Random Systems: Analysis, Modelling and Simulation (EP/S023925/1).} \\
Department of Mathematics \\
Imperial College London \\
\texttt{lorenzo.lucchese17@imperial.ac.uk} \\
\And
Mikko S. Pakkanen\\
Department of Mathematics \\
Imperial College London \\
\texttt{m.pakkanen@imperial.ac.uk} \\
\And
Almut E. D. Veraart\\
Department of Mathematics \\
Imperial College London \\
\texttt{a.veraart17@imperial.ac.uk}
}

\date{\today}

\begin{document}
	
	\maketitle
	%\tableofcontents
	\allowdisplaybreaks
	
	\begin{abstract}
		The aim of this paper is to develop estimation and inference methods for the drift parameters of multivariate \Levydriven continuous-time autoregressive processes of order $p\in\mathbb{N}$. Starting from a continuous-time observation of the process, we develop consistent and asymptotically normal maximum likelihood estimators. We then relax the unrealistic assumption of continuous-time observation by considering natural discretizations based on a combination of Riemann-sum, finite difference, and thresholding approximations. The resulting estimators are also proven to be consistent and asymptotically normal under a general set of conditions, allowing for both finite and infinite jump activity in the driving \Levy process. When discretizing the estimators, allowing for irregularly spaced observations is of great practical importance. In this respect, CAR($p$) models are not just relevant for ``true'' continuous-time processes: a CAR($p$) specification provides a natural continuous-time interpolation for modeling irregularly spaced data -- even if the observed process is inherently discrete. As a practically relevant application, we consider the setting where the multivariate observation is known to possess a graphical structure. We refer to such a process as GrCAR and discuss the corresponding drift estimators and their properties. The finite sample behavior of all theoretical asymptotic results is empirically assessed by extensive simulation experiments. 		
	\end{abstract}

        {\hyperlink{https://mathscinet.ams.org/mathscinet/msc/msc2020.html}{\textit{Mathematics Subject Classification 2020:}} 62F12, 60G30.}\\
	{\textit{Keywords:} Estimation, Inference, Asymptotic properties, Autoregressive, Continuous-time, \Levy process, \\ Maximum likelihood, Ornstein-Uhlenbeck process, State-space models, Graph topology.}
	
	\section{Introduction}
	\subsection{Motivation: from OU processes to CAR processes}
	Let $\mathbb{L} = \{\mathbf{L}_t, \ t \geq 0\}$ denote a $d$-dimensional \Levy process with characteristics $(\mathbf{b}, \Sigma, F)$ on the probability space $(\Omega', \mathcal{F}', \mathbb{P}')$, i.e.\ $\mathbf{b}\in\mathbb{R}^d$, $\Sigma$ a symmetric positive semi-definite $d\times d$ matrix and $F$ a \Levy measure on $\mathbb{R}^d$, that is $F(\{0\}) = 0$ and $\int_{\mathbb{R}^d} (\|x\|^2\wedge 1) F(dx) <\infty$. For a detailed treatment of \Levy processes see \citet{Sato_1999}. The Ornstein-Uhlenbeck (OU) process driven by $\mathbb{L}$ with drift matrix $A\in\mathcal{M}_d(\mathbb{R})$ is the unique strong solution $\mathbb{Y} = \{\mathbf{Y}_t,\ t\geq 0\}$ to the stochastic differential equation (SDE)
	\begin{equation} \label{eqn:OU_intro}
		\mathrm{d}\mathbf{Y}_t = - A \mathbf{Y}_t\, \mathrm{d}t + \mathrm{d}\mathbf{L}_t,\ t\geq 0.
	\end{equation}
	This process has been the subject of great research interest from a theoretical and practical point of view. When $\mathbb{L}$ is taken to be a Brownian motion, its properties are now well understood, with the first studies dating back to \citet{OU_1930}. OU processes with more general \Levy drivers have also been extensively analyzed, for example, \citet{Sato_Yamazato_1984} related stationary OU-type processes with operator self-decomposable distributions while \cite{Masuda_2004} studied their transition semigroups, transition densities, and mixing conditions. Such processes have been used for modeling purposes in various fields, including physics, finance, and biology. For a concrete financial application, see \citet{barndorff_nielsen_shephard_2001} where a positive stationary OU process is used to model the volatility process in a stochastic volatility model.
	
	Formally, the OU process \eqref{eqn:OU_intro} satisfies
	\begin{equation} \label{eqn:OU_AR}
		p(D)\mathbf{Y}_t = D\mathbf{L}_t, \ t\geq 0,
	\end{equation}
	where $p(z) = z + A$ is a first order polynomial and $D$ denotes differentiation with respect to time $t$. If we instead interpret $D$ as a backshift operator we immediately notice why OU processes can be understood as the continuous-time analog of discrete-time autoregressive processes of order 1, AR(1) for short. A deeper connection between the two classes of processes is present: under appropriate conditions on $A$ and $\mathbb{L}$, for any $h>0$ the $h$-skeleton of the OU process $\mathbb{Y}$, defined by $\mathbb{Y}_h:=\{\mathbf{Y}_{nh}, \ n\in\mathbb{N}\}$, is an AR(1) process, see for example \citet[Theorem~ 4.3]{Masuda_2004}. In view of this connection, we can naturally generalize OU processes in the same way as AR($p$) processes generalize AR($1$) processes, i.e.\ for $p\in\mathbb{N}$ and $A_1,\ldots,A_p\in\mathcal{M}_d(\mathbb{R})$ we consider the $(p-1)$-times ``differentiable'' process $\mathbb{Y}$ which formally satisfies Equation \eqref{eqn:OU_AR} where 
	\[p(z)= z^p + A_1z^{p-1} +\ldots + A_{p-1}z + A_p\] 
	is a so-called autoregressive polynomial and $D$ denotes differentiation with respect to time $t$. Such a process, which will be rigorously defined in Section \ref{sec:MCAR}, is known as a Continuous Autoregressive (CAR) process and falls in the wider class of Continuous Autoregressive Moving Average (CARMA) processes, the continuous-time counterpart of ARMA time series models. An overview of CARMA processes is given in \citet[Chapter~11.5]{brockwell_davis}, focusing mainly on the results for Gaussian CARMA processes obtained by \citet{jones1980} and \citet{chan_tong}. Properties of \Levydriven CARMA processes are investigated by \citet{brockwell_CARMA} and \citet{brockwell_lindner} in the univariate case and by \citet{Marquardt_Stelzer_2007} in the more general multivariate setting. It is important to note that the study of the Lévy-driven OU and CARMA processes is an active area of research, and many researchers have made significant contributions to their understanding. The mentioned works are merely a small fraction of the research conducted in this field and should not be regarded as an exhaustive compilation.
	
	\subsection{Contributions and literature review}
	Firstly and most prominently we develop consistent and asymptotically normal estimators for the drift coefficients of multivariate \Levydriven CAR($p$) processes both when continuous-time and (possibly irregularly spaced) discrete-time observations are available. As a necessary step in the discretization procedure, in Section \ref{sec:approx_derivatives} we provide a rigorous treatment of the finite difference approximation of the derivatives of a differentiable stochastic process which, to the best of our knowledge, is a novelty in the field. Moreover, we believe the proofs of the asymptotic results for the OU drift estimators appearing in \citet{mai_OU} might contain some inconsistencies. By carefully modifying the arguments we show the stated properties hold under slightly different assumptions as a special case of our more general asymptotic results for CAR($p$) processes. Finally, when the multivariate observation has a known underlying network structure, we reduce the model parametrization by introducing the GrCAR architecture, a natural extension of the GrOU process introduced in \citet{Courgeau_Veraart_2021}, and the corresponding estimators in Section \ref{sec:GrCAR}.
	
	\paragraph{Estimation of OU drift coefficients} When $p=1$ our task reduces to estimating the drift coefficient of the \Levydriven OU process in Equation \eqref{eqn:OU_intro}. In this setting, a wide range of estimators have been investigated in the literature under different assumptions on the driving \Levy process $\mathbb{L}$, cf.\ Table \ref{table:OU_estimators}. Other than the well-studied Gaussian case, specific results were developed in the context of non-decreasing \Levy drivers in \citet{jongbloed_OU} and \citet{brockwell_OU} for Davis-McCormick $M$-estimators. As discussed in the previous section, this class of processes is of particular relevance for volatility modeling \citep{barndorff_nielsen_shephard_2001}. Assuming the \Levy driver has a non-degenerate Gaussian component, \citet{mai_OU} develops continuous-time maximum likelihood estimators and analyzes their discrete-time counterparts, naturally allowing for irregularly spaced data. \citet{gushchin_OU} extend the continuous-time likelihood results when $\mathbb{L}$ has heavier tails and \cite{Courgeau_Veraart_2022} consider a specific multivariate extension, imposing a graphical structure on the observations. More general (symmetric) \Levy drivers are considered in \citet{fasen_OU} and \citet{masuda_LAD} under a discrete-time uniform sampling scheme, where the authors analyze least squares and least absolute deviation estimators respectively. 
	
	In the present work, we aim to revisit the results in \citet{mai_OU}, treated in \citet{Courgeau_Veraart_2022} in the multivariate case, and extend them to multivariate CAR($p$) processes of general order $p\in\mathbb{N}$. Working under similar assumptions, we will consider explicit continuous-time maximum likelihood estimators and their corresponding discretizations. 	
	
	\bgroup
	\def\arraystretch{1.25}% 
	\begin{table}[ht]
		\centering
		\begin{tabular}{|C{2.5cm}|C{3.75cm}|C{3.5cm}|C{2cm}|C{2.3cm}|}
			\cline{2-5}
			\multicolumn{1}{c|}{} & \multicolumn{4}{c|}{\Levy driver $\mathbb{L}$ with characteristics $(\mathbf{b}, \Sigma, F)$} \\ 
			\cline{2-5}
			\multicolumn{1}{c|}{} & \multicolumn{2}{c|}{Square-integrable\textsuperscript{1}} & \multicolumn{2}{c|}{Heavy tails\textsuperscript{2}} \\
			\cline{2-5}
			\multicolumn{1}{c|}{} & Non-decreasing\textsuperscript{3} & $\Sigma>0$\textsuperscript{4} & $\Sigma>0$\textsuperscript{4} & \\
			\hline
			continuous-time & \cite{brockwell_OU} & \cite{mai_OU}, \cite{Courgeau_Veraart_2022} & \cite{gushchin_OU} & \\ 
			\hline
			discrete-time (uniform) & \cite{brockwell_OU}, \cite{jongbloed_OU} & \cite{mai_OU}, \cite{Courgeau_Veraart_2022} & &\cite{fasen_OU}, \cite{masuda_LAD} \\
			\hline
			discrete-time (irregular) &  & \cite{mai_OU}, \cite{Courgeau_Veraart_2022} & & \\
			\hline
		\end{tabular}
		\caption{Estimation of OU drift coefficients. \\ \textsuperscript{1}$\int \|x\|^2 F(\mathrm{d}x)<\infty$. \textsuperscript{2}$F$ is regularly varying of order $\alpha\in(0,2)$. \textsuperscript{3}$\mathrm{supp}(F)\subseteq (0,\infty),\ \Sigma=0$. \\ \textsuperscript{4}$\Sigma$ strictly positive definite.} \label{table:OU_estimators}
	\end{table}
	\egroup 
	
	\paragraph{Estimation of CAR drift coefficients}
	Fewer results are available for the estimation of CAR($p$) processes of order $p\in\mathbb{N}$, some of which are summarized in Table \ref{table:CAR_estimators}. In the more general setting of \Levydriven CARMA($p$, $q$) models, to the best of our knowledge, all available estimators are based on uniform discrete-time observations and their statistical analysis relies on the properties of the $h$-skeleton of the process. Following this approach, \cite{brockwell_CARMA_estimation} develops least squares estimators with non-negative \Levy drivers while \cite{schlemm_stelzer_CARMA} provides mixing conditions which set the scene for  Gaussian quasi maximum likelihood estimation. As previously mentioned, in this work we will take a diametrically opposed approach by starting from continuous-time maximum likelihood estimators. In this spirit, we aim to extend the results for Gaussian CAR($p$) processes presented in \cite{brockwell_gaussian_CAR} by allowing for general \Levy drivers and rigorously proving the asymptotic properties of the discretized estimators.

	\bgroup
	\def\arraystretch{1.25}% 
	\begin{table}[ht]
		\centering
		\begin{tabular}{|C{2.5cm}|C{3cm}|C{3cm}|C{3cm}|}
			\cline{2-4}
			\multicolumn{1}{c|}{} & \multicolumn{3}{c|}{\Levy driver $\mathbb{L}$ with characteristics $(\mathbf{b}, \Sigma, F)$} \\ 
			\cline{2-4}
			\multicolumn{1}{c|}{} & \multicolumn{3}{c|}{Square-integrable\textsuperscript{1}} \\
			\cline{2-4}
			\multicolumn{1}{c|}{} & Gaussian\textsuperscript{2} & Non-decreasing\textsuperscript{3} & $\Sigma>0$\textsuperscript{4}\\
			\hline
			continuous-time & \cite{brockwell_gaussian_CAR} &  & $(*)$ \\ 
			\hline
			discrete-time (uniform) & \cite{brockwell_gaussian_CAR} &\cite{brockwell_CARMA_estimation} & \cite{schlemm_stelzer_CARMA}, $(*)$ \\
			\hline
			discrete-time (irregular) & & & $(*)$ \\
			\hline
		\end{tabular}
		\caption{Estimation of CAR drift coefficients. $(*)$ contribution of this paper. \\
			\textsuperscript{1}$\int \|x\|^2 F(\mathrm{d}x)<\infty$. \textsuperscript{2}$F\equiv 0$, $\Sigma>0$. \textsuperscript{3}$\mathrm{supp}(F)\subseteq (0,\infty),\ \Sigma=0$. \\ \textsuperscript{4}$\Sigma$ strictly positive definite.} \label{table:CAR_estimators}
	\end{table}
	\egroup
	
	\paragraph{Outline of paper} The rest of the paper is organized as follows. In Section \ref{sec:MCAR}, we define multivariate CAR models via their state-space representation. In Section \ref{sec:MCAR_cont}, we develop continuous-time maximum likelihood theory based on fundamental results for more general stochastic processes \citep{Jacod_Memin_1976, Jacod_Shiryaev_1987, Sorensen_1991, Kuchler_Sorensen_1997}. In Section \ref{sec:MCAR_discr}, we discretize the maximum likelihood estimators and prove the asymptotic properties are preserved. The discretizations are obtained by approximating the integrals by Riemann sums, approximating the derivatives by finite differences, and disentangling the continuous and jump parts via thresholding. In Section \ref{sec:GrCAR}, we consider a specific application of CAR models to high-dimensional multivariate observations with a graphical structure, allowing for a parsimonious parametrization of the drift coefficients. Finally, Section \ref{sec:conclusions} summarizes the main findings and discusses possible avenues for future work.
	
	\section{Multivariate Continuous Autoregressive processes} \label{sec:MCAR}
	
	In the following we denote by $\mathcal{M}_{n,d}(\mathbb{R})$ the set of $n\times d$ matrices over $\mathbb{R}$. When working with square matrices we write $\mathcal{M}_{d}(\mathbb{R}) = \mathcal{M}_{d,d}(\mathbb{R})$ and denote by $\sigma(A)$ the spectrum of $A\in\mathcal{M}_{d}(\mathbb{R})$. For a vector space $V$, we denote by $V^k$ its $k$-th Cartesian power, i.e.\ $V^k:=\{(v_1, \ldots, v_k): v_1,\ldots,v_k\in V\}$. We start by making the definition of a multivariate \Levydriven CAR process rigorous.
	
	\begin{definition}[MCAR process] \label{def:MCAR}
		Let $\mathbb{L} = \{\mathbf{L}_t, \ t \geq 0\}$ be a $d$-dimensional \Levy process with characteristics $(\mathbf{b}, \Sigma, F)$. Let $p\in \mathbb{N}$ and $\mathbf{A} = (A_1, \ldots, A_p) \in (\mathcal{M}_d(\mathbb{R}))^p$. Define 
		\begin{equation} \label{eqn:design_A_E}
			\mathcal{A}_{\mathbf{A}} = \begin{pmatrix}
				0_{d \times d} & I_{d \times d} & 0_{d \times d} & \cdots & 0_{d \times d} \\
				0_{d \times d} & 0_{d \times d} & I_{d \times d} & \cdots & 0_{d \times d} \\
				\vdots & \vdots & \vdots & \ddots & \vdots \\
				0_{d \times d} & 0_{d \times d} & 0_{d \times d} & \cdots & I_{d \times d} \\
				- A_p & -A_{p-1} & -A_{p-2} & \cdots & -A_1 \\
			\end{pmatrix}\in \mathcal{M}_{pd}(\mathbb{R}), 
			\quad \mathcal{E} = \begin{pmatrix}
				0_{d \times d} \\ 0_{d \times d} \\ \vdots \\ 0_{d \times d} \\ I_{d \times d}
			\end{pmatrix} \in \mathcal{M}_{pd, d}(\mathbb{R}).
		\end{equation}
		Then, for any set of matrices $A_1, \ldots, A_p \in \mathcal{M}_d(\mathbb{R})$, we define a $d$-dimensional multivariate continuous-time autoregressive process of order $p$, MCAR($p$) process for short, driven by $\mathbb{L}$ with coefficient matrices $A_1, \ldots, A_p \in \mathcal{M}_d(\mathbb{R})$ and initial state-space representation $\xi$, an $\mathcal{F}'_0$-measurable random variable independent of $\mathbb{L} = \{\mathbf{L}_t, \ t \geq 0\}$, to be $\mathbb{Y} = \{\mathbf{Y}_t, \ t\geq0\}$ such that
		\begin{equation}
			\mathbf{Y}_t =  
			\begin{pmatrix}
				I_{d \times d} & 0_{d \times d} & \cdots & 0_{d \times d}
			\end{pmatrix} 
			\mathbf{X}_t, \quad t \geq 0, \label{eqn:CAR(p)}
		\end{equation}
		where $\mathbb{X} = \{\mathbf{X}_t,\ t \geq 0\}$ is the unique  $pd$-dimensional solution-process to the SDE
		\begin{equation} \label{eqn:SDE_def}
			d \mathbf{X}_t = \mathcal{A}_{\mathbf{A}} \mathbf{X}_t dt + \mathcal{E} d \mathbf{L}_t, \quad \mathbf{X}_0 = \xi.
		\end{equation}
	\end{definition}
	
	\begin{remark} \label{rem:stat_ergodic}
		We stress the following important properties of the MCAR($p$) process:
		\begin{itemize}
			\item The MCAR($p$) process is well-defined since the stochastic differential Equation \eqref{eqn:SDE_def} has a unique solution-process $\mathbb{X}=\{\mathbf{X}_t,\ t\geq 0\}$ with the following representation:
			\[\mathbf{X}_t = e^{\mathcal{A}_{\mathbf{A}}t}\xi + \int_0^t e^{\mathcal{A}_{\mathbf{A}}(t-s)} \mathcal{E} d\mathbf{L}_s, \quad t\geq 0.\]
			\item If, in particular, we assume
			\begin{itemize}
				\item $\mathbb{L}$ satisfies the log-moment condition $\displaystyle \int_{\|x\|\geq 1} \log \|x\| F(dx) <\infty$,
				\item the coefficient matrices are such that $\sigma(\mathcal{A}_{\mathbf{A}}) \subseteq (-\infty, 0) + i\mathbb{R}$, and
				\item the initial condition $\xi$ has law $\displaystyle \mathcal{L}\left( \int_0^\infty e^{\mathcal{A}_{\mathbf{A}} s} \mathcal{E}\, d\mathbf{L}_s\right)$,
			\end{itemize}
			then the $d$-dimensional MCAR($p$) process $\mathbb{Y} = \{\mathbf{Y}_t,\ t \geq 0\}$ and its state-space representation $\mathbb{X} = \{\mathbf{X}_t,\ t\geq 0\}$ are strictly stationary, cf.\ \citet[Proposition~2.2]{Masuda_2004} and \citet[Theorem~4.1~and~Theorem~4.2]{Sato_Yamazato_1984}. If, moreover, the driving \Levy process is assumed to be $r$-integrable for some $r>0$ they are also ergodic, cf.\ \citet[Proposition~3.34]{Marquardt_Stelzer_2007}.
			\item If we extend $\mathbb{L}$ to be a two-sided \Levy process, i.e.\ $\mathbb{L} = \{\mathbf{L}_t,\ t \in \mathbb{R}\}$ defined by
			\[ \mathbf{L}_t := \mathbf{L}^1_t \mathds{1}_{[0, \infty)}(t) - \mathbf{L}^2_{-t-} \mathds{1}_{(-\infty, 0)}(t),\quad t\in\mathbb{R}, \]
			where $\mathbb{L}^1 = \{\mathbf{L}^1_t,\ t \geq 0\}$ and $\mathbb{L}^2 = \{\mathbf{L}^2_t,\ t \geq 0\}$ are independent identically distributed one-sided \Levy processes, we can write the state-space representation of the strictly stationary MCAR($p$) process $\mathbb{Y} = \{\mathbf{Y}_t,\ t\in\mathbb{R}\}$ as
			\[\mathbf{X}_t = \int_{-\infty}^t e^{\mathcal{A}_{\mathbf{A}}(t-s)}\mathcal{E} d\mathbf{L}_s, \quad t\in\mathbb{R}.\]
		\end{itemize}
	\end{remark}
	
	The process $\mathbb{X}$ is known as the $pd$-dimensional state-space process of $\mathbb{Y}$. Formally, the MCAR($p$) process $\mathbb{Y} = \{\mathbf{Y}_t,\ t\geq 0\}$ satisfies
	\begin{equation} \label{eqn:autoregressive}
		p(D) \mathbf{Y}_t = D^p\mathbf{Y}_t + A_1 D^{p-1}\mathbf{Y}_t + \ldots + A_{p-1}  D\mathbf{Y}_t + A_{p} \mathbf{Y}_t = D\mathbf{L}_t,
	\end{equation}
	where $p(z) = z^p + A_1z^{p-1} + \ldots + A_{p-1}z + A_p$ is the autoregressive polynomial and the operator $D$ denotes differentiation with respect to $t$. Note that $D\mathbf{L}_t$ is in general not defined as a random function and hence we make sense of this formal equation by introducing the $\mathbb{R}^{pd}$-dimensional state-space representation $\mathbb{X} = \{\mathbf{X}_t,\ t\geq 0\}$. Note that, due to this construction, from a realization of the CAR($p$) process $\mathbb{Y}$ we can recover its state-space representation via right differentiation.
	
	\begin{proposition} \label{prop:right_diff}
		Let $D$ denote right differentiation, i.e\ for a function $f:[0, \infty)\rightarrow \mathbb{R}^d$ we say $f$ is right differentiable at $t\in[0,\infty)$ if the limit
        \[\lim_{\delta\downarrow 0} \frac{f(t+\delta) - f(t)}{\delta},\] 
        exists, in which case we call this the right derivative of $f$ at $t\in[0,\infty)$ and denote it by $Df(t)$. We say $f$ is right differentiable if it is right differentiable at every $t\in[0,\infty)$. Let $\mathbb{Y}$ be an MCAR($p$) process with state-space representation $\mathbb{X}=(\mathbb{X}^{(1),\mathrm{T}},\ldots, \mathbb{X}^{(p),\mathrm{T}})^\mathrm{T}$. Then we have that almost surely $\mathbb{Y}$ is $(p-1)$-times right differentiable and
		\[\mathbb{X}^{(j+1)} = D^j\mathbb{Y},\ \mathrm{for}\ j=0,\ldots, p-1.\]
	\end{proposition}
	\begin{proof}
		Note that if $f(t) = \int_0^tg(s) \, ds$ for a \cadlag function $g$ (the integral is well-defined in the Lebesgue sense since $g$ has countably many discontinuities) then $f$ is right differentiable and at any $t\in[0,\infty)$ we have $Df(t) = g(t)$. This can be seen as follows: fix $\epsilon>0$, and choose $\delta>0$ such that if $s\in[t, t+\delta)$ then $|g(t)-g(s)|<\epsilon$. Then for any $0<\delta'\leq \delta$,
		\[ \left|\frac{f(t+\delta') - f(t)}{\delta'} - g(t)\right| =  \left|\frac{1}{\delta'} \int_t^{t+\delta'} g(s)\, ds- g(t)\right| \leq  \frac{1}{\delta'} \int_t^{t+\delta'}\left| g(s)- g(t)\right| \, ds \leq \epsilon.\]
		Thus, we have that, by definition of the MCAR process $\mathbb{Y}$,
		\[\mathbf{X}^{(1)}_t = \mathbf{Y}_t, \ t\geq 0,\]
		and, inductively, for $j\geq 1$
		\[D^{j-1}\mathbf{Y}_t = \mathbf{X}^{(j)}_t = \int_0^t \mathbf{X}^{(j+1)}_s \, ds, \ t\geq 0 \implies \mathbf{X}^{(j+1)}_t = D^j\mathbf{Y}_t, \ t\geq 0. \]
		In particular, we have that $D^{p-1}\mathbb{Y}$ can have (infinitely many) jumps, since it jumps every time $\mathbb{L}$ does and by the same magnitude, while $D^j\mathbb{Y}$ is almost surely continuous for all $j<p-1$.
	\end{proof}
	
	\begin{figure}
		\centering            \includegraphics[width=\textwidth]{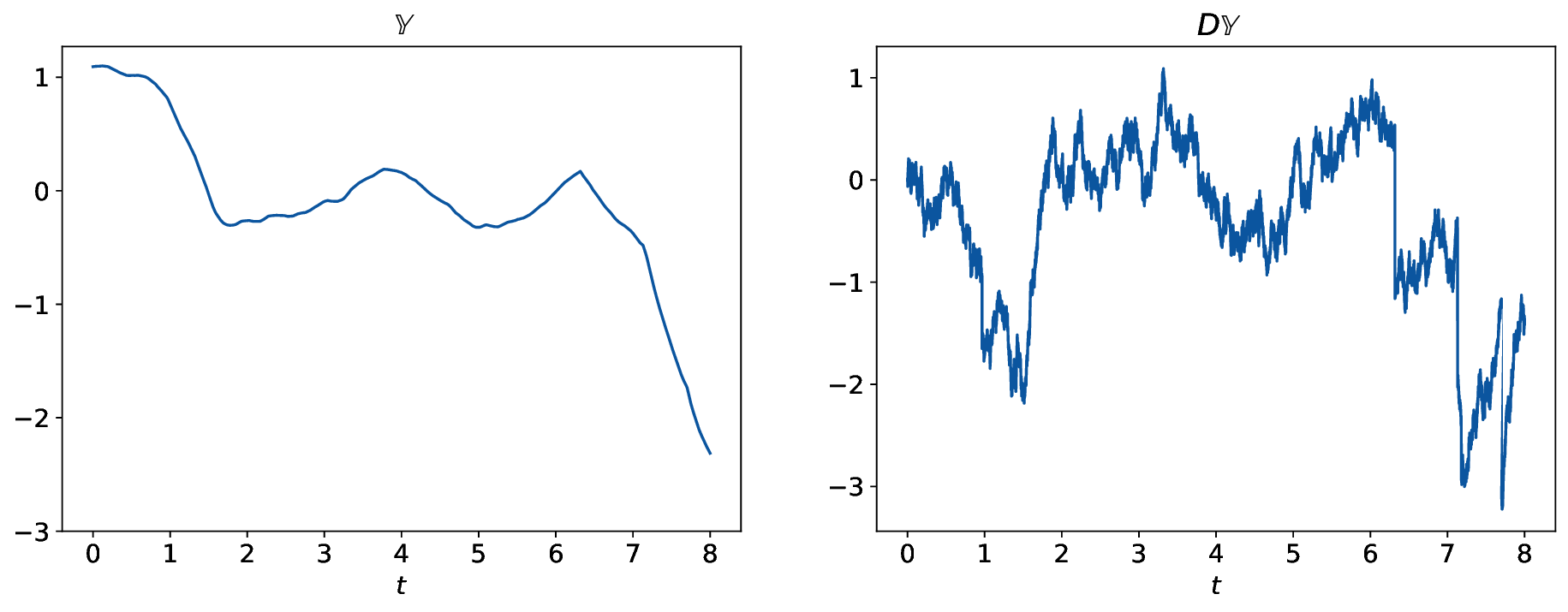}
		\caption{State-space path realization $\mathbb{X} = (\mathbb{Y}, D\mathbb{Y})^\textrm{T}$ of a $d=1$ dimensional MCAR process of order $p=2$ with drift parameter $\mathbf{A} = (A_1, A_2) = (1, 2)$ and \Levy triplet $(b=0, \Sigma=1, F(dx) = \phi(x)dx)$ where $\phi(x)$ is a $N(0, 1)$ probability density. The path is simulated over a discrete time grid by using the exact scheme described in Appendix \ref{app:simulate_MCAR_exact}.}
		\label{fig:MCAR_realization}
	\end{figure}
	
	Throughout this paper, we make the following assumptions on the driving \Levy process.
	\begin{assumption} \label{ass:levy_process}
		The driving \Levy process $\mathbb{L}=\{\mathbf{L}_t,\ t\geq0\}$ has characteristics $(\mathbf{b}, \Sigma, F)$ and:
		\begin{itemize}
			\item is square-integrable, i.e.\ $\mathbb{E}[\|\mathbf{L}_1\|^2]<\infty;$
			\item has a non-degenerate Gaussian part, i.e.\ $\Sigma$ is strictly positive definite.
		\end{itemize}
	\end{assumption}
	\begin{remark}
		Assuming $\Sigma$ is strictly positive definite, we denote by $\Sigma^{1/2}$ its unique symmetric positive definite square root. This assumption is necessary to obtain equivalence between measures on path space and develop a likelihood theory. Note our results will thus not be applicable to non-decreasing driving \Levy processes, i.e.\ \Levy subordinators, which are considered instead in \citet{brockwell_CARMA_estimation}. 
	\end{remark}

	\section{Continuous observations: explicit maximum likelihood} \label{sec:MCAR_cont}
	In order to develop a likelihood theory for MCAR processes we first need to understand the probability measures these processes induce on the path space and how such measures are related to each other via a likelihood function. Given an observation of an MCAR process $\mathbb{Y}_{[0,t]} := \{\mathbf{Y}_s, \ s\in[0,t]\}$ one can then estimate the drift coefficients $\mathbf{A}^*\in (\mathcal{M}(\mathbb{R}^d))^p$ by the set of parameters $\hat{\mathbf{A}}(\mathbb{Y}_{[0,t]}) \in(\mathcal{M}(\mathbb{R}^d))^p$ corresponding to the measure which makes the observation ``most likely'', i.e.\ which maximize the likelihood. 
	
	Fixing an initial condition for the state-space process $\mathbf{x}_0\in\mathbb{R}^{pd}$ and a set of coefficient matrices $\mathbf{A}\in(\mathcal{M}(\mathbb{R}^d))^p$, by applying the definition of an MCAR process and Proposition \ref{prop:right_diff} we note there is a one-to-one relationship between the MCAR probability measure $\mu_{\mathbf{A}, \mathbf{x}_0}$ on the space of $\mathbb{R}^d$-valued paths and the state-space probability measure $\mathbb{P}_{\mathbf{A}, \mathbf{x}_0}$ on the space of \cadlag $\mathbb{R}^{pd}$-valued paths. Let $\mathfrak{C}^{p-1}([0, \infty);\mathbb{R}^d)$ denote the set of $(p-1)$-times right differentiable \cadlag functions so that the support of $\mu_{\mathbf{A}, \mathbf{x}_0}$ is contained in $\mathfrak{C}^{p-1}([0, \infty);\mathbb{R}^d) \subseteq D([0, \infty);\mathbb{R}^d)$. Then
	\[ \mathbb{P}_{\mathbf{A}, \mathbf{x}_0} = f_\star \mu_{\mathbf{A}, \mathbf{x}_0}\ \mathrm{and}\ \mu_{\mathbf{A}, \mathbf{x}_0} = g_\star \mathbb{P}_{\mathbf{A}, \mathbf{x}_0}, \]
	where $f:\mathfrak{C}^{p-1}([0, \infty);\mathbb{R}^d)\rightarrow D([0, \infty);\mathbb{R}^{pd})$ maps
	\[ \{\mathbf{y}_t,\ t\geq 0\} \mapsto \{(\mathbf{y}^{\mathrm{T}}_t, D\mathbf{y}^{\mathrm{T}}_t,\ldots, D^{p-1}\mathbf{y}^{\mathrm{T}}_t)^{\mathrm{T}},\ t\geq 0\},\]
	and $g:D([0, \infty);\mathbb{R}^{pd})\rightarrow D([0, \infty);\mathbb{R}^d)$ maps 
	\[\{(\mathbf{x}^{(1),\mathrm{T}}_t, \mathbf{x}^{(2),\mathrm{T}}_t,\ldots, \mathbf{x}^{(p),\mathrm{T}}_t)^{\mathrm{T}},\ t\geq 0\} \mapsto \{\mathbf{x}^{1}_t,\ t\geq 0\}.\]
	In view of this correspondence, we can work directly with the probability measures of the state-space representation of the MCAR process, i.e.\ $\mathbb{P}_{\mathbf{A}, \mathbf{x}_0}$. These turn out to be easier to analyze since they are solution-measures of the $\mathbb{R}^{pd}$-valued \Levydriven OU SDE \eqref{eqn:SDE_def}. The main issues in the analysis arise from the degeneracy of the driving \Levy noise $\mathcal{E}\mathbb{L} = \{\mathcal{E}\mathbf{L}_t, \ t\geq 0\}$.
	
	\subsection{Deriving the likelihood}   
	
	\sloppy For a parameter $\mathbf{A} \in (\mathcal{M}(\mathbb{R}^{d}))^p$ consider the $pd$-dimensional state-space representation $\mathbb{X}_{\mathbf{A},\mathbf{x}_0} = \{\mathbf{X}_{\mathbf{A},\mathbf{x}_0,t},\ t\geq 0\}$ of a MCAR($p$) process driven by a \Levy process $\mathbb{L}$ with characteristics $(\mathbf{b}, \Sigma, F)$, parameter $\mathbf{A}$ and deterministic initial condition $\xi = \mathbf{x}_0$ on an arbitrary stochastic basis $(\Omega', \mathcal{F}', \{\mathcal{F}'_t, \ t\geq 0\}, \mathbb{P}')$. This satisfies the stochastic differential Equation \eqref{eqn:SDE_def} which, using the \LevyIto representation of $\mathbb{L}$, we can write as
	\begin{equation} \label{eqn:SDE}
		d \mathbf{X}_t = (\mathcal{A}_{\mathbf{A}} \mathbf{X}_t + \mathcal{E} b) dt + \mathcal{E} \Sigma^{1/2} d \mathbf{W}_t + \int_{\|z\| \leq 1} \mathcal{E} z \{\mu(dt, dz) - F(dz) d t\} + \int_{\|z\| > 1} \mathcal{E} z \mu(d t, dz),
	\end{equation}
	for a $d$-dimensional standard Wiener process $\mathbb{W} = \{\mathbf{W}_t,\ t\geq 0\}$ and a Poisson random measure $\mu(dt, dz)$ on $\mathbb{R}_+ \times \mathbb{R}^d\setminus \{0\}$ with compensator $F(dz) dt$. It is easy to check that the coefficients of the SDE satisfy local Lipschitz and linear growth conditions. We can thus combine \citet[Theorem~II.2.32~and~Theorem~II.2.33]{Jacod_Shiryaev_1987} to deduce that $\forall \mathbf{A} \in (\mathcal{M}_d(\mathbb{R}))^p$ there exists a unique solution-measure of the SDE \eqref{eqn:SDE} on the canonical state-space $(\Omega = D([0,\infty), \mathbb{R}^{pd}), \mathcal{F}, \{\mathcal{F}_t, \ t\geq 0\})$ for any initial condition $\mathbf{x}_0 \in \mathbb{R}^{pd}$. We denote this measure by $\mathbb{P}_{\mathbf{A}, \mathbf{x}_0}$. Under $\mathbb{P}_{\mathbf{A}, \mathbf{x}_0}$ the canonical process $\mathbb{X} = \{\mathbf{X}_t,\ t\geq 0\}$ defined by $\mathbf{X}_t(\omega) = \omega(t), \ \forall t\geq 0$ is a solution of SDE \eqref{eqn:SDE}. 
	
	\begin{remark}
		We note that $c_t :=  \mathcal{E} \Sigma \mathcal{E}^\mathrm{T}$ is not strictly positive definite and thus \citet[Condition~C]{Sorensen_1991} (which is the same as the hypothesis of \citet[Theorem~ II.2.34]{Jacod_Shiryaev_1987}) does not hold. This condition is used by \citet{Sorensen_1991} to show uniqueness of the solution-measure to SDEs for more general diffusions with jumps. 
	\end{remark}
	
	In the analysis of measures induced by general semimartingales, such as $\mathbb{P}_{\mathbf{A}, \mathbf{x}_0}$, the characteristics \citep[Definition~II.2.6]{Jacod_Shiryaev_1987} provide a valuable tool. It is a standard result \citep[Theorem~III.2.26]{Jacod_Shiryaev_1987} that the set of all solution-measures to SDE \eqref{eqn:SDE} with initial condition $\mathbf{X}_0 = \mathbf{x}_0$ is the set of solutions to the martingale problem on the canonical space $(\Omega, \mathcal{F}, \{\mathcal{F}_t, \ t\geq 0\})$ with deterministic initial condition $\eta = \delta_{\mathbf{x}_0}$ and semimartingale characteristics $(\mathbf{B}_{\mathbf{A}}, C, \nu)$ where for $\omega\in\Omega$, $t\geq 0$
	\begin{equation} \label{eqn:loc_char_P}
		\begin{gathered}
			\mathbf{B}_{\mathbf{A}, t}(\omega) = \int_0^t (\mathcal{A}_{\mathbf{A}} \omega(s) + \mathcal{E} b) ds = \mathcal{A}_{\mathbf{A}}\int_0^t \omega(s) ds + \mathcal{E} bt,
			\\
			C_t(\omega) = \int_0^t  \mathcal{E} \Sigma \mathcal{E}^\mathrm{T} ds = \mathcal{E} \Sigma \mathcal{E}^\mathrm{T} t, \quad
			\nu(\omega; dt, dx) = dt \int_{\mathbb{R}^d} \mathds{1}_{\{dx\}}(\mathcal{E} z) F(dz) = dt\, \delta_0 (dx_{-p})\, F(dx_{p}),
		\end{gathered}
	\end{equation}
	where we write $x = (x^\mathrm{T}_1, \ldots, x^\mathrm{T}_p)^\mathrm{T}\in\mathbb{R}^{pd}$ and $x_{-p} =(x_{1}^\mathrm{T}, \ldots, x^\mathrm{T}_{p-1})^\mathrm{T}$, for $x_1, \ldots, x_p \in \mathbb{R}^d$. Under $\mathbb{P}_{\mathbf{A}, \mathbf{x}_0}$ the canonical process $\mathbb{X}$ is a diffusion with jumps with local characteristics $(\mathbf{B}_\mathbf{A}, C, \nu)$. 
	
	\begin{remark}
		See \citet[Definition~II.2.4~and~Definition~II.2.24]{Jacod_Shiryaev_1987} for the definitions of solutions to martingale problems and solution-measures. Two remarks:
		\begin{itemize}
			\item Note that a solution-measure for an SDE as \eqref{eqn:SDE} defined on a general stochastic basis $(\Omega', \mathcal{F}', \{\mathcal{F}'_t, \ t\geq 0\}, \mathbb{P}')$ is a probability measure on the canonical space $(\Omega, \mathcal{F}, \{\mathcal{F}_t, \ t\geq 0\})$. As we vary $\mathbf{A} \in (\mathcal{M}_d(\mathbb{R}))^p$ in \eqref{eqn:SDE} the original stochastic basis $(\Omega', \mathcal{F}', \{\mathcal{F}'_t, \ t\geq 0\}, \mathbb{P}')$ does not change but the solution measure $\mathbb{P}_{\mathbf{A}, \mathbf{x}_0}$ on $(\Omega, \mathcal{F}, \{\mathcal{F}_t, \ t\geq 0\})$ does indeed.
			\item Note that solution-measures and solutions to martingale problems depend on the initial condition, i.e.\ the initial distribution $\eta(\cdot) = \mathbb{P}'(\mathbf{X}_0\in \cdot)$. In the case where this is non-deterministic the canonical space needs to be enlarged so that $\mathcal{F}_0$ contains the $\sigma$-algebra generated by $\mathbf{X}_0$. In the following though, we will only need to consider deterministic initial conditions $\mathbf{x}_0\in\mathbb{R}^{pd}$.
		\end{itemize} 
	\end{remark}
	Under $\mathbb{P}_{\mathbf{A}, \mathbf{x}_0}$, $\mathbb{X}^c_{\mathbf{A}, \mathbf{x}_0} = \{\mathbf{X}^c_{\mathbf{A},\mathbf{x}_0,t},\ t\geq 0\}$ defined by
	\begin{equation} \label{eqn:cont_mart_part}
		\mathbf{X}_{\mathbf{A}, \mathbf{x}_0, t}^c = \mathbf{X}_t-\mathbf{x}_0- \sum_{s\leq t} \Delta \mathbf{X}_s \mathds{1}_{\{\|\Delta \mathbf{X}_s\|\geq 1\}} - \mathbf{B}_{\mathbf{A},t}
		- \int_0^t\int_{\|x\|\leq 1} x(\mu_\mathbb{X} - \nu)(dt, dx), \ t\geq0,
	\end{equation}
	is a continuous local martingale, where $\mu$ is the jump measure over $\mathbb{R}_{+} \times \mathbb{R}^{pd}\setminus \{0\}$ associated with the canonical process $\mathbb{X}$, i.e.\
	\begin{equation} \label{eqn:jump_meas}
		\mu_\mathbb{X}(\omega; dt, dx) = \sum_{s\geq 0:\Delta \mathbf{X}_s(\omega) \neq 0} \delta_{(s, \Delta \mathbf{X}_s(\omega))}(dt, dx), 
	\end{equation}
	and $\Delta \mathbb{X} = \{\Delta \mathbf{X}_t,\ t\geq 0\}$ is the jump process of $\mathbb{X}$, i.e.\ $\Delta \mathbf{X}_t =  \mathbf{X}_t - \mathbf{X}_{t-}$. Note that by construction of the state-space representation of the MCAR($p$) process, jumps will only occur in the last $d$ entries of $\mathbb{X}$. Let \[\mathbf{A}^{(0)} = (0_{d\times d}, \ldots, 0_{d\times d}) \in (\mathcal{M}_d(\mathbb{R}))^p,\] and note that the jump measure $\nu$ does not depend on $\mathbf{A} \in (\mathcal{M}_d(\mathbb{R}))^p$, i.e.\ $Y_{\mathbf{A}}(t, x, y) = 1$ in \citet{Sorensen_1991} and \citet{Jacod_Memin_1976}. 
	\begin{proposition} \label{prop:exist_likelihood}
		Assume Assumption \ref{ass:levy_process} holds. Fix $t>0$ and let $\mathbb{P}^t_{\mathbf{A}, \mathbf{x}_0}$ and $\mathbb{P}^t_{\mathbf{A}^{(0)}, \mathbf{x}_0}$ denote the restrictions of $\mathbb{P}_{\mathbf{A}, \mathbf{x}_0}$ and $\mathbb{P}_{\mathbf{A}^{(0)}, \mathbf{x}_0}$ to $\mathcal{F}_t$ respectively. Then $\forall \mathbf{A} \in (\mathcal{M}_d(\mathbb{R}))^p$
		\[\mathbb{P}^t_{\mathbf{A}, \mathbf{x}_0} \sim \mathbb{P}^t_{\mathbf{A}^{(0)}, \mathbf{x}_0}, \]
		and $\mathbb{P}^t_{\mathbf{A}^{(0)}, \mathbf{x}_0}$-almost surely
		\begin{equation} \label{eqn:likelihood_MCAR}
			\frac{d \mathbb{P}^t_{\mathbf{A}, \mathbf{x}_0}}{d\mathbb{P}^t_{\mathbf{A}^{(0)}, \mathbf{x}_0}}  = \mathbb{E}^{\mathbb{P}_{\mathbf{A}^{(0)}, \mathbf{x}_0}}\left[ \frac{d \mathbb{P}_{\mathbf{A}, \mathbf{x}_0}}{d\mathbb{P}_{\mathbf{A}^{(0)}, \mathbf{x}_0}} \Big| {\mathcal{F}_t}\right] = \exp\Big\{ - \int_0^t \mathbf{X}^\mathrm{T}_{s-} \tilde{\mathcal{A}}_\mathbf{A}^\mathrm{T} \tilde{\Sigma}^{-1}\ d \mathbf{X}_{\mathbf{A}^{(0)}, s}^c - \frac{1}{2} \int_0^t \mathbf{X}^\mathrm{T}_{s}  \tilde{\mathcal{A}}_\mathbf{A}^\mathrm{T} \tilde{\Sigma}^{-1} \tilde{\mathcal{A}}_\mathbf{A} \mathbf{X}_s ds  \Big\}, 
		\end{equation}
		where 
		\begin{equation} \label{eqn:A_tilde} \tilde{\mathcal{A}}_\mathbf{A} = 
			\begin{pmatrix}
				0_{d \times d} & 0_{d \times d} & \cdots & 0_{d \times d} \\
				\vdots & \vdots & \ddots & \vdots \\
				0_{d \times d} & 0_{d \times d} & \cdots & 0_{d \times d} \\
				- A_p & -A_{p-1} & \cdots & -A_1 \\
			\end{pmatrix},
		\end{equation}
		and
		\begin{equation} \label{eqn:Sigma_tilde}
			\tilde{\Sigma}^{-1} := \mathcal{E} \Sigma^{-1} \mathcal{E}^\mathrm{T} = 
			\begin{pmatrix}
				0_{d \times d} & \cdots & 0_{d \times d} & 0_{d \times d} \\
				\vdots & \ddots & \vdots & \vdots \\
				0_{d \times d} & \cdots & 0_{d \times d} & 0_{d \times d} \\
				0_{d \times d} &  \cdots & 0_{d \times d} & \Sigma^{-1} \\
			\end{pmatrix}.
		\end{equation}
	\end{proposition}
	\begin{proof}
		We follow a similar proof strategy as that in \citet[Theorem~2.1]{Sorensen_1991}, relying on the results of \citet{Jacod_Memin_1976} and \citet{Jacod_Shiryaev_1987}, but we note that our process does not satisfy \citet[Condition~C]{Sorensen_1991} since $c_t = \tilde{\Sigma} = \mathcal{E} \Sigma \mathcal{E}^\mathrm{T}$ is not strictly positive definite. The outline of the proof is as follows:
		\begin{itemize}
			\item We start by considering the state-space representation of an MCAR process $\mathbb{X}_{\mathbf{A},\mathbf{x}_0}$ satisfying \eqref{eqn:SDE} on an arbitrary stochastic basis $(\Omega', \mathcal{F}', \{\mathcal{F}'_t,\ t\geq 0\}, \mathbb{P}')$ and the associated unique solution-measure $\mathbb{P}_{\mathbf{A}, \mathbf{x}_0}$ on the canonical state-space $(\Omega, \mathcal{F}, \{\mathcal{F}_t, \ t\geq 0\})$, i.e.\ the unique solution to the martingale problem $(\mathbb{X}, (\mathbf{B}_{\mathbf{A}}, C, \nu), \delta_{\mathbf{x}_0})$.
			\item Next, fixing $t>0$, we construct the unique solution $\mathbb{Q}_{\mathbf{A}, \mathbf{x}_0}$ to the martingale problem with stopped characteristics $(\mathbb{X}, (\mathbf{B}^t_{\mathbf{A}}, C^t, \nu^t), \delta_{\mathbf{x}_0})$, i.e.\ the solution measure of SDE \eqref{eqn:SDE} stopped at $t>0$. By construction, $\mathbb{P}^t_{\mathbf{A}, \mathbf{x}_0} = \mathbb{Q}^t_{\mathbf{A}, \mathbf{x}_0}$.
			\item Fixing a reference measure corresponding to the parameter $\mathbf{A}^{(0)}= (0_{d\times d}, \ldots, 0_{d\times_d})$, we then apply \citet[Theorem~4.2.(c)]{Jacod_Memin_1976} to $\mathbb{Q}_{\mathbf{A}, \mathbf{x}_0} \in (\mathbb{X}, (\mathbf{B}^t_{\mathbf{A}}, C^t, \nu^t), \delta_{\mathbf{x}_0})$ and $\mathbb{Q}_{\mathbf{A}^{(0)}, \mathbf{x}_0}\in (\mathbb{X}, (\mathbf{B}^t_{\mathbf{A}^{(0)}}, C^t, \nu^t), \delta_{\mathbf{x}_0})$ to deduce equivalence of the measures. This step crucially relies on the specific structure of the problem as we will be able to write
			\begin{equation} \label{eqn:B_change}
				\mathbf{B}^t_{\mathbf{A}, s} = \mathbf{B}^t_{\mathbf{A}^{(0)}, s} + \int_0^s dC^t_u \mathbf{Z}_{\mathbf{A}, u},
			\end{equation}
			where $\mathbb{Z}_{\mathbf{A}} = \{\mathbf{Z}_{\mathbf{A},s},\ s\geq0\}$ is the predictable process defined by $\mathbf{Z}_{\mathbf{A},s} = \tilde{\Sigma}^{-1}\tilde{\mathcal{A}}_{\mathbf{A}}\mathbf{X}_{s-}$. To apply the theorem we check the assumptions:
			\begin{enumerate}
				\item ``$\tau_n$-uniqueness of the martingale problem $(\mathbb{X}, (\mathbf{B}_\mathbf{A}^t, C^t, \nu^t), \delta_{\mathbf{x}_0})$'': We prove this by showing the stronger local uniqueness property. In \citet{Sorensen_1991}, under the assumption that $c_t$ is strictly positive definite, local uniqueness is deduced by a simple application of \citet[Corollary~III.2.41]{Jacod_Shiryaev_1987}. In our setting though this property does not hold and we need to work a bit harder to prove local uniqueness. We check directly the conditions of the more general result \citet[Theorem~III.2.40]{Jacod_Shiryaev_1987} for the martingale problem $(\mathbb{X}, (\mathbf{B}_\mathbf{A}^t, C^t, \nu^t), \delta_{\mathbf{x}_0})$:
				\begin{enumerate}
					\item \citet[Assumption~III.2.13]{Jacod_Shiryaev_1987} is satisfied because we are working on the canonical space.
					\item \citet[Assumption~III.2.39]{Jacod_Shiryaev_1987} is related to the Markovian structure of our martingale problem.
					\item The solution to the martingale problem defines a transition kernel as the characteristics are shifted in time and the initial condition is shifted in space: this is obtained via an application of the $\pi-\lambda$ theorem.
					\item Uniqueness of the original martingale problem $(\mathbb{X}, (\mathbf{B}_\mathbf{A}^t, C^t, \nu^t), \delta_{\mathbf{x}_0})$ has already been checked.
				\end{enumerate}
				\item Equivalence of initial conditions: We work with fixed initial condition $\delta_{\mathbf{x}_0}$.
				\item The random variable $A^t_{\mathbf{A}, \infty}$, given by the limit as $s\uparrow\infty$ of
				\begin{equation} \label{eqn:A_t}
					A^t_{\mathbf{A},s} = \int_0^s \mathbf{Z}^{\mathrm{T}}_{\mathbf{A},u} dC_u^t\mathbf{Z}_{\mathbf{A},u} = \int_0^{s\wedge t} \mathbf{X}^{\mathrm{T}}_{u} \tilde{\mathcal{A}}^\mathrm{T}_\mathbf{A}\tilde{\Sigma}^{-1}\tilde{\mathcal{A}}_\mathbf{A}\mathbf{X}_{u} du,
				\end{equation}
				is $\mathbb{Q}_{\mathbf{A}, \mathbf{x}_0}$ and $\mathbb{Q}_{\mathbf{A}^{(0)}, \mathbf{x}_0}$ a.s.\ finite: This follows from bounds obtained from the assumption of square integrability of the driving \Levy process and explicit formulas for the second moments of $\mathbb{X}_{\mathbf{A}}$.
			\end{enumerate}
			\item Finally, we apply \citet[Theorem~4.5.(b)]{Jacod_Memin_1976} to deduce the form of the conditional Radon–Nikodym derivative of $\mathbb{Q}_{\mathbf{A}, \mathbf{x}_0}$ w.r.t.\ $\mathbb{Q}_{\mathbf{A}^{(0)}, \mathbf{x}_0}$. We check the assumptions:
			\begin{enumerate}
				\item $\mathbb{Q}_{\mathbf{A}, \mathbf{x}_0}$ satisfies the martingale representation condition as it is the unique solution of the martingale problem $(\mathbb{X}, (\mathbf{B}_\mathbf{A}^t, C^t, \nu^t), \delta_{\mathbf{x}_0})$, this follows from \citet[Theorem~4.4]{Jacod_Memin_1976}.
				\item $\mathbb{A}^t_{\mathbf{A}} = \{A^t_{\mathbf{A}, s},\ s\geq 0\}$ is $\mathbb{Q}_{\mathbf{A}^{(0)}, \mathbf{x}_0}$ a.s.\ continuous and its limit is finite: this simplifies the form of the conditional Radon–Nikodym process as we can drop any dependency on the localizing sequence $(\tau_n)_n$.
			\end{enumerate}
		\end{itemize}		
		For full details, see Appendix \ref{app:proof_prop_exist_likelihood}.
	\end{proof}
	
	\begin{remark}
		We make some important remarks on the proof:
		\begin{itemize}
			\item The proof crucially relies on being able to write Equation \eqref{eqn:B_change}, for which we had to leverage the specific forms of $\mathcal{A}_\mathbf{A}$ and $\tilde{\Sigma}$. Our proof, therefore, does not immediately generalize to OU processes with arbitrary matrix coefficient $\mathcal{A}$ and non-negative definite Gaussian covariance $\tilde{\Sigma}$ of the driving \Levy process. In the particular case where $\tilde{\Sigma}$ is strictly positive definite, though, it is a standard result that the same result holds for any matrix coefficient $\mathcal{A}$.
			\item In the proof it is essential to stop the SDE \eqref{eqn:SDE} at $t>0$ since, if instead we attempt to apply the theorems by \citet{Jacod_Memin_1976} directly to $\mathbb{P}_{\mathbf{A}, \mathbf{x}_0}$ and $\mathbb{P}_{\mathbf{A}^{(0)}, \mathbf{x}_0}$, then we cannot guarantee that $A_{\mathbf{A},\infty}<\infty$, $\mathbb{P}_{\mathbf{A}, \mathbf{x}_0}$- and $\mathbb{P}_{\mathbf{A}^{(0)}, \mathbf{x}_0}$-almost surely, where for $t\geq0$
			\[A_{\mathbf{A}, s} := \int_0^s  \mathbf{Z}^\mathrm{T}_{\mathbf{A},u} dC_u \mathbf{Z}_{\mathbf{A}, u} = \int_0^s \mathbf{X}_u ^\mathrm{T} \tilde{\mathcal{A}}_\mathbf{A}^\mathrm{T} \tilde{\Sigma}^{-1} \tilde{\mathcal{A}}_\mathbf{A} \mathbf{X}_u du.\]
		\end{itemize}
	\end{remark}
	
	%	\begin{remark}
		%	Note that for any $\mathbb{P}'$-a.s. \cadlag process $\mathbb{N} = \{\mathbf{N}_t,\ t\geq 0\}$
		%	\[ \int_0^t \mathbf{N}_{s}(\omega) ds = \int_0^t \mathbf{N}_{s-}(\omega) ds = \int_0^t \mathbf{N}_{s+}(\omega) ds,\quad \mathrm{$\mathbb{P}'$-a.s.}, \]
		%	since the process has $\mathbb{P}'$-a.s.\ countably many jumps and thus the integrands coincide on a set of Lebesgue measure one $\mathbb{P}'$-almost surely. And, similarly, for any \cadlag path $\omega\in\Omega=D([0,\infty))$
		%	\[ \int_0^t \omega(s) ds = \int_0^t \omega(s-) ds = \int_0^t \omega(s+) ds.\]
		%\end{remark}
		
		\subsection{Maximum likelihood estimation}
		
		Note that Proposition \ref{prop:exist_likelihood} yields the likelihood of parameter $\mathbf{A}$ with respect to reference parameter $\mathbf{A}^{(0)}$ as the Radon-Nikodym derivative of two measures on the canonical space of $\mathbb{R}^{pd}$-valued paths $(\Omega, \mathcal{F}, \{\mathcal{F}_t,\ t\geq 0\})$ restricted to $\mathcal{F}_t$, that is an $\mathcal{F}_t$-measurable random variable. This is thus a mapping from paths in $D([0,t];\mathbb{R}^{pd})$ to $\mathbb{R}$, i.e.\ a function of $\mathbb{X}_{[0,t]}$. 
		%       Due to the specific construction of the measures as solution measures to SDE \eqref{eqn:SDE}, we note the following relationship holds $\mathbb{P}_{\mathbf{A}^{(0)}, \mathbf{x}_0}$ (and $\mathbb{P}_{\mathbf{A}, \mathbf{x}_0}$) almost surely between the components of the canonical process $\mathbb{X} := (\mathbb{X}^{(1), \mathrm{T}}, \ldots, \mathbb{X}^{(p), \mathrm{T}})^\mathrm{T}$:
		% \[\mathbf{X}^{(j)}_s = D^{j-1}\mathbf{X}^{(1)}_s, \quad j\in\{1,\ldots,p\},\quad s\in[0, t],\]
		% with the following condition necessarily satisfied
		% \[D^{j-1}\mathbf{X}^{(1)}_0 = \mathbf{x}_0^{(j)}, \quad j\in\{1,\ldots,p\}.\]
		% Thus, any function of 
		% \[\mathbb{X}_{[0,t]} = \{\mathbf{X}_s,\ s\in[0,t]\} = \{\omega(s),\ s\in[0,t]\}\]
		% is completely determined by $\mathbb{X}^{(1)}_{[0,t]} = \{\mathbf{X}^{(1)}_s,\ s\in[0,t]\}$. 
		Setting $\mathbb{Y}_{[0,t]} = \{\mathbf{Y}_s,\ s\in [0,t]\} = \mathbb{X}^{(1)}_{[0,t]}$, understood as an observation of the MCAR process, in view of Proposition \ref{prop:right_diff} we can write the likelihood \eqref{eqn:likelihood_MCAR} as a function of $\mathbb{Y}_{[0,t]}$ (and its right derivatives). This is formalized by the following lemma in which we define the vectorization operator for elements $\mathbf{V}\in(\mathcal{M}(\mathbb{R}^d))^p$ by $\mathrm{vec}: (\mathcal{M}(\mathbb{R}^d))^p \mapsto \mathbb{R}^{pd^2}$ as
		\[\mathbf{V} = (V_1, \ldots, V_p) \mapsto \mathrm{vec}(\mathbf{V}) := (V^{(1,1)}_1, \ldots, V^{(d, 1)}_1, \ldots, V^{(1, d)}_1, \ldots, V^{(d,d)}_1, V^{(1,1)}_2, \ldots, V^{(d,d)}_p)^\mathrm{T}.\]
		\begin{lemma} \label{lemma:likelihood_MCAR}
			Assume Assumption \ref{ass:levy_process} holds. Let $\mathbb{Y}_{[0,t]} = \{\mathbf{Y}_s, \ s\in[0,t]\}$ be the observation of a MCAR($p$) process. Suppose the process starts in $\mathbf{y}_0 \in\mathbb{R}^d$ with initial derivatives $D^1\mathbf{y}_0, \ldots, D^{p-1}\mathbf{y}_0\in\mathbb{R}^d$, i.e.\ $\mathbf{x}_0 = (\mathbf{y}_0^\mathrm{T}, \ldots, D^{p-1}\mathbf{y}_0^\mathrm{T})^\mathrm{T}$. Let 
			\[D^{p-1}\mathbb{Y}^c_{\mathbf{A}^{(0)}, [0,t]}= \{D^{p-1}\mathbf{Y}^c_{\mathbf{A}^{(0)}, s} \ s\in[0,t]\}\]
			be the last $d$ entries of $\mathbb{X}^{c}_{\mathbf{A}^{(0)},\mathbf{x}_0}$, the continuous martingale part of the state-space representation with parameter $\mathbf{A}^{(0)} = (0_{d\times d}, \ldots, 0_{d\times d})$ as defined in \eqref{eqn:cont_mart_part} -- note we drop dependence on $\mathbf{x}_0$ as $D^{p-1}\mathbb{Y}^c_{\mathbf{A}^{(0)}, [0,t]}$ will only appear as an integrator and thus the initial condition becomes superfluous. Then the likelihood function of the MCAR($p$) parameters  $\mathbf{A} \in (\mathcal{M}(\mathbb{R}^{d}))^p$ given the observation $\mathbb{Y}_{[0,t]}$ is 
			\begin{equation} \label{eqn:likelihood}
				\mathcal{L}(\mathbf{A}; \mathbb{Y}_{[0,t]}) := \mathbb{E}^{\mathbb{P}_{\mathbf{A}^{(0)}, \mathbf{x}_0}}\left[ \frac{d \mathbb{P}_{\mathbf{A}, \mathbf{x}_0}}{d\mathbb{P}_{\mathbf{A}^{(0)}, \mathbf{x}_0}} \Big| {\mathcal{F}_t}\right] = \exp\left\{\mathrm{vec} (\mathbf{A})^\mathrm{T} \mathbf{H}_t - \frac{1}{2} \mathrm{vec}(\mathbf{A})^\mathrm{T} [\mathbf{H}]_t \mathrm{vec}(\mathbf{A})\right\},
			\end{equation}
			where
			\begin{equation} \label{eqn:H}
				\mathbf{H}_t = -
				\begin{pmatrix}
					\int_0^t D^{p-1}Y^{(1)}_{s-} \Sigma^{-1} dD^{p-1}\mathbf{Y}^c_{\mathbf{A}^{(0)}, s} \\
					\vdots \\
					\int_0^t  D^{p-1}Y^{(d)}_{s-} \Sigma^{-1} dD^{p-1}\mathbf{Y}^c_{\mathbf{A}^{(0)}, s} \\
					\vdots \\
					\int_0^t Y^{(1)}_{s-} \Sigma^{-1} dD^{p-1}\mathbf{Y}^c_{\mathbf{A}^{(0)}, s} \\
					\vdots \\
					\int_0^t Y^{(d)}_{s-} \Sigma^{-1} dD^{p-1}\mathbf{Y}^c_{\mathbf{A}^{(0)}, s}
					\\
				\end{pmatrix},
			\end{equation}
			with quadratic covariation matrix
			\begin{equation} \label{eqn:[H]}
				[\mathbf{H}]_t =
				\begin{pmatrix}
					\int_0^t D^{p-1}Y^{(1)}_{s-} \Sigma^{-1} D^{p-1}Y^{(1)}_{s-}\ ds & \cdots & \int_0^t D^{p-1}Y^{(1)}_{s-} \Sigma^{-1} Y^{(d)}_{s-}\ ds \\
					\vdots & \ddots & \vdots \\
					\int_0^t Y^{(d)}_{s-} \Sigma^{-1} D^{p-1}Y^{(1)}_{s-}\ ds & \cdots & \int_0^t Y^{(d)}_{s-} \Sigma^{-1} Y^{(d)}_{s-}\ ds \\
				\end{pmatrix}.
			\end{equation}
		\end{lemma}
		\begin{proof}
			The form of the likelihood follows by expanding \eqref{eqn:likelihood_MCAR} and grouping in terms of $A^{(i, j)}_k$ for $k=1, \ldots, p$ and $i,j=1,\ldots,d$. To determine that $[\mathbf{H}]_t$ is indeed the quadratic covariation of $\mathbb{H}=\{\mathbf{H}_t,\ t\geq0\}$ we use the fact that $[D^{p-1}\mathbf{Y}^c_{\mathbf{A}^{(0)}}]_t = \Sigma t$ and the multidimensional version of \citet[Theorem~II.29]{Protter_1990}, i.e.\ if $\mathbb{X} = \{\mathbf{X}_t,\ t\geq0\}, \mathbb{Z} = \{\mathbf{Z}_t,\ t\geq0\}$ are semimartingales and $\{\mathbf{G}_t,\ t\geq 0\},\{\mathbf{K}_t,\ t\geq 0\}$ are \caglad adapted processes then
			\[ \left[\int_0^\cdot \mathbf{G}_s d\mathbf{X}_s, \int_0^\cdot \mathbf{K}_s d\mathbf{Z}_s \right]_t = \int_0^t  \mathbf{G}_s d[\mathbf{X}, \mathbf{Z}]_s \mathbf{K}^\mathrm{T}_s.\] 
		\end{proof}
		\begin{remark}
			Note that, given the observation $\mathbb{Y}_{[0,t]}$, the process $D^{p-1}\mathbb{Y}^c_{\mathbf{A}^{(0)}, [0,t]}$ can be expressed as
			\begin{align}\label{eqn:cont_mart_part_Dp-1Y}
				\begin{split}
					D^{p-1}\mathbf{Y}^c_{\mathbf{A}^{(0)}, s} = D^{p-1}\mathbf{Y}_s-D^{p-1}\mathbf{y}_{0} - \mathbf{b}s - \sum_{0\leq u\leq s} &\Delta D^{p-1} \mathbf{Y}_u \mathds{1}_{\{\|\Delta D^{p-1} \mathbf{Y}_u \|\geq 1\}} \\
					- &\int_0^s\int_{\|y\|\leq 1} x\{\tilde{\mu}(du, dy) - du\,F(dy)\},\quad s\in[0, t],
				\end{split}
			\end{align}
			where $\tilde{\mu}$ is the jump measure of $D^{p-1}\mathbb{Y}_{[0, t]} = \{D^{p-1}\mathbf{Y}_s,\ s\in[0,t]\}$, i.e.\
			\[\tilde{\mu}(ds, dy) = \sum_{0\leq u\leq t:\Delta D^{p-1}\mathbf{Y}_u \neq 0} \delta_{(u, \Delta D^{p-1}\mathbf{Y}_u)}(ds, dy). \]
			Under $\mathbb{P}_{\mathbf{A}^{(0)}, \mathbf{x}_0}$, this process is a Brownian motion with covariance $\Sigma$. For any other $\mathbf{A}\in(\mathcal{M}_d(\mathbb{R}))^p$ we can write
			\[ D^{p-1}\mathbf{Y}_{\mathbf{A}, s}^c = D^{p-1}\mathbf{Y}^c_{\mathbf{A}^{(0)}, s}  + \sum_{j=1}^p\int_0^s A_j D^{p-j}\mathbf{Y}_u\ du, \quad s\in[0,t],\]
			which is a Brownian motion with covariance $\Sigma$ under $\mathbb{P}_{\mathbf{A}, \mathbf{x}_0}$.
		\end{remark}
		
		We can now use the explicit form of the likelihood for the MCAR model to obtain a maximum likelihood estimator when continuous-time observations are available. In order to prove its asymptotic properties, we follow the approach outlined in \citet[Section~11.7]{Kuchler_Sorensen_1997}, which first requires establishing the following result for the score process.
		
		% Under the assumption of square integrability of the driving \Levy process $\mathbb{L}$, the likelihood $\mathcal{L}_t(\phi) := \mathcal{L}(\mathrm{vec}^{-1}(\phi); \mathbb{Y}_{[0,t]})$ for $t\geq 0$ is the natural exponential family generated by $\mathbb{H}$ with natural parameter set $\Phi = \mathbb{R}^{pd^2}$ as defined in \citet[Chapter~11.2]{Kuchler_Sorensen_1997}. To see this note that for any $\phi\in\mathbb{R}^{pd^2}$, the process $\phi^\mathrm{T}\mathbb{H}$ \tetxcolor{red}{satisfies Novikov's condition on any $[0, T]$}. In particular, this defines an exponential family with a time-continuous likelihood function (even if the observed state-space process $\mathbb{X}$ may be discontinuous).
		
		\begin{lemma}\label{lemma:loc_mart_MCAR}
			Under Assumption \ref{ass:levy_process}, for any $\mathbf{A}\in\mathcal{M}_d(\mathbb{R})$ the score 
			\begin{equation} \label{eqn:score} \nabla_{\mathrm{vec}(\mathbf{A}^*)} \log \mathcal{L}(\mathbf{A}; \mathbb{Y}_{[0,t]}) = \mathbf{H}_t - [\mathbf{H}]_t \mathrm{vec}(\mathbf{A}^*),\ t\geq0, \end{equation}
			defines a mean-zero continuous square-integrable $\mathbb{P}_{\mathbf{A}, \mathbf{x}_0}$-martingale with quadratic variation $\{[\mathbf{H}]_t,\ t\geq0\}$.
		\end{lemma}
		\begin{proof} See Appendix \ref{app:proof_loc_mart_MCAR}.
		\end{proof}
		
		The main results of this section are given in the following theorem which establishes the existence and uniqueness of the continuous-time maximum likelihood estimator as well as its consistency and asymptotic normality.
		
		\begin{theorem} \label{thm:cons_asymp}
			Assume Assumption \ref{ass:levy_process} holds and let \begin{equation} \label{eqn:param_set}
				\mathfrak{A}=\{\mathbf{A} \in (\mathcal{M}_d(\mathbb{R}))^p : \sigma(\mathcal{A}_\mathbf{A}) \subseteq (-\infty, 0) + i\mathbb{R}\}.
			\end{equation}
			Define the maximum likelihood estimator
			\[\hat{\mathbf{A}}(\mathbb{Y}_{[0,t]}) = \mathrm{argmax}_{\mathbf{A}\in(\mathcal{M}_d(\mathbb{R}))^p} \mathcal{L}(\mathbf{A};\mathbb{Y}_{[0,t]}), \]
			then the MLE exists and is uniquely defined $\forall t>0$ by
			\begin{equation} \label{eqn:MLE}
				\hat{\mathbf{A}}(\mathbb{Y}_{[0,t]}) = \mathrm{vec}^{-1} ( [\mathbf{H}]_t^{-1}\mathbf{H}_t),
			\end{equation}
			where $\mathbf{H}_t$ and $[\mathbf{H}]_t$ are given in Equation \eqref{eqn:H} and \eqref{eqn:[H]} respectively. Moreover, if $\mathbf{A}^*\in \mathfrak{A}$ then the MLE is
			\begin{itemize}
				\item (weakly) consistent, i.e.\ 
				\[\hat{\mathbf{A}}(\mathbb{Y}_{[0,t]}) \overset{\mathbb{P}_{ \mathbf{A}^*, \mathbf{x}_0}}{\longrightarrow} \mathbf{A}^*, \ t\rightarrow \infty; \] 
				\item asymptotically normal, i.e.\ 
				\[ \sqrt{t} \left(\mathrm{vec}(\hat{\mathbf{A}}(\mathbb{Y}_{[0,t]})) - \mathrm{vec}(\mathbf{A}^*)\right) \overset{\mathcal{L}}{\longrightarrow} \mathbf{Z} \sim N(\mathbf{0}, \mathcal{H}^{-1}_\infty), \ t \rightarrow \infty, \]
				where $\mathcal{H}_\infty\in\mathcal{M}_{pd^2}(\mathbb{R})$ is a symmetric positive definite matrix such that $t^{-1}[\mathbf{H}]_t \overset{\mathbb{P}_{ \mathbf{A}^*, \mathbf{x}_0}}{\longrightarrow} \mathcal{H}_\infty, \ t\rightarrow \infty$.
			\end{itemize}
		\end{theorem}
		\begin{remark} \label{rem:feasible_CLT}
			The asymptotic normality result immediately yields a feasible version of the limit
			\[ [\mathbf{H}]_t^{1/2} \left(\mathrm{vec}(\hat{\mathbf{A}}(\mathbb{Y}_{[0,t]})) - \mathrm{vec}(\mathbf{A}^*)\right) \overset{\mathcal{L}}{\longrightarrow} \mathbf{Z} \sim N(\mathbf{0}, I_{pd^2}), \ t \rightarrow \infty. \]
		\end{remark}
		\begin{remark} \label{rem:cons_asymp_stat_ergodic}
			We can quite easily extend these results to the setting where the initial conditions for $\mathbb{Y}_{[0,t]}$ and its first $p-1$ derivatives are assumed to be random with distribution given by the stationary distribution of the MCAR($p$) state-space representation with parameters $\mathbf{A}^*\in\mathfrak{A}$, i.e.
			\[ \mathcal{L}([\mathbf{Y}^\mathrm{T}_0, \ldots, D^{p-1}\mathbf{Y}^\mathrm{T}_0]^\mathrm{T}) = \mathcal{L}\left( \int_0^{\infty} e^{s\mathcal{A}_{\mathbf{A}^*}} \mathcal{E}\, d\mathbf{L}_s\right).\]
			Under such measure $\mathbb{P}_{\mathbf{A}^*}$ the observed MCAR($p$) process $\mathbb{Y}$ and its state-space representation are stationary and ergodic, cf.\ Remark \ref{rem:stat_ergodic}. All other assumptions withstanding, we can show consistency and asymptotic normality hold under $\mathbb{P}_{\mathbf{A}^*}$ by a simple application of dominated convergence theorem.
		\end{remark}
		\begin{proof}
			The proof follows standard arguments for exponential families of stochastic processes, see \citet{Kuchler_Sorensen_1997}. The log-likelihood
			\[ \ell(\mathrm{vec}(\mathbf{A});\mathbb{Y}_{[0,t]}) = \mathrm{vec} (\mathbf{A})^\mathrm{T} \mathbf{H}_t - \frac{1}{2} \mathrm{vec}(\mathbf{A})^\mathrm{T} [\mathbf{H}]_t \mathrm{vec}(\mathbf{A}), \]
			is a quadratic form and $[\mathbf{H}]_t$ is $\mathbb{P}_{\mathbf{A}, \mathbf{x}_0}$-a.s.\ symmetric positive definite. We can thus write
			\[\mathrm{vec} (\hat{\mathbf{A}}(\mathbb{Y}_{[0,t]}) ) =  [\mathbf{H}]_t^{-1}\mathbf{H}_t = \mathrm{vec}(\mathbf{A}^*) +  [\mathbf{H}]_t^{-1}(\mathbf{H}_t - [\mathbf{H}]_t \mathrm{vec}(\mathbf{A}^*)), \]
			where the score $\{\mathbf{H}_t - [\mathbf{H}]_t \mathrm{vec}(\mathbf{A}^*),\ t\geq0\}$ is a mean-zero continuous square-integrable $\mathbb{P}_{\mathbf{A}^*, \mathbf{x}_0}$-martingale with quadratic variation $[\mathbf{H}]_t$ by Lemma \ref{lemma:loc_mart_MCAR}. The desired results then follow by martingale CLTs. For the full details, see Appendix \ref{app:proof_cons_asymp}.
		\end{proof}
		
		\section{Discrete observations: high-frequency estimators} \label{sec:MCAR_discr}
		In the previous section, we worked under the unrealistic assumption of a continuous-time observation of the MCAR process $\mathbb{Y}_{[0,t]} = \{\mathbf{Y}_s,\ s \in [0,t]\}$. Instead, we now assume that we observe the process on a discrete (possibly irregularly spaced) time grid. For each $t>0$, let $N_t\in\mathbb{N}$ and 
		\[\mathcal{P}_t := \{0=s_0< s_1<\ldots <s_{N_t} = t\}, \]
		be a partition of $[0,t]$ with mesh $\Delta_{\mathcal{P}_t} := \sup_{0\leq i< N_t} (s_{i+1} - s_i)$. The discretely observed process is then given by 
		\[ \mathbb{Y}_{\mathcal{P}_t} := \{\mathbf{Y}_{s}, \ s\in\mathcal{P}_t\}. \]
		As before, we assume the driving \Levy process parameters $(\mathbf{b}, \Sigma, F)$ are given. In this section, we work under the probability measure $\mathbb{P}_{\mathbf{A}^*}$ corresponding to a stationary and ergodic MCAR($p$) process with parameter $\mathbf{A}^*\in \mathfrak{A}$. We start by noting that the maximum likelihood estimator \[\hat{\mathbf{A}}(\mathbb{Y}_{[0,t]}) =  \mathrm{vec}^{-1} ( [\mathbf{H}]_t^{-1}\mathbf{H}_t), \]
		is a function of integrals of the form
		\begin{equation} \label{eqn:integrals}
			\int_0^t D^{l} Y^{(i)}_{s-} dD^{p-1} Y^{c, (j)}_{\mathbf{A}^{(0)}, s} \ \mathrm{and} \ \int_0^t D^{l} Y^{(i)}_s D^{k} Y^{(j)}_s\, ds, 
		\end{equation}
		for $i,j\in\{1,\ldots,d\}$ and $l,k\in\{0,\ldots, p-1\}$ where 
		%		\begin{multline*}
			%			D^{p-1}\mathbf{Y}^c_{\mathbf{A}^{(0)}, s} = D^{p-1}\mathbf{Y}_s - D^{p-1}\mathbf{Y}_{0} - \mathbf{b}s-\sum_{0\leq u\leq s} \Delta D^{p-1} \mathbf{Y}_u \mathds{1}_{\{\|\Delta D^{p-1} \mathbf{Y}_u \|\geq 1\}} \\
			%			- \int_0^s\int_{\|y\|\leq 1} x\{\tilde{\mu}(du, dy) - du\,F(dy)\},\quad s\in[0,t].
			%		\end{multline*}
		\begin{equation} \label{eqn:integrator_process}
			dD^{p-1}\mathbf{Y}^c_{\mathbf{A}^{(0)}, s} = dD^{p-1}\mathbf{Y}_s - \mathbf{b}\,ds- d\mathbf{J}_s
			- d\mathbf{M}_s,\quad s\in[0,t],
		\end{equation}
		and the jump process $\mathbb{J}_{[0,t]} = \{\mathbf{J}_s,\ s\in[0,t]\}$ and $\mathbb{M}_{[0,t]} = \{\mathbf{M}_s,\ s\in[0,t]\}$ are defined by
		\begin{align}
			\mathbf{J}_s &= \sum_{0\leq u\leq s} \Delta D^{p-1} \mathbf{Y}_u \mathds{1}_{\{\|\Delta D^{p-1} \mathbf{Y}_u \|\geq 1\}}, \label{eqn:jump_J}\\
			\mathbf{M}_s &= \int_0^s\int_{\|y\|\leq 1} y\{\tilde{\mu}(du, dy) - du\,F(dy)\}. \label{eqn:jump_M}
		\end{align}
		To write an estimator based only on $\mathbb{Y}_{\mathcal{P}_t}$ we thus need to make the following approximations:
		\begin{enumerate}
			\item Approximate the integrals appearing in Equation \eqref{eqn:integrals} by Riemann sums computed using discretized observations of the processes, i.e.\ replace
			\[ \int_0^t \ast_s \, \mathrm{d}\cdot_s\ \mathrm{with} \ \sum_{n=0}^{N_t -1} \ast_{s_n} (\cdot_{s_{n+1}} - \cdot_{s_n});\]
			\item Approximate the time derivatives in equations \eqref{eqn:integrals} and \eqref{eqn:integrator_process} by finite differences, i.e.\ replace 
			\[ D^{j}\mathbf{Y}_{s_n}\ \mathrm{with}\ \hat{D}^{j}\mathbf{Y}_{s_n}\ \mathrm{for}\ j=1,\ldots, p-1;\]
			\item Filter out the discretized jump increments $\Delta_{\mathcal{P}_t}^n \mathbf{J}$ and $\Delta_{\mathcal{P}_t}^n \mathbf{M}$ from $\Delta_{\mathcal{P}_t}^n \hat{D}^{p-1}\mathbf{Y}$ in Equation \eqref{eqn:integrator_process} via thresholding, i.e.\ replace 
			\[ \Delta^n_{\mathcal{P}_t} D^{p-1}\mathbf{Y}^{c}_{\mathbf{A}^{(0)}}\ \mathrm{with}\ \left[ \Delta_{\mathcal{P}_t}^n \hat{D}^{p-1}\mathbf{Y} - \mathbf{b} \Delta_{\mathcal{P}_t}^n \right] \odot \mathds{1}_{\left\{\left|\Delta_{\mathcal{Q}_t}^m \hat{D}^{p-1}\mathbf{Y} - \mathbf{b} \Delta_{\mathcal{P}_t}^n\right| \leq \boldsymbol{\nu}^n_t\right\}}.\]
		\end{enumerate}
		As we will see in Section \ref{sec:approx_derivatives}, when introducing finite difference approximations, we will need to work on two granularity levels: a finer partition $\mathcal{P}_t$ to approximate the derivatives and a coarser partition $\mathcal{Q}_t$ to approximate the integrals and filter the jumps. Thus, throughout this section, $\mathcal{P}_t$ will denote the finest grid on which observations are available and, when required, $\mathcal{Q}_t$ an appropriately chosen coarsening.
		
		\begin{remark} \label{rem:finite_activity}
			Let $\mathbb{L}_{[0, t]}$ denote the background driving \Levy process. If we assume $\mathbb{L}_{[0, t]}$ to have finite jump activity, i.e.\ $F(\mathbb{R}^d)<\infty$, then we can write the continuous martingale part of $D^{p-1}\mathbb{Y}_{[0,t]}$ as
			\[
			D^{p-1}\mathbf{Y}^c_{\mathbf{A}^{(0)}, s} = D^{p-1}\mathbf{Y}_s - D^{p-1}\mathbf{Y}_{0} - \tilde{\mathbf{b}}s -\sum_{0\leq u\leq s} \Delta D^{p-1} \mathbf{Y}_u, \quad s\in[0,t],
			\]
			where $(\tilde{\mathbf{b}}, \Sigma, F)$ is the characteristic triplet when no truncation function is used, i.e.\ 
			\[\mathbf{L}_t = \tilde{\mathbf{b}}t + \Sigma^{1/2}\mathbf{W}_t + \sum_{0\leq s\leq t}\Delta \mathbf{L}_s, \quad t\geq 0,\]
			where the sum is almost surely finite. In this case we will work with only one jump process, $\tilde{\mathbb{J}}_{[0,t]} = \{\tilde{\mathbf{J}}_s,\ s\in[0,t]\}$ given by
			\begin{equation} \label{eqn:jump_J_tilde}
				\tilde{\mathbf{J}}_s = \sum_{0\leq u\leq s} \Delta D^{p-1} \mathbf{Y}_u,
			\end{equation}
			in which case the integrator becomes 
			\begin{equation}
				D^{p-1}\mathbf{Y}^c_{\mathbf{A}^{(0)}, s} = D^{p-1}\mathbf{Y}_s - D^{p-1}\mathbf{Y}_{0} - \tilde{\mathbf{b}}s- \tilde{\mathbf{J}}_s,\quad s\in[0,t].
			\end{equation}
		\end{remark}
		
		Note that by rearranging the integral form of the formal Equation \eqref{eqn:autoregressive}, i.e.\ the last $d$ entries of SDE \eqref{eqn:SDE_def}, we can retrieve the background driving \Levy process $\mathbb{L}_{[0,t]} = \{\mathbf{L}_s,\ s\in[0,t]\}$ from $\mathbb{Y}_{[0,t]}$ by
		\begin{equation} \label{eqn:L=p(D)Y}
			\mathbf{L}_s = D^{p-1}\mathbf{Y}_{s} - D^{p-1}\mathbf{Y}_{0} + \sum_{j=1}^p\int_0^s A^*_j D^{p-j}\mathbf{Y}_u\ du, \quad s\in[0,t].
		\end{equation}
		Under $\mathbb{P}_{\mathbf{A}^*}$ this is a \Levy process with \Levy triplet $(\mathbf{b},\Sigma, F)$.
		
		\begin{remark}
			Recall that $D^{p-1}\mathbb{Y}$ jumps whenever $\mathbb{L}$ does and by the same magnitude, i.e.\ $\Delta D^{p-1}\mathbb{Y} = \Delta \mathbb{L}$. Assuming we can observe the jumps of the process $D^{p-1}\mathbb{Y}$ is equivalent to assuming we observe the jumps of the driving noise $\mathbb{L}$.
		\end{remark}
		
		The following Lemma establishes the framework to prove the subsequent approximations preserve the asymptotic properties of the estimator. 
		
		%We denote by $\mathcal{T}$ a countable set of times diverging to $\infty$.
		
		\begin{lemma} \label{lemma:approx}
			Let $\hat{\mathbf{A}}_{1, t} = \mathrm{vec}^{-1}([\mathbf{H}]^{-1}_{1, t} \mathbf{H}_{1, t})$ with $\mathbf{H}_{1, t}\in\mathbb{R}^{pd^2}, [\mathbf{H}]_{1, t}\in\mathbb{R}^{pd^2\times pd^2}$ be a consistent and asymptotically normal estimator for $\mathbf{A}^*\in\mathfrak{A}$, i.e.
			\begin{itemize}
				\item $\displaystyle \hat{\mathbf{A}}_{1,t} \overset{\mathbb{P}_{ \mathbf{A}^*}}{\longrightarrow} \mathbf{A}^*, \ t\rightarrow \infty;$ 
				\item $\displaystyle \sqrt{t} \left(\mathrm{vec}(\hat{\mathbf{A}}_{1,t}) - \mathrm{vec}(\mathbf{A}^*)\right) \overset{\mathcal{L}}{\longrightarrow} \mathbf{Z} \sim N(\mathbf{0}, \mathcal{H}^{-1}_\infty), \ t \rightarrow \infty,$ where $\mathcal{H}_\infty\in\mathcal{M}_{pd^2}(\mathbb{R})$ is a symmetric positive definite matrix such that $t^{-1}[\mathbf{H}]_{1,t} \overset{\mathbb{P}_{ \mathbf{A}^*}}{\longrightarrow} \mathcal{H}_\infty, \ t\rightarrow \infty$.
			\end{itemize}
			Suppose that
			\begin{equation} \label{eqn:H_limit}
				t^{-1/2}\left( \mathbf{H}_{2, t} - \mathbf{H}_{1,t}\right) \overset{\mathbb{P}_{\mathbf{A}^*}}{\longrightarrow} 0,\quad t\rightarrow\infty,
			\end{equation}
			and 
			\begin{equation} \label{eqn:[H]_limit} 
				t^{-1/2}\left([\mathbf{H}]_{2, t} - [\mathbf{H}]_{1, t} \right) \overset{\mathbb{P}_{ \mathbf{A}^*}}{\longrightarrow} 0, \quad t\rightarrow \infty,
			\end{equation}
			for $\mathbf{H}_{2, t}\in\mathbb{R}^{pd^2}, [\mathbf{H}]_{2, t}\in\mathbb{R}^{pd^2\times pd^2}$. Then the estimator $\hat{\mathbf{A}}_{2, t} = \mathrm{vec}^{-1}([\mathbf{H}]^{-1}_{2, t} \mathbf{H}_{2, t})$ is also consistent and asymptotically normal for $\mathbf{A}^*\in\mathfrak{A}$, i.e.
			\begin{itemize}
				\item $\displaystyle \hat{\mathbf{A}}_{2,t} \overset{\mathbb{P}_{ \mathbf{A}^*}}{\longrightarrow} \mathbf{A}^*, \ t\rightarrow \infty;$ 
				\item $\displaystyle \sqrt{t} \left(\mathrm{vec}(\hat{\mathbf{A}}_{2,t}) - \mathrm{vec}(\mathbf{A}^*)\right) \overset{\mathcal{L}}{\longrightarrow} \mathbf{Z} \sim N(\mathbf{0}, \mathcal{H}^{-1}_\infty), \ t \rightarrow \infty,$ and $t^{-1}[\mathbf{H}]_{2,t} \overset{\mathbb{P}_{ \mathbf{A}^*}}{\longrightarrow} \mathcal{H}_\infty, \ t\rightarrow \infty.$
			\end{itemize}
		\end{lemma}
		
		\begin{proof}
			We start by noting that Equation \eqref{eqn:[H]_limit} immediately implies 
			\begin{equation*} 
				t^{-1}\left([\mathbf{H}]_{2, t} - [\mathbf{H}]_{1,t} \right) \overset{\mathbb{P}_{ \mathbf{A}^*}}{\longrightarrow} 0, \ t\rightarrow \infty. 
			\end{equation*}
			Combining this with $t^{-1}[\mathbf{H}]_{1,t} \overset{\mathbb{P}_{ \mathbf{A}^*}}{\longrightarrow} \mathcal{H}_\infty, \ t\rightarrow \infty,$ by uniqueness of limits we have
			\begin{equation} \label{eqn:ergodic_[H_hat]} t^{-1}[\mathbf{H}]_{2, t} \overset{\mathbb{P}_{ \mathbf{A}^*}}{\longrightarrow} \mathcal{H}_\infty, \ t\rightarrow \infty. 
			\end{equation}
			Next, we show that 
			\begin{equation} \label{eqn:limit_discr_est} \sqrt{t}\left( [\mathbf{H}]_{2, t}^{-1} \mathbf{H}_{2, t} - [\mathbf{H}]_{1,t}^{-1} \mathbf{H}_{1,t} \right) \overset{\mathcal{L}}{\longrightarrow} 0, \ t\rightarrow \infty. \end{equation}			
			To do so we split the limiting variable as
			\begin{equation} \label{eqn:split_limit_discr_est}
				\begin{split}
					\sqrt{t} \left( [\mathbf{H}]_{2,t}^{-1} \mathbf{H}_{2, t} - [\mathbf{H}]_{1,t}^{-1} \mathbf{H}_{1,t}\right)
					&=\sqrt{t} \left( [\mathbf{H}]_{2,t}^{-1} \mathbf{H}_{2, t} - [\mathbf{H}]_{2, t}^{-1} \mathbf{H}_{1,t}\right) + \sqrt{t} \left([\mathbf{H}]_{2, t}^{-1} \mathbf{H}_{1,t}- [\mathbf{H}]_{1,t}^{-1} \mathbf{H}_{1,t}\right) 
				\end{split}
			\end{equation}
			and show that each of these two terms converges to zero in probability. We can write the first term as
			\[ \sqrt{t} \left( [\mathbf{H}]_{2, t}^{-1} \mathbf{H}_{2, t} - [\mathbf{H}]_{2, t}^{-1} \mathbf{H}_{1,t}\right) = \left(t[\mathbf{H}]_{2, t}^{-1}\right) t^{-1/2} \left( \mathbf{H}_{2, t} - \mathbf{H}_{1,t}\right). \]
			By Equation \eqref{eqn:ergodic_[H_hat]}, continuity of the inversion operator and invertibility of $\mathcal{H}_\infty$ we have \[t[\mathbf{H}]_{2, t}^{-1}\overset{\mathbb{P}_{\mathbf{A}^*}}{\longrightarrow} \mathcal{H}_\infty^{-1}, \ \rightarrow\infty.\]
			Thus by multiplicativity of $\mathbb{P}_{\mathbf{A}^*}$-limits and Equation \eqref{eqn:H_limit} the first term in \eqref{eqn:split_limit_discr_est} converges to zero in probability under $\mathbb{P}_{\mathbf{A}^*}$. The second term in \eqref{eqn:split_limit_discr_est} can be written as
			\begin{align*}  
				\sqrt{t} \left([\mathbf{H}]_{2, t}^{-1} \mathbf{H}_{1,t}- [\mathbf{H}]_{1,t}^{-1} \mathbf{H}_{1,t}\right)  
				&= \left(t [\mathbf{H}]_{2, t}^{-1}\right) t^{-1/2} \left([\mathbf{H}]_{1,t} - [\mathbf{H}]_{2, t}\right)  ( [\mathbf{H}]^{-1}_{1,t} \mathbf{H}_{1,t}).
			\end{align*}
			We note that:
			\begin{itemize}
				\item $t [\mathbf{H}]_{2, t}^{-1}$ converges in probability under $\mathbb{P}_{\mathbf{A}^*}$ to $\mathcal{H}_\infty^{-1}$;
				\item $t^{-1/2} \left([\mathbf{H}]_{1,t} - [\mathbf{H}]_{2, t}\right)$ converges in probability under $\mathbb{P}_{\mathbf{A}^*}$ to zero by Equation \eqref{eqn:[H]_limit};
				\item $[\mathbf{H}]^{-1}_{1,t} \mathbf{H}_{1,t}$ converges in probability under $\mathbb{P}_{\mathbf{A}^*}$ to $\mathrm{vec}(\mathbf{A}^*)$ by consistency of $\hat{\mathbf{A}}_{1,t}$ and continuity of $\mathrm{vec}$.
			\end{itemize} 
			Thus, by multiplicativity of $\mathbb{P}_{\mathbf{A}^*}$-limits, we get that also the second term converges to zero in probability under $\mathbb{P}_{\mathbf{A}^*}$.
			
			Equation \eqref{eqn:limit_discr_est} immediately implies 
			\begin{equation*} [\mathbf{H}]_{2, t}^{-1} \mathbf{H}_{2, t} - [\mathbf{H}]_{1,t}^{-1} \mathbf{H}_{1,t} \overset{\mathbb{P}_{ \mathbf{A}^* }}{\longrightarrow} 0, \ t\rightarrow \infty, \end{equation*}
			and hence consistency of the estimator $\hat{\mathbf{A}}_{2,t}$ follows from consistency of $\hat{\mathbf{A}}_{1,t}$ and continuity of the $\mathrm{vec}$ operator. Combining Equation \eqref{eqn:limit_discr_est} with the asymptotic normality of $\hat{\mathbf{A}}_{1,t}$ yields
			\[ t^{1/2} \left( [\mathbf{H}]_{2, t}^{-1} \mathbf{H}_{2, t} - \mathrm{vec}^{-1}(\mathbf{A}^*) \right) \overset{\mathcal{L}}{\longrightarrow} \mathbf{Z} \sim N(\mathbf{0}, \mathcal{H}_\infty^{-1}), \ t\rightarrow \infty, \]
			which, together with Equation \eqref{eqn:ergodic_[H_hat]}, gives the desired asymptotic normality result for $\hat{\mathbf{A}}_{2,t}$
		\end{proof}
		
		\subsection{Approximating the integrals} \label{sec:approx_integrals}
		First, we assume we have access to the sampled time series of $\mathbb{Y}_{[0,t]}$ and its first $p-1$ derivatives over the partition $\mathcal{P}_t$, i.e.\ $D^i\mathbb{Y}_{\mathcal{P}_t} := \{D^i\mathbf{Y}_{s}, \ s\in\mathcal{P}_t\}$
		for $i=0,\ldots, p-1$, as well as the jump processes $\mathbb{J}_{\mathcal{P}_t} := \{\mathbf{J}_{s}, \ s\in\mathcal{P}_t\},$ $\mathbb{M}_{\mathcal{P}_t} := \{\mathbf{M}_{s}, \ s\in\mathcal{P}_t\}$.
		In this setting we can perfectly recover $D^{p-1}\mathbb{Y}^c_{\mathbf{A}^{(0)}, \mathcal{P}_t}$ by
		\[D^{p-1}\mathbf{Y}^c_{\mathbf{A}^{(0)}, s} = D^{p-1}\mathbf{Y}_{s} - D^{p-1}\mathbf{Y}_0 - \mathbf{b}s- \mathbf{J}_s - \mathbf{M}_s, \quad s\in\mathcal{P}_t.\]
		Approximating the integrals in  \eqref{eqn:integrals}
		%		\[ \int_0^t D^l Y^{(i)}_{s-} dD^{p-1} Y^{c, (j)}_{\mathbf{A}^{(0)},s} \ \mathrm{and} \ \int_0^t D^l Y^{(i)}_s D^{k} Y^{(j)}_s\, ds, \]
		by
		\[ \sum_{n = 0}^{N_t-1} D^l Y^{(i)}_{s_n} \left(D^{p-1} Y^{c, (j)}_{\mathbf{A}^{(0)}, s_{n+1}} - D^{p-1} Y^{c, (j)}_{\mathbf{A}^{(0)}, s_{n}}\right) \ \mathrm{and} \  \sum_{n = 0}^{N_t-1} D^l Y^{(i)}_{s_n} D^{k} Y^{(j)}_{s_n} (s_{n+1} - s_n), \]
		for $i,j\in\{1,\ldots,d\}$ and $l,k\in\{0,\ldots, p-1\}$, we obtain natural approximations for $\mathbf{H}_t$ and $[\mathbf{H}]_t$ respectively.
		
		%		We can thus define the modified estimator
		%		\[\hat{\mathbf{A}}(\mathbb{Y}_{\mathcal{P}_t}, \ldots, D^{p-1}\mathbb{Y}_{\mathcal{P}_t},\mathbb{J}_{\mathcal{P}_t}, \mathbb{M}_{\mathcal{P}_t}) := \mathrm{vec}^{-1}\left([\mathbf{H}]_{\mathcal{P}_t}^{-1} \mathbf{H}_{\mathcal{P}_t}\right).  \]
		
		In order to derive the asymptotic properties of this discretized estimator, classic in-fill approximations
		%      \footnote{Note that, for fixed $t>0$, by definition of the integrals
			% \[\sum_{n = 0}^{N_t-1} D^l Y^{(i)}_{s_n} \left(D^{p-1} Y^{c, (j)}_{\mathbf{A}^{(0)}, s_{n+1}} - D^{p-1} Y^{c, (j)}_{\mathbf{A}^{(0)}, s_{n}}\right) \overset{L^2(\Omega, \mathcal{F}, \mathbb{P}_{\mathbf{A}^*})}{\longrightarrow} \int_0^t D^l Y^{(i)}_{s-} dD^{p-1} Y^{c, (j)}_s, \] 
			% \[\sum_{n = 0}^{N_t-1} D^l Y^{(i)}_{s_n} D^{k} Y^{(j)}_{s_n} (s_{n+1} - s_n) \overset{\mathbb{P}_{*}\mathrm{-a.s.}}{\longrightarrow} \int_0^t D^l Y^{(i)}_s D^{k} Y^{(j)}_s\, ds, \] 
			% as the partition's mesh size $\Delta_{\mathcal{P}_t}\rightarrow 0$. 
			% Hence 
			% \[{\mathbf{H}}_{\mathcal{P}_t} \overset{L^2(\Omega, \mathcal{F}, \mathbb{P}_{*})}{\longrightarrow} \mathbf{H}_t\ \mathrm{and} \ [\mathbf{H}]_{\mathcal{P}_t} \overset{\mathbb{P}_{*}\mathrm{-a.s.}}{\longrightarrow} [\mathbf{H}]_t,\]
			% and, by Slutsky's lemma and continuity of the vectorization operator,
			% \[\hat{\mathbf{A}}(\mathbb{Y}_{\mathcal{P}_t}, \ldots, D^{p-1}\mathbb{Y}_{\mathcal{P}_t},\mathbb{J}_{\mathcal{P}_t}, \mathbb{M}_{\mathcal{P}_t}) \overset{\mathcal{L}}{\longrightarrow}  \hat{\mathbf{A}}(\mathbb{Y}_{[0,t]}), \]
			% as $\Delta_{\mathcal{P}_t}\rightarrow 0$. And thus, up to a subsequence, $\mathbb{P}_{*}\mathrm{-a.s.}$} 
		are not sufficient. Instead, we need to combine in-fill and long-span asymptotics by taking the limits as both the mesh size vanishes and the observation horizon increases to infinity. We consider a countable sequence of partitions $\{\mathcal{P}_t,\ t\in\mathcal{T}\}$ with mesh $\Delta_{\mathcal{P}_t}\rightarrow 0$ as $t\rightarrow\infty$ which satisfy the following ``high-frequency sampling condition''.
		
		\begin{assumption} \label{ass:HF_sampling}
			The sequence of partitions $\{\mathcal{P}_t,\ t\in\mathcal{T}\}$ satisfies $\Delta_{\mathcal{P}_t} t\rightarrow 0$ as $t\rightarrow\infty$.
		\end{assumption}
		
		Under this assumption, we can prove the discretized estimator preserves the asymptotic properties of $\hat{\mathbf{A}}(\mathbb{Y}_{[0,t]})$.
		
		\begin{theorem} \label{thm:cons_asymp_first_approx}
			Let $\mathbb{Y}_{\mathcal{P}_t}$ denote the sampled observation of a stationary and ergodic MCAR($p$) process with parameters $\mathbf{A}^*\in\mathfrak{A}$ under $\mathbb{P}_{\mathbf{A}^*}$ such that the driving \Levy process satisfies Assumption \ref{ass:levy_process}. Let $\{\mathcal{P}_t,\ t\in\mathcal{T}\}$ be a countable sequence of partitions satisfying Assumption \ref{ass:HF_sampling}. 
			%			Let \[ D^i\mathbb{Y}_{\mathcal{P}_t} := \{D^i\mathbf{Y}_{s}, \ s\in\mathcal{P}_t\},\ i=0,\ldots, p-1,\quad \mathbb{J}_{\mathcal{P}_t} := \{\mathbf{J}_{s}, \ s\in\mathcal{P}_t\},\quad \mathbb{M}_{\mathcal{P}_t} := \{\mathbf{M}_{s}, \ s\in\mathcal{P}_t\}, \] 
			%			be the observations of a MCAR($p$) process, its first $p-1$ derivatives and its jump processes over the partition $\mathcal{P}_t$ for $t\in\mathcal{T}$. 
			Define the discretized estimator
			\begin{equation} \label{eqn:discr_estimator_A} \hat{\mathbf{A}}(\mathbb{Y}_{\mathcal{P}_t}, \ldots, D^{p-1}\mathbb{Y}_{\mathcal{P}_t},\mathbb{J}_{\mathcal{P}_t}, \mathbb{M}_{\mathcal{P}_t}) = \mathrm{vec}^{-1}\left([\mathbf{H}]_{\mathcal{P}_t}^{-1} \mathbf{H}_{\mathcal{P}_t}\right),
			\end{equation} 
			where 
			\begin{align}\label{eqn:H_hat}
				\mathbf{H}_{\mathcal{P}_t} &= -
				\sum_{n = 0}^{N_t-1}
				\begin{pmatrix}
					D^{p-1} Y^{(1)}_{s_n} \Sigma^{-1} \\
					\vdots \\
					D^{p-1} Y^{(d)}_{s_n} \Sigma^{-1}\\
					\vdots \\
					Y^{(1)}_{s_n} \Sigma^{-1}\\
					\vdots \\
					Y^{(d)}_{s_n} \Sigma^{-1}\\
				\end{pmatrix} \left(D^{p-1} \mathbf{Y}^{c}_{\mathbf{A}^{(0)}, s_{n+1}} - D^{p-1} \mathbf{Y}^{c}_{\mathbf{A}^{(0)}, s_{n}}\right) ,\\
				\label{eqn:[H]_hat}
				[\mathbf{H}]_{\mathcal{P}_t} &=
				\sum_{n = 0}^{N_t-1} \begin{pmatrix}
					D^{p-1} Y^{(1)}_{s_n} \Sigma^{-1} D^{p-1}Y^{(1)}_{s_n} & \cdots &  D^{p-1}Y^{(1)}_{s_n} \Sigma^{-1} Y^{(d)}_{s_n} \\
					\vdots & \ddots & \vdots \\
					Y^{(d)}_{s_n} \Sigma^{-1} D^{p-1}Y^{(1)}_{s_n} & \cdots &  Y^{(d)}_{s_n} \Sigma^{-1} Y^{(d)}_{s_n}\ \\
				\end{pmatrix}(s_{n+1} - s_n).\end{align}
			%			and
			%			\[D^{p-1}\mathbf{Y}^c_{\mathbf{A}^{(0)}, s} = D^{p-1}\mathbf{Y}_{s} - D^{p-1}\mathbf{Y}_0 - \mathbf{b}s- \mathbf{J}_s - \mathbf{M}_s, \quad s\in\mathcal{P}_t.\]
			The estimator is 
			\begin{itemize}
				\item (weakly) consistent, i.e.\ 
				\[ \hat{\mathbf{A}}(\mathbb{Y}_{\mathcal{P}_t}, \ldots, D^{p-1}\mathbb{Y}_{\mathcal{P}_t},\mathbb{J}_{\mathcal{P}_t}, \mathbb{M}_{\mathcal{P}_t}) \overset{\mathbb{P}_{ \mathbf{A}^* }}{\longrightarrow} \mathbf{A}^*, \ t\rightarrow \infty; \] 
				\item asymptotically normal, i.e.\ 
				\[ \sqrt{t} \left(\mathrm{vec}(\hat{\mathbf{A}}(\mathbb{Y}_{\mathcal{P}_t}, \ldots, D^{p-1}\mathbb{Y}_{\mathcal{P}_t},\mathbb{J}_{\mathcal{P}_t}, \mathbb{M}_{\mathcal{P}_t})) - \mathrm{vec}(\mathbf{A}^*)\right) \overset{\mathcal{L}}{\longrightarrow} \mathbf{Z} \sim N(\mathbf{0}, \mathcal{H}^{-1}_\infty), \ t \rightarrow \infty, \]
				where $\mathcal{H}_\infty\in\mathcal{M}_{pd^2}(\mathbb{R})$ is a symmetric positive definite matrix such that $t^{-1}[\mathbf{H}]_{\mathcal{P}_t} \overset{\mathbb{P}_{ \mathbf{A}^*}}{\longrightarrow} \mathcal{H}_\infty, \ t\rightarrow \infty$.
			\end{itemize}
		\end{theorem}
		\begin{proof}
			We apply Lemma \ref{lemma:approx} with $\hat{\mathbf{A}}_{1,t} = \hat{\mathbf{A}}(\mathbb{Y}_{[0,t]})$ and $\hat{\mathbf{A}}_{2,t} = \hat{\mathbf{A}}(\mathbb{Y}_{\mathcal{P}_t}, \ldots, D^{p-1}\mathbb{Y}_{\mathcal{P}_t},\mathbb{J}_{\mathcal{P}_t}, \mathbb{M}_{\mathcal{P}_t})$
			by decomposing the integrator as
			\begin{equation} \label{eqn:DY=W-A}
				D^{p-1}\mathbf{Y}^c_{\mathbf{A}^{(0)}, s} = D^{p-1}\mathbf{Y}_{\mathbf{A}^*, s}^c  - \sum_{j=1}^p\int_0^s A^*_j D^{p-j}\mathbf{Y}_u\ du = \Sigma^{1/2} \mathbf{W}_s - \sum_{j=1}^p\int_0^s A^*_j D^{p-j}\mathbf{Y}_u\ du,
			\end{equation}
			where $\mathbb{W} = \{\mathbf{W}_s,\  s\geq 0\}$ is a standard Brownian motion under $\mathbb{P}_{\mathbf{A}^*}$. Under Assumption \ref{ass:levy_process} by Theorem \ref{thm:cons_asymp} (and Remark \ref{rem:cons_asymp_stat_ergodic}) we have that $\hat{\mathbf{A}}_{1,t}$ is consistent and asymptotic normal. It thus remains to prove equations \eqref{eqn:H_limit} and \eqref{eqn:[H]_limit}. See Appendix \ref{app:proof_them_cons_asymp_first_approx} for the full details.			
		\end{proof}
		
		\subsection{Approximating the derivatives} \label{sec:approx_derivatives}
		Next, we assume we have access to the sampled MCAR($p$) process only, i.e.\ $\mathbb{Y}_{\mathcal{P}_t} := \{\mathbf{Y}_{s}, \ s\in\mathcal{P}_t\}$,
		and, as in the previous section, the jump processes $ \mathbb{J}_{\mathcal{P}_t}$ and $\mathbb{M}_{\mathcal{P}_t}$.
		Motivated by Proposition \ref{prop:right_diff} we approximate the first $p-1$ derivatives of $\mathbb{Y}_{[0,t]}$ at each point in $\mathcal{P}_t$ by forward finite differences. We can define these iteratively by
		\begin{equation} \label{eqn:D_hat}
			\hat{D}^k\mathbb{Y}_{\mathcal{P}_t} := \{ \hat{D}^k\mathbf{Y}_{s_n}, \ n=0,\ldots, N_t - k\}\ \mathrm{s.t.}\ \hat{D}^k\mathbf{Y}_{s_n} = 
			\frac{\hat{D}^{k-1}\mathbf{Y}_{s_{n+1}} - \hat{D}^{k-1}\mathbf{Y}_{s_{n}}}{s_{n+1} - s_{n}},
		\end{equation} 
		for $k\in\{1,\ldots, p-1\}$ and $\hat{D}^0\mathbb{Y}_{\mathcal{P}_t} = \mathbb{Y}_{\mathcal{P}_t}$. 
		In order to obtain appropriate convergence of the (higher-order) finite difference approximations to the corresponding derivatives we assume the observation partitions $\mathcal{P}_t$ become increasingly uniform. This is what we will refer to as the ``convergence to uniform spacing'' condition.
		\begin{assumption} \label{ass:evenly_spaced_lim}
			The sequence of partitions $\{\mathcal{P}_t,\ t\in\mathcal{T}\}$ satisfies
			\begin{equation} \label{eqn:c_P}
				|c_{\mathcal{P}_t}^{p-1} - 1| = O\left(\Delta_{\mathcal{P}_t}^{p-3/2}\right), \quad t\rightarrow \infty,
			\end{equation}
			where
			\begin{equation*}
				c_{\mathcal{P}_t} := \frac{\inf_{0\leq i< N_t} (s_{i+1} - s_i)}{\sup_{0\leq i< N_t} (s_{i+1} - s_i)}.
			\end{equation*}
		\end{assumption}
		\begin{remark} 
			To understand this condition we note that
			\begin{itemize}
				\item for $p\geq 2$, it implies $c_{\mathcal{P}_t}\rightarrow 1,\ t\rightarrow\infty$ and Equation \eqref{eqn:c_P} can be understood as the ``speed of convergence'';
				\item when $p=1$, it is trivially satisfied for any sequence of partitions $\{\mathcal{P}_t,\ t\in\mathcal{T}\}$. This is what one would desire: the resulting estimator does not depend on $\mathcal{P}_t$ as there is no approximation of derivatives.
			\end{itemize}
			We will also make use of the following immediate consequences of the definition of $c_{\mathcal{P}_t}$:
			\begin{itemize}
				\item $c_{\mathcal{P}_t}\in[0, 1]$, thus $|c_{\mathcal{P}_t}^k - 1| = O\left(\Delta_{\mathcal{P}_t}^{k-1/2}\right),\ t\rightarrow \infty$ for all $k=1, \ldots, p-1$, and
				\item for any interval $[s_n, s_{n+1}]\in\mathcal{P}_t$ one has
				\[c_{\mathcal{P}_t}\Delta_{\mathcal{P}_t} \leq |s_{n+1} - s_n| \leq \Delta_{\mathcal{P}_t}.\]
			\end{itemize}
			A weaker assumption might be sufficient, where we control the ratio between the smallest and the largest interval in each group of $p-1$ successive intervals, i.e.\ the maximum length over which a forward difference is defined.
		\end{remark}
		
		\begin{lemma} \label{lemma:D-D_L2_bound}
			Let $\mathbb{Y}_{\mathcal{P}_t}$ denote the sampled observation of a stationary and ergodic MCAR($p$) process with parameters $\mathbf{A}^*\in\mathfrak{A}$ under $\mathbb{P}_{\mathbf{A}^*}$ such that the driving \Levy process satisfies Assumption \ref{ass:levy_process}. Let $\{\mathcal{P}_t,\ t\in\mathcal{T}\}$ be a countable sequence of partitions satisfying Assumption \ref{ass:evenly_spaced_lim}. For $k\in\{1,\ldots, p-1\}$ let $\hat{D}^k\mathbb{Y}_{\mathcal{P}_t}$ denote the $k$-th forward finite difference of $\mathbb{Y}_{\mathcal{P}_t}$ defined in Equation \eqref{eqn:D_hat}. Then for any $n\in\{0,\ldots, N_t-k\}$ we have
			\begin{equation} \label{eqn:D-D_L2_bound}
				\mathbb{E}_{\mathbf{A}^*}\left[\left\|\hat{D}^k\mathbf{Y}_{s_n} - D^k\mathbf{Y}_{s_n}\right\|^2\right] = O(\Delta_{\mathcal{P}_t}), \quad t\rightarrow \infty,
			\end{equation}
		\end{lemma}
		\begin{proof}
			See Appendix \ref{app:proof_D-D_L2_bound}.
		\end{proof}
		
		Having approximated the derivatives on the partition $\mathcal{P}_t$, we then proceed as in Section \ref{sec:approx_integrals} to approximate the integrals. To preserve the asymptotic properties of the estimator, we will need to work on a coarser partition. For all $t\in\mathcal{T}$ let 
		\[\mathcal{Q}_t = \{0=u_0<u_1<\ldots < u_{N'_t} = t\},\]
		be a coarsening of $\mathcal{P}_t$. Note that $N'_t \leq N_t$ and $\Delta_{\mathcal{P}_t} \leq \Delta_{\mathcal{Q}_t}$. As before, we will assume the sequence $\{\mathcal{Q}_t, \ t\in\mathcal{T}\}$ satisfies the high-frequency sampling condition Assumption \ref{ass:HF_sampling}, and, hence, so does $\{\mathcal{P}_t, \ t\in\mathcal{T}\}$, but we will also control the rate at which the number of observations increases.
		\begin{assumption} \label{ass:controlled_sampling}
			The sequence of partitions $\{\mathcal{Q}_t,\ t\in\mathcal{T}\}$ satisfies $c_{\mathcal{Q}_t}\rightarrow c>0$ as $t\rightarrow\infty$.
		\end{assumption}
		\begin{remark}
			This is a slightly stronger condition than $N'_t \Delta_{\mathcal{Q}_t} = O(t) $ as $t\rightarrow\infty$, which is used in \cite{mai_OU} and \cite{Courgeau_Veraart_2022}. The stronger condition stated in Assumption \ref{ass:controlled_sampling} will only be needed when jointly approximating the derivatives and thresholding the jumps in Section \ref{sec:thresholding}. In any case, both versions of this condition can always be achieved by subsampling $\mathcal{Q}_t$, i.e.\ discarding observations.
		\end{remark}
		Finally, we control the relative mesh sizes of $\mathcal{P}_t$ and $\mathcal{Q}_t$ via the following joint condition. This ensures the derivative approximations converge faster than the Riemann sums and the resulting estimator does not degenerate. The relationship between $\mathcal{P}_t$ and $\mathcal{Q}_t$ is illustrated in Figure \ref{fig:joint_partitions}.
		\begin{assumption} \label{ass:joint_mesh}
			The sequences of partitions $\{\mathcal{P}_t, \ t\in\mathcal{T}\}$ and $\{\mathcal{Q}_t, \ t\in\mathcal{T}\}$ are such that 
			$\Delta_{\mathcal{P}_t} = o(t^{-1} \Delta_{\mathcal{Q}_t}^2)$ as $t\rightarrow \infty.$
		\end{assumption}
		\begin{remark}
			For example, since $t\Delta_{\mathcal{Q}_t}\rightarrow 0$ as $t\rightarrow 0$ by Assumption \ref{ass:HF_sampling}, a sufficient condition is that  $\Delta_{\mathcal{P}_t} \propto \Delta^{3}_{\mathcal{Q}_t}$.
		\end{remark}
		
		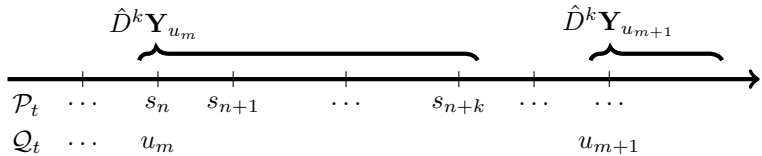
\begin{figure}[ht]
			\centering
			\begin{tikzpicture}
				% draw horizontal line   
				\draw[ultra thick, ->] (0,0) -- (10,0);
				
				% draw vertical lines
				\foreach \x in {1, 2,3,4.5,6,7,8}
				\draw (\x cm,3pt) -- (\x cm,-3pt);
				
				% P_t
				\draw[ultra thick] (0.25,0) node[below=2pt,thick] {$\mathcal{P}_t$} node[above=3pt] {};]
				\draw[ultra thick] (1,0) node[below=5pt,thick] {$\ldots$} node[above=3pt] {};
				\draw[ultra thick] (2,0) node[below=3pt,thick] {$s_n$} node[above=3pt] {};
				\draw[ultra thick] (3,0) node[below=3pt,thick] {$s_{n+1}$} node[above=3pt] {};
				\draw[ultra thick] (4.5,0) node[below=5pt,thick] {$\ldots$} node[above=3pt] {};
				\draw[ultra thick] (6,0) node[below=3pt,thick] {$s_{n+k}$} node[above=3pt] {};
				\draw[ultra thick] (7,0) node[below=5pt,thick] {$\ldots$} node[above=3pt] {};
				\draw[ultra thick] (8,0) node[below=5pt,thick] {$\ldots$} node[above=3pt] {};
				
				% Q_t
				\draw[ultra thick] (0.25,0) node[below=16pt,thick] {$\mathcal{Q}_t$} node[above=3pt] {};
				\draw[ultra thick] (1,0) node[below=20pt,thick] {$\ldots$} node[above=3pt] {};
				\draw[ultra thick] (2,0) node[below=18pt,thick] {$u_m$} node[above=3pt] {};
				\draw[ultra thick] (8,0) node[below=18pt,thick] {$u_{m+1}$} node[above=3pt] {};
				
				\draw[black, ultra thick, decorate, decoration={brace, amplitude=5pt, aspect=0.05}] (1.75,0.25) -- (6.25,0.25) node [black, midway, above=4pt, xshift=-2.05cm] {$\hat{D}^{k}\mathbf{Y}_{u_m}$};
				
				\draw[black, ultra thick, decorate, decoration={brace, amplitude=5pt, aspect=0.15}] (7.75,0.25) -- (9.5,0.25) node [black, midway, above=4pt, xshift=-0.5cm] {$\hat{D}^{k}\mathbf{Y}_{u_{m+1}}$};
			\end{tikzpicture}
			\caption{Approximating the derivatives on $\mathcal{P}_t$ and the integrals on $\mathcal{Q}_t$.} \label{fig:joint_partitions}
		\end{figure}
		
		\begin{theorem}               \label{thm:cons_asymp_second_approx}
			Let $\mathbb{Y}_{\mathcal{P}_t}$ denote the sampled observation of a stationary and ergodic MCAR($p$) process with parameters $\mathbf{A}^*\in\mathfrak{A}$ under $\mathbb{P}_{\mathbf{A}^*}$ such that the driving \Levy process satisfies Assumption \ref{ass:levy_process}. Let $\{\mathcal{P}_t,\ t\in\mathcal{T}\}$ be a countable sequence of partitions satisfying Assumption \ref{ass:evenly_spaced_lim} with refinements $\{\mathcal{Q}_t,\ t\in\mathcal{T}\}$ satisfying Assumption \ref{ass:HF_sampling}, Assumption \ref{ass:controlled_sampling} and Assumption \ref{ass:joint_mesh}. Define the discretized estimator
			\begin{equation} \label{eqn:discr_estimator_A_2}
				\hat{\mathbf{A}}(\mathbb{Y}_{\mathcal{P}_t}, \mathbb{J}_{\mathcal{Q}_t}, \mathbb{M}_{\mathcal{Q}_t}) := \mathrm{vec}^{-1}\left([\mathbf{H}]_{\mathcal{P}_t, \mathcal{Q}_t}^{-1} \mathbf{H}_{\mathcal{P}_t, \mathcal{Q}_t}\right),
			\end{equation} 
			where 
			\begin{align} \label{eqn:H_hat_hat}
				\mathbf{H}_{\mathcal{P}_t, \mathcal{Q}_t} &= -
				\sum_{m = 0}^{M_t-1}
				\begin{pmatrix}
					\hat{D}^{p-1} Y^{(1)}_{u_m} \Sigma^{-1} \\
					\vdots \\
					\hat{D}^{p-1} Y^{(d)}_{u_m} \Sigma^{-1}\\
					\vdots \\
					Y^{(1)}_{u_m} \Sigma^{-1}\\
					\vdots \\
					Y^{(d)}_{u_m} \Sigma^{-1}\\
				\end{pmatrix} \left(\hat{D}^{p-1} \mathbf{Y}^{c}_{\mathbf{A}^{(0)}, u_{m+1}} - \hat{D}^{p-1} \mathbf{Y}^{c}_{\mathbf{A}^{(0)}, u_{m}}\right),\\ \label{eqn:[H]_hat_hat}
				[\mathbf{H}]_{\mathcal{P}_t, \mathcal{Q}_t} &=
				\sum_{m = 0}^{M_t-1} \begin{pmatrix}
					\hat{D}^{p-1}Y^{(1)}_{u_m} \Sigma^{-1} \hat{D}^{p-1}Y^{(1)}_{u_m} & \cdots &  \hat{D}^{p-1}Y^{(1)}_{u_m} \Sigma^{-1} Y^{(d)}_{u_m} \\
					\vdots & \ddots & \vdots \\
					Y^{(d)}_{u_m} \Sigma^{-1} \hat{D}^{p-1}Y^{(1)}_{u_m} & \cdots &  Y^{(d)}_{u_m} \Sigma^{-1} Y^{(d)}_{u_m}\ \\
				\end{pmatrix}(u_{m+1} - u_m),
			\end{align}
			are the approximated integrals over the partition $\mathcal{Q}_t$ with the derivatives replaced by the finite difference approximations, $M_t := \max\{m\in\{0,\ldots, N'_t\}:u_{m} \leq s_{N_t-p}\}$ (i.e.\ need $\hat{D}^{p-1}\mathbf{Y}_{u_{M_t}}$ to be defined) and
			\begin{equation} \label{eqn:approx_cont_mart_part}
				\hat{D}^{p-1}\mathbf{Y}^c_{\mathbf{A}^{(0)}, u_m} = \hat{D}^{p-1}\mathbf{Y}_{u_m} - \hat{D}^{p-1}\mathbf{Y}_0 - \mathbf{b} u_m - \mathbf{J}_{u_m} - \mathbf{M}_{u_m}, \quad m = 0,\ldots, M_t.
			\end{equation}
			The estimator is 
			\begin{itemize}
				\item (weakly) consistent, i.e.\ 
				\[ \hat{\mathbf{A}}(\mathbb{Y}_{\mathcal{P}_t}, \mathbb{J}_{\mathcal{Q}_t}, \mathbb{M}_{\mathcal{Q}_t}) \overset{\mathbb{P}_{ \mathbf{A}^* }}{\longrightarrow} \mathbf{A}^*, \ t\rightarrow \infty; \] 
				\item asymptotically normal, i.e.\ 
				\[ \sqrt{t} \left(\mathrm{vec}(\hat{\mathbf{A}}(\mathbb{Y}_{\mathcal{P}_t}, \mathbb{J}_{\mathcal{Q}_t}, \mathbb{M}_{\mathcal{Q}_t})) - \mathrm{vec}(\mathbf{A}^*)\right) \overset{\mathcal{L}}{\longrightarrow} \mathbf{Z} \sim N(\mathbf{0}, \mathcal{H}^{-1}_\infty), \ t \rightarrow \infty, \]
				where $\mathcal{H}_\infty\in\mathcal{M}_{pd^2}(\mathbb{R})$ is a symmetric positive definite matrix such that $t^{-1}[\mathbf{H}]_{\mathcal{P}_t, \mathcal{Q}_t} \overset{\mathbb{P}_{ \mathbf{A}^*}}{\longrightarrow} \mathcal{H}_\infty, \ t\rightarrow \infty$.
			\end{itemize}
		\end{theorem}
		\begin{proof}
			We again apply Lemma \ref{lemma:approx}, this time with $\hat{\mathbf{A}}_{1,t} = \hat{\mathbf{A}}(\mathbb{Y}_{\mathcal{Q}_t}, \ldots, D^{p-1}\mathbb{Y}_{\mathcal{Q}_t},\mathbb{J}_{\mathcal{Q}_t}, \mathbb{M}_{\mathcal{Q}_t})$ and $\hat{\mathbf{A}}_{2,t} = \hat{\mathbf{A}}(\mathbb{Y}_{\mathcal{P}_t}, \mathbb{J}_{\mathcal{Q}_t}, \mathbb{M}_{\mathcal{Q}_t})$. Consistency and asymptotic normality of $\hat{\mathbf{A}}_{1,t}$ are given by Theorem \ref{thm:cons_asymp_first_approx} under Assumption \ref{ass:levy_process} on the driving \Levy process and Assumption \ref{ass:HF_sampling} on $\{\mathcal{Q}_t,\ t\in\mathcal{T}\}$. As before, it thus remains to prove equations \eqref{eqn:H_limit} and \eqref{eqn:[H]_limit}, for which we rely on the result of Lemma \ref{lemma:D-D_L2_bound}. See Appendix \ref{app:proof_them_cons_asymp_second_approx} for full details.
		\end{proof}
		
		\subsection{Thresholding the jumps} \label{sec:thresholding}
		Finally, we assume that we have access to the sampled MCAR($p$) process $\mathbb{Y}_{\mathcal{P}_t}$ only, the most realistic assumption. To remove the dependency on the jump processes $\mathbb{J}_{\mathcal{P}_t}$ and $\mathbb{M}_{\mathcal{P}_t}$ we will use thresholding techniques to disentangle the continuous and jump components of $D^{p-1}\mathbb{Y}_{[0,t]}$. For ease of notation let us write the increments of the partition $\mathcal{Q}_t$ with mesh $\Delta_{\mathcal{Q}_t}$ as
		\[\Delta_{\mathcal{Q}_t}^m := u_{m+1} - u_m, \quad m=0,\ldots, M_t-1,\]
		and the increment of a process $\mathbb{Z}$ over the interval $[u_m, u_{m+1}] \in\mathcal{Q}_t$ as
		\[\Delta_{\mathcal{Q}_t}^m \mathbf{Z} := (\mathbf{Z}_{u_{m+1}} - \mathbf{Z}_{u_m}), \quad m=0,\ldots, M_t-1. \]
		In the following, we aim to approximate the increments of the process given in Equation \eqref{eqn:approx_cont_mart_part}, i.e.\
		\begin{align*}
			\Delta_{\mathcal{Q}_t}^m \hat{D}^{p-1}\mathbf{Y}^c_{\mathbf{A}^{(0)}} = \Delta_{\mathcal{Q}_t}^m \hat{D}^{p-1}\mathbf{Y} - \mathbf{b} \Delta_{\mathcal{Q}_t}^m - \Delta_{\mathcal{Q}_t}^m \mathbf{J} - \Delta_{\mathcal{Q}_t}^m\mathbf{M}, \quad m=0,\ldots, M_t-1,
		\end{align*}
		by the thresholded increments
		\begin{align} \label{eqn:thresholded_incr}
			\Delta^{m}_{\mathcal{Q}_t, \boldsymbol{\nu}^m_t} \hat{D}^{p-1}\mathbf{Y}^c_{\mathbf{A}^{(0)}} := \left[ \Delta_{\mathcal{Q}_t}^m \hat{D}^{p-1}\mathbf{Y} - \mathbf{b} \Delta_{\mathcal{Q}_t}^m \right] \odot \mathds{1}_{\left\{\left|\Delta_{\mathcal{Q}_t}^m \hat{D}^{p-1}\mathbf{Y} - \mathbf{b} \Delta_{\mathcal{Q}_t}^m\right| \leq \boldsymbol{\nu}^m_t\right\}} , \quad m=0,\ldots, M_t-1,
		\end{align}
		where $\boldsymbol{\nu}^m_t\in\mathbb{R}^d$ is a thresholding vector, $\mathds{1}_{\{|\mathbf{x}|\leq \mathbf{y}\}}$ for $\mathbf{y}\in\mathbb{R}^d$ denotes the vector with $i$-th entry given by $\mathds{1}_{\{|x_i|\leq y_i\}}$ and $\odot$ denotes element-wise product. The motivation behind the thresholding approach is to interpret the (de-trended) increments 
		\[\Delta_{\mathcal{Q}_t}^m \hat{D}^{p-1}Y^{(i)} - {b}^{(i)} \Delta_{\mathcal{Q}_t}^m,\]
		greater than the threshold $\nu^{(i),m}_t$ to be mainly driven by the jump part of $D^{p-1}Y^{(i)}$ and those smaller than $\nu^{(i), m}_t$ to be mainly due to the continuous part of $D^{p-1}Y^{(i)}$. By an appropriate choice of scaling for $\{(\boldsymbol{\nu}^m_t)_{m=0}^{M_t-1},\ t\in\mathcal{T}\}$ this procedure aims to disentangle the continuous and the jump components of the process as $t\rightarrow\infty$. We proceed in a similar way to \citet{mai_OU} and \citet{Courgeau_Veraart_2022} where the thresholding technique is used to develop an estimator for the drift of a (multivariate) OU process. We treat the cases where the driving \Levy process has finite and infinite jump activity separately.
		
		\begin{remark}
			Note that, while in the OU setting \citet{mai_OU} and \citet{Courgeau_Veraart_2022} could apply thresholding to $\mathbb{Y}$ directly, here we can only threshold the \textit{approximated} process $\hat{D}^{p-1}\mathbb{Y}$. Moreover, in the following, we allow for the slightly more general case of a (known) non-zero drift $\mathbf{b}$ in the driving \Levy process.
		\end{remark}
		
		In this section, we thus define the discretized estimator resulting from the combination of Riemann sum, finite difference, and thresholding approximations by 
		\begin{equation} \label{eqn:discr_estimator_A_3}
			\hat{\mathbf{A}}(\mathbb{Y}_{\mathcal{P}_t}; \mathcal{Q}_t, \boldsymbol{\nu}_t) := \mathrm{vec}^{-1}\left([\mathbf{H}]^{-1}_{\mathcal{P}_t, \mathcal{Q}_t} \mathbf{H}_{\mathcal{P}_t, \mathcal{Q}_t, \boldsymbol{\nu}_t} \right),\end{equation}
		where $[\mathbf{H}]^{-1}_{\mathcal{P}_t, \mathcal{Q}_t}$ is given by Equation \eqref{eqn:[H]_hat_hat} and
		\begin{align} \label{eqn:H_hat_hat_hat}
			\mathbf{H}_{\mathcal{P}_t, \mathcal{Q}_t, \boldsymbol{\nu}_t} &= -
			\sum_{m = 0}^{M_t-1}
			\begin{pmatrix}
				\hat{D}^{p-1} Y^{(1)}_{u_m} \Sigma^{-1} \\
				\vdots \\
				\hat{D}^{p-1} Y^{(d)}_{u_m} \Sigma^{-1}\\
				\vdots \\
				Y^{(1)}_{u_m} \Sigma^{-1}\\
				\vdots \\
				Y^{(d)}_{u_m} \Sigma^{-1}\\
			\end{pmatrix} \Delta^m_{\mathcal{Q}_t,\boldsymbol{\nu}_t} \hat{D}^{p-1}\mathbf{Y}^c_{\mathbf{A}^{(0)}}.
		\end{align}
		
		\begin{remark} \label{rem:approx_only_H}
			To show the asymptotic properties are preserved under appropriate conditions, we will apply Lemma \ref{lemma:approx} one last time with $\hat{\mathbf{A}}_{1,t} = \hat{\mathbf{A}}(\mathbb{Y}_{\mathcal{P}_t}, \mathbb{J}_{\mathcal{Q}_t}, \mathbb{M}_{\mathcal{Q}_t})$ and $\hat{\mathbf{A}}_{2,t} = \hat{\mathbf{A}}(\mathbb{Y}_{\mathcal{P}_t}; \mathcal{Q}_t, \boldsymbol{\nu}_t)$. In this case, it will be sufficient to show that Equation \eqref{eqn:H_limit} holds, as the approximator of the quadratic variation $[\mathbf{H}]_{\mathcal{P}_t, \mathcal{Q}_t}$ does not change and hence \eqref{eqn:[H]_limit} trivially holds.
		\end{remark}

		\subsubsection{Finite jump activity}
		If we assume that the driving \Levy process has finite jump activity then by Remark \ref{rem:finite_activity} we can write
		\begin{equation}
			\Delta_{\mathcal{Q}_t}^m D^{p-1}\mathbf{Y}_{\mathbf{A}^{(0)}}^c = \Delta_{\mathcal{Q}_t}^m D^{p-1}\mathbf{Y} - \tilde{\mathbf{b}} \Delta_{\mathcal{Q}_t}^m - \Delta_{\mathcal{Q}_t}^m \tilde{\mathbf{J}}, \label{eqn:cont_mart_finite_act}
		\end{equation}
		where $\tilde{\mathbb{J}}$ is the compound Poisson process given in Equation \eqref{eqn:jump_J_tilde} with elementwise Poisson counting processes $\mathbf{N}_t = (N_t^{(1)},\ldots, N_t^{(d)})^{\mathrm{T}}$. We denote by $\lambda^{(i)}$ the jump rate of the process $\{N_t^{(i)},\ t\geq 0\}$ and by $\tilde{F}^{(i)}(\cdot)$ the probability distribution of jump sizes in the $i$-th component, i.e.\ 
		$F^{(i)}(\cdot) = \lambda^{(i)}\tilde{F}^{(i)}(\cdot), $
		where $F^{(i)}(\cdot)$ is the $i$-th marginal of the \Levy measure $F$. We impose the following joint condition on the sequence of partitions $\{\mathcal{Q}_t,\ t\in\mathcal{T}\}$ and the thresholding sequences $\{(\boldsymbol{\nu}^m_t)_{m=0}^{M_t-1},\ t\in\mathcal{T}\}$:
		\begin{assumption} \label{ass:finite_thresholding}
			The sequence of partitions $\{\mathcal{Q}_t,\ t\in\mathcal{T}\}$ and the sequence of jump thresholds $\{(\boldsymbol{\nu}^m_t)_{m=0}^{M_t-1},\ t\in\mathcal{T}\}$ satisfy for $i\in\{1,\ldots, d\}$
			\begin{enumerate}[label=(\roman*)]
				\item $\nu_t^{(i),m} = \left(\Delta^m_{\mathcal{Q}_t}\right)^{\beta^{(i)}}$ with $\beta^{(i)}\in(0,\, 1/2)$;
				\item $\displaystyle t \Delta^{1 - 2 \beta^{(i)}}_{\mathcal{Q}_t} \rightarrow 0, \ t\rightarrow \infty$;
				\item $\displaystyle t \tilde{F}^{(i)}\left(\left(-2\Delta^{\beta^{(i)}}_{\mathcal{Q}_t},\  2\Delta^{\beta^{(i)}}_{\mathcal{Q}_t}\right)\right) \rightarrow 0, \ t\rightarrow \infty.$
			\end{enumerate}
		\end{assumption}
		\begin{remark} \label{rem:ass_finite_activity}
			Note that $\beta^{(i)}\in(0,\,1/2)$ relates the strength of the thresholding, via Assumption \ref{ass:finite_thresholding}$.(i)$, and the speed at which the mesh $\Delta_{\mathcal{Q}_t}$ vanishes, via Assumptions \ref{ass:finite_thresholding}$.(ii)$ and \ref{ass:finite_thresholding}$.(iii)$.
			\begin{itemize}
				\item Since $1-2\beta^{(i)}\in(0,\,1)$ Assumption \ref{ass:finite_thresholding}$.(ii)$ is strictly stronger than the high-frequency sampling condition for $\mathcal{Q}_t$, i.e.\ Assumption \ref{ass:HF_sampling}.
				\item Intuitively, Assumption \ref{ass:finite_thresholding}$.(ii)$ controls the rate at which the probability of observing a continuous increment bigger than the threshold vanishes, while Assumption \ref{ass:finite_thresholding}$.(iii)$ controls the rates at which the probability of observing a jump increment smaller than (twice) the threshold vanishes. The ``optimal'' $\beta^{(i)}$ is thus chosen by appropriately balancing the two conditions. For example, if $\tilde{F}^{(i)}(\cdot)$ has a bounded Lebesgue density, Assumption \ref{ass:finite_thresholding}$.(iii)$ becomes \[t \Delta^{\beta^{(i)}}_{\mathcal{Q}_t}\rightarrow 0, \ t \rightarrow \infty,\] and we see the optimal trade-off is given by $\beta^{(i), *} = 1/3$. 
				\item From a theoretical point of view, given a sequence of partitions $\{\mathcal{Q}_t,\ t\in\mathcal{T}\}$, one can choose $\beta^{(i)}$ depending on the driving \Levy process such that Assumption \ref{ass:finite_thresholding}$.(ii)$ and \ref{ass:finite_thresholding}$.(iii)$ are ``optimally'' fulfilled. In practice, given the observation partition $\mathcal{Q}_t$, one chooses $\beta^{(i)}$ based on the data, we will discuss this in more detail in Section \ref{sec:practical_considerations}.
			\end{itemize}
		\end{remark}
		
		In order to understand the choice of thresholding sequence, we start sketching the main idea of the proof. Let us focus on dimension $i\in\{1,\ldots,d\}$ and assume for the moment that we can perfectly observe an increment of $D^{p-1}Y^{(i)}$ over $[u_m, u_{m+1}]$, i.e.\ the finite difference approximation is close enough to the target value. In this case, one can bound the probability that the thresholding technique fails by considering whether
		\begin{itemize}
			\item an increment is discarded due to a large move in the continuous component, or 
			\item an increment is kept in the presence of a jump.
		\end{itemize} 
		To do so one requires bounds on 
		\begin{itemize}
			\item the probability of continuous increments larger than the threshold $\nu^{(i), m}_t = (\Delta^m_{\mathcal{Q}_t})^{\beta^{(i)}}$: This is controlled for $\beta^{(i)}\in(0,1/2)$ by Assumption \ref{ass:finite_thresholding}$.(i)$ and \ref{ass:finite_thresholding}$.(iii)$,
			\item the probability of jumps smaller than $2(\Delta^m_{\mathcal{Q}_t})^{\beta^{(i)}}$: This is controlled by Assumption \ref{ass:finite_thresholding}$.(ii)$.
		\end{itemize}
		One can then show that the probability of the thresholding technique being successful approaches one and, on this set, the $L^1(\Omega, \mathcal{F}, \mathbb{P}_{\mathbf{A}^*})$ distance between the original and the thresholded approximator, i.e.\ $\mathbf{H}_{\mathcal{P}_t, \mathcal{Q}_t}$ and $\mathbf{H}_{\mathcal{P}_t, \mathcal{Q}_t, \boldsymbol{\nu}_t}$, approaches zero. The following Theorem, along with Theorem \ref{thm:cons_asymp_third_approx_infinite_activity} in the infinite activity setting, states the asymptotic properties of the only feasible, i.e.\ depending on data $\mathbb{Y}_{\mathcal{P}_t}$ only, estimator \eqref{eqn:discr_estimator_A_3} for $\mathbf{A}^*\in\mathfrak{A}$, the main results of this section. 
		
		\begin{theorem}               \label{thm:cons_asymp_third_approx_finite_activity}
			Let $\mathbb{Y}_{\mathcal{P}_t}$ denote the sampled observation of a stationary and ergodic MCAR($p$) process with parameters $\mathbf{A}^*\in\mathfrak{A}$ under $\mathbb{P}_{\mathbf{A}^*}$ such that the driving \Levy process satisfies Assumption \ref{ass:levy_process} and has finite jump activity. Let $\{\mathcal{P}_t,\ t\in\mathcal{T}\}$ be a countable sequence of partitions satisfying Assumption \ref{ass:evenly_spaced_lim} with refinements $\{\mathcal{Q}_t,\ t\in\mathcal{T}\}$ satisfying Assumption \ref{ass:HF_sampling}, Assumption \ref{ass:controlled_sampling} and Assumption \ref{ass:joint_mesh}. Furthermore, assume the sequence of partitions $\{\mathcal{Q}_t,\ t\in\mathcal{T}\}$ and the thresholding sequences $\{(\boldsymbol{\nu}^m_t)_{m=0}^{M_t-1},\ t\in\mathcal{T}\}$ satisfy Assumption \ref{ass:finite_thresholding}.
			Then the estimator $\hat{\mathbf{A}}(\mathbb{Y}_{\mathcal{P}_t}; \mathcal{Q}_t, \boldsymbol{\nu}_t)$ defined in Equation \eqref{eqn:discr_estimator_A_3} is 
			\begin{itemize}
				\item (weakly) consistent, i.e.\ 
				\[ \hat{\mathbf{A}}(\mathbb{Y}_{\mathcal{P}_t}; \mathcal{Q}_t, \boldsymbol{\nu}_t) \overset{\mathbb{P}_{ \mathbf{A}^* }}{\longrightarrow} \mathbf{A}^*, \ t\rightarrow \infty; \] 
				\item asymptotically normal, i.e.\ 
				\[ \sqrt{t} \left(\mathrm{vec}(\hat{\mathbf{A}}(\mathbb{Y}_{\mathcal{P}_t}; \mathcal{Q}_t, \boldsymbol{\nu}_t)) - \mathrm{vec}(\mathbf{A}^*)\right) \overset{\mathcal{L}}{\longrightarrow} \mathbf{Z} \sim N(\mathbf{0}, \mathcal{H}^{-1}_\infty), \ t \rightarrow \infty, \]
				where $\mathcal{H}_\infty\in\mathcal{M}_{pd^2}(\mathbb{R})$ is a symmetric positive definite matrix such that $t^{-1}[\mathbf{H}]_{\mathcal{P}_t, \mathcal{Q}_t} \overset{\mathbb{P}_{ \mathbf{A}^*}}{\longrightarrow} \mathcal{H}_\infty, \ t\rightarrow \infty$.
			\end{itemize}
		\end{theorem}
		\begin{proof}
			We apply Lemma \ref{lemma:approx} with $\hat{\mathbf{A}}_{1,t} = \hat{\mathbf{A}}(\mathbb{Y}_{\mathcal{P}_t}, \mathbb{J}_{\mathcal{Q}_t}, \mathbb{M}_{\mathcal{Q}_t})$ and $\hat{\mathbf{A}}_{2,t} = \hat{\mathbf{A}}(\mathbb{Y}_{\mathcal{P}_t}; \mathcal{Q}_t, \boldsymbol{\nu}_t)$. Consistency and asymptotic normality of $\hat{\mathbf{A}}_{1,t}$ are given by Theorem \ref{thm:cons_asymp_second_approx}. By Remark \ref{rem:approx_only_H} it suffices to show Equation \eqref{eqn:H_limit}. The arguments of the proof are sketched above, see Appendix \ref{app:proof_thm_cons_asymp_third_approx_finite_activity} for full details.
		\end{proof}		
		
		\subsubsection{Infinite jump activity}
		We now consider the more general case, where the driving \Levy process $\mathbb{L}$ might have infinite jump activity, i.e.\ $F(\mathbb{R}^d)=\infty$. In this case, by Equation \eqref{eqn:integrator_process}, we can write the continuous martingale increments as
		\begin{equation}
			\Delta_{\mathcal{Q}_t}^m D^{p-1}\mathbf{Y}_{\mathbf{A}^{(0)}}^c = \Delta_{\mathcal{Q}_t}^m D^{p-1}\mathbf{Y} - \mathbf{b}\Delta_{\mathcal{Q}_t}^m - \Delta_{\mathcal{Q}_t}^m \mathbf{J} - \Delta_{\mathcal{Q}_t}^m \mathbf{M}, \label{eqn:cont_mart_infinite_act}
		\end{equation}
		where $\mathbb{J}$ is the finite activity pure jump \Levy process with jumps greater than one given in Equation \eqref{eqn:jump_J} and $\mathbb{M}$ is the pure jump martingale with jumps smaller or equal to one given in Equation \eqref{eqn:jump_M}. Note that $\mathbb{J}$ is a compound Poisson process with elementwise Poisson counting processes $\mathbf{N}_t = (N_t^{(1)},\ldots, N_t^{(d)})^{\mathrm{T}}$. We denote by $\lambda^{(i)}$ the jump rate of the process $\{N_t^{(i)},\ t\geq 0\}$ and by $\tilde{F}^{(i)}(\cdot)$ the probability distribution of jump sizes in the $i$-th component of $\mathbb{J}$, i.e.\ the $i$-th marginal of $\tilde{F}:= F|_{\{\mathbf{x}:\|\mathbf{x}\|>1\}} / F(\{\mathbf{x}:\|\mathbf{x}\|>1\})$. We impose the following joint condition on the sequence of partitions $\{\mathcal{Q}_t,\ t\in\mathcal{T}\}$ and the thresholding sequence $\{(\boldsymbol{\nu}^m_t)_{m=0}^{M_t-1},\ t\in\mathcal{T}\}$:
		\begin{assumption} \label{ass:infinite_thresholding}
			The sequence of partitions $\{\mathcal{Q}_t,\ t\in\mathcal{T}\}$ and the sequence of jump thresholds $\{(\boldsymbol{\nu}^m_t)_{m=0}^{M_t-1},\ t\in\mathcal{T}\}$ satisfy for $i\in\{1,\ldots, d\}$ 
			\begin{enumerate}[label=(\roman*)]
				\item $\nu^{(i), m}_t = (\Delta^m_{\mathcal{Q}_t})^{\beta^{(i)}}$ with $\beta^{(i)}\in(0,\, 1/4)$;
				\item $\displaystyle t \Delta^{1 - 4 \beta^{(i)}}_{\mathcal{Q}_t} \rightarrow 0, \ t\rightarrow \infty$;
				\item $t\tilde{F}^{(i)}\left(\left(-4\Delta_{\mathcal{Q}_t}^{\beta^{(i)}}, 4\Delta_{\mathcal{Q}_t}^{\beta^{(i)}}\right)\right)\rightarrow 0,\ t\rightarrow\infty$ and 
				$\exists \epsilon^{(i)} >0$ such that $\displaystyle t \Delta^{2\epsilon^{(i)}}_{\mathcal{Q}_t} \rightarrow 0, \ t\rightarrow \infty$ and
				\[\mathbb{E}_{\mathbf{A}^*}\left[ |{M}^{(i)}_s| \mathds{1}_{\left\{|{M}^{(i)}_s| \leq s^{\beta^{(i)}}\right\}}\right] = O\left(s^{1+\epsilon^{(i)}}\right) \ \mathrm{as}\  s\downarrow 0.\]
			\end{enumerate}
			Moreover, the driving \Levy process $\mathbb{L}$ has finite fourth-moment, i.e.\ $\mathbb{E}[\|\mathbf{L}_1\|^4]<\infty$.
		\end{assumption}
		\begin{remark}
			Some important remarks are in place. If we compare these assumptions to those considered in the finite jump activity case, i.e.\ Assumption \ref{ass:finite_thresholding}, we note that:
			\begin{itemize}
				\item In Assumption \ref{ass:infinite_thresholding}$.(i)$ the upper bound of allowed $\beta^{(i)}$'s is lower. This implies the thresholding cannot be too strong, i.e.\ enough observations need to be kept to identify the parameters. A choice of $\beta^{(i)}\in(1/4,1/2)$ would void Assumption \ref{ass:infinite_thresholding}$.(ii)$ which is required to control the rate at which the probability of observing increments of $\mathbb{M}$ bigger than the threshold vanishes.
				\item On the other hand, Assumption \ref{ass:infinite_thresholding}$.(iii)$ implicitly requires additional control on $\beta^{(i)}$ to ensure the jump component of the kept increments is negligible. The control on the rate at which the probability of observing a jump increment due to the compound Poisson process $\mathbb{J}$ smaller than (four times) the threshold vanishes -- also present in Assumption \ref{ass:finite_thresholding}$.(iii)$ -- is complemented by a condition controlling the magnitude of the increments due to the infinite activity jump part $\mathbb{M}$.  We note that, in the second part of Assumption \ref{ass:infinite_thresholding}$.(iii)$, the existence and choice of $\epsilon^{(i)}$ depends on both $\mathbb{M}$ and $\beta^{(i)}$. 
				\item As in Assumption \ref{ass:finite_thresholding} the ``optimal'' $\beta^{(i)}$ is chosen by appropriately balancing  conditions $(ii)$ and $(iii)$. For example, if $\tilde{F}^{(i)}(\cdot)$ has a bounded Lebesgue density and we write $\epsilon^{(i)} = f\left(\beta^{(i)}\right)$, the optimal trade-off is given by \[\beta^{*, (i)} = \underset{\beta^{(i)}\in(0,1/4)}{\mathrm{argmax}} \min\left\{1 - 4 \beta^{(i)}, \beta^{(i)}, 2f\left(\beta^{(i)}\right)\right\},\] i.e.\ by balancing the rates $t \Delta^{1 - 4 \beta^{(i)}}_{\mathcal{Q}_t}, t \Delta^{ \beta^{(i)}}_{\mathcal{Q}_t}, t \Delta^{2f\left(\beta^{(i)}\right)}_{\mathcal{Q}_t}$. In practice, as for the finite activity case, we usually choose the thresholding sequence from data as discussed in Section \ref{sec:practical_considerations}.
				\item When the jump components of $\mathbb{L}$ are independent, the first part of Assumption \ref{ass:infinite_thresholding}$.(iii)$ is trivially satisfied since $\tilde{F}^{(i)}([-1,1]) = 0$. This holds trivially when $d=1$.
			\end{itemize}
			We note that our assumptions differ from those imposed in \citet{mai_OU} to deduce the asymptotic properties of the drift estimator for one-dimensional \Levydriven OU processes. 
			% We could not reproduce the arguments which were made in that paper and thus had to consider a slightly different set of assumptions. We stress the differences in the assumptions are not due to the more general set of processes under consideration, i.e.\ multidimensional \Levydriven CAR processes, but to the irreproducibility of the proofs in \citet{mai_OU}. In particular, we could not understand the argument at 
			We had to modify the assumptions after noticing some inconsistencies in the arguments applied at the bottom of \citet[page~944]{mai_OU} when splitting the proof of \citet[Equation~(35)]{mai_OU} into two steps denoted by $(i)$ and $(ii)$ therein. Similar arguments were repeatedly applied in the paper (and ceteris paribus extended in \cite{Courgeau_Veraart_2022} to the multivariate setting). Moreover, we wish to underline the following additional differences.
			\begin{itemize}
				\item We will require finiteness of 4-th moments of the driving \Levy process. This was explicitly imposed in \cite{Courgeau_Veraart_2022}, but was a tacit assumption in \cite{mai_OU}.
				\item The assumptions in \citet{mai_OU} include an additional symmetric condition on $\mathbb{M}$. As discussed in  Appendix \ref{app:infinite_activity_symmetric}, including this additional condition allows to ``weaken'' Assumption \ref{ass:infinite_thresholding}$.(iii)$.
			\end{itemize}
			Following multiple numerical experiments we believe the assumptions considered here might be stronger than required but, as illustrated in Example \ref{example:Gamma_process}, they are not too restrictive.
		\end{remark}
		It is important to explore the attainability of Assumption \ref{ass:infinite_thresholding}$.(iii)$, to ensure it is not a void condition. Example \ref{example:Gamma_process} considers the case where the driving \Levy process is a Gamma process.
		\begin{example} \label{example:Gamma_process}
			Note that if $\{M^{(i)}_t,\ t\geq 0\}$ has a transition probability density $p(s, x, y)$ then
			\[\mathbb{E}_{\mathbf{A}^*}\left[ |{M}^{(i)}_s| \mathds{1}_{\left\{|{M}^{(i)}_s| \leq s^{\beta^{(i)}}\right\}}\right] = \int_{-s^{\beta^{(i)}}}^{s^{\beta^{(i)}}} |x| p(s, 0, x)\, dx.\]
			If $M^{(i)}$ is a standard Gamma process, i.e.\ shape $k=1$ and scale $\theta=1$, then $p(s, 0, x)= \frac{1}{\Gamma(s)} x^{s-1}e^{-x}, x> 0$. This is an infinite activity process with intensity measure \[F(dx)=x^{-1}e^{-x}\mathds{1}_{(0,\infty)} (x) dx.\] 
			It has infinite activity since
			\[F((0,\infty)) = \int_0^\infty x^{-1}e^{-x} dx \geq e^{-1}\int_0^1 x^{-1} dx = e^{-1} [\ln(1) - \ln(0)] = \infty,\]
			and finite moments for any $n\in\mathbb{N}$
			\[\mathbb{E}_{\mathbf{A}^*}\left[ |{M}^{(i)}_s|^n\right] = \int_0^\infty \frac{1}{\Gamma(s)}x^{s+n-1} e^{-x} \, dx = \frac{\Gamma(s + n )}{\Gamma(s)}.\]
			We have
			\begin{align*}
				\mathbb{E}_{\mathbf{A}^*}\left[ |{M}^{(i)}_s| \mathds{1}_{\left\{|{M}^{(i)}_s| \leq s^{\beta^{(i)}}\right\}}\right] &= \int_{0}^{s^{\beta^{(i)}}} \frac{1}{\Gamma(s)}  x^{s}e^{-x}\, dx \\
				&\leq \frac{1}{\Gamma(s)} \int_{0}^{s^{\beta^{(i)}}} x^{s}\, dx =  \frac{1}{\Gamma(s)} \frac{s^{\beta^{(i)}(s+1)}}{s+1} \leq s^{1+\beta^{(i)}},
			\end{align*}
			using $\Gamma(s)\geq s^{-1}$ for $s>0$. Thus if $\beta^{(i)}\in(0,\, 1/4)$, we can set $\epsilon^{(i)}=\beta^{(i)}$ to obtain
			\[\mathbb{E}_{\mathbf{A}^*}\left[ |{M}^{(i)}_s| \mathds{1}_{\left\{|{M}^{(i)}_s| \leq s^{\beta^{(i)}}\right\}}\right] =  O\left(s^{1+ \epsilon^{(i)}}\right),\]
			i.e.\ Assumption \ref{ass:infinite_thresholding}$.(iii)$ is satisfied. In this case the optimal thresholding power is given by $\beta^{(i), *} = 1/5$. Similar arguments can be carried out for non-standard (a-)symmetric Gamma drivers.
		\end{example}
		We can now state the consistency and asymptotic normality result for the high-frequency estimator \eqref{eqn:discr_estimator_A_3} in the infinite-activity setting.
		
		\begin{theorem}               \label{thm:cons_asymp_third_approx_infinite_activity}
			Let $\mathbb{Y}_{\mathcal{P}_t}$ denote the sampled observation of a stationary and ergodic MCAR($p$) process with parameters $\mathbf{A}^*\in\mathfrak{A}$ under $\mathbb{P}_{\mathbf{A}^*}$ such that the driving \Levy process satisfies Assumption \ref{ass:levy_process}. Let $\{\mathcal{P}_t,\ t\in\mathcal{T}\}$ be a countable sequence of partitions satisfying Assumption \ref{ass:evenly_spaced_lim} with refinements $\{\mathcal{Q}_t,\ t\in\mathcal{T}\}$ satisfying Assumption \ref{ass:HF_sampling}, Assumption \ref{ass:controlled_sampling} and Assumption \ref{ass:joint_mesh}. Furthermore, assume the sequence of partitions $\{\mathcal{Q}_t,\ t\in\mathcal{T}\}$ and the thresholding sequences $\{(\boldsymbol{\nu}^m_t)_{m=0}^{M_t-1},\ t\in\mathcal{T}\}$ satisfy Assumption \ref{ass:infinite_thresholding}.
			Then the estimator $\hat{\mathbf{A}}(\mathbb{Y}_{\mathcal{P}_t}; \mathcal{Q}_t, \boldsymbol{\nu}_t)$ defined in Equation \eqref{eqn:discr_estimator_A_3} is 
			\begin{itemize}
				\item (weakly) consistent, i.e.\ 
				\[ \hat{\mathbf{A}}(\mathbb{Y}_{\mathcal{P}_t}; \mathcal{Q}_t, \boldsymbol{\nu}_t) \overset{\mathbb{P}_{ \mathbf{A}^* }}{\longrightarrow} \mathbf{A}^*, \ t\rightarrow \infty; \] 
				\item asymptotically normal, i.e.\ 
				\[ \sqrt{t} \left(\mathrm{vec}(\hat{\mathbf{A}}(\mathbb{Y}_{\mathcal{P}_t}; \mathcal{Q}_t, \boldsymbol{\nu}_t)) - \mathrm{vec}(\mathbf{A}^*)\right) \overset{\mathcal{L}}{\longrightarrow} \mathbf{Z} \sim N(\mathbf{0}, \mathcal{H}^{-1}_\infty), \ t \rightarrow \infty, \]
				where $\mathcal{H}_\infty\in\mathcal{M}_{pd^2}(\mathbb{R})$ is a symmetric positive definite matrix such that $t^{-1}[\mathbf{H}]_{\mathcal{P}_t, \mathcal{Q}_t} \overset{\mathbb{P}_{ \mathbf{A}^*}}{\longrightarrow} \mathcal{H}_\infty, \ t\rightarrow \infty$.
			\end{itemize}
		\end{theorem}
		\begin{proof}
			We proceed exactly as in the proof of Theorem \ref{thm:cons_asymp_third_approx_finite_activity} by applying Lemma \ref{lemma:approx} with $\hat{\mathbf{A}}_{1,t} = \hat{\mathbf{A}}(\mathbb{Y}_{\mathcal{P}_t}, \mathbb{J}_{\mathcal{Q}_t}, \mathbb{M}_{\mathcal{Q}_t})$ and $\hat{\mathbf{A}}_{2,t} = \hat{\mathbf{A}}(\mathbb{Y}_{\mathcal{P}_t}; \mathcal{Q}_t, \boldsymbol{\nu}_t)$. Consistency and asymptotic normality of $\hat{\mathbf{A}}_{1,t}$ are given by Theorem \ref{thm:cons_asymp_second_approx} and thus, by Remark \ref{rem:approx_only_H}, it suffices to show Equation \eqref{eqn:H_limit}. The computation requires splitting the limiting variable in multiple components and showing each one converges to zero in $\mathbb{P}_{\mathbf{A}^*}$. See Appendix \ref{app:proof_thm_cons_asymp_third_approx_infinite_activity} for full details.
		\end{proof}
		
		Table \ref{tab:assumptions} summarizes how the assumptions of Theorem \ref{thm:cons_asymp_third_approx_finite_activity} and Theorem \ref{thm:cons_asymp_third_approx_infinite_activity} relate the driving \Levy process $\mathbb{L}$, the sequences of observation partitions $\{\mathcal{Q}_t,\ t\in\mathcal{T}\}$ and $\{\mathcal{P}_t,\ t\in\mathcal{T}\}$, and the threshold powers $\beta^{(i)},\ i=1,\ldots,d$.
		
		% \begin{table}[htbp]
			%   \centering
			%     \begin{tabular}{|C{2cm}|C{1.5cm}C{1.5cm}C{1.5cm}C{1.5cm}C{1.5cm}C{1.5cm}|}
				%     \hline
				%           & \multicolumn{1}{C{1.5cm}|}{Ass. \ref{ass:levy_process}} & \multicolumn{1}{C{1.5cm}|}{Ass. \ref{ass:HF_sampling}} & \multicolumn{1}{C{1.5cm}|}{Ass. \ref{ass:evenly_spaced_lim}} & \multicolumn{1}{C{1.5cm}|}{Ass. \ref{ass:controlled_sampling}} & \multicolumn{1}{C{1.5cm}|}{Ass. \ref{ass:joint_mesh}} & \multicolumn{1}{C{1.5cm}|}{Ass. \ref{ass:finite_thresholding}/\ref{ass:infinite_thresholding}} \\
				%     \hline
				%     $\mathbb{L}$ & \cellcolor[rgb]{ .125,  .216,  .392} &       &       &       &       & \cellcolor[rgb]{ .125,  .216,  .392} \\
				% \hhline{-~~~~~~|}    $\{\mathcal{Q}_t,\ t\in\mathcal{T}\}$  &       & \cellcolor[rgb]{ .125,  .216,  .392} &       & \cellcolor[rgb]{ .125,  .216,  .392} & \cellcolor[rgb]{ .125,  .216,  .392} & \cellcolor[rgb]{ .125,  .216,  .392} \\
				% \hhline{-~~~~~~|}     $\{\mathcal{P}_t,\ t\in\mathcal{T}\}$ &       &       & \cellcolor[rgb]{ .125,  .216,  .392} &       & \cellcolor[rgb]{ .125,  .216,  .392} &  \\
				% \hhline{-~~~~~~|}  $\{\boldsymbol{\nu}_t,\ t\in\mathcal{T}\}$ &       &       &       &       &       & \cellcolor[rgb]{ .125,  .216,  .392} \\
				%     \hline
				%     \end{tabular}
			%   \caption{The assumptions of Theorem \ref{thm:cons_asymp_third_approx_finite_activity} and Theorem \ref{thm:cons_asymp_third_approx_infinite_activity} .}
			%   \label{tab:assumptions}%
			% \end{table}%

            % \cellcolor[rgb]{ .125,  .216,  .392} or \Xsolid
		
		\bgroup
		\def\arraystretch{1.75}% 
		\begin{table}[H]
			\centering
			\begin{tabular}{|p{2.15cm}p{5.05cm}|C{1.5cm}C{1.85cm}C{1.85cm}C{1.5cm}|}
				\hline
				& & \multicolumn{1}{C{1.5cm}|}{$\mathbb{L}$} & \multicolumn{1}{C{1.85cm}|}{$\{\mathcal{Q}_t,\ t\in\mathcal{T}\}$} & \multicolumn{1}{C{1.85cm}|}{$\{\mathcal{P}_t,\ t\in\mathcal{T}\}$} & \multicolumn{1}{C{1.5cm}|}{$\beta^{(i)}$'s} \\
				\hline
				\multirow{2}{*}{Assumption \ref{ass:levy_process}} & $\mathbb{E}[\|\mathbf{L}_1\|^2]<\infty$ & \cellcolor[rgb]{ .125,  .216,  .392} &       &       &      \\
				\hhline{~-~~~~|}
				& $\Sigma$ positive definite & \cellcolor[rgb]{ .125,  .216,  .392} &       &       &      \\
				\hhline{--~~~~|}   
				Assumption \ref{ass:HF_sampling}  & $\Delta_{\mathcal{Q}_t} t \rightarrow 0$ &  &  \cellcolor[rgb]{ .125,  .216,  .392}   &       &       \\ 
				\hhline{--~~~~|}   
				Assumption \ref{ass:evenly_spaced_lim}  &  $c_{\mathcal{P}_t} \rightarrow 1$ ``fast enough'' &  &       &    \cellcolor[rgb]{ .125,  .216,  .392}   &      \\ 
				\hhline{--~~~~|}   
				Assumption \ref{ass:controlled_sampling} & $c_{\mathcal{Q}_t}\rightarrow c>0$  &  &     \cellcolor[rgb]{ .125,  .216,  .392}   &       &         \\ 
				\hhline{--~~~~|}   
				Assumption \ref{ass:joint_mesh}  & $\Delta_{\mathcal{P}_t} = o(t^{-1}\Delta^2_{\mathcal{Q}_t})$ &   &   \cellcolor[rgb]{ .125,  .216,  .392}    &    \cellcolor[rgb]{ .125,  .216,  .392}   &          \\ 
				\hhline{--~~~~|}   
				\multirow{2}{*}{Assumption \ref{ass:finite_thresholding}} & $t\Delta_{\mathcal{Q}_t}^{1-2\beta^{(i)}}\rightarrow 0$ &  &   \cellcolor[rgb]{ .125,  .216,  .392}    &       &    \cellcolor[rgb]{ .125,  .216,  .392}  \\
				\hhline{~-~~~~|}
				& $t\tilde{F}^{(i)}\Big(\Big(-2\Delta_{\mathcal{Q}_t}^{\beta^{(i)}}, 2\Delta_{\mathcal{Q}_t}^{\beta^{(i)}}\Big)\Big)\rightarrow 0$ & \cellcolor[rgb]{ .125,  .216,  .392} &   \cellcolor[rgb]{ .125,  .216,  .392}    &       &  \cellcolor[rgb]{ .125,  .216,  .392}    \\
				\hhline{--~~~~|}   
				\multirow{4}{*}{Assumption \ref{ass:infinite_thresholding}} & $t\Delta_{\mathcal{Q}_t}^{1-4\beta^{(i)}}\rightarrow 0$ &  &    \cellcolor[rgb]{ .125,  .216,  .392}   &       &  \cellcolor[rgb]{ .125,  .216,  .392}    \\
				\hhline{~-~~~~|}
				& $t\tilde{F}^{(i)}\Big(\Big(-4\Delta_{\mathcal{Q}_t}^{\beta^{(i)}}, 4\Delta_{\mathcal{Q}_t}^{\beta^{(i)}}\Big)\Big)\rightarrow 0$ & \cellcolor[rgb]{ .125,  .216,  .392} &    \cellcolor[rgb]{ .125,  .216,  .392}   &       &   \cellcolor[rgb]{ .125,  .216,  .392}   \\
				\hhline{~-~~~~|}
				& control on $\mathbb{M}$ in terms of $\beta^{(i)}, \Delta_{\mathcal{Q}_t}$ & \cellcolor[rgb]{ .125,  .216,  .392} &    \cellcolor[rgb]{ .125,  .216,  .392}    &       &   \cellcolor[rgb]{ .125,  .216,  .392}   \\
				\hhline{~-~~~~|}
				& $\mathbb{E}[\|\mathbf{L}_1\|^4]<\infty$ & \cellcolor[rgb]{ .125,  .216,  .392} &       &       &      \\
				\hline
			\end{tabular}
			\caption{The assumptions of Theorem \ref{thm:cons_asymp_third_approx_finite_activity} and Theorem \ref{thm:cons_asymp_third_approx_infinite_activity} where the thresholds are assumed to be of the form $\nu_t^{(i),m} = (\Delta^m_{\mathcal{Q}_t})^{\beta^{(i)}}$ for $i=1,\ldots,d$ and $m=0,\ldots. M_t-1$ for any $t\in\mathcal{T}$.}
			\label{tab:assumptions}%
		\end{table}%
		\egroup

		\subsection{Simulation study} \label{sec:simulation_study}
		In this section, we empirically check how well the asymptotic results proved in Theorem \ref{thm:cons_asymp_third_approx_finite_activity} and Theorem \ref{thm:cons_asymp_third_approx_infinite_activity} perform in finite samples by carrying out a simulation study. We consider a $d=1$ dimensional MCAR process of order $p=2$ with drift parameter $\mathbf{A}^* = (A_1, A_2) = (1, 2)$. The driving \Levy process has characteristic triplet $(0, \Sigma, F)$ with variance $\Sigma = 1$. We consider three different jump regimes: 
		\begin{itemize}
			\item[(BM)] when no jumps are present, i.e.\ the process is driven by a standard Brownian motion only;
			\item[(CP)] when the jumps follow a compound Poisson process with rate $\lambda = 1$ and $N(0,1)$ jump sizes;
			\item[($\Gamma$)] when the jumps follow a symmetric standard Gamma process, i.e.\ with shape $k=1$ and scale $\theta=1$.
		\end{itemize}
		In all three cases the \Levy driver satisfies Assumption \ref{ass:levy_process}. In the first two, the process has finite jump activity and hence we can apply Theorem \ref{thm:cons_asymp_third_approx_finite_activity}, in the last case the process has infinite activity with finite 4th moment and thus Theorem \ref{thm:cons_asymp_third_approx_infinite_activity} applies.
		
		We consider a sequence of uniform partitions $\{\mathcal{P}_t,\ t\in\mathcal{T}\}$ with step size $\Delta_{\mathcal{P}_t} = t^{-6}$ and uniform refinements $\{\mathcal{Q}_t,\ t\in\mathcal{T}\}$ with step size $\Delta_{\mathcal{Q}_t} = t^{-2}$ so that Assumption \ref{ass:HF_sampling}, Assumption \ref{ass:evenly_spaced_lim}, Assumption \ref{ass:controlled_sampling} and Assumption \ref{ass:joint_mesh} hold. Noting that all partitions are refinements of each other we simulate the processes at each point in the finest grid we will use for estimation, i.e.\ $\mathcal{P}_t$ with $t=8$. We use the exact procedure described in Appendix \ref{app:simulate_MCAR_exact} to simulate the finite activity CAR processes and the approximate Euler-Maruyama procedure described in Appendix \ref{app:simulate_MCAR_approx} to simulate the Gamma-driven CAR process. The code for simulating and estimating MCAR processes is available at \url{https://github.com/lorenzolucchese/mcar}.
		
		We assume the \Levy triplet $(0,\Sigma, F)$ is known and hence select the theoretically optimal threshold powers $\beta^{*} =-\infty, 1/3, 1/5$, i.e.\ no thresholding in the Brownian case and the values given in Remark \ref{rem:ass_finite_activity} and Example \ref{example:Gamma_process} for the finite activity and Gamma drivers. Under these choices of observation partitions and thresholding vectors Assumption \ref{ass:finite_thresholding} is satisfied by regimes (BM) and (CP) while Assumption \ref{ass:infinite_thresholding} is satisfied by ($\Gamma$).
		
		\begin{figure}
			\centering
			\includegraphics[width=\textwidth]{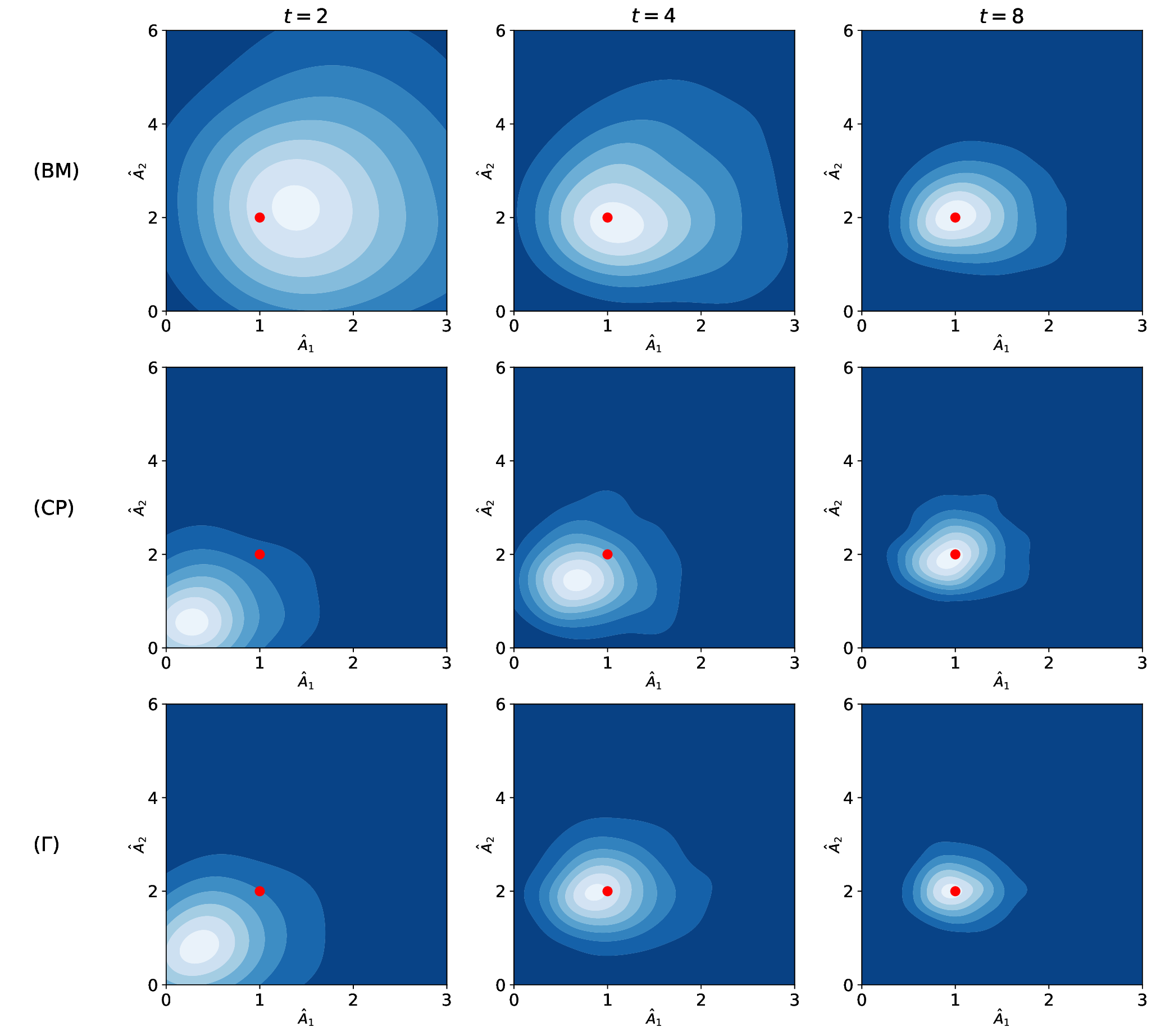}
			\caption{Empirical distribution of the estimator $\hat{\mathbf{A}}(\mathbb{Y}_{\mathcal{P}_t}; \mathcal{Q}_t, \boldsymbol{\nu}_t)$ for 1000 Monte Carlo samples under the three jump regimes for $t=2, 4, 8$.}
			\label{fig:kde_contour_estimators}
		\end{figure}
		
		In Figure \ref{fig:kde_contour_estimators} we report the empirical distribution of the estimator $\hat{\mathbf{A}}(\mathbb{Y}_{\mathcal{P}_t}; \mathcal{Q}_t, \boldsymbol{\nu}_t)$ for 1000 Monte Carlo samples under the three jump regimes for $t=2, 4, 8$. We see that as $t$ increases the distributions concentrate around the true values $\mathbf{A}^*=(1, 2)$, validating the consistency result. Moreover, in Figure \ref{fig:kde_contour_statistic}, we check how well the pre-asymptotic normal approximation holds for the bivariate statistic $\mathbf{Z}_t$ of the feasible CLT \eqref{eqn:feasible_CLT}. We see that as $t$ increases the empirical distribution of the $\mathbf{Z}_t$'s approaches a $N(\mathbf{0}, I_{2\times 2})$ distribution.
		
		\begin{figure}
			\centering            \includegraphics[width=\textwidth]{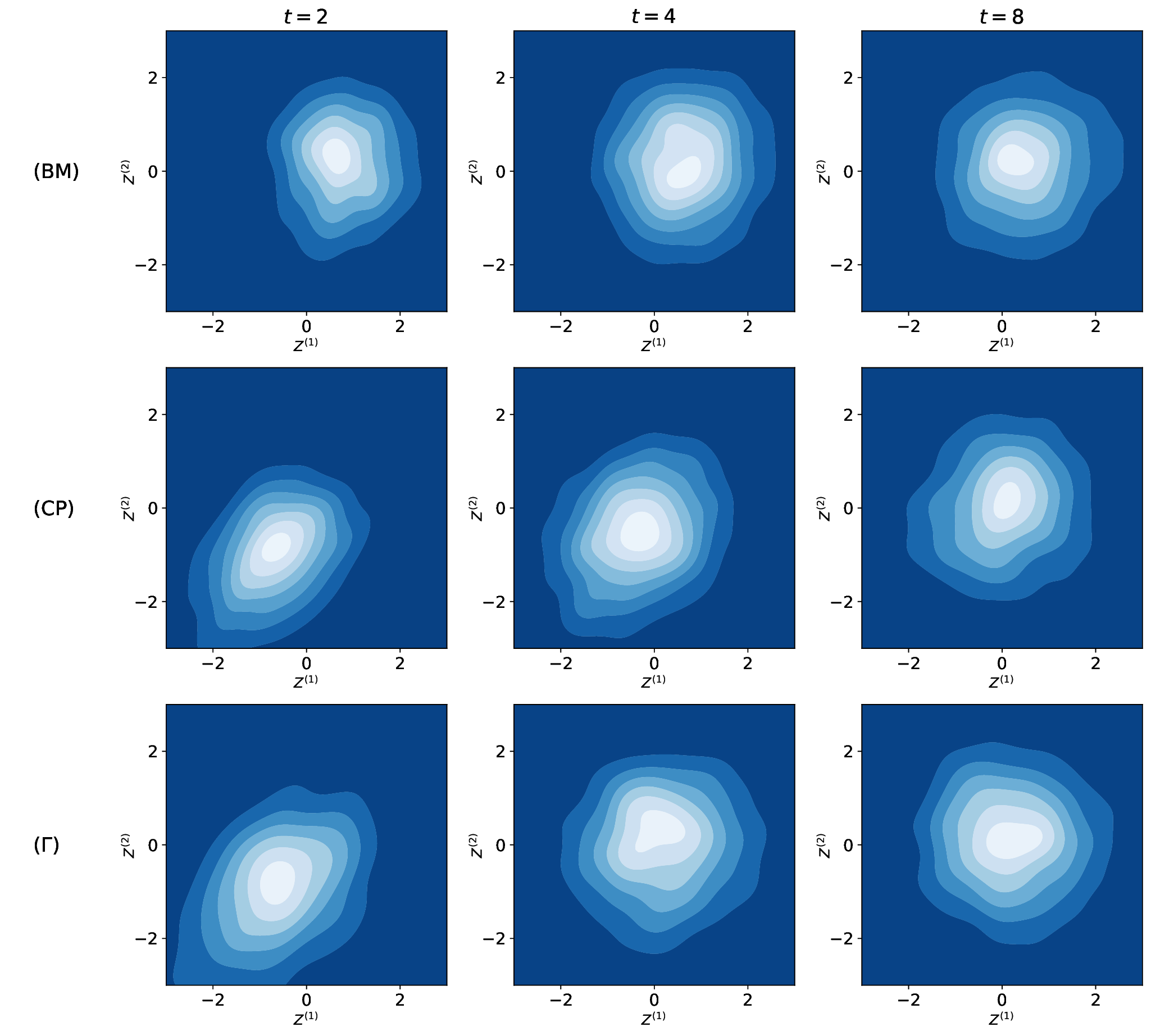}
			\caption{Empirical distribution of the statistic $\mathbf{Z}_t$ for 1000 Monte Carlo samples under the three jump regimes for $t=2, 4, 8$.}
			\label{fig:kde_contour_statistic}
		\end{figure}
		
		\subsection{Practical considerations} \label{sec:practical_considerations} 
		In this section we focus on the feasible estimator  $\hat{\mathbf{A}}(\mathbb{Y}_{\mathcal{P}_t}; \mathcal{Q}_t, \boldsymbol{\nu}_t)$ defined in Equation \eqref{eqn:discr_estimator_A_3} and investigate some natural questions which arise when applying this to observed data.
		
		First, we note that in Theorem \ref{thm:cons_asymp_third_approx_finite_activity} and Theorem \ref{thm:cons_asymp_third_approx_infinite_activity} the estimator $\hat{\mathbf{A}}(\mathbb{Y}_{\mathcal{P}_t}; \mathcal{Q}_t, \boldsymbol{\nu}_t)$ depends, more or less explicitly, on the driving \Levy triplet $(\mathbf{b},\Sigma, F)$ as follows:
		\begin{itemize}
			\item the drift $\mathbf{b}\in\mathbb{R}^d$ is used to de-trend the thresholded increments in Equation \eqref{eqn:thresholded_incr};
			\item the symmetric positive definite covariance matrix $\Sigma\in\mathcal{M}_d(\mathbb{R})$ appears in both $\mathbf{H}_{\mathcal{P}_t, \mathcal{Q}_t, \boldsymbol{\nu}_t}$, cf.\ Equation \eqref{eqn:H_hat_hat_hat}, and  $[\mathbf{H}]_{\mathcal{P}_t, \mathcal{Q}_t}$, cf.\ Equation \eqref{eqn:[H]_hat_hat};
			\item the \Levy measure $F$ is used to determine the activity level of the \Levy driver and thence choose the optimal thresholding powers $\beta^{(i)}$ for $i\in\{1,\ldots,d\}$ via Assumption \ref{ass:finite_thresholding} and Assumption \ref{ass:infinite_thresholding}.
		\end{itemize}
		In practice, these quantities are not known, but some natural and data-driven choices are possible.
		\begin{itemize}
			\item The driving \Levy process is assumed to have zero drift with respect to the truncation function, i.e.\ $\mathbf{b}\equiv 0$. This is a natural choice when considering the \Levy driver as ``noise''.
			\item We note that in the computation of $\hat{\mathbf{A}}(\mathbb{Y}_{\mathcal{P}_t}; \mathcal{Q}_t, \boldsymbol{\nu}_t)$ the $\Sigma$'s appearing in $\mathbf{H}_{\mathcal{P}_t, \mathcal{Q}_t, \boldsymbol{\nu}_t}$ and  $[\mathbf{H}]_{\mathcal{P}_t, \mathcal{Q}_t}$ cancel out. This removes the dependency on $\Sigma$ when estimating $\mathbf{A}^*$ but, as discussed later in this section, does not remove the dependency on $\Sigma$ when making the CLT feasible for inference.
			\item The thresholding powers $\beta^{(i)}$ for $i\in\{1,\ldots,d\}$ can be chosen using a data-driven approach. To be more precise, in practice, we can choose the thresholds $\nu_t^{m, (i)} = (\Delta_{\mathcal{Q}_t}^m)^{\beta^{(i)}}$ directly, dropping their functional form in terms of $\Delta_{\mathcal{Q}_t}^m$ and $\beta^{(i)}$. These thresholds should be selected as small as possible to eliminate the jump part without throwing away too many continuous increments, cf. \citet[Section~6.2.2]{HFFE} where increments of a general semimartingale are thresholded to estimate the integrated volatility. Assuming $\Delta_{\mathcal{Q}_t}^m \hat{D}^{p-1}\mathbf{Y}$ is close enough to $\Delta_{\mathcal{Q}_t}^m D^{p-1}\mathbf{Y}$, i.e.\ $\Delta_{\mathcal{P}_t}$ in Lemma \ref{lemma:D-D_L2_bound} is negligible compared to $\Delta_{\mathcal{Q}_t}^m$, we can write
			\[\Delta_{\mathcal{Q}_t}^m \hat{D}^{p-1}\mathbf{Y} \approx \Delta_{\mathcal{Q}_t}^m \mathbf{L} - \sum_{j=1}^p\int_{u_{m}}^{u_{m+1}} A^*_j D^{p-j}\mathbf{Y}_u\ du \approx  \Delta_{\mathcal{Q}_t}^m \mathbf{L}, \]
			assuming $\Delta_{\mathcal{Q}_t}^m$ is small enough so that, for each $j\in\{1,\ldots,p\}$,
			\[ \mathbb{E}_{\mathbf{A}^*}\left[\left\|\int_{0}^{\Delta_{\mathcal{Q}_t}^m} D^{p-j}\mathbf{Y}_u \ du \right\|^2\right] \leq (\Delta_{\mathcal{Q}_t}^m)^2 \mathbb{E}_{\mathbf{A}^*}\left[\int_{0}^{\Delta_{\mathcal{Q}_t}^m} \frac{\|D^{p-j}\mathbf{Y}_u\|^2}{\Delta_{\mathcal{Q}_t}^m} \ du \right] = O\left((\Delta_{\mathcal{Q}_t}^m)^2\right), \]
			is negligible compared to
			\[\mathbb{E}_{\mathbf{A}^*}\left[\left\|\Delta_{\mathcal{Q}_t}^m \mathbf{L}\right\|^2\right] = \Omega(\Delta_{\mathcal{Q}_t}^m), \]
			due to the presence of the Brownian component. We can hence threshold $\Delta_{\mathcal{Q}_t}^m \hat{D}^{p-1}\mathbf{Y}$ as if we were disentangling the increments of a \Levy process, cf.\ Appendix \ref{app:disentangling}. In this setting, i.e.\ when choosing thresholds $\nu_t^{m, (i)}$, one may work on each component $\hat{D}^{p-1}Y^{(i)}$ independently.
		\end{itemize}
		
		Next, we discuss how to produce feasible asymptotic inference from Theorem \ref{thm:cons_asymp_third_approx_finite_activity} and Theorem \ref{thm:cons_asymp_third_approx_infinite_activity}, i.e.\ how to compute explicit confidence regions for the parameters $\mathbf{A}^*\in\mathfrak{A}$. Under the conditions of the theorems $t^{-1} [\mathbf{H}]_{\mathcal{P}_t,\mathcal{Q}_t}$ is a consistent estimator for the inverse of the limiting variance-covariance matrix $\mathcal{H}_\infty$. Thus, in the spirit of Remark \ref{rem:feasible_CLT}, we can apply Slutsky's lemma to write the feasible CLT
		\begin{equation} \label{eqn:feasible_CLT} \mathbf{Z}_t := [\mathbf{H}]_{\mathcal{P}_t,\mathcal{Q}_t}^{1/2} \left(\mathrm{vec}(\hat{\mathbf{A}}(\mathbb{Y}_{\mathcal{P}_t}; \mathcal{Q}_t, \boldsymbol{\nu}_t)) - \mathrm{vec}(\mathbf{A}^*)\right) \overset{\mathcal{L}}{\longrightarrow} \mathbf{Z} \sim N(\mathbf{0}, I_{pd^2}), \quad t \rightarrow \infty. 
		\end{equation}
		By working with the distribution of a standard $pd^2$-dimensional normal random variable one can easily obtain asymptotic confidence regions for the parameters, e.g.\ elliptic regions centred at $\hat{\mathbf{A}}(\mathbb{Y}_{\mathcal{P}_t}; \mathcal{Q}_t, \boldsymbol{\nu}_t)$. The main caveat here is that, as discussed at the beginning of the section, $[\mathbf{H}]_{\mathcal{P}_t,\mathcal{Q}_t}$ depends on the, usually unknown, covariance matrix $\Sigma$. In this case though, having estimated the parameters $\hat{\mathbf{A}}(\mathbb{Y}_{\mathcal{P}_t}; \mathcal{Q}_t, \boldsymbol{\nu}_t)$, we can plug these into \eqref{eqn:L=p(D)Y} along with finite difference and Riemann sum approximations to recover the \Levy increments 
		\begin{equation} \label{eqn:recover_Levy}
			\left\{\Delta_{\mathcal{Q}_t}^m \mathbf{L} \approx \Delta_{\mathcal{Q}_t}^m \hat{D}^{p-1}\mathbf{Y} + \sum_{j=1}^p \sum_{[s_n,s_{n+1}]\subseteq[u_m, u_{m+1}]} \hat{A}_j(\mathbb{Y}_{\mathcal{P}_t}; \mathcal{Q}_t, \boldsymbol{\nu}_t) \hat{D}^{p-j}\mathbf{Y}_{s_n} \Delta_{\mathcal{P}_t}^n ,\ m = 0,\ldots,M_t-1 \right\}, \end{equation}
		and use standard methods to estimate $\Sigma$, for example, one can use the critical region-based estimator discussed in Appendix \ref{app:disentangling}. The resulting estimator $\hat{\Sigma}$ can then be plugged into $[\mathbf{H}]_{\mathcal{P}_t,\mathcal{Q}_t}$ to compute explicit confidence regions. In Appendix \ref{app:realistic_simulations} we repeat the simulation study of Section \ref{sec:simulation_study} assuming $\Sigma$ and $F$ are unknown. We use the data-driven procedures described above to select the thresholds $(\boldsymbol{\nu}_t^m)_{m=0}^{M_t-1}$ and estimate $\Sigma$. The simulations yield similar results to those of Section \ref{sec:simulation_study} for the estimator and z-statistic of the drift parameter $\mathbf{A}$ and, moreover, seem to suggest the estimator $\hat{\Sigma}$ is consistent.
		
		Finally, the choice of the hyperparameter $p$, the order of the MCAR specification, may or may not be motivated by the modeling task at hand. In the latter case, one may tackle such problem from a model selection perspective. For example, one might use a likelihood-based criterion, such as the Aikake Information Criterion (AIC), by plugging in the approximators $\mathbf{H}_{\mathcal{Q}_t, \mathcal{P}_t, \boldsymbol{\nu}_t}$ and $[\mathbf{H}]_{\mathcal{Q}_t, \mathcal{P}_t}$ in the continuous-time likelihood \eqref{eqn:likelihood_MCAR}.

		\section{Application: GrCAR} \label{sec:GrCAR}
		
		The parametrization of a $d$-dimensional MCAR process is given by the tuple of coefficient matrices 
		\[\mathbf{A} = (A_1, \ldots, A_p) \in (\mathcal{M}_d(\mathbb{R}))^p. \]
		Clearly this means the number of parameters increases quadratically in the dimension $d$ and thus, for a high dimensional observed process, the number of parameters to be estimated can be quite large. When additional structure of the observed process is known, one can significantly reduce such parametrization. For example, as in \citet{Courgeau_Veraart_2022}, one can interpret the components of an MCAR process $\mathbb{Y}$ as the nodes of a graph structure linked together through a collection of edges. The edges of the graph are encoded in an adjacency (or graph topology) matrix $A \in \mathcal{M}_d(\{0, 1\})$ such that $A^{(i,j)} = 1 $ if and only if there exists a link from node $i$ to node $j$. We set $A^{(i,i)} = 0$ for all $i \in\{1,\ldots, d\}$. For any node $i\in\{1,\ldots, d\}$ the in-degree is defined as $n_i = 1 \vee \sum_{j=1}^d A^{(j,i)} $, i.e.\ the number of in-going edges to node $i$. We can thus define the column-normalized adjacency matrix 
		\[\bar{A} = A\, \mathrm{diag}(n^{-1}_1, \ldots, n^{-1}_d) \in \mathcal{M}_d(\mathbb{R}). \]
		
		\begin{definition}[Graph CAR process] \label{def:GrCAR}
			Let $\mathbb{L} = \{\mathbf{L}_t,\ t \geq 0\}$ be a $d$-dimensional \Levy process with characteristics $(\mathbf{b}, \Sigma, F)$. Fix $p\in\mathbb{N}$ and let $\theta \in \mathbb{R}^{p \times 2}$. Define
			\[A_k(\theta) = \theta^{(k, 1)} I_{d\times d} + \theta^{(k, 2)} \bar{A}^\mathrm{T} \in \mathcal{M}_d(\mathbb{R}), \quad k\in\{1,\ldots, p\}.\]
			A graph continuous-time autoregressive process of order $p$, GrCAR($p$) process for short, driven by $\mathbb{L}$ with parameter $\theta\in\mathbb{R}^{p\times 2}$ and adjacency matrix $A\in \mathcal{M}_d(\{0, 1\})$ is a is a $d$-dimensional MCAR($p$) process driven by $\mathbb{L}$ with coefficient matrices $\mathbf{A}(\theta)= (A_1(\theta), \ldots, A_p(\theta))\in(\mathcal{M}_d(\mathbb{R}))^p$ and initial state-space representation $\xi$.
		\end{definition}
		\begin{remark} 
			\begin{itemize}
				\item In the definition of the GrCAR process we use the normalized matrix $\bar{A}$. This is a modeling choice consistent with the GrOU definition of \citet{Courgeau_Veraart_2022}  ensuring the strength of the network effect is proportional to the average value of the (derivatives of the) process at the influencing nodes. Alternatively, one could replace $\bar{A}$ by the un-normalized adjacency matrix $A$, in which case the network effect becomes proportional to the sum of the (derivatives of the) process at the influencing nodes.
				\item When defining the GrCAR process we transpose the matrix $\bar{A}$ so that if $A^{(i,j)} = 1$ values of node $i$ influence those of node $j$. In the social network literature instead a link from node $i$ to node $j$ is interpreted as ``node $i$ follows node $j$'' and thus the values of node $j$ influence those of node $i$. By using the row-normalized version of $A$ instead of $\bar{A}^\mathrm{T}$ in the construction of the network effect, as in \citet{Courgeau_Veraart_2022}, the interpretation of the GrCAR parameters is consistent with such literature.
			\end{itemize}
		\end{remark}
		
		We now leverage the graphical structure of the GrCAR model to write the corresponding explicit estimator for the $\theta\in\mathbb{R}^{p\times 2}$ parameters when continuous-time and discrete-time observations are available. In the following, we consider the vectorization operator for elements $\theta\in\mathbb{R}^{p\times 2}$ given by \ $\mathrm{vec}: \mathbb{R}^{p\times 2} \mapsto \mathbb{R}^{2p}$ such that
		\[ \theta \mapsto \mathrm{vec}(\theta) := (\theta^{(1,1)}, \theta^{(1, 2)},\ldots,\theta^{(p,1)}, \theta^{(p, 2)})^\mathrm{T}.\]	
		\begin{theorem} \label{thm:GrCAR}
			Let $\mathbb{Y}_{[0,t]} = \{\mathbf{Y}_s, \ s\in[0,t]\}$ be the observation of a stationary and ergodic GrCAR($p$) process with parameters 
			\begin{equation} \label{eqn:param_set_Theta}
				\theta^* \in\Theta := \{\theta \in \mathbb{R}^{p \times 2} : \sigma(\mathcal{A}_{\mathbf{A}(\theta)}) \subseteq (-\infty, 0) + i\mathbb{R}\},
			\end{equation}
			and known adjacency matrix $A\in\mathcal{M}(\{0,1\})$ under $\mathbb{P}_{\theta^*}$ such that the driving \Levy process satisfies Assumption \ref{ass:levy_process}. For $\theta^{(0)} = 0_{p\times 2}$, let $D^{p-1}\mathbb{Y}^c_{\theta^{(0)}, [0,t]}= \{D^{p-1}\mathbf{Y}^c_{\theta^{(0)}, t}, \ s\in[0,t]\}$ be defined as in Equation \eqref{eqn:cont_mart_part_Dp-1Y}. Let $\mathbb{K} = \{\mathbf{K}_t,\ t\geq 0\}$ be the process defined by 
			\begin{equation} \label{eqn:K}
				\mathbf{K}_t = -
				\begin{pmatrix}
					\int_0^t \langle D^{p-1}\mathbf{Y}_{s-}, dD^{p-1}\mathbf{Y}^c_{\theta^{(0)}, s}\rangle_{\Sigma} \\
					\int_0^t \langle \bar{A}^\mathrm{T} D^{p-1}\mathbf{Y}_{s-}, dD^{p-1}\mathbf{Y}^c_{\theta^{(0)}, s}\rangle_{\Sigma}  \\
					\vdots\\
					\int_0^t \langle \mathbf{Y}_{s-}, dD^{p-1}\mathbf{Y}^c_{\theta^{(0)}, s}\rangle_{\Sigma} \\
					\int_0^t \langle \bar{A}^\mathrm{T} \mathbf{Y}_{s-}, dD^{p-1}\mathbf{Y}^c_{\theta^{(0)}, s}\rangle_{\Sigma}  \\
				\end{pmatrix}
			\end{equation}
			with quadratic covariation matrix
			\begin{equation} \label{eqn:[K]}
				[\mathbf{K}]_t =
				\begin{pmatrix}
					\int_0^t \langle D^{p-1}\mathbf{Y}_{s-}, D^{p-1}\mathbf{Y}_{s-}\rangle_{\Sigma}\ ds & \cdots & \int_0^t \langle D^{p-1}\mathbf{Y}_{s-}, \bar{A}^\mathrm{T} \mathbf{Y}_{s-}\rangle_{\Sigma}\ ds\\
					\int_0^t \langle \bar{A}^\mathrm{T} D^{p-1}\mathbf{Y}_{s-}, D^{p-1}\mathbf{Y}_{s-}\rangle_{\Sigma}\ ds & \cdots & \int_0^t \langle \bar{A}^\mathrm{T} D^{p-1}\mathbf{Y}_{s-}, \bar{A}^\mathrm{T} \mathbf{Y}_{s-}\rangle_{\Sigma}\ ds\\
					\vdots & \ddots & \vdots \\
					\int_0^t \langle \mathbf{Y}_{s-}, D^{p-1}\mathbf{Y}_{s-}\rangle_{\Sigma}\ ds & \cdots & \int_0^t \langle \mathbf{Y}_{s-}, \bar{A}^\mathrm{T} \mathbf{Y}_{s-}\rangle_{\Sigma}\ ds\\
					\int_0^t \langle \bar{A}^\mathrm{T} \mathbf{Y}_{s-}, D^{p-1}\mathbf{Y}_{s-}\rangle_{\Sigma}\ ds & \cdots & \int_0^t \langle \bar{A}^\mathrm{T} \mathbf{Y}_{s-}, \bar{A}^\mathrm{T} \mathbf{Y}_{s-}\rangle_{\Sigma}\ ds\\
				\end{pmatrix}.
			\end{equation}
			Then the MLE given by
			\begin{equation} \label{eqn:MLE_GrCAR}
				\hat{\theta}(\mathbb{Y}_{[0,t]}) = \mathrm{vec}^{-1} ( [\mathbf{K}]_t^{-1}\mathbf{K}_t),
			\end{equation}
			is 
			\begin{itemize}
				\item (weakly) consistent, i.e.\ 
				\[\hat{\theta}(\mathbb{Y}_{[0,t]}) \overset{\mathbb{P}_{ \theta^*}}{\longrightarrow} \theta^*, \ t\rightarrow \infty; \] 
				\item asymptotically normal, i.e.\ 
				\[ \sqrt{t} \left(\mathrm{vec}(\hat{\theta}(\mathbb{Y}_{[0,t]})) - \mathrm{vec}(\theta^*)\right) \overset{\mathcal{L}}{\longrightarrow} \mathbf{Z} \sim N(\mathbf{0}, \mathcal{K}^{-1}_\infty), \ t \rightarrow \infty, \]
				where $\mathcal{K}_\infty\in\mathcal{M}_{2p}(\mathbb{R})$ is a symmetric positive definite matrix such that $t^{-1}[\mathbf{K}]_t \overset{\mathbb{P}_{ \theta^*}}{\longrightarrow} \mathcal{K}_\infty, \ t\rightarrow \infty$.
			\end{itemize}
			Next, assume $\mathbb{Y}_{\mathcal{P}_t}$ is a discretely sampled version of $\mathbb{Y}_{[0,t]}$ where $\{\mathcal{P}_t,\ t\in\mathcal{T}\}$ is a countable sequence of partitions satisfying Assumption \ref{ass:evenly_spaced_lim} with refinements $\{\mathcal{Q}_t,\ t\in\mathcal{T}\}$ satisfying Assumption \ref{ass:HF_sampling}, Assumption \ref{ass:controlled_sampling} and Assumption \ref{ass:joint_mesh}. Also, assume the sequence of partitions $\{\mathcal{Q}_t,\ t\in\mathcal{T}\}$ and the thresholding sequences $\{(\boldsymbol{\nu}^m_t)_{m=0}^{M_t-1},\ t\in\mathcal{T}\}$ either
			\begin{itemize}
				\item satisfy Assumption \ref{ass:finite_thresholding} and the driving \Levy process has finite jump activity; or
				\item satisfy Assumption \ref{ass:infinite_thresholding}.
			\end{itemize} 
			Then the estimator \begin{equation}
				\hat{\theta}(\mathbb{Y}_{\mathcal{P}_t}; \mathcal{Q}_t, \boldsymbol{\nu}_t) = \mathrm{vec}^{-1}\left([\mathbf{K}]^{-1}_{\mathcal{P}_t, \mathcal{Q}_t} \mathbf{K}_{\mathcal{P}_t, \mathcal{Q}_t, \boldsymbol{\nu}_t} \right),
			\end{equation}
			where
			\begin{align}
				\label{eqn:K_hat}
				\mathbf{K}_{\mathcal{P}_t, \mathcal{Q}_t, \boldsymbol{\nu}_t} &= -
				\sum_{m = 0}^{M_t-1}
				\begin{pmatrix}
					\hat{D}^{p-1} \mathbf{Y}_{u_m}^\mathrm{T} \\
					\hat{D}^{p-1} \mathbf{Y}_{u_m}^\mathrm{T} \bar{A} \\
					\vdots \\
					\mathbf{Y}_{u_m}^\mathrm{T}\\
					\mathbf{Y}_{u_m}^\mathrm{T} \bar{A} \\
				\end{pmatrix} \Sigma^{-1} \left[ \Delta_{\mathcal{Q}_t}^m \hat{D}^{p-1}\mathbf{Y} - \mathbf{b} \Delta_{\mathcal{Q}_t}^m \right] \odot \mathds{1}_{\left\{\left|\Delta_{\mathcal{Q}_t}^m \hat{D}^{p-1}\mathbf{Y} - \mathbf{b} \Delta_{\mathcal{Q}_t}^m\right| \leq \boldsymbol{\nu}^m_t\right\}} , \\
				\label{eqn:[K]_hat}
				[\mathbf{K}]_{\mathcal{P}_t, \mathcal{Q}_t} &=
				\sum_{m = 0}^{M_t-1} \begin{pmatrix}
					\langle \hat{D}^{p-1} \mathbf{Y}_{u_m}, \hat{D}^{p-1}\mathbf{Y}_{u_m}\rangle_{\Sigma} & \cdots &  \langle \hat{D}^{p-1} \mathbf{Y}_{u_m}, \bar{A}^\mathrm{T} \mathbf{Y}_{u_m}\rangle_{\Sigma}\\ 
					\langle \bar{A}^\mathrm{T} \hat{D}^{p-1} \mathbf{Y}_{u_m}, \hat{D}^{p-1}\mathbf{Y}_{u_m}\rangle_{\Sigma} & \cdots &  \langle \bar{A}^\mathrm{T} \hat{D}^{p-1} \mathbf{Y}_{u_m}, \bar{A}^\mathrm{T} \mathbf{Y}_{u_m}\rangle_{\Sigma}\\ 
					\vdots & \ddots & \vdots \\
					\langle \mathbf{Y}_{u_m}, \hat{D}^{p-1}\mathbf{Y}_{u_m}\rangle_{\Sigma} & \cdots &  \langle \mathbf{Y}_{u_m}, \bar{A}^\mathrm{T} \mathbf{Y}_{u_m}\rangle_{\Sigma}\\ 
					\langle \bar{A}^\mathrm{T}  \mathbf{Y}_{u_m}, \hat{D}^{p-1}\mathbf{Y}_{u_m}\rangle_{\Sigma} & \cdots &  \langle \bar{A}^\mathrm{T}  \mathbf{Y}_{u_m}, \bar{A}^\mathrm{T} \mathbf{Y}_{u_m}\rangle_{\Sigma}\\ 
				\end{pmatrix}(u_{m+1} - u_m),
			\end{align}
			is 
			\begin{itemize}
				\item (weakly) consistent, i.e.\ 
				\[\hat{\theta}(\mathbb{Y}_{\mathcal{P}_t}; \mathcal{Q}_t, \boldsymbol{\nu}_t) \overset{\mathbb{P}_{ \theta^*}}{\longrightarrow} \theta^*, \ t\rightarrow \infty; \] 
				\item asymptotically normal, i.e.\ 
				\[ \sqrt{t} \left(\mathrm{vec}(\hat{\theta}(\mathbb{Y}_{\mathcal{P}_t}; \mathcal{Q}_t, \boldsymbol{\nu}_t)) - \mathrm{vec}(\theta^*)\right) \overset{\mathcal{L}}{\longrightarrow} \mathbf{Z} \sim N(\mathbf{0}, \mathcal{K}^{-1}_\infty), \ t \rightarrow \infty, \]
				with $t^{-1}[\mathbf{K}]_{\mathcal{P}_t,
					\mathcal{Q}_t} \overset{\mathbb{P}_{ \theta^*}}{\longrightarrow} \mathcal{K}_\infty, \ t\rightarrow \infty$.
			\end{itemize}
			
		\end{theorem}
		\begin{remark}
			Note that, unlike the estimator for the MCAR process $\hat{\mathbf{A}}(\mathbb{Y}_{\mathcal{P}_t}; \mathcal{Q}_t, \boldsymbol{\nu}_t)$ given in Equation \eqref{eqn:discr_estimator_A_3}, we cannot easily remove the dependence of $\hat{\theta}(\mathbb{Y}_{\mathcal{P}_t}; \mathcal{Q}_t, \boldsymbol{\nu}_t)$ on $\Sigma$. Extending \citet[Definition~2.6]{Courgeau_Veraart_2022} from the GrOU to the GrCAR setting we can construct an alternative estimator $\hat{\theta}$ for $\theta$ from an estimator  $\hat{\mathbf{A}}$ for $\mathbf{A}$ by ``inverting'' the relationship
			\[A_k(\theta) = \theta^{(k, 1)} I_{d\times d} + \theta^{(k, 2)} \bar{A}^\mathrm{T} \in \mathcal{M}_d(\mathbb{R}), \quad k\in\{1,\ldots, p\}.\]
			As long as each node has at least one neighbour the resulting estimator $\hat{\theta}$ depends linearly on $\hat{\mathbf{A}}$, ensuring the asymptotic properties of $\hat{\mathbf{A}}$ are preserved. In particular, by choosing $\hat{\mathbf{A}} = \hat{\mathbf{A}}(\mathbb{Y}_{\mathcal{P}_t}; \mathcal{Q}_t, \boldsymbol{\nu}_t)$ given in Equation \eqref{eqn:discr_estimator_A_3} the resulting estimator $\hat{\theta}$ does not depend on $\Sigma$.
		\end{remark}
		
		\begin{proof}
			The proof follows by similar arguments to Theorem \ref{thm:cons_asymp}, Theorem \ref{thm:cons_asymp_third_approx_finite_activity} and Theorem \ref{thm:cons_asymp_third_approx_infinite_activity}. See Appendix \ref{app:thm_GrCAR_proof}.
		\end{proof}
		
		In a similar spirit to Section \ref{sec:simulation_study} we carry out a finite sample simulation study. Here, we consider a $d=5$ dimensional GrCAR process of order $p=2$ with parameters 
		\[\theta^* = \begin{pmatrix} \theta^{(1, 1)}& \theta^{(1, 2)}\\
			\theta^{(2, 1)}& \theta^{(2, 2)}
		\end{pmatrix} = \begin{pmatrix} 2 & 1\\
			1& 0.5
		\end{pmatrix},  \]
		and fully connected underlying adjacency matrix 
		\[ A = \begin{pmatrix} 
			0& 1& 1& 1& 1 \\
			1& 0& 1& 1& 1 \\
			1& 1& 0& 1& 1 \\
			1& 1& 1& 0& 1 \\
			1& 1& 1& 1& 0 \\
		\end{pmatrix}. \]
		The driving \Levy process has characteristics $(\mathbf{0}, \Sigma, F)$ with covariance matrix $\Sigma = I_{5 \times 5}$. As in Section \ref{sec:simulation_study} we consider the three jump regimes (BM), (CP) and ($\Gamma$), and use the numerical methods described in Appendix \ref{app:simulate_MCAR} to simulate the processes. Again, we work with uniformly spaced partitions $\mathcal{P}_t$ and $\mathcal{Q}_t$ with mesh sizes $\Delta_{\mathcal{P}_t} =t^{-6}$ and $\Delta_{\mathcal{Q}_t}=t^{-2}$ with given optimal $\beta^* = -\infty, 1/3, 1/5$, so that all the Assumptions of Theorem \ref{thm:GrCAR} are satisfied. The resulting empirical distributions for 1000 Monte Carlo simulations are plotted in $(\theta^{(1,1)}, \theta^{(1,2)})$- and $(\theta^{(2,1)}, \theta^{(2,2)})$-space in Figure \ref{fig:graph__estimators}.
		
		\begin{figure}
			\centering
			\includegraphics[height=0.9\textheight]{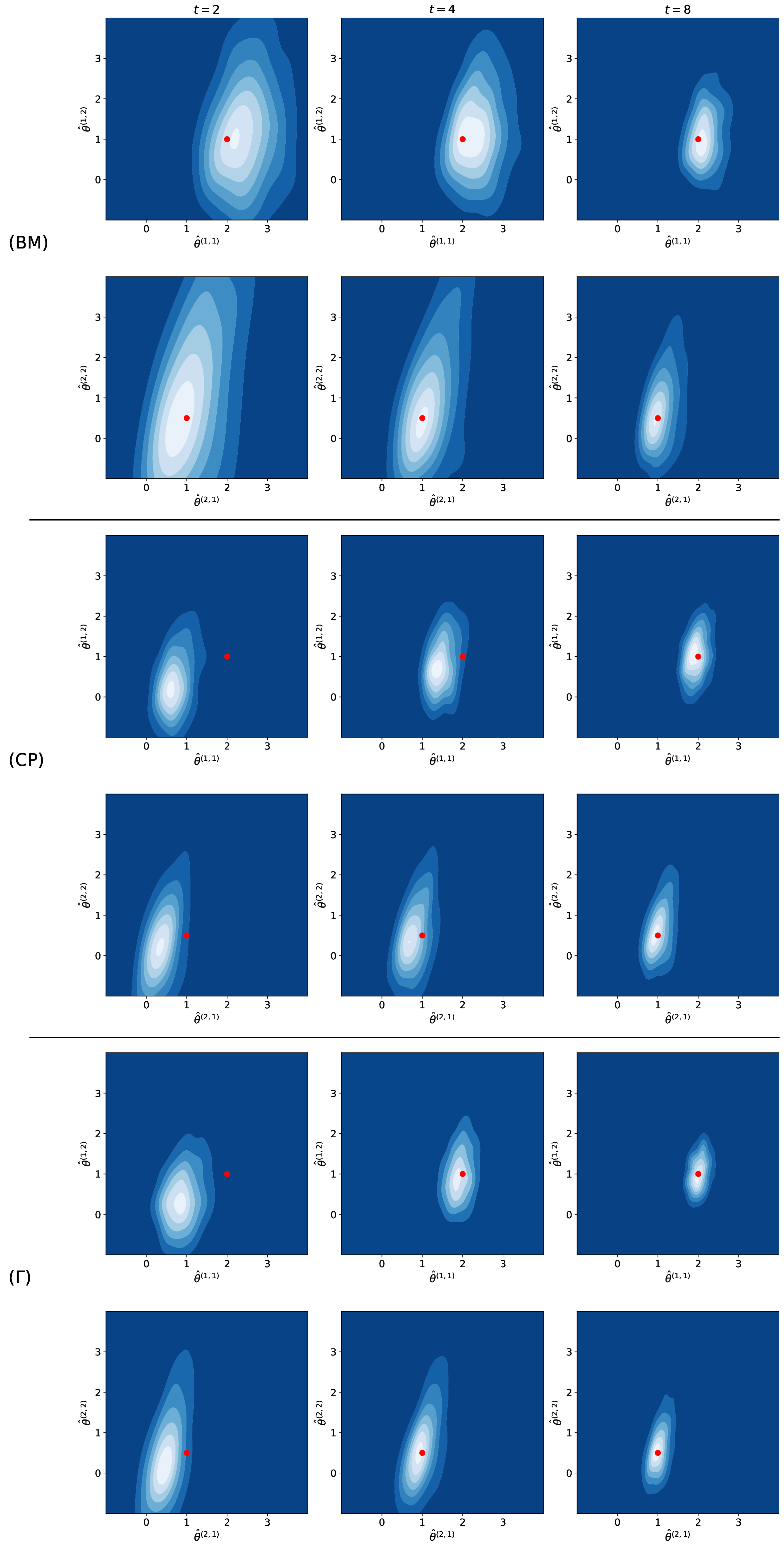}
			\caption{Empirical distribution of the estimator $\hat{\theta}(\mathbb{Y}_{\mathcal{P}_t}; \mathcal{Q}_t, \boldsymbol{\nu}_t) $ for 1000 Monte Carlo samples under the three jump regimes for $t=2, 4, 8$.}            \label{fig:graph__estimators}
		\end{figure}

		\section{Conclusions and Outlook} \label{sec:conclusions}
		This paper introduced novel estimation and inference methods for the drift parameter of a general multivariate \Levydriven CAR process. The main contributions can be summarized as follows. By working on the path space of the state-space process, we showed the maximum-likelihood estimators based on continuous-time observations are consistent and asymptotically normal. Next, we addressed the more realistic scenario of discrete-time observations. We considered natural discretizations of the quantities appearing in the continuous-time setting and proved that the resulting estimators possess the same asymptotic properties. Importantly, the proposed methods also accommodate irregularly spaced observations, which often pose a practical modeling challenge. To prove the desired asymptotic results we carried out a rigorous analysis of the finite difference approximation of the derivatives of the CAR process. To the best of our knowledge this was the first time such kind of analysis was done for a differentiable stochastic processes. Moreover, we addressed what we deem to be inconsistencies in the proofs of the asymptotic properties of the corresponding estimator for drift parameter of a univariate OU process found in \citet{mai_OU}. Under slightly different assumptions we provided modified arguments establishing the desired results in the more general context of multivariate CAR processes. Finally, considering the case where the multivariate observation possesses a known underlying network structure, we introduced the GrCAR architecture, a natural extension of the GrOU process from \citet{Courgeau_Veraart_2021}, along with the corresponding continuous-time and discrete-time estimators. The reduced parametrization offers an additional modeling advantage when dealing with graphically structured data.
		
		Moving forward, there are several directions for future research that can build upon the findings of this paper. Firstly, it is important to rigorously explore the practical challenges discussed in Section \ref{sec:practical_considerations}, for example by developing a joint estimation framework for the CAR drift and the parameters of the driving \Levy process. Investigating methods that simultaneously estimate these components could potentially enhance modeling accuracy and capture the interplay between structure of the CAR process and the underlying \Levy driver.
		From a practitioner's perspective, exploring different modeling applications of the GrCAR specification represents an exciting area for future empirical work. In such an applied context the GrCAR architecture provides a flexible framework to capture relatively simple relationships in very large networks from (possibly irregular) high-frequency observations. A related interesting problem is to consider the scenario where the adjacency matrix is unknown and needs to be estimated from an observation of the GrCAR process. Developing methodologies that effectively estimate the adjacency matrix can provide insights into the underlying network structure and improve the understanding of complex systems with a graphical topology.		
		Finally, we believe it would be of great research interest to explore the estimation and inference of more general CARMA processes, possibly by extending the results discussed in this paper for CAR processes. CARMA models incorporate moving average components, allowing for more flexible modeling of time series data. Developing estimators for CARMA processes with strong theoretical underpinnings would provide a more comprehensive toolbox for modeling and analyzing complex (spatio-)temporal data.
		
		\bibliographystyle{apalike}
		\bibliography{references}
		
		\newpage
		
		\appendix
		\section{Preliminaries} \label{sec:prelim}
		\subsection{\Levy processes}
		We consider a filtered probability space $(\Omega',\mathcal{F}', \{\mathcal{F}'_t,\ t\in\mathbb{R}\}, \mathbb{P}')$ to which all stochastic processes are adapted. Let $\mathbb{L} = \{\mathbf{L}_t,\ t \geq 0\}$ be a  $d$-dimensional \Levy process, i.e.\ a continuous in probability stochastic processes with stationary and independent increments such that $\mathbf{L}_0 = \mathbf{0},\ \mathbb{P}'- \mathrm{a.s.}$, which we assume without loss of generality to be càdlàg. Two well-known results about \Levy processes are the \Levy - Khintchine representation and the \LevyIto decomposition:
		\begin{itemize}
			\item (\Levy - Khintchine representation, \citet[Theorem~8.1]{Sato_1999}) For each $t>0$ the law $\mathcal{L}(\mathbf{L}_t) = (\mathbf{L}_t)_* \mathbb{P}'$ of the \Levy process $\mathbb{L}$ is infinitely divisible with characteristic triplet $(t\mathbf{b}, t\Sigma, tF)$, i.e.\ $\mathbf{b}\in\mathbb{R}^d$, $\Sigma$ is a symmetric positive semi-definite $d\times d$ matrix and $F$ is a \Levy measure on $\mathbb{R}^d$, that is $F(\{0\}) = 0$ and $\int_{\mathbb{R}^d} (\|x\|^2\wedge 1) F(dx) <\infty$, such that
			\[\mathbb{E}'\left[\exp\{iz^\mathrm{T}\mathbf{L}_t\}\right] = \exp\left\{t\left[i \mathbf{b}^\mathrm{T}z - \frac{1}{2}z^\mathrm{T}\Sigma z + \int_{\mathbb{R}^d} (e^{iy^\mathrm{T}z} - 1 - iy^\mathrm{T} z \tau(y)) F(dy) \right]\right\}, \quad z\in\mathbb{R}^d, \]
			where $\tau:\mathbb{R}^d\rightarrow \mathbb{R}$ is a fixed truncation function, i.e.\ a bounded measurable function such that $\tau(x) = 1 + o(\|x\|)$ as $\|x\|\rightarrow 0$ and $\tau(x) = O(\|x\|^{-1})$ as $\|x\|\rightarrow \infty$. We say that $(\mathbf{b}, \Sigma, F)$ are the characteristics of $\mathbb{L}$.
			\begin{remark}
				Strictly speaking, the triplet parameter $b$ depends on the choice of truncation function in the \Levy - Khintchine representation  (cf. \citet[Remark~8.4]{Sato_1999}). In this paper we will always take the truncation function to be $\tau(x):=\mathds{1}_{\{\|x\|\leq 1\}}$ but other choices are possible.
			\end{remark}
			\item (\LevyIto decomposition, \citet[Theorem~19.2]{Sato_1999}) One has the following representation for the \Levy process $\mathbb{L}$ with characteristics $(\mathbf{b}, \Sigma, F)$:
			\begin{multline}
				\mathbf{L}_t = \mathbf{b}t + \Sigma^{1/2} \mathbf{W}_t + \int_0^t \int_{\{x:\|x\| > 1\}} x \mu(dt, dx) \\
				+\int_0^t \int_{\{x:\|x\| \leq 1\}} x \{\mu(dt, dx) - dt\, F(dx)\}, \quad t\geq 0, \mathbb{P}'\mathrm{-a.s.}
			\end{multline}
			where:
			\begin{itemize}
				\item $\mathbb{W} = \{ \mathbf{W}_t,\ t\geq 0\}$ is a standard Brownian motion on $(\Omega',\mathcal{F}', \{\mathcal{F}'_t,\ t\in\mathbb{R}\}, \mathbb{P}')$ and $\Sigma^{1/2}$ is a $d\times d$ matrix s.t. $\Sigma = \Sigma^{1/2} \Sigma^{1/2, \mathrm{T}}$;
				\item $\{\mu(B): B\in\mathcal{B}((0,\infty)\times \mathbb{R}^d\setminus\{0\})\}$ is a Poisson random measure independent of $\mathbb{W}$ such that
				\begin{equation*} \mu(dt, dx) = \sum_{s\geq 0: \Delta \mathbf{L}_s \neq 0} \delta_{(s, \Delta \mathbf{L}_s)}(dt, dx), \end{equation*}
				and $\Delta \mathbb{L} = \{\Delta \mathbf{L}_t,\ t\geq 0\}$ is the jump process of $\mathbb{L}$, i.e.\ $\Delta \mathbf{L}_t =  \mathbf{L}_t - \mathbf{L}_{t-}$, 
				\item we define the integral \[\int_0^t \int_{\{x:\|x\| \leq 1\}} x \{\mu(dt, dx) - dtF(dx)\} := \lim_{\epsilon \downarrow 0} \int_0^t \int_{\{x: \epsilon < \|x\| \leq 1\}} x \{\mu(dt, dx) - dtF(dx)\}.\]
			\end{itemize}
		\end{itemize}
		\subsection{Stochastic Integration}
		In order to define an MCAR process we require a notion of stochastic differential, or, more formally, that of the stochastic integral
		\[(H \cdot \mathbf{L})_t := \int_0^t H_s d\mathbf{L}_s, \quad t\geq 0, \]
		for general (possibly stochastic) $\mathbb{R}^{m \times d}$-valued integrands $\mathbb{H} = \{H_t,\ t\geq 0\}$. When considering a \Levy process $\mathbb{L}$ as integrator we will define the integral above via the definition of \Levy type stochastic integral given in \citet[Section~4.3.3]{Applebaum_2009} and \citet[Example~4.3.2]{Applebaum_2009}. For any predictable process $\mathbb{H} = \{H_t,\ t\geq 0\}$ s.t. $\int_0^t \|H_s\|^2 ds < \infty$, $\mathbb{P}'$-a.s.\ we use the \Levy - \Ito decomposition to define for any $t\geq 0$
		\begin{multline}
			\int_0^t H_s d\mathbf{L}_s = \int_0^t H_s \mathbf{b}\, ds + \int_0^t H_s \Sigma^{1/2} d\mathbf{W}_s \\ + \int_0^t \int_{\{x:\|x\| > 1\}} H_s x  \mu(ds, dx) +  \int_0^t \int_{\{x:\|x\| \leq 1\}} H_s x \{\mu(ds, dx) - ds\, F(dx)\}.
		\end{multline}
		Here we use the following definitions:
		\begin{itemize} 
			\item The first integral $\int_0^t H_s \mathbf{b}\, ds$ is a (stochastic) Lebesgue-Stieljes integral w.r.t.\ the Lebesgue measure on $\mathbb{R}$, i.e.\ for each $\omega\in\Omega', t\geq 0$ one defines $\int_0^t H_s \mathbf{b}\, ds (\omega) := \int_0^t H_s(\omega) \mathbf{b}\, ds$ as a Lebsgue-Stieljes integral.
			\item The second integral $\int_0^t H_s \Sigma^{1/2} d\mathbf{W}_s$ is a stochastic \Ito integral w.r.t.\ Brownian motion defined
			via $L_2$-limit construction. The original theory was developed by \citet{Ito_1944} for Brownian motion and then extended to general square-integrable martingales by \citet{Kunita_Watanabe_1967}. Expositions of the general theory can be found in \citet{Karatzas_Shreve_1996} when the integrator is continuous or in \citet[Theorem~I.4.40]{Jacod_Shiryaev_1987}, when the integrator is discontinuous (slightly different classes of integrands need to be considered in the two cases, progressively measurable integrands when the integrator is continuous and predictable integrands, a smaller class, when the integrator is discontinuous). \citet[Section~4.3.1]{Applebaum_2009} gives a slightly different but equivalent construction by treating $\mathbb{W}$ as an $\mathbb{R}^d$-valued martingale measure on $\mathbb{R}$.
			\item The third integral is defined as $\displaystyle \int_0^t \int_{\{x:\|x\| > 1\}} H_s x  \mu(ds, dx) = \sum_{0\leq s\leq t: \|\Delta \mathbf{L}_s\| > 1} H_s \Delta \mathbf{L}_s$.
			\item The last integral $\int_0^t \int_{\{x:\|x\| \leq 1\}} H_s x \{\mu(ds, dx) - ds\, F(dx)\}$ is a stochastic integral w.r.t.\ the $(2, F)$-martingale valued measure $M(dt, dx):= \mu(dt, dx) - dt\, F(dx)$, cf.\ \citet[Section~4.2.2]{Applebaum_2009}.
		\end{itemize}
		
		\begin{remark}
			There are less general definitions of integral which are consistent with this construction when suitable stricter conditions are satisfied:
			\begin{itemize}
				\item When $\mathbb{L}$ is of finite variation this coincides $\mathbb{P}'$-a.s.\ with a stochastic pathwise Lebesgue-Stieljes integral, see \citet[page~314]{Millar_1972}.
				% 			\item Imposing additional conditions on $\mathbb{H}$ we could define the integral using Young's construction pathwise (?), but be careful one needs the integrator to be continuous... kill? [elaborate].
			\end{itemize}
			
			There are also more general definitions which are consistent with the above construction of the integral:
			\begin{itemize}
				\item Every \Levy process is a semimartingale [add reference], so any definition of stochastic integral for semimartingale integrators applies to \Levy integrators:
				\begin{itemize}
					\item The uniformly on compacts in probability (ucp) limit construction of \citet{Protter_1990}. In the simplest case the integrand $\mathbb{H}$ is an adapted \caglad process but the construction can be extended to more general predictable processes $\mathbb{H}$ which satisfy joint conditions with $\mathbb{X}$, cf.\ \citet{Protter_1990}.
					\item The $L_2$-limit construction of \citet{Dellacherie_Meyer_1978} (also treated in \citet[Theorem~I.4.31]{Jacod_Shiryaev_1987}). This approach uses the decomposition of a semimartingale into a predictable finite variation process and a local martingale to define the stochastic integral, as discussed below.
					% \item Rough path theory, but be careful one needs the integrator to be continuous... kill? [elaborate]
				\end{itemize}
				\item Every \Levy process can be viewed as an $\mathbb{R}^d$-valued \Levy random measure on $\mathbb{R}$, so definitions of stochastic integral w.r.t.\ infinitely divisible random measures apply. For example, \citet{Rajput_Rosinski_1989}, construct the integral using the $\mathbb{P}'$-limit construction of \citet{Urbanik_Woyczynski}.
			\end{itemize}
		\end{remark}
		
		In the following section we will also be integrating with respect to general semimartingales, i.e.\ we will be considering processes $\mathbb{H} \cdot \mathbb{X} = \{ (H \cdot \mathbf{X})_t,\ t \geq 0\}$ where 
		\[ (H \cdot \mathbf{X})_t := \int_0^t H_s d\mathbf{X}_s \]
		for $\mathbb{R}^d$-valued semimartingale $\mathbb{X} = \{\mathbf{X}_t,\ t\geq 0\}$ and suitable $\mathbb{R}^{m\times d}$-valued integrand $\mathbb{H} = \{H_t,\ t\geq 0\}$. When considering such integrals, we will be using the definition of stochastic integral obtained by extending the classical $L_2$ integration theory of square-integrable martingales to semimartingales. To do so, one decomposes $\mathbb{X} = \mathbb{A} + \mathbb{M}$ where $\mathbb{A} = \{\mathbf{A}_t,\ t\geq 0\}$ is a predictable finite variation process and $\mathbb{M}=\{\mathbf{M}_t,\ t\geq 0\}$ is a local martingale and defines
		\[\int_0^t H_s d\mathbf{X}_s = \int_0^t H_s d\mathbf{A}_s + \int_0^t H_s d\mathbf{M}_s, \]
		where the first integral is defined as a stochastic Lebesgue-Stieljes integral while the latter is obtained by an extension via localization of the classical $L_2$-theory for square-integrable martingales by \citet{Kunita_Watanabe_1967}. For further details see \citet[Theorem~I.4.31]{Jacod_Shiryaev_1987}, which follows a similar exposition as that in \citet{Dellacherie_Meyer_1978}. Note that the integral in the most general form is defined for any locally bounded predictable process $\mathbb{H}$, that is $\mathbb{H} = \{H_t,\ t\geq 0\}$ needs to be predictable, i.e.\ measurable w.r.t.\ the $\sigma$-algebra generated by all left-continuous adapted processes, and locally bounded, i.e.\ there exists a localizing sequence of stopping times $(\tau_n)_{n\geq 0}$ increasing to infinity such that the stopped process $\mathbb{H}^{\tau_n}$ is uniformly bounded for all $n\geq 0$. In particular, it holds that $\mathbb{H} \cdot \mathbb{X}$ is also a semimartingale and, whenever $\mathbb{X}$ is a local martingale, then so is $\mathbb{H} \cdot \mathbb{X}$. The way these integrals are defined automatically ensures that when observations of the continuous-time processes are only available on a discrete grid, the integral constructed from the observed piecewise-linear processes yields a natural approximator for the stochastic integral. This will be of fundamental importance when constructing estimators from discrete observations in Section \ref{sec:MCAR_discr}.
		
		\subsection{Canonical Space}
		Let $(\Omega, \mathcal{F}, \{\mathcal{F}_t, \ t\geq 0\})$ denote the canonical space of $\mathbb{R}^{pd}$-valued paths, i.e.\ 
		\begin{itemize}
			\item $\Omega = D([0,\infty);\mathbb{R}^{pd})$ the space of \cadlag processes $\omega:[0,\infty)\rightarrow\mathbb{R}^{pd}$,
			\item $\{\mathcal{F}_t, \ t\geq 0\}$ is the filtration s.t. $\forall t\geq 0$
			\[\mathcal{F}_t = \bigcap_{u>t} \sigma\Big(\omega(s): 0\leq s\leq u\Big),\]
			\item $\mathcal{F}$ is the smallest sigma algebra containing $\{\mathcal{F}_t, \ t\geq 0\}$.
		\end{itemize}
		This is the (filtered) probability space on which the probability measures $\mathbb{P}_{\mathbf{A}, \mathbf{x}_0}$ for $\mathbf{x}_0\in\mathbb{R}^{pd}$ and $\mathbf{A}\in(\mathcal{M}(\mathbb{R}^d))^p$ are defined.
		\begin{remark} \label{rem:sigma_algebra}
			Here $\sigma(\omega(s): s\in\mathcal{I})$ indicates the smallest sigma algebra such that $\omega\in D([0,\infty)) \mapsto \omega(s)\in\mathbb{R}^{pd}$ is measurable for each $s\in\mathcal{I}$. We have that $\mathcal{F} = \sigma \left(\omega(s) : s \geq 0 \right)$ and hence $\mathcal{F}$ can be characterized as the restriction to $D([0, \infty))$ of the product $\sigma$-algebra $\otimes_{\mathbb{R}_+}\mathcal{B}(\mathbb{R}^{pd})$ defined as the smallest $\sigma$-algebra such that $\pi \in \{f\ \mathrm{s.t.}\ f:[0,\infty)\rightarrow \mathbb{R}^{pd}\}\mapsto \pi(s)\in\mathbb{R}^{pd}$ is measurable for each $s\in\mathbb{R}_+$. Alternatively, we can define $\mathcal{F}$ as the $\sigma$-algebra generated by cylinder sets of the form
			\[C = \{\omega\in D([0,\infty)) \ \mathrm{s.t.}\ \omega(t_1) \in A_1, \ldots, \omega(t_n) \in A_N\},\]
			for any $n\in\mathbb{N}, t_1< \ldots< t_n\in\mathbb{R}_{+}$ and $A_1,\ldots, A_n\in\mathcal{B}(\mathbb{R}^{pd})$.
			Moreover, by \citet[Proposition~III.7.1]{Ethier_Kurtz_1986}, this $\sigma$-algebra is the same as the Borel $\sigma$-algebra generated by the Skorokhod topology on $D([0,\infty))$. It is a well-known property that $(\Omega, \mathcal{F})$ is a Polish space.
		\end{remark}
		
		\sloppy For a parameter $\mathbf{A} \in (\mathcal{M}(\mathbb{R}^{d}))^p$ consider the $pd$-dimensional state-space representation $\mathbb{X}_{\mathbf{A},\mathbf{x}_0} = \{\mathbf{X}_{\mathbf{A},\mathbf{x}_0,t},\ t\geq 0\}$ of a MCAR($p$) process driven by a \Levy process $\mathbb{L}$ with characteristics $(\mathbf{b}, \Sigma, F)$, parameter $\mathbf{A}$ and deterministic initial condition $\xi = \mathbf{x}_0$ on an arbitrary stochastic basis $(\Omega', \mathcal{F}', \{\mathcal{F}'_t, \ t\geq 0\}, \mathbb{P}')$. Under $\mathbb{P}'$, the process $\mathbb{X}_{\mathbf{A},\mathbf{x}_0}$ is a semimartingale with local characteristics $(\mathbf{B}'_{\mathbf{A}}, C', \nu')$ on $(\Omega', \mathcal{F}', \{\mathcal{F}'_t, \ t\geq 0\}, \mathbb{P}')$ given by
		\begin{equation} \label{eqn:loc_char_P'}
			\begin{gathered}
				\mathbf{B}'_{\mathbf{A}, t}(\omega) = \int_0^t (\mathcal{A}_{\mathbf{A}} \mathbf{X}_{\mathbf{A},\mathbf{x}_0, s}(\omega) + \mathcal{E} b) ds = \mathcal{A}_{\mathbf{A}}\int_0^t \mathbf{X}_{\mathbf{A},\mathbf{x}_0, s}(\omega) ds + \mathcal{E} bt,
				\\
				C'_t(\omega) = \int_0^t  \mathcal{E} \Sigma \mathcal{E}^\mathrm{T} ds = \mathcal{E} \Sigma \mathcal{E}^\mathrm{T} t, \quad
				\nu'(\omega; dt, dx) = dt \int_{\mathbb{R}^d} \mathds{1}_{\{dx\}}(\mathcal{E} z) F(dz) = dt\, \delta_0 (dx_{-p})\, F(dx_{p}),
			\end{gathered}
		\end{equation}
		where we write $x = (x_{1}^\mathrm{T}, \ldots, x^\mathrm{T}_p)^\mathrm{T}\in\mathbb{R}^{pd}$ and $x_{-p} =( x_{1}^\mathrm{T}, \ldots, x^\mathrm{T}_{p-1})^\mathrm{T}$ for $x_1, \ldots, x_p \in \mathbb{R}^d.$    
		\section{Properties of the state-space representation of the MCAR(p) process}
		Note that the state-space representation of an $\mathbb{R}^d$-dimensional MCAR($p$) process $\mathbb{Y} = \{\mathbf{Y}_t,\ t\geq 0\}$ driven by $\mathbb{L} = \{\mathbf{L}_t,\ t\geq 0\}$ with coefficient matrices $\mathbf{A} = (A_1, \ldots, A_p) \in (\mathcal{M}_d(\mathbb{R}))^p$ and initial state-space representation $\mathbf{X}_0$ is the $\mathbb{R}^{pd}$-dimensional Ornstein-Uhlenbeck process $\mathbb{X} = \{\mathbf{X}_t,\ t\geq 0\}$ driven by the (degenerate) \Levy process $\mathcal{E}\mathbb{L}:= \{\mathcal{E}\mathbf{L}_t,\ t\geq 0\}$ and dynamics matrix 
		\[\mathcal{A} = \begin{pmatrix}
			0_{d \times d} & I_{d \times d} & 0_{d \times d} & \cdots & 0_{d \times d} \\
			0_{d \times d} & 0_{d \times d} & I_{d \times d} & \cdots & 0_{d \times d} \\
			\vdots & \vdots & \vdots & \ddots & \vdots \\
			0_{d \times d} & 0_{d \times d} & 0_{d \times d} & \cdots & I_{d \times d} \\
			- A_p & -A_{p-1} & -A_{p-2} & \cdots & -A_1 \\
		\end{pmatrix},
		\]
		i.e.\ it can be written as
		\[\mathbf{X}_t = e^{\mathcal{A} t}\mathbf{X}_0 + \int_0^t e^{\mathcal{A} (t-s)} \mathcal{E} d\mathbf{L}_s, \quad t\geq 0.\]
		
		We can thus apply standard results from the theory of multivariate Ornstein–Uhlenbeck processes driven by a general \Levy process, i.e.\ \cite{Masuda_2004}, \cite{Sato_Yamazato_1984}. For example, we have that for $t>0$, the law of $\mathbf{X}_t$ given $\mathbf{X}_{0} = \mathbf{x}_0$ is infinitely divisible with characteristics $(b_{t,\mathbf{x}_0, \mathcal{A}}, \Sigma_{t, \mathcal{A}}, F_{t, \mathcal{A}})$ given by
		\begin{align*}
			b_{t,\mathbf{x}_0,\mathcal{A}} &= e^{t\mathcal{A}} \mathbf{x}_0 + \int_0^t e^{s \mathcal{A}} \mathcal{E} \mathbf{b} \, ds \\
			&\quad\quad+\int_{\mathbb{R}^{pd}}\int_0^t e^{s \mathcal{A}} z\left(\mathds{1}_{\{y: \|y\| \leq 1\}}\left(e^{s \mathcal{A}} z\right) - \mathds{1}_{\{y: \|y\| \leq 1\}}(z)\right) ds\, \delta_0(dz_{-p})\,F(dz_p), \\
			\Sigma_{t, \mathcal{A}} &= \int_0^t e^{s \mathcal{A}} \mathcal{E}\Sigma\mathcal{E}^\mathrm{T}  e^{s \mathcal{A}^{\mathrm{T}}} ds, \\
			F_{t,\mathcal{A}}(E) &= \int_{\mathbb{R}^{pd}}  \int_0^t \mathds{1}_{E}\left(e^{s\mathcal{A}}z\right) ds \, \delta_0(dz_{-p})\, F(dz_p), \quad \forall E\in\mathcal{B}(\mathbb{R}^{pd}),
		\end{align*}
		where $(\mathbf{b}, \Sigma, F)$ are the characteristics of the driving \Levy process $\mathbb{L} = \{\mathbf{L}_{t}, \ t\geq 0\}$, cf.\ \citet[Theorem~3.1]{Sato_Yamazato_1984}. If the driving \Levy process $\mathbb{L}$ is assumed to have second moments then also $\mathbb{X}$ will have second moments and we can explicitly compute
		\begin{align*}
			\mathbb{E}[\|\mathbf{X}_t - \mathbf{X}_0\|^2] &= \mathrm{tr}\left(\Sigma_{t,\mathcal{A}}\right) + \int_{\mathbb{R}^{pd}} \|z\|^2 F_{t, \mathcal{A}}(dz) + \mathbb{E}\left[\left\| \mathbf{X}_0 - \left(b_{t,\mathbf{X}_0, \mathcal{A}} + \int_{\{z:\|z\|>1\}} z F_{t,\mathcal{A}}(dz)\right) \right\|^2 \right] \\
			&= \mathrm{tr}\left(\int_0^t e^{s \mathcal{A}} \mathcal{E}\Sigma\mathcal{E}^\mathrm{T}  e^{s \mathcal{A}^{\mathrm{T}}} ds\right) + \int_{\mathbb{R}^{d}}\int_0^t \|e^{s\mathcal{A}}\mathcal{E}y\|^2 F(dy) ds + \\
			&\quad\quad\quad\quad\quad\quad\quad\quad\quad\quad\quad\quad\quad\quad\quad +\mathbb{E}\left[\left\|(e^{t\mathcal{A}}- I) \mathbf{X}_0 + \int_0^t e^{(t-s)\mathcal{A}} \mathcal{E} \mathbb{E}[\mathbf{L}_1] ds \right\|^2 \right]. \\
		\end{align*}
		And one can thus obtain the bound
		\begin{equation} \label{eqn:big_O_X2}
			\mathbb{E}[\|\mathbf{X}_t - \mathbf{X}_0\|^2] = O\left(\int_0^t e^{2s\|\mathcal{A}\|}ds + t\|\mathcal{A}\|e^{t\|\mathcal{A}\|}\right) = O\left( t(e^{2t\|\mathcal{A}\|} + e^{t\|\mathcal{A}\|})\right) = O(t), \quad t\rightarrow 0.
		\end{equation}
		
		\section{Continuous observations: proofs}
		
		\subsection{Proof of Proposition \ref{prop:exist_likelihood}} \label{app:proof_prop_exist_likelihood}
		\begin{proof}
			Let $\mathbb{X}_{\mathbf{A},\mathbf{x}_0} = \{\mathbf{X}_{\mathbf{A},\mathbf{x}_0,s},\ t\geq 0\}$ denote a MCAR($p$) process driven by a \Levy process $\mathbb{L}$ with characteristics $(\mathbf{b}, \Sigma, F)$ on the stochastic basis $(\Omega', \mathcal{F}', \{\mathcal{F}'_s, \ s\geq 0\}, \mathbb{P}')$ and parameter $\mathbf{A}$, i.e.\ such as it satisfies SDE \eqref{eqn:SDE}. Recall that $\mathbb{P}_{\mathbf{A}, \mathbf{x}_0}$ represents the unique solution-measure to the SDE \eqref{eqn:SDE} with initial condition $\mathbf{X}_0 = \mathbf{x}_0$, i.e.\ the unique solution to the martingale problem for $(\mathbf{B}_{\mathbf{A}}, C, \nu)$ on the canonical space with initial condition $\eta = \delta_{\mathbf{x}_0}$, where $(\mathbf{B}_{\mathbf{A}}, C, \nu)$ are defined in Equation \eqref{eqn:loc_char_P}.
			
			Fix deterministic horizon $t>0$, then, under $\mathbb{P}'$, the stopped process $\mathbb{X}_{\mathbf{A},\mathbf{x}_0}^t = \{\mathbf{X}^t_{\mathbf{A},\mathbf{x}_0, s},\ s\geq 0\}$ such that $\mathbf{X}^t_{\mathbf{A},\mathbf{x}_0,s}(\omega) = \mathbf{X}_{\mathbf{A},\mathbf{x}_0, s\wedge t} (\omega)$ satisfies the stopped SDE
			\begin{equation}
				\begin{split}
					d \mathbf{X}_s = (\mathcal{A}_{\mathbf{A}} \mathbf{X}_s + \mathcal{E} b)\mathds{1}_{[0,t]}(s) ds + \mathcal{E} &\Sigma^{1/2}\mathds{1}_{[0,t]}(s) d \mathbf{W}_s \\
					+\int_{\|z\| \leq 1} &\mathcal{E} z \{\mu^t(ds, dz) - \mathds{1}_{[0,t]}(s)F(dz) d s\} + \int_{\|z\| > 1} \mathcal{E} z \mu^t(d s, dz),
				\end{split} \label{eqn:SDE_stopped}
			\end{equation}
			where $\mu^t(ds, dz)= \mathds{1}_{[0,t]}(s)\mu^t(ds, dz)$ is again a Poisson random measure on $\mathbb{R}_+ \times \mathbb{R}^d$ with compensator $\mathds{1}_{[0,t]}(s) F(dz) ds$ (abusing notation slightly). The process $\mathbb{X}_{\mathbf{A},\mathbf{x}_0}^t$ is a semimartingale with local characteristics $(\mathbf{B}^{\prime, t}_\mathbf{A}, C^{\prime,t}, \nu^{\prime,t})$ \citep[Proposition~2.4]{Jacod_Memin_1976} where $\mathbf{B}^{\prime, t}_\mathbf{A}$ and $C^{\prime, t}$ denote the stopped at $t>0$ processes $\mathbf{B}'_\mathbf{A}$ and $C'$ respectively and
			\[\nu^{\prime, t}(A) = \int_A \mathds{1}_{[0,t]}(s) \nu'(ds, dx), \quad \forall A \in \mathcal{B}(\mathbb{R}_+\times \mathbb{R}^{pd}).\]
			We start by showing that for any $\mathbf{A}\in(\mathcal{M}_d(\mathbb{R}))^p$ the martingale problem for $(\mathbf{B}^t_{\mathbf{A}}, C^t, \nu^t)$ with deterministic initial condition has a unique solution. The characteristics $(\mathbf{B}^t_{\mathbf{A}}, C^t, \nu^t)$ are defined on the canonical space by
			\begin{equation} 
				\begin{gathered}            \label{eqn:loc_char_P_stopped}
					\mathbf{B}^t_{\mathbf{A}, s}(\omega) = \mathcal{A}_{\mathbf{A}}\int_0^{s\wedge t} \omega(u) du + \mathcal{E} b(s\wedge t),
					\quad
					C^t_s(\omega) = \tilde{\Sigma} (s\wedge t), \\
					\nu^t(\omega; ds, dx) = \mathds{1}_{[0,t]}(s)\, ds\, \delta_0 (dx_{-p})\, F(dx_{p}), 
				\end{gathered}
			\end{equation}
			and we write $x = (x^\mathrm{T}_1, \ldots, x^\mathrm{T}_p)^\mathrm{T}\in\mathbb{R}^{pd}$ and $x_{-p} =(x_{1}^\mathrm{T}, \ldots, x^\mathrm{T}_{p-1})^\mathrm{T}$ for $x_1, \ldots, x_p \in \mathbb{R}^d$.
			\begin{lemma} \label{lemma:uniqueness_mart_problem}
				Let $\mathbb{X}$ denote the canonical process on $(\Omega, \mathcal{F}, \{\mathcal{F}_s,\ s\geq0\})$. For any $\mathbf{A}\in(\mathcal{M}_d(\mathbb{R}))^p$ the martingale problem $(\mathbb{X}, (\mathbf{B}^t_{\mathbf{A}}, C^t, \nu^t), \delta_{\mathbf{x}_0})$ has a unique solution
			\end{lemma}
			\begin{proof}
				Uniqueness can be proved as follows. First, we note that there exists a solution-process to the stopped SDE \eqref{eqn:SDE_stopped} starting at $\mathbf{x}_0$ on $(\Omega', \mathcal{F}', \{\mathcal{F}'_s, \ s\geq 0\}, \mathbb{P}')$ given by $\mathbb{X}_{\mathbf{A},\mathbf{x}_0}^t$ (and this implies the existence of a solution-measure to SDE \eqref{eqn:SDE_stopped}). If $\mathbb{Y}$ is another solution-process of \eqref{eqn:SDE_stopped} starting at $\mathbf{x}_0$ then the difference $\mathbb{D}_\mathbf{A} = \mathbb{Y} - \mathbb{X}_{\mathbf{A},\mathbf{x}_0}^t$ must satisfy the ODE 
				\[d\mathbf{D}_s = \mathcal{A}_\mathbf{A} \mathbf{D}_s \mathds{1}_{[0,t]} ds, \quad \mathbf{D}_0 = 0,\]
				$\mathbb{P}'$-a.s.\, and hence $\mathbb{D}_\mathbf{A} \equiv 0$, $\mathbb{P}'$-a.s. Thus \eqref{eqn:SDE_stopped} has a unique solution-process (for any stochastic basis $(\Omega', \mathcal{F}', \{\mathcal{F}'_s, \ s\geq 0\}, \mathbb{P}')$). By \citet[Theorem~III.2.33]{Jacod_Shiryaev_1987} this implies uniqueness of the solution-measure with initial condition $\delta_{\mathbf{x}_0}$. In turn, by \citet[Theorem~III.2.26]{Jacod_Shiryaev_1987}, this implies uniqueness of the solution to the martingale problem on the canonical space $(\Omega, \mathcal{F}, \{\mathcal{F}_s, \ s\geq 0\})$ for the characteristics $(\mathbf{B}_\mathbf{A}^t, C^t, \nu^t)$ and initial condition $\eta = \delta_{\mathbf{x}_0}$. The same reasoning holds with $\mathbf{A}$ replaced by $\mathbf{A}^{(0)}$.
			\end{proof}
			Next, let $\mathbb{Q}_{\mathbf{A}, \mathbf{x}_0}$ and $\mathbb{Q}_{\mathbf{A}^{(0)}, \mathbf{x}_0}$ denote the unique measures on canonical path space $(\Omega, \mathcal{F}, \{\mathcal{F}_s,\ s\geq 0\})$ such that the canonical process is a semimartingale with characteristics $(\mathbf{B}_\mathbf{A}^t, C^t, \nu^t)$ and $(\mathbf{B}_{\mathbf{A}^{(0)}}^t, C^t, \nu^t)$ respectively. Then we clearly have that the restricted measures on $(\Omega, \mathcal{F}_t)$ coincide, i.e.
			\[
			\mathbb{Q}^t_{\mathbf{A}, \mathbf{x}_0} = \mathbb{P}^t_{\mathbf{A}, \mathbf{x}_0} \quad \mathrm{and} \quad \mathbb{Q}^t_{\mathbf{A}^{(0)}, \mathbf{x}_0} = \mathbb{P}^t_{\mathbf{A}^{(0)}, \mathbf{x}_0}.
			\]
			Note that for $\omega \in \Omega$, $s\geq 0$
			\begin{align*}
				\mathbf{B}^t_{\mathbf{A}, s}(\omega) = \mathbf{B}^t_{\mathbf{A}^{(0)}, s}(\omega) + (\mathcal{A}_\mathbf{A} - \mathcal{A}_{\mathbf{A}^{(0)}}) \int_0^s \omega(u) du = \mathbf{B}^t_{\mathbf{A}^{(0)}, s}(\omega) + \tilde{\mathcal{A}}_\mathbf{A} \int_0^s \omega(u-) du.
			\end{align*}
			% \begin{align*}
				%     \mathbf{B}^t_{\mathbf{A}, s} = \mathbf{B}^t_{\mathbf{A}^{(0)}, s} + (\mathcal{A}_\mathbf{A} - \mathcal{A}_{\mathbf{A}^{(0)}}) \int_0^s \mathbf{X}_{u} du = \mathbf{B}^t_{\mathbf{A}^{(0)}, s} + \tilde{\mathcal{A}}_\mathbf{A} \int_0^s \mathbf{X}_{u-} du.
				% \end{align*}
			Due to the specific structure of $\tilde{\mathcal{A}}_\mathbf{A}:= \mathcal{A}_\mathbf{A} - \mathcal{A}_{\mathbf{A}^{(0)}}$, we can write $\forall \mathbf{x}\in\mathbb{R}^{pd}$
			\[ \tilde{\mathcal{A}}_\mathbf{A} \mathbf{x} = \tilde{\Sigma} \tilde{\Sigma}^{-1} \tilde{\mathcal{A}}_\mathbf{A} \mathbf{x}, \]
			where we define the pseudo-inverse $\tilde{\Sigma}^{-1}$ for $\tilde \Sigma$ as
			\[\tilde{\Sigma}^{-1} := 
			\begin{pmatrix}
				0_{d \times d} & \cdots & 0_{d \times d} & 0_{d \times d} \\
				\vdots & \ddots & \vdots & \vdots \\
				0_{d \times d} & \cdots & 0_{d \times d} & 0_{d \times d} \\
				0_{d \times d} &  \cdots & 0_{d \times d} & \Sigma^{-1} \\
			\end{pmatrix}.
			\]
			We can thus define the $\mathbb{R}^{pd}$-valued predictable process $\mathbb{Z}_\mathbf{A} = \{\mathbf{Z}_{\mathbf{A},s},\ s\geq 0\}$ on $(\Omega, \mathcal{F}, \{\mathcal{F}_{s},\ s\geq0\})$ such that $\mathbf{Z}_{\mathbf{A},s} (\omega)= \tilde{\Sigma}^{-1} \tilde{\mathcal{A}}_\mathbf{A} \omega({s-})$ and for $s\geq0$, $\omega\in\Omega$
			\begin{equation} \label{eqn:B_change_2} 
				\mathbf{B}^t_{\mathbf{A}, s}(\omega) = \mathbf{B}^t_{\mathbf{A}^{(0)}, s}(\omega) + \int_0^s dC^t_u \mathbf{Z}_{\mathbf{A}, u}(\omega).
			\end{equation}
			The results from \citet{Jacod_Memin_1976} can be immediately generalized to the multivariate setting, cf.\ \citet{Sorensen_1991}. We can thus check the hypothesis of \citet[Theorem~4.2.(c)]{Jacod_Memin_1976} for $\mathbb{Q}_{\mathbf{A}, \mathbf{x}_0}$ and $\mathbb{Q}_{\mathbf{A}^{(0)}, \mathbf{x}_0}$:
			\begin{itemize}
				\item For $s\in\mathbb{R}_+$ let
				\[A^t_{\mathbf{A}, s}(\omega) := \int_0^s  \mathbf{Z}^\mathrm{T}_{\mathbf{A},u}(\omega) dC^t_u \mathbf{Z}_{\mathbf{A}, u}(\omega) = \int_0^{s\wedge t} \omega(u)^\mathrm{T} \tilde{\mathcal{A}}_\mathbf{A}^\mathrm{T} \tilde{\Sigma}^{-1} \tilde{\mathcal{A}}_\mathbf{A} \omega(u) du, \]
				and define the $(\Omega, \mathcal{F}, \{\mathcal{F}_s,\ s\geq 0\})$-stopping times (not that, by construction, this is a strict stopping time in the sense of \citet[Definition~III.2.35]{Jacod_Shiryaev_1987}): 
				\[\tau_n(\omega) = \inf\{s\geq 0 :A^t_{\mathbf{A}, s}(\omega) \geq n\}. \]
				In order to check $\tau_n$-uniqueness of the martingale problem $(\mathbb{X}, (\mathbf{B}_\mathbf{A}^t, C^t, \nu^t), \delta_{\mathbf{x}_0})$, in the sense of \citet[Section~4]{Jacod_Memin_1976}, we show the (stronger) condition of local uniqueness, as defined in \citet[Definition~III.2.37]{Jacod_Shiryaev_1987}. Due to the Markovian structure of the martingale problem we can apply \citet[Theorem~III.2.40]{Jacod_Shiryaev_1987} to the characteristics $(\mathbf{B}_\mathbf{A}^t, C^t, \nu^t)$ to deduce local uniqueness (note that, unlike \citet{Sorensen_1991}, we cannot apply \citet[Corollary~III.2.41]{Jacod_Shiryaev_1987} directly because \citet[Condition~C]{Sorensen_1991} is not satisfied). For $s\geq 0, u\geq 0, \omega\in\Omega$ define
				\begin{align*}
					&(p_s \mathbf{B}_\mathbf{A}^t)_u(\omega) = \mathbf{B}_{\mathbf{A},u}^{(t-s)\vee 0}(\omega) = \mathcal{A}_{\mathbf{A}}\int_0^{u\wedge [(t-s)\vee 0]} \omega(r) dr + \mathcal{E} \mathbf{b} (u\wedge [(t-s)\vee 0]), \\
					&(p_s C^t)_u(\omega) = C_u^{(t-s)\vee 0}(\omega) = \tilde{\Sigma} (u \wedge [(t-s)\vee 0]), \\
					&(p_s \nu^t)(\omega; du, dx) = \nu^{(t-s)\vee 0}(\omega; du, dx) = \mathds{1}_{[0,(t-s)\vee 0]}(u)\, du\, \delta_0 (dx_{-p})\, F(dx_{p}).
				\end{align*}
				Note that for fixed $u\geq 0$ and $A\in \mathcal{B}(\mathbb{R}_+ \times \mathbb{R}^{pd})$ the maps 
				\[(s,\omega)\mapsto (p_s\mathbf{B}^t_{\mathbf{A}})_u(\omega), \ (s,\omega)\mapsto (p_sC^t)_u(\omega),\ \mathrm{and}\  (s,\omega)\mapsto (p_s \nu^t)(\omega; A)\]
				are clearly $\mathcal{F} \otimes \mathcal{B}(\mathbb{R}_+)$-measurable. Moreover, letting $\theta: \Omega \rightarrow \Omega$ denote the shift operator, it can be easily verified that the following relations are satisfied $\forall s,u\geq 0$, $\forall\omega\in\Omega$, $\forall A \in  \mathcal{B}(\mathbb{R}^{pd})$
				\begin{align*}
					&(p_s \mathbf{B}_\mathbf{A}^t)_u(\theta_s\omega) = \mathbf{B}_{\mathbf{A},s+u}^{t}(\omega) -  \mathbf{B}_{\mathbf{A},s}^{t}(\omega), \\
					&(p_s C^t)_u(\theta_s\omega) = C_{s+u}^{t}(\omega) - C_{s}^{t}(\omega), \\
					&(p_s \nu^t)(\theta_s\omega; (0,u] \times A) = \nu^{t}(\omega; (s, s+u] \times A), 
				\end{align*}
				i.e.\ the mapping defined by $p_s$ satisfies \citet[Assumption~III.2.39]{Jacod_Shiryaev_1987} (moreover \citet[Assumption~III.2.13]{Jacod_Shiryaev_1987} is implicit as we are working on the canonical space). Next, we note that, by Lemma \ref{lemma:uniqueness_mart_problem}, the martingale problem for 
				\[(\mathbb{X}, (\mathbf{B}_\mathbf{A}^{(t-s)\vee 0}, C^{(t-s)\vee 0}, \nu^{(t-s)\vee 0}), \delta_{x})\]
				admits a unique solution, $\mathbb{Q}_{\mathbf{A}, s, x}$, which, for fixed $\mathbf{A}$, can be shown to define a transition kernel
				from $(\mathbb{R}_+ \times \mathbb{R}^{pd}, \mathcal{B}(\mathbb{R}_+ \times \mathbb{R}^{pd}))$ to $(\Omega, \mathcal{F})$ by
				\[\mathcal{Q}_{s,x} (d\omega) = \mathbb{Q}_{\mathbf{A}, s, x} (d\omega),\] 
				cf.\ Lemma \ref{lemma:transition_kernel}. Moreover, by setting $s=0$, we know that $(\mathbb{X}, (\mathbf{B}_\mathbf{A}^{t}, C^{t}, \nu^{t}), \delta_{\mathbf{x}_0})$ admits a unique solution. All hypothesis of \citet[Theorem~III.2.40]{Jacod_Shiryaev_1987} are thus satisfied and we can hence conclude that the martingale problem for $(\mathbb{X}, (\mathbf{B}_\mathbf{A}^{t}, C^{t}, \nu^{t}), \delta_{\mathbf{x}_0})$ is locally unique. In particular, we have $\tau_n$-uniqueness in the sense of \citet[Section~4]{Jacod_Memin_1976} for $(\mathbb{X}, (\mathbf{B}_\mathbf{A}^{t}, C^{t}, \nu^{t}), \delta_{\mathbf{x}_0})$ for all $n$.
				\item The initial conditions are both $\delta_{\mathbf{x}_0}$ and thus trivially equivalent.
				\item Finally we note that 
				% 	\[ A^t_\infty =  \int_0^t  \mathbf{X}^\mathrm{T}_{s-} \tilde{\mathcal{A}}_\mathbf{A}^\mathrm{T} \tilde{\Sigma}^{-1} \tilde{\mathcal{A}}_\mathbf{A} \mathbf{X}_{s-} ds \leq t \|\Sigma^{-1}\|\|\tilde{\mathcal{A}}_\mathbf{A}\| \sup_{s\in[0,t]} \|\mathbf{X}_{s-}\|^2 < \infty \]
				\begin{equation} \label{eqn:A^t_infty}
					A^t_{\mathbf{A}, \infty} =  \int_0^t  \omega(u)^\mathrm{T} \tilde{\mathcal{A}}_\mathbf{A}^\mathrm{T} \tilde{\Sigma}^{-1} \tilde{\mathcal{A}}_\mathbf{A} \omega(u) du \leq   \|\Sigma^{-1}\|\|\tilde{\mathcal{A}}_\mathbf{A}\| \int_0^t \|\omega(u)\|^2 du < \infty, \end{equation}
				$\mathbb{Q}_{\mathbf{A}, \mathbf{x}_0}$ and $\mathbb{Q}_{\mathbf{A}^{(0)}, \mathbf{x}_0}$ almost surely, where $\|\cdot\|$ denote appropriate arbitrary norms on the underlying (finite dimensional) spaces. Since the driving \Levy process $\mathbb{L}$ is assumed to have finite second moments on $(\Omega', \mathcal{F}', \{\mathcal{F}'_s, \ s\geq 0\}, \mathbb{P}')$, we have that, for $s\leq t$, a solution of the stopped SDE \eqref{eqn:SDE_stopped} with initial condition $\mathbf{X}_{0} = \mathbf{x}_0$ has law $\mathbb{P}'(\mathbf{X}^t_{\mathbf{A},\mathbf{x}_0,s} \in \cdot) = \mathbb{P}'(\mathbf{X}_{\mathbf{A},\mathbf{x}_0,s} \in \cdot) = P_{s, \mathbf{x}_0, \mathbf{A}}(\cdot)$ which is infinitely  divisible with second moments and characteristics $(b_{s,\mathbf{x}_0, \mathbf{A}}, \Sigma_{s, \mathbf{A}}, F_{s, \mathbf{A}})$ given by
				\begin{align} \label{eqn:triplet_X}
					\begin{split}
						b_{s,\mathbf{x}_0,\mathbf{A}} &= e^{s\mathcal{A}_\mathbf{A}} \mathbf{x}_0 + \int_0^s e^{u \mathcal{A}_\mathbf{A}} \mathcal{E} \mathbf{b} \, du \\
						&\quad\quad+\int_{\mathbb{R}^{pd}}\int_0^s e^{u \mathcal{A}_\mathbf{A}} z\left(\mathds{1}_{\{y: \|y\| \leq 1\}}\left(e^{u \mathcal{A}_\mathbf{A}} z\right) - \mathds{1}_{\{y: \|y\| \leq 1\}}(z)\right) du\, \delta_0(dz_{-p})\,F(dz_p), \\
						\Sigma_{s, \mathbf{A}} &= \int_0^s e^{s \mathcal{A}_\mathbf{A}} \mathcal{E}\Sigma\mathcal{E}^\mathrm{T}  e^{u \mathcal{A}^{\mathrm{T}}_\mathbf{A}} du, \\
						F_{s,\mathbf{A}}(E) &= \int_{\mathbb{R}^{pd}}  \int_0^s \mathds{1}_{E}\left(e^{u\mathcal{A}_{\mathbf{A}}}z\right) du \, \delta_0(dz_{-p})\, F(dz_p), \quad \forall E\in\mathcal{B}(\mathbb{R}^{pd}),
					\end{split}
				\end{align}
				where $(\mathbf{b}, \Sigma, F)$ are the characteristics of the driving \Levy process $\mathbb{L} = \{\mathbf{L}_{s},\ s\geq 0\}$, cf. Theorem 3.1 Sato and Yamazato 1984. Thus we can bound uniformly in $s\in[0,t]$
				\begin{align} \label{eqn:L2_bound}
					\mathbb{E}^{\mathbb{P}'}\left[\|\mathbf{X}_{\mathbf{A},\mathbf{x}_0,s}\|^2\right] &= \mathrm{tr}(\mathrm{Var}^{\mathbb{P}'}[\mathbf{X}_{\mathbf{A},\mathbf{x}_0,s}]) + \|\mathbb{E}^{\mathbb{P}'}[\mathbf{X}_{\mathbf{A},\mathbf{x}_0,s}]\|^2 \nonumber\\
					&= \mathrm{tr}(\Sigma_{s, \mathbf{A}}) + \int_{\mathbb{R}^{pd}} \|z\|^2 F_{s, \mathbf{A}}(dz) + \left\|b_{s,x, \mathbf{A}} + \int_{\{z:\|z\|>1\}}  z F_{s, \mathbf{A}}(dz)\right\|^2& \\
					&\leq \kappa(t, \mathbf{x}_0, \mathcal{A}_\mathbf{A}, \mathbf{b}, \Sigma, F), \nonumber
				\end{align}
				for all $s\in[0,t]$ where $\kappa(t) = \kappa(t, \mathbf{x}_0, \mathcal{A}_\mathbf{A}, \mathbf{b}, \Sigma, F)$ is a constant which depends on $t, \mathbf{x}_0, \mathcal{A}_\mathbf{A}, \mathbf{b}, \Sigma,$ and $F$. Thus we can show that $A^t_{\mathbf{A},\infty}$ is bounded in $L^1(\Omega, \mathcal{F}, \mathbb{Q}_{\mathbf{A}, \mathbf{x}_0})$
				\begin{multline*}
					\mathbb{E}^{\mathbb{Q}_{\mathbf{A}, \mathbf{x}_0}}[A^t_{\mathbf{A},\infty}] \leq \|\Sigma^{-1}\|\|\tilde{\mathcal{A}}_\mathbf{A}\| \int_0^t \mathbb{E}^{\mathbb{Q}_{\mathbf{A}, \mathbf{x}_0}}\left[\|\omega(s)\|^2 \right] ds \\
					= \|\Sigma^{-1}\|\|\tilde{\mathcal{A}}_\mathbf{A}\|\int_0^t \mathbb{E}^{\mathbb{P}'}\left[\|\mathbf{X}_{\mathbf{A},\mathbf{x}_0,s}\|^2\right] ds \leq t \|\Sigma^{-1}\|\|\tilde{\mathcal{A}}_\mathbf{A}\| \kappa(t, \mathbf{A}) < \infty,
				\end{multline*} 
				and hence
				\[A^t_{\mathbf{A},\infty} <\infty, \quad \text{$\mathbb{Q}_{\mathbf{A}, \mathbf{x}_0}$-a.s.}\]
				Similarly, we get 
				\begin{equation*}
					\mathbb{E}^{\mathbb{Q}_{\mathbf{A}^{(0)}, \mathbf{x}_0}}[A^t_{\mathbf{A},\infty}] \leq t \|\Sigma^{-1}\|\|\tilde{\mathcal{A}}_\mathbf{A}\| \kappa(t, \mathbf{A}^{(0)}) < \infty,
				\end{equation*} 
				and hence
				\[A^t_{\mathbf{A},\infty} <\infty, \quad \text{$\mathbb{Q}_{\mathbf{A}^{(0)}, \mathbf{x}_0}$-a.s.}\]
			\end{itemize}
			The assumptions of \citet[Theorem~4.2.(c)]{Jacod_Memin_1976} for $\mathbb{Q}_{\mathbf{A}, \mathbf{x}_0}$ and $\mathbb{Q}_{\mathbf{A}^{(0)}, \mathbf{x}_0}$ are satisfied and hence we have
			\[ \mathbb{Q}_{\mathbf{A}, \mathbf{x}_0} \sim \mathbb{Q}_{\mathbf{A}^{(0)}, \mathbf{x}_0} \]
			which by construction implies
			\[ \mathbb{P}^t_{\mathbf{A}, \mathbf{x}_0} \sim \mathbb{P}^t_{\mathbf{A}^{(0)}, \mathbf{x}_0}. \]
			In order to deduce the form of the likelihood we apply \citet[Theorem~4.5.(b)]{Jacod_Memin_1976} to $\mathbb{Q}_{\mathbf{A}, \mathbf{x}_0}$ and $\mathbb{Q}_{\mathbf{A}^{(0)}, \mathbf{x}_0}$. Note that, by the arguments above, $\mathbb{Q}_{\mathbf{A}^{(0)}, \mathbf{x}_0}$ is the (unique) solution to the martingale problem $(\mathbb{X}, (\mathbf{B}_{\mathbf{A}^{(0)}}^t, C^t, \nu^t), \delta_{\mathbf{x}_0})$ on the canonical space $(\Omega, \mathcal{F}, \{\mathcal{F}_s, \ s\geq 0\})$. Moreover we have that $\mathbb{Q}_{\mathbf{A}, \mathbf{x}_0}$ is the unique solution to the martingale problem $(\mathbb{X}, (\mathbf{B}_{\mathbf{A}}^t, C^t, \nu^t), \delta_{\mathbf{x}_0})$ on $(\Omega, \mathcal{F}, \{\mathcal{F}_s, \ s\geq 0\})$ and hence satisfies the martingale representation condition  \citep[Theorem~4.4]{Jacod_Memin_1976}. Finally we note that by construction $\{A^t_{\mathbf{A},s},\ s\geq 0\}$ has continuous trajectories $\mathbb{Q}_{\mathbf{A}^{(0)}, \mathbf{x}_0}$-a.s. and hence for $\mathbb{Q}_{\mathbf{A}^{(0)}, \mathbf{x}_0}$-a.e. $\omega \in \Omega$ and $\forall s\geq 0$
			\begin{align*}
				\mathbb{E}^{\mathbb{Q}_{\mathbf{A}^{(0)}, \mathbf{x}_0}}\left[ \frac{d \mathbb{Q}_{\mathbf{A}, \mathbf{x}_0}}{d\mathbb{Q}_{\mathbf{A}^{(0)}, \mathbf{x}_0}} \Big| {\mathcal{F}_s}\right](\omega)
				&= \exp\Big\{ - \int_0^s \mathbf{Z}^{\mathrm{T}}_{\mathbf{A}, u} \, d \mathbf{X}_{\mathbf{A}^{(0)}, u}^c(\omega) - \frac{1}{2} \int_0^s  \mathbf{Z}^\mathrm{T}_{\mathbf{A},u}(\omega) dC^t_u \mathbf{Z}_{\mathbf{A}, u}(\omega)  \Big\} \\
				&= \exp\Big\{ - \int_0^s \mathbf{X}^\mathrm{T}_{u-} \tilde{\mathcal{A}}_\mathbf{A}^\mathrm{T} \tilde{\Sigma}^{-1}\ d \mathbf{X}_{\mathbf{A}^{(0)}, u}^c - \frac{1}{2} \int_0^{s\wedge t} \mathbf{X}^\mathrm{T}_{u}  \tilde{\mathcal{A}}_\mathbf{A}^\mathrm{T} \tilde{\Sigma}^{-1}  \tilde{\mathcal{A}}_\mathbf{A} \mathbf{X}_u du  \Big\},
			\end{align*}
			and noting that, by construction, if we set $s=t$, for $\mathbb{P}^t_{\mathbf{A}^{(0)}, \mathbf{x}_0}$-a. e. $\omega\in\Omega$ 
			\[
			\mathbb{E}^{\mathbb{Q}_{\mathbf{A}^{(0)}, \mathbf{x}_0}}\left[ \frac{d \mathbb{Q}_{\mathbf{A}, \mathbf{x}_0}}{d\mathbb{Q}_{\mathbf{A}^{(0)}, \mathbf{x}_0}} \Big| {\mathcal{F}_t}\right](\omega) 
			= \frac{d \mathbb{Q}^t_{\mathbf{A}, \mathbf{x}_0}}{d\mathbb{Q}^t_{\mathbf{A}^{(0)}, \mathbf{x}_0}}(\omega) 
			=  \frac{d \mathbb{P}^t_{\mathbf{A}, \mathbf{x}_0}}{d\mathbb{P}^t_{\mathbf{A}^{(0)}, \mathbf{x}_0}}(\omega) 
			= \mathbb{E}^{\mathbb{P}_{\mathbf{A}^{(0)}, \mathbf{x}_0}}\left[ \frac{d \mathbb{P}_{\mathbf{A}, \mathbf{x}_0}}{d\mathbb{P}_{\mathbf{A}^{(0)}, \mathbf{x}_0}} \Big| {\mathcal{F}_t}\right](\omega), \]
			which completes the proof.
		\end{proof}
		
		\begin{remark}
			In the general setting of \citet[Theorem~4.2~and~Theorem~4.5]{Jacod_Memin_1976} the condition $A_{\mathbf{A},t} = A^t_{\mathbf{A},\infty} < \infty$, $\mathbb{Q}_{\mathbf{A}, \mathbf{x}_0}$-a.s.\ ensures the stochastic integral
			\[ (\mathbf{Z}^{\mathrm{T}}_{\mathbf{A}} \cdot \mathbf{X}_{\mathbf{A}}^c)_s := \int_0^s \mathbf{Z}^{\mathrm{T}}_{\mathbf{A}, u} \, d \mathbf{X}_{\mathbf{A}^{(0)}, u}^c = \int_0^s \mathbf{X}^\mathrm{T}_{u-} \tilde{\mathcal{A}}_\mathbf{A}^\mathrm{T} \tilde{\Sigma}^{-1}\ d \mathbf{X}_{\mathbf{A}^{(0)}, u}^c, \]
			is well-defined for $s\leq t$. Note that one can write
			\[ \mathbf{X}_{\mathbf{A}^{(0)}, s}^c = \mathbf{X}_{\mathbf{A}, s}^c + (\mathbf{B}^t_{\mathbf{A}^{(0)},s} - \mathbf{B}^t_{\mathbf{A}, s}), \quad s\in[0,t],\]
			where $\{\mathbf{X}_{\mathbf{A}, s}^c, s\in[0, t]\}$ is a continuous local martingale under $\mathbb{Q}_{\mathbf{A}, \mathbf{x}_0}$ and $\{\mathbf{B}^t_{\mathbf{A}^{(0)}, s} - \mathbf{B}^t_{\mathbf{A}, s},\ s\in[0, t]\}$ is a finite variation process. Then one defines
			\[ \int_0^s \mathbf{Z}^{\mathrm{T}}_{\mathbf{A}, u} \, d \mathbf{X}_{\mathbf{A}^{(0)}, u}^c = \int_0^s \mathbf{Z}^{\mathrm{T}}_{\mathbf{A}, u} \, d \mathbf{X}_{\mathbf{A}, u}^c + \int_0^s \mathbf{Z}^{\mathrm{T}}_{\mathbf{A}, u} \, d (\mathbf{B}^t_{\mathbf{A}^{(0)}, u} - \mathbf{B}^t_{\mathbf{A}, u}), \quad s\in[0,t],\]
			where the first integral is well-defined as the integrand is in \[ \mathcal{P}^*(\mathbf{X}_{\mathbf{A}}^c; \Omega, \mathcal{F}, \{\mathcal{F}_s,\ s\geq 0\}, \mathbb{Q}_{\mathbf{A}, \mathbf{x}_0}) := \left\{\mathbb{H}\ \mathrm{prog.\, meas.\, s.t. }\ \mathbb{Q}_{\mathbf{A}, \mathbf{x}_0}\left(\int_0^t\mathbf{H}_s^{\mathrm{T}}d\langle \mathbf{X}_{\mathbf{A}}^c\rangle_s \mathbf{H}_s < \infty \right) = 1\right\},\] 
			cf.\ \citet[Theorem~3.2.23]{Karatzas_Shreve_1996} Theorem 3.2.23, and the second integral is defined for any optional (and thus any predictable) integrand, cf.\ \citet[Definition~3.4]{Jacod_Shiryaev_1987}. In our specific setting the processes $\{\mathbf{X}_{\mathbf{A}, s}^c, s\in[0, t]\}$ and $\{\mathbf{Z}_{\mathbf{A}, s}, s\in[0, t]\}$ satisfy stricter conditions: the first is actually a true square-integrable martingale under $\mathbb{Q}_{\mathbf{A}, \mathbf{x}_0}$ since $\mathbb{E}^{\mathbb{Q}_{\mathbf{A},\mathbf{x}_0}}[\langle \mathbf{X}_{\mathbf{A}}^c\rangle_s] = \tilde{\Sigma} s <\infty $, \citet[Corollary~II.3]{Protter_1990}, and the latter is in 
			\[ \mathcal{L}^*(\mathbf{X}_{\mathbf{A}}^c; \Omega, \mathcal{F}, \{\mathcal{F}_s,\ s\geq 0\}, \mathbb{Q}_{\mathbf{A}, \mathbf{x}_0}) := \left\{\mathbb{H}\ \mathrm{prog.\, meas.\, s.t. } \ \mathbb{E}_{\mathbb{Q}_{\mathbf{A}, \mathbf{x}_0}}\left[\int_0^t\mathbf{H}_s^{\mathrm{T}}d\langle \mathbf{X}_{\mathbf{A}}^c\rangle_s \mathbf{H}_s \right] < \infty\right\},\]
			since $A^t_{\mathbf{A}, \mathbf{x}_0}$ is bounded in $L^1(\Omega, \mathcal{F}, \mathbb{Q}_{\mathbf{A},\mathbf{x}_0})$. We can hence apply \citet[Definition~3.2.9]{Karatzas_Shreve_1996} (or from the original work of \citet{Kunita_Watanabe_1967}) to deduce that the stochastic integral
			\[ (\mathbf{Z}^{\mathrm{T}}_{\mathbf{A}} \cdot \mathbf{X}_{\mathbf{A}}^c)_s = \int_0^s \mathbf{Z}^{\mathrm{T}}_{\mathbf{A}, u} \, d \mathbf{X}_{\mathbf{A}, u}^c, \quad s\in[0, t],\]
			is not just well-defined but is actually a continuous square-integrable martingale. This observation is fundamentally linked to Lemma \ref{lemma:loc_mart_MCAR}: the process $(\mathbb{Z}^{\mathrm{T}}_{\mathbf{A}} \cdot \mathbb{X}_{\mathbf{A}}^c)$ is the sum of the components of the score \eqref{eqn:score} and is proved to be a $\mathbb{P}_{\mathbf{A}, \mathbf{x}_0}$-square-integrable martingale (i.e.\ a $\mathbb{Q}_{\mathbf{A}, \mathbf{x}_0}$-square-integrable martingale up to time $t$).
		\end{remark}
		
		\begin{lemma} \label{lemma:transition_kernel}
			Fix $\mathbf{A}\in(\mathcal{M}_d(\mathbb{R}))^p$ and let $\mathcal{Q}:\mathbb{R}_+ \times \mathbb{R}^{pd} \times \mathcal{F} \rightarrow [0,1]$ be given by 
			\[(s, x, A) \mapsto \mathcal{Q}_{s,x} (A) = \mathbb{Q}_{\mathbf{A}, s, x} (A),\]
			where $\mathbb{Q}_{\mathbf{A}, s, x}$ is the unique solution to the martingale problem for 
			\[(\mathbb{X}, (\mathbf{B}_\mathbf{A}^{(t-s)\vee 0}, C^{(t-s)\vee 0}, \nu^{(t-s)\vee 0}), \delta_{x}).\] Then $\mathcal{Q}$ defines a transition kernel from $(\mathbb{R}_+ \times \mathbb{R}^{pd}, \mathcal{B}(\mathbb{R}_+ \times \mathbb{R}^{pd}))$ to $(\Omega, \mathcal{F})$.
		\end{lemma}
		\begin{proof}
			% 			The fact this is indeed a kernel follows from \citet[Proposition~4.1]{Getoor_1975} by noting that $(\Omega,\mathcal{F})$ is Polish and hence a Lusin space, and the map\footnote{Here $\mathcal{B}_b(E;\mathbb{R})$ denotes the space of bounded $\mathbb{R}$-valued measurable functions on the measurable space $(E, \mathcal{B}(E))$ where $\mathcal{B}(E)$ is the Borel $\sigma$-algebra corresponding to the standard topology on $E$ (for $E=\mathbb{R}$ this is the topology generated by open intervals while for $E=\Omega = D_{\mathbb{R}^{pd}}([0,\infty))$ we take this to be the Skorokhod topology).} $T:\mathcal{B}_b(\Omega;\mathbb{R})\rightarrow \mathcal{B}_b(\mathbb{R}_+\times \mathbb{R}^{pd};\mathbb{R})$ such that $f \mapsto Tf $ where
			% 			\[ (Tf)(t,x) = \int_\Omega f(\omega) \mathcal{Q}_{t,x}(d\omega) = \mathbb{E}_{\mathbf{A}, t, x}[f(\mathbb{X})] = \mathbb{E}'\left[f\left(\mathbb{X}_{x, \mathbf{A}}^{(T-t)\vee 0}\right)\right] \]
			% 			for $(t,x)\in\mathbb{R}_+\times\mathbb{R}^{pd}$ is linear, positive, and, by the monotone convergence theorem, satisfies \citet[Property~(3.3)]{Getoor_1975}, i.e.\ if $(f_n)_n\subset\mathcal{B}_b(\Omega;\mathbb{R})$ with $0\leq f_n\uparrow f \in \mathcal{B}_b(\Omega;\mathbb{R})$ then $Tf_n\uparrow Tf$. Note, missing most important part, i.e.\ $Tf\in\mathcal{B}_b(\mathbb{R}_+\times \mathbb{R}^{pd};\mathbb{R})$.
			Clearly, for fixed $(s,x)$, $\mathcal{Q}_{s,x}(\cdot)$ defines a probability measure on $(\Omega, \mathcal{F})$. 
			Next, to show that for fixed $A\in\mathcal{F}$, $(s, x) \mapsto Q_{s, x}(A)$ is $\mathcal{B}(\mathbb{R}_+ \times \mathbb{R}^{pd})$-measurable, we use the $\pi-\lambda$ theorem. Let $\mathcal{C}\subset\mathcal{F}$ denote the $\pi$-system of finite cylinder sets, i.e.\ $C\in\mathcal{C}$ of the form 
			\[C = \{\omega\in\Omega\ \mathrm{s.t.}\ \omega(s_1) \in A_1, \ldots, \omega(s_n) \in A_N\}\in\mathcal{F},\]
			for any $n\in\mathbb{N}, s_1< \ldots< s_n\in\mathbb{R}_{+}$ and $A_1,\ldots, A_n\in\mathcal{B}(\mathbb{R}^{pd})$. Define \[\mathcal{L}=\{L\in\mathcal{F} \ \mathrm{s.t.}\ (s,x)\mapsto \mathcal{Q}_{s, x}(L)\ \mathrm{is}\ \mathcal{B}(\mathbb{R}_+\times \mathbb{R}^{pd})\mathrm{-measurable}\}.\] 
			This is a $\lambda$-system, since:
			\begin{enumerate}
				\item $\Omega\in\mathcal{L}$ since $(s, x) \mapsto \mathcal{Q}_{s, x}(\Omega) = \mathbb{Q}_{\mathbf{A},s,x}(\Omega) = 1$ is trivially measurable;
				\item if $A, B\in\mathcal{L}$ and $A\subseteq B$ then $B\setminus A\in \mathcal{L}$ since 
				\[(s, x) \mapsto \mathcal{Q}_{s, x}(B\setminus A) = \mathbb{Q}_{\mathbf{A},s,x}(B\setminus A) = \mathbb{Q}_{\mathbf{A},s,x}(B) -\mathbb{Q}_{\mathbf{A},s,x}(A) = \mathcal{Q}_{s, x}(B) - \mathcal{Q}_{s, x}(A)\]
				is the difference of measurable functions and thus measurable;
				\item if $A_1\subseteq A_2\subseteq \cdots \in \mathcal{L}$ then $\bigcup_{n\geq 1} A_n\in\mathcal{L}$ since
				\begin{align*}(s, x) \mapsto \mathcal{Q}_{s, x}(\bigcup_{n\geq 1} A_n) &= \mathbb{Q}_{\mathbf{A},s,x}(\bigcup_{n\geq 1} A_n)
					% 	\\ &= \mathbb{Q}_{\mathbf{A},s,x}(A_1) + \sum_{n\geq 1} \mathbb{Q}_{\mathbf{A},s,x}(A_{n+1} \setminus A_n) \\ &= \mathbb{Q}_{\mathbf{A},s,x}(A_1) + \sum_{n\geq 1}\left[ \mathbb{Q}_{\mathbf{A},s,x}(A_{n+1}) -  \mathbb{Q}_{\mathbf{A},s,x}(A_n)\right] \\ &= \lim_{n\rightarrow\infty} \mathbb{Q}_{\mathbf{A},s,x}(A_{n}) \\ 
					% 	&
					= \lim_{n\rightarrow\infty} \mathcal{Q}_{s,x}(A_{n})
				\end{align*}
				is the pointwise limit of measurable functions and thus measurable.
			\end{enumerate}
			Moreover $\mathcal{C}\subset \mathcal{L}$ since for $C\in\mathcal{C}$ we can write
			\begin{align*}(s, x) \mapsto \mathcal{Q}_{s,x}(C) = \mathbb{Q}_{\mathbf{A},s,x}(C) &= \int_{A_1} \cdots \int_{A_n}  Q_{\mathbf{A}, s}(s_{n-1}, x_{n-1}, s_n, dx_n) \cdots Q_{\mathbf{A}, s}(0, x, s_1, dx_1) 
				% \\
				% 	&= \int_{A_1} \cdots \int_{A_k}  P_{\mathbf{A}}(x_k, (T-t)\vee 0 - t_k, A_{k+1}\cap \cdots \cap A_n) \\
				% 	&\quad \quad \quad \quad \quad \quad \quad \times P_{\mathbf{A}}(x_{k-1}, t_k-t_{k-1}, dx_k) \cdots P_{\mathbf{A}}(x, t_1, dx_1)
			\end{align*}
			where 
			%$k = \max\{j : t_j< (T-t)\vee 0\}$,
			$Q_{\mathbf{A}, s}(\cdot, \cdot, \cdot, \cdot): \mathbb{R}_{+} \times \mathbb{R}^{pd} \times \mathbb{R}_{+} \times \mathcal{B}(\mathbb{R}^{pd}) \rightarrow [0,1]$ are the time-inhomogenous Markov transition probabilities associated to $\mathbb{X}_{\mathbf{A}}^{(t-s)\vee 0}$ given by
			\[Q_{\mathbf{A}, s}(s_1, y, s_2, A) = 
			\begin{cases}
				P_\mathbf{A}(y, s_2-s_1, A), \quad &\mathrm{if}\ 0\leq s_1\leq s_2\leq (t-s)\vee 0, \\
				P_\mathbf{A}(y, (t-s)- s_1, A), \quad &\mathrm{if}\ 0 \leq s_1 \leq (t-s)\vee 0 \leq s_2,\\
				\mathds{1}_{A}(y), \quad &\mathrm{if}\ 0 \leq (t-s)\vee 0 \leq s_1 \leq s_2,\\
			\end{cases}\]
			for $0\leq s_1\leq s_2$, $y\in\mathbb{R}^{pd}$, $A\in\mathcal{B}(\mathbb{R}^{pd})$ and $P_\mathbf{A}(\cdot, \cdot, \cdot): \mathbb{R}^{pd} \times \mathbb{R}_{+} \times \mathcal{B}(\mathbb{R}^{pd}) \rightarrow [0,1]$ denoting the time-homogeneous Markov transition probabilities associated to $\mathbb{X}_{\mathbf{A},\mathbf{x}_0}$ given in \citet[Proposition~2.1]{Masuda_2004}, or \citet[Theorem~3.1]{Sato_Yamazato_1984}. The map $(s,x)\mapsto \mathcal{Q}_{s,x}(C)$ is:
			\begin{itemize}
				\item measurable in $x\in\mathbb{R}^{pd}$, using measurability of $P_\mathbf{A}(x, s_1, A)$/$P_\mathbf{A}(x, t-s, A)$/$\mathds{1}_{A}(x)$ to obtain measurability of $Q_{\mathbf{A}, s}(0, x, s_1, A)$ for general $A\in\mathcal{B}(\mathbb{R}^{pd})$ and then extending from simple functions to obtain measurability of general integrals w.r.t.\\ $Q_{\mathbf{A}, s}(0, x, s_1, dx_1)$,
				\item right-continuous in $s\in\mathbb{R}_+$, using continuity in $s$ of $P_\mathbf{A}$ and treating the cases $s<t$ and $s\geq t$ separately,
			\end{itemize}
			and thus jointly measurable in $(s, x)$. Note that $\mathcal{F} = \sigma(\mathcal{C})$ and by the $\pi-\lambda$ theorem we have $\sigma(\mathcal{C}) \subseteq \mathcal{L}$. We can hence conclude that measurability in $(s,x)$ holds for any $A\in\mathcal{F}$.
		\end{proof}
		
		\subsection{Proof of Lemma \ref{lemma:loc_mart_MCAR}} \label{app:proof_loc_mart_MCAR}
		
		\begin{proof} Under $\mathbb{P}_{\mathbf{A}^{(0)}, \mathbf{x}_0}$, the process $\mathbb{H}$ is a continuous local martingale as it given by the integral of a locally bounded predictable process against the $\mathbb{P}_{\mathbf{A}^{(0)}, \mathbf{x}_0}$ Brownian motion $D^{p-1}\mathbb{Y}^c_{\mathbf{A}^{(0)}}$, i.e.\ $\mathbb{H}$ is a semimartingale with characteristcs $(0, [\mathbf{H}], 0)$ under $\mathbb{P}_{\mathbf{A}^{(0)}, \mathbf{x}_0}$. Hence by \citet[Corollary~11.3.2]{Kuchler_Sorensen_1997}, $\mathbb{H}$ is a semimartingale with characteristcs $([\mathbf{H}]\mathrm{vec}(\mathbf{A}), [\mathbf{H}], 0)$ under $\mathbb{P}_{\mathbf{A}, \mathbf{x}_0}$ and thus the score $\{\mathbf{H}_t - [\mathbf{H}]_t \mathrm{vec}(\mathbf{A}),\ t\geq 0\}$ is a mean-zero continuous local martingale with quadratic variation $[\mathbf{H}]_t$. Since, moreover, the driving \Levy process $\mathbb{L}$ is assumed to be square-integrable the score has $\mathbb{P}_{\mathbf{A}, \mathbf{x}_0}$-integrable quadratic variation (here $\|\cdot\|$ denotes the Frobenius norm for matrices and $\|\cdot\|_\Sigma$ the norm induced by $\langle \cdot\, , \cdot \rangle_\Sigma$ for vectors):
			\begin{align*}
				\mathbb{E}^{\mathbb{P}_{\mathbf{A}, \mathbf{x}_0}}\left[ \|[\mathbf{H}]_t\| \right] &\leq \mathbb{E}^{\mathbb{P}_{\mathbf{A}, \mathbf{x}_0}}\left[\left(\sum_{l,k =0}^{p-1} \sum_{i, j = 1}^{d} \sum_{m, n = 1}^{d}  \left( \int_0^t D^{l}Y^{(i)}_s \Sigma^{-1}_{m, n} D^{k}Y^{(j)}_s \, ds \right)^2 \right)^{1/2}\right] \\
				& \leq d \|\Sigma^{-1} \| \mathbb{E}^{\mathbb{P}_{\mathbf{A}, \mathbf{x}_0}}\left[\left(\int_0^t \int_0^t  \sum_{l, k=0}^{p-1} \sum_{i, j= 1}^{d}  D^{l}Y^{(i)}_s D^{k}Y^{(j)}_s  D^{l}Y^{(i)}_r D^{k}Y^{(j)}_r \,ds dr \right)^{1/2}\right] \\
				& \leq d \|\Sigma^{-1} \| \mathbb{E}^{\mathbb{P}_{\mathbf{A}, \mathbf{x}_0}}\left[\left(\int_0^t \int_0^t  \left(\sum_{l=0}^{p-1} \sum_{i= 1}^{d}  D^{l}Y^{(i)}_s D^{l}Y^{(i)}_r \right)^2 \,ds dr \right)^{1/2}\right] \\
				& \leq d \|\Sigma^{-1} \| \mathbb{E}^{\mathbb{P}_{\mathbf{A}, \mathbf{x}_0}}\left[\left(\int_0^t \int_0^t  \sum_{l=1}^{p-1} \sum_{i= 1}^{d} \left( D^{l}Y^{(i)}_s D^{l}Y^{(i)}_r \right)^2 \,ds dr \right)^{1/2}\right] \\
				& \leq d \|\Sigma^{-1} \| \mathbb{E}^{\mathbb{P}_{\mathbf{A}, \mathbf{x}_0}}\left[\int_0^t \sum_{l=0}^{p-1} \sum_{i= 1}^{d} \left( D^{l}Y^{(i)}_s \right)^2 \,ds\right] \\
				& \leq d \|\Sigma^{-1} \| \int_0^t \mathbb{E}^{\mathbb{P}_{\mathbf{A}, \mathbf{x}_0}} [\|\mathbf{X}_s\|^2] ds = d \|\Sigma^{-1} \| \int_0^t \mathbb{E}^{\mathbb{P}'} [\|\mathbf{X}_{\mathbf{A},\mathbf{x}_0,s}\|^2] \, ds 
				<\infty,
			\end{align*}
			by noting in the last step that $\{\mathbf{X}_s,\ s\in[0,t]\}$ and $\{\mathbf{X}_{\mathbf{A},\mathbf{x}_0,s}, \ s\in[0,t]\}$ have the same law under $\mathbb{P}_{\mathbf{A}, \mathbf{x}_0}$ and $\mathbb{P}'$ respectively and using the bound in Equation \eqref{eqn:L2_bound}. Thus the $\mathbb{P}_{\mathbf{A}, \mathbf{x}_0}$ continuous local martingale $\{\mathbf{H}_t - [\mathbf{H}]_t \mathrm{vec}(\mathbf{A}),\ t\geq 0\}$ is indeed a $\mathbb{P}_{\mathbf{A}, \mathbf{x}_0}$ continuous square-integrable martingale by \citet[Corollary~II.3]{Protter_1990}. In particular, $\mathbf{H}_t$ and $[\mathbf{H}]_t$ are finite $\mathbb{P}_{\mathbf{A}, \mathbf{x}_0}$-a.s.
		\end{proof}
		\begin{remark}
			Alternatively one can explicitly substitute 
			\[D^{p-1}\mathbf{Y}_{\mathbf{A}^{(0)}, t}^c = D^{p-1}\mathbf{Y}_{\mathbf{A}, t}^c  - \sum_{j=1}^p\int_0^t A_j D^{p-j}\mathbf{Y}_s\ ds\]
			into the score. This yields an integral with respect to the $\mathbb{P}_{\mathbf{A}, \mathbf{x}_0}$-Brownian motion $D^{p-1}\mathbb{Y}_{\mathbf{A}}^c$, which one can then check to have quadratic variation $[\mathbf{H}]_t$.
		\end{remark}
		
		\subsection{Proof of Theorem \ref{thm:cons_asymp}} \label{app:proof_cons_asymp}
		\begin{proof}
			For $\phi\in\mathbb{R}^{pd^2}$ the log-likelihood is a quadratic form given by:
			\[ l_t(\phi) := \log \mathcal{L}(\mathrm{vec}^{-1}(\phi);\mathbb{Y}_{[0,t]}) = \phi^\mathrm{T} \mathbf{H}_t - \frac{1}{2} \phi^\mathrm{T} [\mathbf{H}]_t \phi, \]
			i.e.\ it has gradient
			\[ \nabla_{\phi} l_t(\phi) = \mathbf{H}_t - [\mathbf{H}]_t \phi, \]
			and Hessian
			\[\mathrm{Hess}_{\phi} l_t(\phi) = - [\mathbf{H}]_t. \]
			First, note that for any $t>0$ the (log)likelihood has finite coefficients $\mathbb{P}_{\mathbf{A}^*, \mathbf{x}_0}$-a.s.\ (see dicussion at the end of the proof of Lemma \ref{lemma:loc_mart_MCAR}). Moreover the matrix $[\mathbf{H}]_t$ is strictly positive definite $\mathbb{P}_{\mathbf{A}^*, \mathbf{x}_0}$-a.s. Let $\phi\in\mathbb{R}^{pd^2}$ with $\phi = (\phi^\mathrm{T}_{1,1}, \ldots, \phi^\mathrm{T}_{d, p})^\mathrm{T}$ for $\phi_{1,1}, \ldots, \phi_{d,p}\in\mathbb{R}^d$. Then
			\begin{align*}
				\phi^{\mathrm{T}} [\mathbf{H}]_t \phi = \int_0^t\|\sum_{i=1}^d \sum_{l=0}^{p-1} \phi_{ i, l} D^{l} Y^{(i)}_s\|_{\Sigma}\, ds \geq 0, \quad  \mathbb{P}_{\mathbf{A}^*, \mathbf{x}_0}\mathrm{-a.s.}
			\end{align*}
			with equality if and only if
			\[ \sum_{i=1}^d \sum_{l=0}^{p-1} \phi_{ i, l} D^{l} Y^{(i)}_s = 0, \quad  \mathrm{Leb}\otimes \mathbb{P}_{\mathbf{A}^*, \mathbf{x}_0}\mathrm{-a.s}.\]
			We define a collection of random variables $\{Z^i,\ i\in\mathcal{I}\}$ on $(\mathbb{P}', \Omega', \mathcal{F}')$ to be affinely independent under $\mathbb{P}'$ if for any $\xi\in\mathbb{R}^{\mathcal{I}}, \xi_0\in\mathbb{R}$
			\[\sum_{i\in\mathcal{I}} \xi_{i} Z^i = \xi_0 \quad        \mathbb{P}'\mathrm{-a.s.}\implies \xi, \xi_0 = 0. \]
			Thus, showing that under $\mathbb{P}_{\mathbf{A}^*, \mathbf{x}_0}$, $\{D^{l}Y_s^{(i)}, \ i=1,\ldots,d, l=0,\ldots, p-1 \}$ is an affinely independent collection of random variables for any $s\in(0,t]$ would imply $\phi \equiv 0$ and thus that $[\mathbf{H}]_t$ is strictly positive definite. Under $\mathbb{P}_{\mathbf{A}^{(0)}, \mathbf{x}_0}$, with $\mathbf{A}^{(0)} = (0_{d\times d}, \ldots, 0_{d\times d})$, we have that
			\[D^{p-1}{Y}_t^{(i)} = L_t^{(i)},\ D^{l}{Y}_t^{(i)} = \int_0^t D^{l+1}Y_s^{(i)} ds \quad  l= p-2, \ldots, 0, \]
			and one can check that the (independent) components of a non-degenerate \Levy process and their first $p-1$ time integrals are affinely independent under $\mathbb{P}_{\mathbf{A}^{(0)}, \mathbf{x}_0}$, and hence -- by equivalence of the restricted measures up to $t$ -- also under $\mathbb{P}_{\mathbf{A}^*, \mathbf{x}_0}$. Thus we have that the (log)likelihood is a strictly concave continuous function on $\mathbb{R}^{pd^2}$ with unique maximizer
			\[\hat{\phi}_t :=  [\mathbf{H}]_t^{-1}\mathbf{H}_t, \]
			and thus 
			\[\hat{\mathbf{A}}(\mathbb{Y}_{[0,t]}) = \mathrm{vec}^{-1} ( [\mathbf{H}]_t^{-1}\mathbf{H}_t). \]
			
			In the remaining part of the proof, we assume $\mathbf{A}^*\in\mathfrak{A}$. Under $\mathbb{P}_{\mathbf{A}^*, \mathbf{x}_0}$, the MCAR($p$) process $\mathbb{Y}$ and its state-space representation are ergodic. Thus, $\mathbb{P}_{\mathbf{A}^*, \mathbf{x}_0}$ almost surely 
			\begin{equation} \label{eqn:ergodic_H_MCAR} \frac{1}{t} [\mathbf{H}]_t \rightarrow \mathcal{H}_\infty, \end{equation}
			where $\mathcal{H}_\infty$ is a strictly positive definite matrix (by similar reasoning as for $[\mathbf{H}]_t$) given by
			\begin{equation} \label{eqn:H_infty}
				\mathcal{H}_\infty =
				\begin{pmatrix}
					\mathbb{E}_{\mathbf{A}^*} [ D^{p-1}Y^{(1)}_{\infty} \Sigma^{-1} D^{p-1}Y^{(1)}_{\infty}] & \cdots & \mathbb{E}_{\mathbf{A}^*} [ D^{p-1}Y^{(1)}_{\infty} \Sigma^{-1} Y^{(d)}_{\infty}]\\
					\vdots & \ddots & \vdots \\
					\mathbb{E}_{\mathbf{A}^*} [ Y^{(d)}_{\infty} \Sigma^{-1} D^{p-1}Y^{(1)}_{\infty}] & \cdots & \mathbb{E}_{\mathbf{A}^*} [ Y^{(d)}_{\infty} \Sigma^{-1} Y^{(d)}_{\infty}]  \\
				\end{pmatrix},
			\end{equation}
			and $\mathcal{L}([\mathbf{Y}^\mathrm{T}_\infty, \ldots, D^{p-1}\mathbf{Y}^\mathrm{T}_\infty]^\mathrm{T}) = \mathcal{L}( \int_0^\infty e^{s\mathcal{A}_{\mathbf{A}^*}} \mathcal{E}\, d\mathbf{L}_s) $ under $\mathbb{P}_{\mathbf{A}^*}$.
			Note that
			\[\mathrm{vec} (\hat{\mathbf{A}}(\mathbb{Y}_{[0,t]})) =  [\mathbf{H}]_t^{-1}\mathbf{H}_t = \mathrm{vec}(\mathbf{A}^*) +  [\mathbf{H}]_t^{-1}(\mathbf{H}_t - [\mathbf{H}]_t \mathrm{vec}(\mathbf{A}^*)), \]
			and by Lemma \ref{lemma:loc_mart_MCAR} the score $\mathbf{H}_t - [\mathbf{H}]_t \mathrm{vec}(\mathbf{A}^*)$ is a mean-zero square-integrable martingale under $\mathbb{P}_{\mathbf{A}^*, \mathbf{x}_0}$. By the central limit theorem for martingales, i.e.\ \citet[Theorem~A.7.7~and~Corollary~A.7.8]{Kuchler_Sorensen_1997}, applied to the score we have that 
			\[ [\mathbf{H}]_t^{-1}(\mathbf{H}_t - [\mathbf{H}]_t \mathrm{vec}(\mathbf{A}^*)) \overset{\mathbb{P}_{\mathbf{A}^*, \mathbf{x}_0}}{\longrightarrow} 0,\]
			and hence 
			\[\mathrm{vec} (\hat{\mathbf{A}}(\mathbb{Y}_{[0,t]})) \overset{\mathbb{P}_{\mathbf{A}^*, \mathbf{x}_0}}{\longrightarrow} \mathrm{vec}(\mathbf{A}^*), \ t\rightarrow \infty. \]
			% 		This implies\footnote{Since $\mathbf{A}^* \in \mathrm{int}(\mathfrak{A})$ there exists $\delta > 0$ such that $B_\delta(\mathbf{A}^*) \subset \mathfrak{A}$. Fix $\epsilon>0$, and write $\phi = \mathrm{vec}(\mathbf{A})$, then
				% 			\begin{align*} \mathbb{P}_{\mathbf{A}^*, \mathbf{x}_0}&(\|\hat{\phi}^{MLE}_t - \phi^* \| > \epsilon) = \\
					% 				&= \mathbb{P}_{\mathbf{A}^*, \mathbf{x}_0}(\|\hat{\phi}^{MLE}_t - \phi^* \| > \epsilon \, |\, \|\hat{\phi}_t - \phi^* \| > \delta) \mathbb{P}_{\mathbf{A}^*, \mathbf{x}_0}(\|\hat{\phi}_t - \phi^* \| > \delta)  + \mathbb{P}_{\mathbf{A}^*, \mathbf{x}_0}(\epsilon < \|\hat{\phi}_t - \phi^* \| \leq \delta) \\
					% 				&\leq  \mathbb{P}_{\mathbf{A}^*, \mathbf{x}_0}(\|\hat{\phi}_t - \phi^* \| > \delta) + \mathbb{P}_{\mathbf{A}^*, \mathbf{x}_0}(\|\hat{\phi}_t - \phi^* \| > \epsilon) \rightarrow 0,\ t\rightarrow \infty.
					% 		\end{align*}} that
			% 		\[\mathrm{vec} (\hat{\mathbf{A}}^{MLE}_t) \overset{\mathbb{P}_{\mathbf{A}^*, \mathbf{x}_0}}{\longrightarrow} \mathrm{vec}(\mathbf{A}^*), \ t\rightarrow \infty. \]
			Again by \citet[Theorem A.7.7 ]{Kuchler_Sorensen_1997} we have
			\[ \frac{1}{\sqrt{t}} \mathcal{H}_\infty^{-1/2} (\mathbf{H}_t - [\mathbf{H}]_t \mathrm{vec}(\mathbf{A}^*)) \overset{\mathcal{L}}{\longrightarrow} \mathbf{W}\sim N(\mathbf{0}, I_{pd^2}), \ t \rightarrow \infty. \]
			Combining this with Equation \eqref{eqn:ergodic_H_MCAR} via Slutsky's lemma yields 
			\[  \sqrt{t} \left([\mathbf{H}]^{-1}_t\mathbf{H}_t - \mathrm{vec}(\mathbf{A}^*)\right) \overset{\mathcal{L}}{\longrightarrow} \mathcal{H}_\infty^{-1/2} \mathbf{W} =: \mathbf{Z} \sim N(\mathbf{0}, \mathcal{H}_\infty^{-1}), \ t \rightarrow \infty.\]
		\end{proof}
		
		\section{Discrete observations: proofs} \label{app:proof_discr}
		\subsection{Proof of Theorem \ref{thm:cons_asymp_first_approx}} \label{app:proof_them_cons_asymp_first_approx}
		\begin{proof}
			To show Equation \eqref{eqn:H_limit} with $\hat{\mathbf{A}}_{1,t} = \hat{\mathbf{A}}(\mathbb{Y}_{[0,t]})$ and $\hat{\mathbf{A}}_{2,t} = \hat{\mathbf{A}}(\mathbb{Y}_{\mathcal{P}_t}, \ldots, D^{p-1}\mathbb{Y}_{\mathcal{P}_t},\mathbb{J}_{\mathcal{P}_t}, \mathbb{M}_{\mathcal{P}_t})$ we work in $L^1(\Omega, \mathcal{F}, \mathbb{P})$ and note that for each $d$-entry batch with $l\in\{0,\ldots,p-1\}$ and $i\in\{1,\ldots, d\}$
			\begin{align*}
				t^{-1/2}&\mathbb{E}_{\mathbf{A}^*}\left[\left\| \sum_{n = 0}^{N_t-1}
				D^lY^{(i)}_{s_n} \Sigma^{-1} (D^{p-1}\mathbf{Y}^{c}_{\mathbf{A}^{(0)}, s_{n+1}} - D^{p-1}\mathbf{Y}^{c}_{\mathbf{A}^{(0)}, s_{n}}) - \int_0^t D^lY^{(i)}_{s-} \Sigma^{-1}\ dD^{p-1}\mathbf{Y}^{c}_{\mathbf{A}^{(0)}, s} \right\|\right]  \\
				% 			&\leq t^{-1/2} \|\Sigma^{-1}\| \mathbb{E}_{\mathbf{A}^*}\left[\left\|  \sum_{n = 0}^{N_t-1} \int_{s_n}^{s_{n+1}} \left(D^lY^{(i)}_{s_n} - D^lY^{(i)}_{s-} \right)\ dD^{p-1}\mathbf{Y}^{c}_{\mathbf{A}^{(0)}, s} \right\|\right]
				% 			\\
				&\overset{(i)}{\leq} t^{-1/2} \|\Sigma^{-1}\| \mathbb{E}_{\mathbf{A}^*}\Big[\Big\|   \sum_{n = 0}^{N_t-1} \int_{s_n}^{s_{n+1}} \Big(D^lY^{(i)}_{s_n} - D^lY^{(i)}_{s-} \Big) \Sigma^{1/2} \ d\mathbf{W}_s \\
				&\quad\quad\quad\quad\quad\quad\quad\quad\quad\quad\quad\quad\quad\quad -   \sum_{n = 0}^{N_t-1} \int_{s_n}^{s_{n+1}} \Big(D^lY^{(i)}_{s_n} - D^lY^{(i)}_{s-} \Big) \sum_{j=1}^p A^*_j D^{p-j}\mathbf{Y}_s\ ds \Big\|\Big]
				\\
				&\overset{(ii)}{\leq} t^{-1/2} \|\Sigma^{-1}\| \Bigg\{ \mathbb{E}_{\mathbf{A}^*}\Big[\Big\|  \int_{0}^{t} \sum_{n = 0}^{N_t-1} \mathds{1}_{[s_n, s_{n+1}]}(s) \big(D^lY^{(i)}_{s_n} - D^lY^{(i)}_{s-} \big) \Sigma^{1/2} \ d\mathbf{W}_s\Big\|^2\Big]^{1/2} \\
				&\quad\quad\quad\quad\quad\quad\quad\quad\quad\quad\quad +   \sum_{n = 0}^{N_t-1} \mathbb{E}_{\mathbf{A}^*}\Big[\Big\|\int_{s_n}^{s_{n+1}} \big(D^lY^{(i)}_{s_n} - D^lY^{(i)}_{s-} \big) \sum_{j=1}^p A^*_j D^{p-j}\mathbf{Y}_s\ ds \Big\|\Big]\Bigg\}
				\\
				&\overset{(iii)}{\leq} t^{-1/2} \|\Sigma^{-1}\| \Bigg\{\mathrm{tr}(\Sigma)^{1/2} \Big(\int_0^t \sum_{n = 0}^{N_t-1} \mathds{1}_{[s_n, s_{n+1}]}(s)\ \mathbb{E}_{\mathbf{A}^*}\Big[ \Big(D^lY^{(i)}_{s_n} - D^lY^{(i)}_{s} \Big)^2\Big] \ ds \Big)^{1/2}  \\
				&\quad\quad\quad\quad\quad\quad\quad\quad\quad\quad\quad + \sum_{n = 0}^{N_t-1} \sum_{j=1}^p \int_{s_n}^{s_{n+1}} \mathbb{E}_{\mathbf{A}^*}\Big[\Big\|\big(D^lY^{(i)}_{s_n} - D^lY^{(i)}_{s} \Big) A^*_j D^{p-j}\mathbf{Y}_s \Big\|\Big]\ ds \Bigg\}
				\\
				&\overset{(iv)}{\leq} t^{-1/2} \|\Sigma^{-1}\| \Bigg\{\mathrm{tr}(\Sigma)^{1/2} \Big(\sum_{n = 0}^{N_t-1} \int_{s_n}^{s_{n+1}}\ \mathbb{E}_{\mathbf{A}^*}\Big[ \Big(D^lY^{(i)}_{s_n} - D^lY^{(i)}_{s} \Big)^2\Big] \ ds \Big)^{1/2}  \\
				&\quad\quad\quad\quad\quad\quad + \sum_{n = 0}^{N_t-1} \sum_{j=1}^p \int_{s_n}^{s_{n+1}} \mathbb{E}_{\mathbf{A}^*}\Big[\big(D^lY^{(i)}_{s_n} - D^lY^{(i)}_{s} \Big)^2\Big]^{1/2}  \mathbb{E}_{\mathbf{A}^*} \Big[\big\|A^*_j D^{p-j}\mathbf{Y}_s \big\|^2\Big]^{1/2}\ ds \Bigg\}
				\\
				&\overset{(v)}{\leq} t^{-1/2} \|\Sigma^{-1}\| \Bigg\{\mathrm{tr}(\Sigma)^{1/2} \Big(\sum_{n = 0}^{N_t-1} \sup_{s\in[s_n, s_{n+1}]} \mathbb{E}_{\mathbf{A}^*}\Big[\big(D^lY^{(i)}_{s_n} - D^lY^{(i)}_{s} \Big)^2\Big] (s_{n+1} - s_n) \Big)^{1/2}  \\
				&\quad +  p \max_{j=1,\ldots,p} \|A^*_j\|  \ \mathbb{E}_{\mathbf{A}^*}\Big[\big\|\mathbf{X}_\infty \big\|^2\Big]^{1/2} \sum_{n = 0}^{N_t-1} \sup_{s\in[s_n, s_{n+1}]} \mathbb{E}_{\mathbf{A}^*}\Big[\big(D^lY^{(i)}_{s_n} - D^lY^{(i)}_{s} \Big)^2\Big]^{1/2} (s_{n+1} - s_n) \Bigg\}
				\\
				&\overset{(vi)}{\leq} t^{-1/2} \|\Sigma^{-1}\| \Bigg\{\mathrm{tr}(\Sigma)^{1/2} \Big( t \sup_{s\in[0, \Delta_{\mathcal{P}_t}]} \mathbb{E}_{\mathbf{A}^*}\Big[\big(D^lY^{(i)}_{0} - D^lY^{(i)}_{s} \Big)^2\Big] \Big)^{1/2}  \\
				&\quad +  p \max_{j=1,\ldots,p} \|A^*_j\|  \ \mathbb{E}_{\mathbf{A}^*}\Big[\big\|\mathbf{X}_\infty \big\|^2\Big]^{1/2} t \sup_{s\in[0, \Delta_{\mathcal{P}_t}]} \mathbb{E}_{\mathbf{A}^*}\Big[\big(D^lY^{(i)}_{0} - D^lY^{(i)}_{s} \Big)^2\Big]^{1/2}  \Bigg\}
				\\
				&\overset{(vii)}{\lesssim} \Delta_{\mathcal{P}_t}^{1/2} + t^{1/2}\Delta_{\mathcal{P}_t}^{1/2} \rightarrow 0, \quad t\rightarrow \infty,
			\end{align*}
			by using $(i)$ the decomposition \eqref{eqn:DY=W-A} of the integrator $\{D^{p-1}\mathbf{Y}^c_{\mathbf{A}^{(0)}, s}, \ s\geq 0\}$, $(iii)$ triangle inequality and domination of $L^1(\Omega, \mathcal{F}, \mathbb{P}_{\mathbf{A}^*})$ norm by $L^2(\Omega, \mathcal{F}, \mathbb{P}_{\mathbf{A}^*})$ norm, $(iii)$ \Ito isometry, triangle inequality and Tonelli's theorem, $(iv)$ Cauchy-Schwartz inequality, $(v)$ sup bounds on integral and stationarity of the state-space representation, $(vi)$ stationarity of the state-space representation again, $(vii)$ Equation \eqref{eqn:big_O_X2} and the high-frequency sampling condition, i.e.\ Assumption \ref{ass:HF_sampling}. 
			
			Next, to show \eqref{eqn:[H]_limit}, we work again in $L^1(\Omega, \mathcal{F}, \mathbb{P}_{\mathbf{A}^*})$ which then implies the corresponding limit in probability. For each $d\times d$ sub-matrix with $l, k\in\{0,\ldots,p-1\}$ and $i, j\in\{1,\ldots, d\}$ we note that
			\begin{align*}
				t^{-1/2}&\mathbb{E}_{\mathbf{A}^*}\left[\left\| \sum_{n = 0}^{N_t-1}
				D^lY^{(i)}_{s_n} \Sigma^{-1} D^{k}Y^{(j)}_{s_n} (s_{n+1} - s_n) - \int_0^t D^lY^{(i)}_{s} \Sigma^{-1} D^{k}Y^{(j)}_{s}\ ds \right\|\right] \\
				% 			&\leq t^{-1/2}\|\Sigma^{-1}\| \mathbb{E}_{\mathbf{A}^*}\left[\left| \sum_{n = 0}^{N_t-1} \int_{s_n}^{s_{n+1}}
				% 			D^lY^{(i)}_{s_n} D^{k}Y^{(j)}_{s_n} - D^lY^{(i)}_{s}  D^{k}Y^{(j)}_{s}\ ds \right|\right]
				% 			\\
				&\overset{(i)}{\leq} t^{-1/2}\|\Sigma^{-1}\|\sum_{n = 0}^{N_t-1} \int_{s_n}^{s_{n+1}} \mathbb{E}_{\mathbf{A}^*}\left[\left|
				D^lY^{(i)}_{s_n} D^{k}Y^{(j)}_{s_n} - D^lY^{(i)}_{s}  D^{k}Y^{(j)}_{s}\ \right|\right] \ ds 
				\\
				&\overset{(ii)}{\leq} t^{-1/2}\|\Sigma^{-1}\|\sum_{n = 0}^{N_t-1} \sup_{s\in[s_n, s_{n+1}]} \mathbb{E}_{\mathbf{A}^*}\left[\left|
				D^lY^{(i)}_{s_n} D^{k}Y^{(j)}_{s_n} - D^lY^{(i)}_{s}  D^{k}Y^{(j)}_{s}\ \right|\right] (s_{n+1} - s_{n})
				\\
				&\overset{(iii)}{\leq} t^{-1/2}\|\Sigma^{-1}\| \sup_{s\in[0, \Delta_{\mathcal{P}_t}]} \mathbb{E}_{\mathbf{A}^*}\left[\left|
				D^lY^{(i)}_{0} D^{k}Y^{(j)}_{0} - D^lY^{(i)}_{s}  D^{k}Y^{(j)}_{s}\ \right|\right] \sum_{n = 0}^{N_t-1} (s_{n+1} - s_{n}) \\
				&\overset{(iv)}{\leq} t^{1/2} \|\Sigma^{-1}\| \sup_{s\in[0, \Delta_{\mathcal{P}_t}]} \left\{ \mathbb{E}_{\mathbf{A}^*}\left[\left|D^lY^{(i)}_{0} ( D^{k}Y^{(j)}_{0} - D^{k}Y^{(j)}_{s}) \right|\right] + \mathbb{E}_{\mathbf{A}^*}\left[\left|D^{k}Y^{(j)}_{s} ( D^lY^{(i)}_{0} - D^lY^{(i)}_{s}) \right|\right]\right\} \\
				&\overset{(v)}{\leq} t^{1/2} \|\Sigma^{-1}\|\, 2\, \mathbb{E}_{\mathbf{A}^*}\left[\left\|\mathbf{X}_{\infty}\right\|^2\right]^{1/2} \sup_{s\in[0, \Delta_{\mathcal{P}_t}]} \mathbb{E}_{\mathbf{A}^*}\left[\left\|\mathbf{X}_{0} - \mathbf{X}_{s} \right\|^2\right]^{1/2} \\
				&\overset{(vi)}{\lesssim} t^{1/2} \Delta_{\mathcal{P}_t}^{1/2} \rightarrow 0, \quad t\rightarrow \infty,
			\end{align*}
			where we have used $(i)$ triangle inequality and Tonelli's theorem, $(ii)$ sup bound on the integral, $(iii)$ stationarity of the state-space representation 
			\[\mathbf{X}_s = 
			\begin{pmatrix} \mathbf{Y}_s \\ \vdots \\ D^{p-1}\mathbf{Y}_s \end{pmatrix}, \ s\geq 0,\]
			of the MCAR($p$) process, $(iv)$ triangle inequality, $(v)$ Cauchy-Schwartz inequality and stationarity of state-space representation, $(vi)$ Equation \eqref{eqn:big_O_X2} and the high-frequency sampling condition, i.e.\ Assumption \ref{ass:HF_sampling}. 
		\end{proof}
		
		\subsection{Proof of Lemma \ref{lemma:D-D_L2_bound}} \label{app:proof_D-D_L2_bound}
		\begin{proof}
			Let $k\in\{1,\ldots,p-1\}$. First note that for $n\in\{0,\ldots, N_t-k\}$ by Lemma \ref{lemma:FD_Y} and Lemma \ref{lemma:Y_iterated_int} we can express the forward difference as
			\begin{equation*} \hat{D}^k\mathbf{Y}_{s_n} = \sum_{i=1}^k (-1)^{k+i} \sum_{\substack{n_1,\ldots, n_i\geq 1\\ n_1 +\cdots + n_i = k}} (s_{n+1} - s_n)^{-n_1} \cdots (s_{n+i} - s_{n+i-1})^{-n_i} (\mathbf{Y}_{s_{n+i}} - \mathbf{Y}_{s_{n+i-1}}), \end{equation*}
			and, for each $i\in\{1,\ldots,k\}$, the increments as
			\begin{equation*}
				\begin{split} \mathbf{Y}_{s_{n+i}} - \mathbf{Y}_{s_{n+i-1}} = \sum_{j=1}^{k-1} D^j \mathbf{Y}_{s_n} &\frac{(s_{n+i} - s_{n})^j - (s_{n+i-1} - s_{n})^j}{j!} + \int_{s_{n+i-1}}^{s_{n+i}} \int_{s_n}^{u_1}\cdots\int_{s_n}^{u_{k-1}} D^k \mathbf{Y}_{u_k}\ du_k\ldots du_1. \end{split}
			\end{equation*}
			These can be combined to write 
			\begin{multline*} 
				\hat{D}^k\mathbf{Y}_{s_n} = \sum_{i=1}^k (-1)^{k+i} \sum_{\substack{n_1,\ldots, n_i\geq 1\\ n_1 +\cdots + n_i = k}} (s_{n+1} - s_n)^{-n_1} \cdots (s_{n+i} - s_{n+i-1})^{-n_i} \\
				\times \Big[\sum_{j=1}^{k-1} D^j \mathbf{Y}_{s_n} \frac{(s_{n+i} - s_{n})^j - (s_{n+i-1} - s_{n})^j}{j!} \\
				+ \int_{s_{n+i-1}}^{s_{n+i}} \int_{s_n}^{u_1}\cdots\int_{s_n}^{u_{k-1}} D^k \mathbf{Y}_{u_k}\ du_k\ldots du_1\Big]. 
			\end{multline*}
			Thus for any $k\in\{1\ldots, p-1\}$ we can write
			% 	\begin{align*} 
				% 	&\hat{D}^k\mathbf{Y}_{s_n} - D^k\mathbf{Y}_{s_n} = \\
				% 	&\left[\sum_{j=1}^{k-1} D^j \mathbf{Y}_{s_n} \left(\sum_{i=1}^k (-1)^{k+i} \sum_{\substack{n_1,\ldots, n_i\geq 1\\ n_1 +\cdots + n_i = k}} (s_{n+1} - s_n)^{-n_1} \cdots (s_{n+i} - s_{n+i-1})^{-n_i} \frac{(s_{n+i} - s_{n})^j - (s_{n+i-1} - s_{n})^j}{j!}\right)\right]\\
				% 	&+ \left(\sum_{i=1}^k (-1)^{k+i} \sum_{\substack{n_1,\ldots, n_i\geq 1\\ n_1 +\cdots + n_i = k}} (s_{n+1} - s_n)^{-n_1} \cdots (s_{n+i} - s_{n+i-1})^{-n_i} \int_{s_{n+i-1}}^{s_{n+i}} \int_{s_n}^{u_1}\cdots\int_{s_n}^{u_{k-1}} \ du_k\ldots du_1\right)(D^k \mathbf{Y}_{u_k} - D^k \mathbf{Y}_{s_n}) \\
				% 	&+ D^k \mathbf{Y}_{s_n} \left(\sum_{i=1}^k (-1)^{k+i} \sum_{\substack{n_1,\ldots, n_i\geq 1\\ n_1 +\cdots + n_i = k}} (s_{n+1} - s_n)^{-n_1} \cdots (s_{n+i} - s_{n+i-1})^{-n_i} \frac{(s_{n+i} - s_{n})^k - (s_{n+i-1} - s_{n})^k}{k!} -1\right).
				% 	\end{align*}
			\begin{equation}\label{eqn:D-D}
				\begin{split}
					\hat{D}^k\mathbf{Y}_{s_n} - D^k\mathbf{Y}_{s_n} =
					\sum_{j=1}^{k-1} D^j \mathbf{Y}_{s_n} C&(s_n,\ldots, s_{n+k};j) +  D^k \mathbf{Y}_{s_n} \left[C(s_n,\ldots, s_{n+k};k) - 1\right] \\ + \Bigg(\sum_{i=1}^k &(-1)^{k+i} \sum_{\substack{n_1,\ldots, n_i\geq 1\\ n_1 +\cdots + n_i = k}} (s_{n+1} - s_n)^{-n_1} \cdots (s_{n+i} - s_{n+i-1})^{-n_i} \\ &\times \int_{s_{n+i-1}}^{s_{n+i}} \int_{s_n}^{u_1}\cdots\int_{s_n}^{u_{k-1}} \ du_k\ldots du_1\Bigg)(D^k \mathbf{Y}_{u_k} - D^k \mathbf{Y}_{s_n}),
				\end{split}
			\end{equation}
			where we define for $j\in\{1, \ldots, k\},$
			\begin{equation} \label{eqn:C_partition}
				\begin{split}
					C(s_n, \ldots, s_{n+k};j) := \sum_{i=1}^k (-1)^{k+i} \sum_{\substack{n_1,\ldots, n_i\geq 1\\ n_1 +\cdots + n_i = k}} (s_{n+1} - s_n)^{-n_1} &\cdots (s_{n+i} - s_{n+i-1})^{-n_i} \\
					&\times \frac{(s_{n+i} - s_{n})^j - (s_{n+i-1} - s_{n})^j}{j!}.
				\end{split}
			\end{equation}
			We note that if $\mathcal{P}_t$ is almost evenly spaced
			\[C(s_n, \ldots, s_{n+k};j) \approx \begin{cases} 0, &j=0,\ldots, k-1,\\
				1, &j=k.\end{cases}\]
			Rigorously we can write (here we consider $k$ even, the case where $k$ is odd is covered by multiplying both inequalities through by $-1$):
			\begin{align*}
				C&(s_n, \ldots, s_{n+k};j) \\ &=\Big(\sum_{\substack{i=1 \\ i \ \mathrm{even}}}^k - \sum_{\substack{i=1 \\ i\ \mathrm{odd}}}^k \Big)\sum_{\substack{n_1,\ldots, n_i\geq 1\\ n_1 +\cdots + n_i = k}} (s_{n+1} - s_n)^{-n_1} \cdots (s_{n+i} - s_{n+i-1})^{-n_i} \frac{(s_{n+i} - s_{n})^j - (s_{n+i-1} - s_{n})^j}{j!} \\
				& \leq \sum_{\substack{i=1 \\ i \ \mathrm{even}}}^k \binom{k-1}{i-1} c_{\mathcal{P}_t}^{-k}\Delta_{\mathcal{P}_t}^{-k}  \frac{(i\Delta_{\mathcal{P}_t})^j - c_{\mathcal{P}_t}^j((i-1)\Delta_{\mathcal{P}_t})^j}{j!} - \sum_{\substack{i=1 \\ i \ \mathrm{odd}}}^k \binom{k-1}{i-1} \Delta_{\mathcal{P}_t}^{-k} \frac{c_{\mathcal{P}_t}^j(i\Delta_{\mathcal{P}_t})^j - ((i-1)\Delta_{\mathcal{P}_t})^j}{j!} \\
				& \leq \frac{\Delta_{\mathcal{P}_t}^{j-k}}{j!} \left\{\sum_{\substack{i=1 \\ i \ \mathrm{even}}}^k \binom{k-1}{i-1} c_{\mathcal{P}_t}^{-k} [i^j - c_{\mathcal{P}_t}^j(i-1)^j] - \sum_{\substack{i=1 \\ i \ \mathrm{odd}}}^k \binom{k-1}{i-1} [c_{\mathcal{P}_t}^j i^j - (i-1)^j]\right\} \\
				& \leq \frac{\Delta_{\mathcal{P}_t}^{j-k}}{j!} \left\{\sum_{i=1}^k (-1)^i \binom{k-1}{i-1} [i^j - (i-1)^j] + \sum_{i=1}^k \binom{k-1}{i-1} \left[(|c_{\mathcal{P}_t}^{-k} - 1| + |c_{\mathcal{P}_t}^j - 1|)i^j +|1 - c_{\mathcal{P}_t}^{j-k}| (i-1)^j\right]\right\},
			\end{align*} 
			and
			\begin{align*}
				C&(s_n, \ldots, s_{n+k};j) \\
				& \geq \sum_{\substack{i=1 \\ i \ \mathrm{even}}}^k \binom{k-1}{i-1} \Delta_{\mathcal{P}_t}^{-k}  \frac{c_{\mathcal{P}_t}^j(i\Delta_{\mathcal{P}_t})^j - ((i-1)\Delta_{\mathcal{P}_t})^j}{j!} - \sum_{\substack{i=1 \\ i \ \mathrm{odd}}}^k \binom{k-1}{i-1} c_{\mathcal{P}_t}^{-k} \Delta_{\mathcal{P}_t}^{-k} \frac{(i\Delta_{\mathcal{P}_t})^j - c_{\mathcal{P}_t}^j((i-1)\Delta_{\mathcal{P}_t})^j}{j!} \\
				& \geq \frac{\Delta_{\mathcal{P}_t}^{j-k}}{j!} \left\{\sum_{\substack{i=1 \\ i \ \mathrm{even}}}^k \binom{k-1}{i-1} [c_{\mathcal{P}_t}^{j}i^j - (i-1)^j] - \sum_{\substack{i=1 \\ i \ \mathrm{odd}}}^k \binom{k-1}{i-1} c_{\mathcal{P}_t}^{-k} [i^j - c_{\mathcal{P}_t}^j(i-1)^j]\right\}\\
				& \geq \frac{\Delta_{\mathcal{P}_t}^{j-k}}{j!} \left\{\sum_{i=1}^k (-1)^i \binom{k-1}{i-1} [i^j - (i-1)^j] - \sum_{i=1}^k \binom{k-1}{i-1} \left[(|c_{\mathcal{P}_t}^{j}-1| + |1- c_{\mathcal{P}_t}^{-k}|) i^j + |c_{\mathcal{P}_t}^{j-k} - 1| (i-1)^j \right]\right\},
			\end{align*} 
			where we use that for any $m\geq 1$ the relationship $c_{\mathcal{P}_t}\Delta_{\mathcal{P}_t} \leq |s_{n+1} - s_n| \leq \Delta_{\mathcal{P}_t}$ implies
			\begin{align*}
				c^m_{\mathcal{P}_t}\Delta^m_{\mathcal{P}_t} \leq |s_{n+1} - s_n|^m \leq \Delta^m_{\mathcal{P}_t}&,\  -\Delta^m_{\mathcal{P}_t} \leq -|s_{n+1} - s_n|^m \leq - c^m_{\mathcal{P}_t}\Delta^m_{\mathcal{P}_t}, \\
				\Delta^{-m}_{\mathcal{P}_t} \leq |s_{n+1} - s_n|^{-m} \leq c_{\mathcal{P}_t}^{-m}\Delta^{-m}_{\mathcal{P}_t}&,\ -c_{\mathcal{P}_t}^{-m}\Delta^{-m}_{\mathcal{P}_t} \leq -|s_{n+1} - s_n|^{-m} \leq -\Delta^{-m}_{\mathcal{P}_t}.
			\end{align*}
			and for any $k\geq 1$, $1\leq i\leq k$
			\[\#\{(n_1, \ldots, n_i): n_1,\ldots, n_i\geq 1, n_1+\cdots + n_i = k\} = \binom{k-1}{i-1}.\] 
			We can rearrange and combine the upper and lower bounds for $C(s_n,\ldots, s_{n+k};j)$ to obtain
			\begin{equation} \label{eqn:bound_C_first}
				\left| C(s_n, \ldots, s_{n+k};j) - F(j,k)\right|
				\leq \frac{\Delta_{\mathcal{P}_t}^{j-k}}{j!} \sum_{i=1}^k \binom{k-1}{i-1} \left[(|c_{\mathcal{P}_t}^{j}-1| + |1- c_{\mathcal{P}_t}^{-k}|) i^j + |c_{\mathcal{P}_t}^{j-k} - 1| (i-1)^j \right],
			\end{equation} 
			where, by Lemma \ref{lemma:F(j,k)}, we have
			\begin{equation*}
				F(j, k) := (-1)^k \frac{\Delta_{\mathcal{P}_t}^{j-k}}{j!} \sum_{i=1}^k (-1)^i \binom{k-1}{i-1} [i^j - (i-1)^j] = \begin{cases} 0, &j=0,\ldots, k-1,\\
					1, &j=k.\end{cases}.
			\end{equation*}
			Next, we can use $c_{\mathcal{P}_t}\in(0,1)$ and Assumption \ref{ass:evenly_spaced_lim} to deduce that for $1\leq j\leq k$
			\[|c_{\mathcal{P}_t}^j - 1| \leq |c_{\mathcal{P}_t}^k - 1|,\quad |c_{\mathcal{P}_t}^{-k} - 1| = \frac{|c_{\mathcal{P}_t}^k - 1|}{|c_{\mathcal{P}_t}^k|} \leq 2^k|c_{\mathcal{P}_t}^k - 1|,\quad |c_{\mathcal{P}_t}^{j-k} - 1| = \frac{|c_{\mathcal{P}_t}^{k-j} - 1|}{|c_{\mathcal{P}_t}^{k-j}|} %\leq \frac{|c_{\mathcal{P}_t}^{k} - 1|}{|c_{\mathcal{P}_t}^{k-j}|}
			\leq 2^{k-j} |c_{\mathcal{P}_t}^k - 1|,\]
			for $t\in\mathcal{T}$ sufficiently large (i.e.\ such that $c_{\mathcal{P}_t} \geq 1/2)$. Due to the ``convergence to uniform spacing'' condition on the sequence of partitions, i.e.\ Assumption \ref{ass:evenly_spaced_lim}, we can bound Equation \eqref{eqn:bound_C_first} by
			\begin{equation} \label{eqn:bound_C_second}
				\left| C(s_n, \ldots, s_{n+k};j) - F(j,k)\right|
				\leq \frac{\Delta_{\mathcal{P}_t}^{j-k}}{j!} |c_{\mathcal{P}_t}^k - 1| \sum_{i=1}^k \binom{k-1}{i-1} \left[(1+2^{k}) i^j + 2^{k-j} (i-1)^j \right] \lesssim \Delta_{\mathcal{P}_t}^{j-1/2}.
			\end{equation} 
			We note this implies
			\begin{align*}
				\mathbb{E}_{\mathbf{A}^*}&\left[\left\|\hat{D}^k\mathbf{Y}_{s_n} - D^k\mathbf{Y}_{s_n}\right\|^2\right]\\
				&\overset{(i)}{\leq} \sum_{j=1}^{k}\sum_{l=1}^k \mathbb{E}_{\mathbf{A}^*}\left[D^j \mathbf{Y}_{s_n}^\mathrm{T} D^l \mathbf{Y}_{s_n}\right] |C(s_n,\ldots, s_{n+k};j) - F(j,k)||C(s_n,\ldots, s_{n+k};l) - F(l,k)| \\ 
				&\quad\quad +2 \sum_{j=1}^k \sup_{u\in[s_n, s_{n+k}]} \mathbb{E}_{\mathbf{A}^*}\left[D^j \mathbf{Y}_{s_n}^{\mathrm{T}} (D^k \mathbf{Y}_{u} - D^k \mathbf{Y}_{s_n}) \right] |C(s_n,\ldots, s_{n+k};j) - F(j,k)| \,\bar{C}(s_n,\ldots, s_{n+k};k) \\
				&\quad\quad + \sup_{u\in[s_n, s_{n+k}]} \mathbb{E}_{\mathbf{A}^*}\left[\|D^k \mathbf{Y}_{u} - D^k \mathbf{Y}_{s_n}\|^2\right] \bar{C}(s_n,\ldots, s_{n+k};k)^2 \\ 
				%--------------------------
				&\overset{(ii)}{\leq} \sum_{j=1}^{k}\sum_{l=1}^k \mathbb{E}_{\mathbf{A}^*}\left[D^j \mathbf{Y}_{\infty}^\mathrm{T} D^l \mathbf{Y}_{\infty}\right] |C(s_n,\ldots, s_{n+k};j) - F(j,k)||C(s_n,\ldots, s_{n+k};l) - F(l,k)| \\ 
				&\quad\quad+ 2\sum_{j=1}^k \mathbb{E}_{\mathbf{A}^*}\left[\|D^j \mathbf{Y}_{\infty}\|^2\right]^{1/2} \sup_{u\in[0,k\Delta_{\mathcal{P}_t}]} \mathbb{E}_{\mathbf{A}^*}\left[\|D^k \mathbf{Y}_{u} - D^k \mathbf{Y}_{0}\|^2 \right]^{1/2} \\
				&\quad\quad\quad\quad\quad\quad\quad\quad\quad\quad\quad\quad\quad\quad\quad\quad\quad\quad\quad\quad\quad\quad\quad \times |C(s_n,\ldots, s_{n+k};j) - F(j,k)|  \,\bar{C}(s_n,\ldots, s_{n+k};k) \\
				&\quad\quad+ \sup_{u\in[0, k\Delta_{\mathcal{P}_t}]} \mathbb{E}_{\mathbf{A}^*}\left[\|D^k \mathbf{Y}_{u} - D^k \mathbf{Y}_{0}\|^2\right] \bar{C}(s_n,\ldots, s_{n+k};k)^2 \\ 
				&\overset{(iii)}{\lesssim} \Big[\sum_{j=1}^k\sum_{l=1}^k \Delta_{\mathcal{P}_t}^{j-1/2}\Delta_{\mathcal{P}_t}^{l-1/2} + 2\sum_{j=1}^k (k\Delta_{\mathcal{P}_t})^{1/2}\Delta_{\mathcal{P}_t}^{j-1/2}+ k\Delta_{\mathcal{P}_t}\Big] \lesssim \Delta_{\mathcal{P}_t},
			\end{align*}
			where 
			\begin{align} \label{eqn:C_bar_partition}
				\bar{C}(s_n,\ldots, s_{n+k};k) &:= \sum_{i=1}^k  \sum_{\substack{n_1,\ldots, n_i\geq 1\\ n_1 +\cdots + n_i = k}} (s_{n+1} - s_n)^{-n_1} \cdots (s_{n+i} - s_{n+i-1})^{-n_i} \frac{(s_{n+i} - s_{n})^k - (s_{n+i-1} - s_{n})^k}{k!} \\
				&\leq \sum_{i=1}^k  \binom{k-1}{i-1} \Delta_{\mathcal{P}_t}^{-k} c_{\mathcal{P}_t}^{-k} \frac{(i\Delta_{\mathcal{P}_t})^k - c_{\mathcal{P}_t}^k((i-1)\Delta_{\mathcal{P}_t})^k}{k!} = \sum_{i=1}^k  \binom{k-1}{i-1}  \frac{2^{k}i^k - (i-1)^k}{k!} \notag
			\end{align}
			by applying $(i)$ Equation \eqref{eqn:D-D}, sup-bounds on the integrals and triangle inequality, $(ii)$ stationarity of the state-space representation and Cauchy-Schwartz inequality and $(iii)$ Equation \eqref{eqn:bound_C_second}, Equation \eqref{eqn:big_O_X2} and $\mathbb{E}_{\mathbf{A}^*}[\mathbf{X}_\infty\mathbf{X}_\infty^{\mathrm{T}}]<\infty$.
		\end{proof}
		
		\subsection{Proof of Theorem \ref{thm:cons_asymp_second_approx}} \label{app:proof_them_cons_asymp_second_approx}
		\begin{proof}
			In order to apply Lemma \ref{lemma:approx} we need to show equations \eqref{eqn:H_limit} and \eqref{eqn:[H]_limit} hold for $\hat{\mathbf{A}}_{1,t} = \hat{\mathbf{A}}(\mathbb{Y}_{\mathcal{Q}_t}, \ldots, D^{p-1}\mathbb{Y}_{\mathcal{Q}_t},\mathbb{J}_{\mathcal{Q}_t}, \mathbb{M}_{\mathcal{Q}_t})$ and $\hat{\mathbf{A}}_{2,t} = \hat{\mathbf{A}}(\mathbb{Y}_{\mathcal{P}_t}, \mathbb{J}_{\mathcal{Q}_t}, \mathbb{M}_{\mathcal{Q}_t})$. For ease of notation let us write the increments of the partition $\mathcal{Q}_t$ with mesh $\Delta_{\mathcal{Q}_t}$ as
			\[\Delta_{\mathcal{Q}_t}^m := u_{m+1} - u_m, \quad m=0,\ldots, M_t-1,\]
			and the increment of a process $\mathbb{Z}$ over the interval $[u_m, u_{m+1}] \in\mathcal{Q}_t$ as
			\[\Delta_{\mathcal{Q}_t}^m \mathbf{Z} := (\mathbf{Z}_{u_{m+1}} - \mathbf{Z}_{u_m}), \quad m=0,\ldots, M_t-1. \]
			
			We start by showing Equation \eqref{eqn:H_limit}. We work in $L^1(\Omega, \mathcal{F}, \mathbb{P}_{\mathbf{A}^*})$ which then implies the corresponding limit in probability. We note that for each $d$-entry batch with $l\in\{0,\ldots,p-1\}$ and $i\in\{1,\ldots, d\}$
			
			\begin{align*}
				t^{-1/2}&\mathbb{E}_{\mathbf{A}^*}\left[\left\| \sum_{m = 0}^{M_t-1}
				\hat{D}^{l}Y^{(i)}_{u_m} \Sigma^{-1} \Delta_{\mathcal{Q}_t}^m \hat{D}^{p-1}\mathbf{Y}^{c}_{\mathbf{A}^{(0)}}  - \sum_{m = 0}^{M_t-1}
				D^{l}Y^{(i)}_{u_m} \Sigma^{-1} \Delta_{\mathcal{Q}_t}^m D^{p-1}\mathbf{Y}^{c}_{\mathbf{A}^{(0)}} \right\|\right]  \\
				%------------------
				&\overset{(i)}{\leq} t^{-1/2}\|\Sigma^{-1}\|\mathbb{E}_{\mathbf{A}^*}\Bigg[\bigg\| \sum_{m = 0}^{M_t-1}
				\hat{D}^{l}Y^{(i)}_{u_m} \Big[\hat{D}^{p-1}\mathbf{Y}^{c}_{\mathbf{A}^{(0)}, u_{m+1}} - \hat{D}^{p-1}\mathbf{Y}^{c}_{\mathbf{A}^{(0)}, u_m} - \Delta_{\mathcal{Q}_t}^m D^{p-1}\mathbf{Y}^{c}_{\mathbf{A}^{(0)}}\Big] \\
				&\quad \quad \quad \quad \quad \quad \quad \quad \quad \quad \quad \quad \quad \quad \quad \quad \quad 
				\quad \quad \quad \quad \quad \quad \quad \quad +
				\sum_{m = 0}^{M_t-1} (\hat{D}^{l}Y^{(i)}_{u_m} - D^{l}Y^{(i)}_{u_m}) \Delta_{\mathcal{Q}_t}^m D^{p-1}\mathbf{Y}^{c}_{\mathbf{A}^{(0)}} \bigg\|\Bigg]  \\
				%------------------
				&\overset{(ii)}{\leq} t^{-1/2}\|\Sigma^{-1}\|\Bigg\{\sum_{m = 0}^{M_t-1}
				\mathbb{E}_{\mathbf{A}^*}\Bigg[\bigg\|\hat{D}^{l}Y^{(i)}_{u_m} \Big[(\hat{D}^{p-1}\mathbf{Y}_{u_{m+1}} - D^{p-1}\mathbf{Y}_{u_{m+1}}) - (\hat{D}^{p-1}\mathbf{Y}_{u_m} - D^{p-1}\mathbf{Y}_{u_m})\Big]\bigg\|\Bigg]  \\
				&\quad \quad \quad \quad \quad \quad \quad \quad \quad \quad \quad \quad \quad \quad 
				\quad \quad \quad \quad \quad \quad \quad + \mathbb{E}_{\mathbf{A}^*}\Bigg[\bigg\|
				\sum_{m = 0}^{M_t-1}(\hat{D}^{l}Y^{(i)}_{u_m} - D^{l}Y^{(i)}_{u_m}) \Delta_{\mathcal{Q}_t}^m D^{p-1}\mathbf{Y}^{c}_{\mathbf{A}^{(0)}} \bigg\|\Bigg] \Bigg\} \\
				%------------------
				&\overset{(iii)}{\leq} t^{-1/2}\|\Sigma^{-1}\| \Bigg\{ \sum_{m = 0}^{M_t-1} \mathbb{E}_{\mathbf{A}^*}\left[\left\|\hat{D}^{l}Y^{(i)}_{u_m}\right\|^2\right]^{1/2}  \mathbb{E}_{\mathbf{A}^*}\left[\left\|\hat{D}^{p-1}\mathbf{Y}_{u_{m+1}} - D^{p-1}\mathbf{Y}_{u_{m+1}} \right\|^2\right]^{1/2} \\
				&\quad\quad\quad\quad\quad\quad\quad\quad\quad + \mathbb{E}_{\mathbf{A}^*}\left[\left\|\hat{D}^{l}Y^{(i)}_{u_m}\right\|^2\right]^{1/2}\mathbb{E}_{\mathbf{A}^*}\left[\left\|\hat{D}^{p-1}\mathbf{Y}_{u_m} - D^{p-1}\mathbf{Y}_{u_m} \right\|^2\right]^{1/2} \\
				&\quad \quad \quad \quad \quad \quad \quad \quad \quad  + \mathbb{E}_{\mathbf{A}^*}\left[\left\|
				\sum_{m = 0}^{M_t-1}(\hat{D}^{l}Y^{(i)}_{u_m} - D^{l}Y^{(i)}_{u_m}) [\Sigma^{1/2} \Delta_{\mathcal{Q}_t}^m \mathbf{W} - \sum_{j=1}^p\int_{u_m}^{u_{m+1}} A^*_j D^{p-j}\mathbf{Y}_u\ du] \right\|\right]\Bigg\}    \\
				%---------------------
				&\overset{(iv)}{\lesssim} t^{-1/2} \Bigg\{\sum_{m = 0}^{M_t-1} \Delta_{\mathcal{P}_t}^{1/2} + \mathbb{E}_{\mathbf{A}^*}\bigg[\bigg\|
				\int_0^t \sum_{m = 0}^{M_t-1} \mathds{1}_{[u_m, u_{m+1}]} (s) (\hat{D}^{l}Y^{(i)}_{u_m} - D^{l}Y^{(i)}_{u_m}) \Sigma^{1/2} d\mathbf{W}_{s}\bigg\|\bigg]\\
				&\quad \quad \quad \quad \quad \quad \quad \quad \quad \quad \quad \quad \quad \quad \quad + \mathbb{E}_{\mathbf{A}^*}\bigg[\bigg\| \sum_{m = 0}^{M_t-1} \sum_{j=1}^p \int_{u_m}^{u_{m+1}} A^*_j (\hat{D}^{l}Y^{(i)}_{u_m} - D^{l}Y^{(i)}_{u_m}) D^{p-j}\mathbf{Y}_u\ du \bigg\|\bigg]\Bigg\} \\
				%---------------------
				&\overset{(v)}{\lesssim} t^{-1/2} \Bigg\{ M_t \Delta^{1/2}_{\mathcal{P}_t} + \mathrm{tr}(\Sigma)^{1/2} \left(
				\int_0^t \sum_{m = 0}^{M_t-1} \mathds{1}_{[u_m, u_{m+1}]} (u) \mathbb{E}_{\mathbf{A}^*}\left[\left(\hat{D}^{l}Y^{(i)}_{u_m} - D^{l}Y^{(i)}_{u_m}\right)^2\right] du \right)^{1/2}\\
				&\quad \quad \quad \quad \quad \quad \quad \quad + \sum_{m = 0}^{M_t-1}\sum_{j=1}^p \int_{u_m}^{u_{m+1}} \|A^*_j\| \mathbb{E}_{\mathbf{A}^*}\bigg[\bigg|  \hat{D}^{l}Y^{(i)}_{u_m} - D^{l}Y^{(i)}_{u_m}\bigg|^2\bigg]^{1/2} \mathbb{E}_{\mathbf{A}^*}\bigg[ \|D^{p-j}\mathbf{Y}_\infty\|^2\bigg]^{1/2} du\Bigg\} \\
				%---------------------
				&\overset{(vi)}{\lesssim}  t^{-1/2}M_t \Delta^{1/2}_{\mathcal{P}_t} + \Delta^{1/2}_{\mathcal{P}_t} + t^{1/2}\Delta^{1/2}_{\mathcal{P}_t}\\
				&\ \lesssim   t^{-1/2}N'_t \Delta^{1/2}_{\mathcal{P}_t} = \left(t^{-1} N'_t \Delta_{\mathcal{Q}_t}\right) \left(t\Delta_{\mathcal{Q}_t}^{-2} \Delta_{\mathcal{P}_t}\right)^{1/2} \rightarrow 0, \quad t\rightarrow\infty,
			\end{align*}
			where we use $(i)$ simple rearranging, $(ii)$ triangle inequality, $(iii)$ triangle inequality, Cauchy-Schwartz inequality and representation \eqref{eqn:DY=W-A}, $(iv)$ bound \eqref{eqn:D-D_L2_bound}, triangle inequality and $\mathbb{E}_{\mathbf{A}^*}[\mathbf{X}_\infty\mathbf{X}_\infty^\mathrm{T}]<\infty$, $(v)$ \Ito isometry and triangle inequlity, Tonelli theorem, Cauchy-Schwartz and stationarity of state-space, $(vi)$ bound \eqref{eqn:D-D_L2_bound} again and, finally, Assumption \ref{ass:controlled_sampling} and Assumption \ref{ass:joint_mesh}.
			
			Next, to show Equation \eqref{eqn:H_limit} holds, we again work in $L^1(\Omega, \mathcal{F}, \mathbb{P})$ and note that for each $d\times d$ sub-matrix with $l, k\in\{0,\ldots,p-1\}$ and $i, j\in\{1,\ldots, d\}$ we note that
			\begin{align*}
				t^{-1/2}&\mathbb{E}_{\mathbf{A}^*}\left[\left\| \sum_{m = 0}^{M_t-1}
				\hat{D}^{l}Y^{(i)}_{u_m} \Sigma^{-1} \hat{D}^{k}Y^{(j)}_{u_m} \Delta_{\mathcal{Q}_t}^m -
				\sum_{m = 0}^{M_t-1}
				D^{l}Y^{(i)}_{u_m} \Sigma^{-1} D^{k}Y^{(j)}_{u_m} \Delta_{\mathcal{Q}_t}^m \right\|\right] \\
				&\overset{(i)}{\leq} t^{-1/2}\|\Sigma^{-1}\|\sum_{m = 0}^{M_t-1} \bigg\{ \mathbb{E}_{\mathbf{A}^*}\left[\left|
				\hat{D}^{l}Y^{(i)}_{u_m} (\hat{D}^{k}Y^{(j)}_{u_m} - D^{k}Y^{(j)}_{u_m})\right|\right] + \mathbb{E}_{\mathbf{A}^*}\left[\left|D^{k}Y^{(j)}_{u_m} (\hat{D}^{l}Y^{(i)}_{u_m} - D^{l}Y^{(i)}_{u_m}) \right|\right] \bigg\} \Delta_{\mathcal{Q}_t}^m \\
				&\overset{(ii)}{\leq} t^{-1/2}\|\Sigma^{-1}\|\sum_{m = 0}^{M_t-1} \bigg\{ \mathbb{E}_{\mathbf{A}^*}\left[\left|
				\hat{D}^{l}Y^{(i)}_{u_m}\right|^2\right]^{1/2} \mathbb{E}_{\mathbf{A}^*}\left[\left|\hat{D}^{k}Y^{(j)}_{u_m} - D^{k}Y^{(j)}_{u_m} \right|^2\right]^{1/2} \\
				&\quad\quad\quad\quad\quad\quad\quad\quad\quad\quad\quad\quad\quad\quad\quad\quad\quad\quad + \mathbb{E}_{\mathbf{A}^*}\left[\left|D^{k}Y^{(j)}_{u_m}\right|^2\right]^{1/2} \mathbb{E}_{\mathbf{A}^*}\left[\left|\hat{D}^{l}Y^{(i)}_{u_m} - D^{l}Y^{(i)}_{u_m} \right|^2\right]^{1/2} \bigg\} \Delta_{\mathcal{Q}_t}^m \\
				%------------
				&\overset{(iii)}{\lesssim} t^{-1/2}
				\sum_{m = 0}^{M_t-1} \left[ \Delta_{\mathcal{P}_t}^{1/2} + \Delta_{\mathcal{P}_t}^{1/2} \right] \Delta_{\mathcal{Q}_t}^m  \\
				&\overset{(iv)}{\lesssim} (t \Delta_{\mathcal{P}_t})^{1/2}\rightarrow 0, \quad t\rightarrow\infty,
			\end{align*}
			where we have used $(i)$ the triangle inequality, $(ii)$ Cauchy-Schwartz inequality, $(iii)$ bound \eqref{eqn:D-D_L2_bound}, the fact that $\hat{D}^l\mathbf{Y}_{s_n}$ is a linear combination of stationary $\mathbf{Y}_{s_{n}}, \ldots, \mathbf{Y}_{s_{n+l}}$ with finite second moments, stationarity of the state-space representation and $\mathbb{E}_{\mathbf{A}^*}[\mathbf{X}_\infty\mathbf{X}_\infty^\mathrm{T}]<\infty$, and $(iv)$ the high-frequency sampling condition, i.e.\ Assumption \ref{ass:HF_sampling}.
		\end{proof}
		\begin{remark}
			Note that the estimator $\hat{\mathbf{A}}_{1,t} = \hat{\mathbf{A}}(\mathbb{Y}_{\mathcal{Q}_t}, \ldots, D^{p-1}\mathbb{Y}_{\mathcal{Q}_t},\mathbb{J}_{\mathcal{Q}_t}, \mathbb{M}_{\mathcal{Q}_t})$ is defined with the summations up to $N'_t-1$ and not $M_t-1$. In order to prove equations \eqref{eqn:H_limit} and \eqref{eqn:[H]_limit} we are thus missing 
			\[\sum_{m = M_t}^{N'_t-1}
			D^{l}Y^{(i)}_{u_m} \Sigma^{-1} D^{k}Y^{(j)}_{u_m} (u_{m+1} - u_m) \overset{\mathbb{P}_{\mathbf{A}^*}}{\rightarrow} 0,\ \mathrm{and}\  \sum_{m = M_t}^{N'_t-1}
			D^{l}Y^{(i)}_{u_m} \Sigma^{-1} \Delta_{\mathcal{Q}_t}^m D^{p-1}\mathbf{Y}^{c}_{\mathbf{A}^{(0)}} \overset{\mathbb{P}_{\mathbf{A}^*}}{\rightarrow} 0.\]
			as $t\rightarrow\infty$ for $0\leq l,k\leq p-1$ and $1\leq i,j\leq d$. The first limit can be shown in $L^1(\Omega, \mathcal{F}, \mathbb{P}_{\mathbf{A}^*})$ by
			\begin{align*}
				\mathbb{E}_{\mathbf{A}^*}\left[\left\| \sum_{m = M_t}^{N'_t-1}
				D^{l}Y^{(i)}_{u_m} \Sigma^{-1} D^{k}Y^{(j)}_{u_m} (u_{m+1} - u_m) \right\|\right] &\leq \|\Sigma^{-1}\| \sum_{m = M_t}^{N'_t-1} \mathbb{E}_{\mathbf{A}^*}\left[\left|
				D^{l}Y^{(i)}_{\infty} D^{k}Y^{(j)}_{\infty}\right|\right] (u_{m+1} - u_m) \\
				&\lesssim (N'_t - M_t - 1) \Delta_{\mathcal{Q}_t} \lesssim (p-1) \Delta_{\mathcal{Q}_t} \rightarrow 0,\quad t\rightarrow \infty,
			\end{align*}
			and similarly the second limit follows by 
			\begin{align*}
				\mathbb{E}_{\mathbf{A}^*}\bigg[\bigg\| \sum_{m = M_t}^{N'_t-1}
				D^{l}Y^{(i)}_{u_m} \Sigma^{-1} \Delta_{\mathcal{Q}_t}^m D^{p-1}\mathbf{Y}^{c}_{\mathbf{A}^{(0)}} \bigg\|\bigg] &\leq \|\Sigma^{-1}\| \sum_{m = M_t}^{N'_t-1} \mathbb{E}_{\mathbf{A}^*}\left[\left\|
				D^{l}Y^{(i)}_{u_m} \Delta_{\mathcal{Q}_t}^m D^{p-1}\mathbf{Y}^{c}_{\mathbf{A}^{(0)}}\right\|\right] \\
				&\lesssim (N'_t - M_t - 1) \Delta_{\mathcal{Q}_t} \lesssim (p-1) \Delta_{\mathcal{Q}_t} \rightarrow 0,\quad t\rightarrow \infty,
			\end{align*}
		\end{remark}

		\subsection{Proof of Theorem \ref{thm:cons_asymp_third_approx_finite_activity}} \label{app:proof_thm_cons_asymp_third_approx_finite_activity}
		\begin{proof}
			Recall that by Lemma \ref{lemma:approx} it suffices to show Equation \eqref{eqn:H_limit} with $\hat{\mathbf{A}}_{1,t} = \hat{\mathbf{A}}(\mathbb{Y}_{\mathcal{P}_t}, \mathbb{J}_{\mathcal{Q}_t}, \mathbb{M}_{\mathcal{Q}_t})$ and $\hat{\mathbf{A}}_{2,t} = \hat{\mathbf{A}}(\mathbb{Y}_{\mathcal{P}_t}; \mathcal{Q}_t, \boldsymbol{\nu}_t)$. To show this we work with a slightly different approach than in Proofs \ref{app:proof_them_cons_asymp_first_approx} and \ref{app:proof_them_cons_asymp_second_approx}. We first introduce the event $A_t$, representing whether the thresholding technique is succesful, and show this event has probability approaching one as $t\rightarrow\infty$, i.e.\ Equation \eqref{eqn:P(A)_conv_one}. Next, we show that on the set $A_t$, $\hat{\mathbf{A}}(\mathbb{Y}_{\mathcal{P}_t}; \mathcal{Q}_t, \boldsymbol{\nu}_t)$ converges to $ \hat{\mathbf{A}}(\mathbb{Y}_{\mathcal{P}_t}, \mathbb{J}_{\mathcal{Q}_t}, \mathbb{M}_{\mathcal{Q}_t})$ in $L^1(\Omega, \mathcal{F}, \mathbb{P}_{\mathbf{A}^*})$, i.e.\ Equation \eqref{eqn:restricted_L1_conv}. Combining these results implies the desired condition \eqref{eqn:H_limit}.
			
			For $t\in\mathcal{T}$, $m\in\{0,\ldots, M_t-1\}$ and $i\in\{1,\ldots,d\}$ let us introduce the event that the thresholding technique is successful on the $m$-th interval of the $t$-th partition in the $i$-th component, i.e.\
			\begin{align}
				A_t^{m, i} :&= \left\{\omega\in\Omega : \mathds{1}_{\left\{\left|\Delta_{\mathcal{Q}_t}^m \hat{D}^{p-1}Y^{(i)} - \tilde{b}^{(i)} \Delta_{\mathcal{Q}_t}^m\right| \leq \nu^{(i), m}_t \right\}}(\omega) = \mathds{1}_{\left\{ \Delta^m_{\mathcal{Q}_t}N^{(i)} = 0 \right\}}(\omega) \right\}, \notag \\
				&= (K_t^{m, i} \cap M_t^{m, i}) \cup (K_t^{m,i, c} \cap M_t^{m,i, c}) = (K_t^{m,i}\setminus M_t^{m,i})^c \cap (M_t^{m,i}\setminus K_t^{m,i})^c , \label{eqn:A=KUM}
			\end{align}
			where 
			\begin{align}
				K_t^{m, i} := \left\{\omega\in\Omega : \left|\Delta_{\mathcal{Q}_t}^m \hat{D}^{p-1}Y^{(i)}(\omega) - \tilde{b}^{(i)} \Delta_{\mathcal{Q}_t}^m\right| \leq \nu^{(i), m}_t \right\}, \quad
				M_t^{m, i} := \left\{\omega\in\Omega :  \Delta^m_{\mathcal{Q}_t}N^{(i)}(\omega) = 0 \right\}. \label{eqn:K,M}
			\end{align}
			On the set $A_t^{m, i}$ either:
			\begin{itemize}
				\item no jumps have occurred and the \textit{approximated} de-trended increment is kept, i.e.\ we perfectly recover $\Delta_{\mathcal{Q}_t}^m \hat{D}^{p-1}Y^{c, (i)}_{\mathbf{A}^{(0)}}$ on $[u_m, u_{m+1}]$; or
				\item jumps have occurred and the  \textit{approximated} de-trended increment is discarded, i.e.\ $\Delta_{\mathcal{Q}_t}^m \hat{D}^{p-1}\mathbf{Y}^c_{\mathbf{A}^{(0)}}$ is set to $0$ on $[u_m, u_{m+1}]$.
			\end{itemize}
			Recall that by definition of $\hat{D}^{p-1}\mathbb{Y}^c_{\mathbf{A}^{(0)}}$, i.e.\ Equation \eqref{eqn:approx_cont_mart_part}, we have
			\begin{equation*}
				\Delta_{\mathcal{Q}_t}^m\hat{D}^{p-1}\mathbf{Y}^c_{\mathbf{A}^{(0)}} - \Delta_{\mathcal{Q}_t}^m D^{p-1}\mathbf{Y}^c_{\mathbf{A}^{(0)}} = \Delta_{\mathcal{Q}_t}^m\hat{D}^{p-1}\mathbf{Y} - \Delta_{\mathcal{Q}_t}^m D^{p-1}\mathbf{Y},
			\end{equation*}
			and thus by Equation \eqref{eqn:D-D_L2_bound} we have the bound
			\begin{equation} \label{eqn:incr_difference_bound}
				\mathbb{E}_{\mathbf{A}^*}\left[ \left\|\Delta_{\mathcal{Q}_t}^m\hat{D}^{p-1}\mathbf{Y}^c_{\mathbf{A}^{(0)}} - \Delta_{\mathcal{Q}_t}^m D^{p-1}\mathbf{Y}^c_{\mathbf{A}^{(0)}}\right\|^2\right] = \mathbb{E}_{\mathbf{A}^*}\left[ \left\|\Delta_{\mathcal{Q}_t}^m\hat{D}^{p-1}\mathbf{Y} - \Delta_{\mathcal{Q}_t}^m D^{p-1}\mathbf{Y}\right\|^2\right] \lesssim \Delta_{\mathcal{P}_t}.
			\end{equation}
			Define the event that the thresholding technique is successful over the whole partition $\mathcal{Q}_t$ 
			\[A_t := \bigcap_{i=1}^{d} \bigcap_{m=0}^{M_t-1} A_t^{m,i}.\]
			We start by showing the sequence of events $\{A_t,\ t\in\mathcal{T}\}$ has probability approaching one as $t\rightarrow\infty$, i.e.\
			\begin{equation}
				\label{eqn:P(A)_conv_one}
				\mathbb{P}_{\mathbf{A}^*}(A_t)\rightarrow 1,\quad t\rightarrow \infty.
			\end{equation}
			Note
			\begin{align*}
				\mathbb{P}_{\mathbf{A}^*}&\left(A_t^{c}\right) \\
				&= \mathbb{P}_{\mathbf{A}^*}\left(\bigcup_{i=1}^{d} \bigcup_{m=0}^{M_t-1} A_t^{m,i, c} \right) \\
				%--------------------------
				&\overset{(i)}{\leq} \sum_{i=1}^{d} \sum_{m=0}^{M_t-1} \left[ \mathbb{P}_{\mathbf{A}^*}\left( K_t^{m,i} \setminus M_t^{m,i}\right) + \mathbb{P}_{\mathbf{A}^*}\left(M_t^{m,i} \setminus K_t^{m,i} \right)\right] \\
				%------------------------------
				&\overset{(ii)}{\leq} \sum_{i=1}^{d} \sum_{m=0}^{M_t-1} \bigg[ \mathbb{P}_{\mathbf{A}^*}\left( \left|\Delta_{\mathcal{Q}_t}^m \hat{D}^{p-1}Y^{(i)} - \tilde{b}^{(i)} \Delta_{\mathcal{Q}_t}^m\right| \leq \nu^{(i), m}_t,\, \Delta^m_{\mathcal{Q}_t}N^{(i)} > 0 \right) \\
				&\quad\quad\quad\quad\quad\quad\quad\quad\quad\quad\quad\ + \mathbb{P}_{\mathbf{A}^*}\left( \left|\Delta_{\mathcal{Q}_t}^m \hat{D}^{p-1}Y^{(i)} - \tilde{b}^{(i)} \Delta_{\mathcal{Q}_t}^m\right| > \nu^{(i), m}_t,\, \Delta^m_{\mathcal{Q}_t}N^{(i)} = 0 \right) \bigg] \\
				%----------------------------
				%  &\overset{(iii)}{\leq} \sum_{i=1}^{d} \sum_{m=0}^{M_t-1} \bigg[ \mathbb{P}_{\mathbf{A}^*}\left( \left|\Delta_{\mathcal{Q}_t}^m \hat{D}^{p-1}Y^{c, (i)}_{\mathbf{A}^{(0)}} + \Delta_{\mathcal{Q}_t}^m \tilde{J}^{(i)}\right| \leq \nu^{(i), m}_t,\, \Delta^m_{\mathcal{Q}_t} N^{(i)} = 1 \right)  \\
				%  &\quad\quad\quad\quad\quad\quad\quad\quad  + \mathbb{P}_{\mathbf{A}^*}\left( \left|\Delta_{\mathcal{Q}_t}^m \hat{D}^{p-1}Y^{c, (i)}_{\mathbf{A}^{(0)}} + \Delta_{\mathcal{Q}_t}^m \tilde{J}^{(i)}\right| \leq \nu^{(i), m}_t ,\, \Delta^m_{\mathcal{Q}_t} N^{(i)} \geq 2 \right)  \\
				%  &\quad\quad\quad\quad\quad\quad\quad\quad\quad\quad\quad\quad\quad\quad\quad\quad\quad\quad\quad\quad\quad + \mathbb{P}_{\mathbf{A}^*}\left( \left|\Delta_{\mathcal{Q}_t}^m \hat{D}^{p-1}Y^{c, (i)}_{\mathbf{A}^{(0)}}\right| > \nu^{(i), m}_t \right) \bigg] \\
				%---------------------------------
				&\overset{(iii)}{\leq} \sum_{i=1}^{d} \sum_{m=0}^{M_t-1} \bigg[ \mathbb{P}_{\mathbf{A}^*}\left( \left|\Delta_{\mathcal{Q}_t}^m \hat{D}^{p-1}Y^{c, (i)}_{\mathbf{A}^{(0)}} + \Delta_{\mathcal{Q}_t}^m \tilde{J}^{(i)}\right| \leq \nu^{(i), m}_t, \, \Delta^m_{\mathcal{Q}_t} N^{(i)} =1 \right)  \\
				&\quad\quad\quad\quad\quad\quad\quad\quad\quad\quad\quad  + \mathbb{P}_{\mathbf{A}*}\left(\Delta^m_{\mathcal{Q}_t}N^{(i)} \geq 2 \right) + \mathbb{P}_{\mathbf{A}^*}\left( \left|\Delta_{\mathcal{Q}_t}^m \hat{D}^{p-1}Y^{c, (i)}_{\mathbf{A}^{(0)}}\right| > \nu^{(i), m}_t \right) \bigg] \\
				%--------------------------------
				&\overset{(iv)}{\leq} \sum_{i=1}^{d} \sum_{m=0}^{M_t-1} \bigg[ \mathbb{P}_{\mathbf{A}^*}\left( \left|\Delta_{\mathcal{Q}_t}^m \hat{D}^{p-1}Y^{c, (i)}_{\mathbf{A}^{(0)}} + \Delta_{\mathcal{Q}_t}^m \tilde{J}^{(i)}\right| \leq \nu^{(i), m}_t,\ \left|\Delta_{\mathcal{Q}_t}^m \tilde{J}^{(i)}\right| > 2\nu^{(i), m}_t, \, \Delta^m_{\mathcal{Q}_t} N^{(i)} =1 \right) \\
				&\quad\quad\quad\quad\quad\quad\quad + \mathbb{P}_{\mathbf{A}^*}\left( \left|\Delta_{\mathcal{Q}_t}^m \hat{D}^{p-1}Y^{c, (i)}_{\mathbf{A}^{(0)}} + \Delta_{\mathcal{Q}_t}^m \tilde{J}^{(i)}\right| \leq \nu^{(i), m}_t,\ \left|\Delta_{\mathcal{Q}_t}^m \tilde{J}^{(i)}\right| \leq 2\nu^{(i), m}_t, \, \Delta^m_{\mathcal{Q}_t} N^{(i)} =1 \right) \\
				&\quad\quad\quad\quad\quad\quad\quad\quad\quad\quad\quad\quad\quad\quad\quad\quad\quad\quad\quad\quad + \left(\lambda^{(i)} \Delta^m_{\mathcal{Q}_t}\right)^2  + \mathbb{P}_{\mathbf{A}^*}\left( \left|\Delta_{\mathcal{Q}_t}^m \hat{D}^{p-1}Y^{c, (i)}_{\mathbf{A}^{(0)}}\right| > \nu^{(i), m}_t \right) \bigg] \\
				%--------------------------------
				&\overset{(v)}{\leq} \sum_{i=1}^{d} \sum_{m=0}^{M_t-1} \bigg[ \mathbb{P}_{\mathbf{A}^*}\left( \left|\Delta_{\mathcal{Q}_t}^m \hat{D}^{p-1}Y^{c, (i)}_{\mathbf{A}^{(0)}} \right| > \nu^{(i), m}_t, \Delta^m_{\mathcal{Q}_t} N^{(i)} =1 \right) \\
				&\quad\quad\quad\quad\quad\quad\quad + \mathbb{P}_{\mathbf{A}^*}\left( \left|\Delta_{\mathcal{Q}_t}^m \tilde{J}^{(i)}\right| \leq 2\nu^{(i), m}_t \,\Big|\, \Delta^m_{\mathcal{Q}_t} N^{(i)} =1 \right) \mathbb{P}_{\mathbf{A}^*}\left( \Delta^m_{\mathcal{Q}_t} N^{(i)} =1 \right) \\
				&\quad\quad\quad\quad\quad\quad\quad\quad\quad\quad\quad\quad\quad\quad\quad\quad\quad\quad\quad + \left(\lambda^{(i)} \Delta^m_{\mathcal{Q}_t}\right)^2  + \mathbb{P}_{\mathbf{A}^*}\left( \left|\Delta_{\mathcal{Q}_t}^m \hat{D}^{p-1}Y^{c, (i)}_{\mathbf{A}^{(0)}}\right| > \nu^{(i), m}_t \right) \bigg] \\
				%-------------------------------
				&\overset{(vi)}{\leq} \sum_{i=1}^{d} \sum_{m=0}^{M_t-1} \bigg[ 2\mathbb{P}_{\mathbf{A}^*}\left( \left|\Delta_{\mathcal{Q}_t}^m \hat{D}^{p-1}Y^{c, (i)}_{\mathbf{A}^{(0)}} \right| > \nu^{(i), m}_t \right) + \tilde{F}^{(i)}\left(\left(-2\nu_t^{(i),m}, 2\nu_t^{(i),m}\right)\right) \lambda^{(i)} \Delta^m_{\mathcal{Q}_t} + \left(\lambda^{(i)} \Delta^m_{\mathcal{Q}_t}\right)^2  \bigg]
				%-------------------------------
			\end{align*}
			where we use in $(i)$ Equation \eqref{eqn:A=KUM}, in $(ii)$ Equation \eqref{eqn:K,M}, in $(iii)$ Equation \eqref{eqn:cont_mart_finite_act}, in $(iv)$ theorem of total probability and 
			$\Delta_{\mathcal{Q}_t}^m N^{(i)} \overset{\mathbb{P}_{\mathbf{A}^*}}{\sim} \mathrm{Poisson}\left(\lambda^{(i)} \Delta_{\mathcal{Q}_t}^m\right)$, in $(v)$ triangle inequality and definition of conditional probability, in $(vi)$ the jump distribution
			$\Delta_{\mathcal{Q}_t}^m \tilde{J}^{(i)} \,\big|\, \Delta^m_{\mathcal{Q}_t} N^{(i)} =1 \overset{\mathbb{P}_{\mathbf{A}^*}}{\sim} \tilde{F}^{(i)}(\cdot)$. Next, we can bound
			\begin{align}
				\mathbb{P}_{\mathbf{A}^*}&\left( \left|\Delta_{\mathcal{Q}_t}^m \hat{D}^{p-1}Y^{c, (i)}_{\mathbf{A}^{(0)}} \right| > \nu^{(i), m}_t \right) \notag \\
				%---------------------
				%  &\overset{(i)}{\leq} \mathbb{P}_{\mathbf{A}^*}\left(
				%  \left|\Delta_{\mathcal{Q}_t}^m \hat{D}^{p-1}Y^{c, (i)}_{\mathbf{A}^{(0)}} - \Delta_{\mathcal{Q}_t}^m D^{p-1}Y^{c, (i)}_{\mathbf{A}^{(0)}} \right| + \left|\Delta_{\mathcal{Q}_t}^m D^{p-1}Y^{c, (i)}_{\mathbf{A}^{(0)}} \right| > \nu^{(i), m}_t \right) \notag \\
				%----------------------
				&\overset{(i)}{\leq}  \mathbb{P}_{\mathbf{A}^*}\left(
				\left|\Delta_{\mathcal{Q}_t}^m \hat{D}^{p-1}Y^{c, (i)}_{\mathbf{A}^{(0)}} - \Delta_{\mathcal{Q}_t}^m D^{p-1}Y^{c, (i)}_{\mathbf{A}^{(0)}} \right| > \frac{\nu_t^{(i)}}{2} \right) + \mathbb{P}_{\mathbf{A}^*}\left( \left|\Delta_{\mathcal{Q}_t}^m D^{p-1}Y^{c, (i)}_{\mathbf{A}^{(0)}} \right| > \frac{\nu_t^{(i)}}{2}\right) \notag \\
				%-------------------------
				&\overset{(ii)}{\leq} \frac{4}{\left(\nu^{(i), m}_t\right)^2} \mathbb{E}_{\mathbf{A}^*}\left[
				\left\|\Delta_{\mathcal{Q}_t}^m \hat{D}^{p-1}\mathbf{Y}^{c}_{\mathbf{A}^{(0)}} - \Delta_{\mathcal{Q}_t}^m D^{p-1}\mathbf{Y}^{c}_{\mathbf{A}^{(0)}} \right\|^2 \right] \notag \\
				&\quad\quad\quad\quad\quad\quad\quad\quad\quad\quad\quad\quad\quad\quad\quad + \mathbb{P}_{\mathbf{A}^*}\Bigg( \Big| \Delta_{\mathcal{Q}_t}^m (\Sigma^{1/2}\mathbf{W})^{(i)} - \int_{u_m}^{u_{m+1}} \sum_{j=1}^p  (A^*_jD^{p-j}\mathbf{Y}_u)^{(i)}\,du \Big| > \frac{\nu_t^{(i)}}{2} \Bigg) \notag \\
				%------------------------
				&\overset{(iii)}{\lesssim} \frac{\Delta_{\mathcal{P}_t}}{\left(\nu^{(i), m}_t\right)^{2}} + \mathbb{P}_{\mathbf{A}^*}\Bigg( \Big|\Delta_{\mathcal{Q}_t}^m(\Sigma^{1/2} \mathbf{W})^{(i)}\Big| > \frac{\nu_t^{(i)}}{4} \Bigg) + \mathbb{P}_{\mathbf{A}^*}\Bigg( \Big| \int_{u_m}^{u_{m+1}} \sum_{j=1}^p (A^*_j D^{p-j}\mathbf{Y}_u)^{(i)} \,du \Big| > \frac{\nu_t^{(i)}}{4} \Bigg)\notag \\
				%------------------------
				&\overset{(iv)}{\lesssim} \frac{\Delta_{\mathcal{P}_t}}{\left(\nu^{(i), m}_t\right)^{2}} +  \frac{\left(\Delta^m_{\mathcal{Q}_t}\right)^{1/2}}{\nu^{(i), m}_t} \exp\left\{ -\frac{1}{32} \nu^{(i), m}_t\left(\Delta^m_{\mathcal{Q}_t}\right)^{-1/2}\right\} \notag\\
				& \quad\quad\quad\quad\quad\quad\quad\quad\quad\quad\quad\quad\quad\quad\quad\quad\quad\quad\quad\quad + \frac{16}{\left(\nu^{(i), m}_t\right)^2} \mathbb{E}_{\mathbf{A}^*}\left[ \Big| \int_{u_m}^{u_{m+1}} \sum_{j=1}^p (A^*_j D^{p-j}\mathbf{Y}_u)^{(i)} \,du \Big|^2 \right]\notag \\
				%------------------------
				&\overset{(v)}{\lesssim} 
				\frac{\Delta_{\mathcal{P}_t}}{\left(\nu^{(i), m}_t\right)^{2}} + \frac{\left(\Delta^m_{\mathcal{Q}_t}\right)^{1/2}}{\nu^{(i), m}_t} \exp\left\{ -\frac{1}{32} \nu^{(i), m}_t\left(\Delta^m_{\mathcal{Q}_t}\right)^{-1/2}\right\} \notag\\
				& \quad\quad\quad\quad\quad\quad\quad\quad\quad\quad\quad\quad\quad\quad\quad\quad\quad\quad\quad\quad + \frac{\Delta_{\mathcal{Q}_t}^m}{\left(\nu^{(i), m}_t\right)^2} \int_{u_m}^{u_{m+1}}\mathbb{E}_{\mathbf{A}^*}\bigg[ \Big|\sum_{j=1}^p (A^*_j D^{p-j}\mathbf{Y}_\infty)^{(i)}\Big|^2 \bigg]\,du 
				\notag \\
				%------------------------
				&\overset{(vi)}{\lesssim} \frac{\Delta_{\mathcal{P}_t}}{\left(\nu^{(i), m}_t\right)^{2}} + \frac{\left(\Delta^m_{\mathcal{Q}_t}\right)^{1/2}}{\nu^{(i), m}_t} \exp\left\{ -\frac{1}{32} \nu^{(i), m}_t\left(\Delta^m_{\mathcal{Q}_t}\right)^{-1/2}\right\} +  \frac{\left(\Delta^m_{\mathcal{Q}_t}\right)^{2}}{\left(\nu^{(i), m}_t\right)^{2}} \notag \\
				%----------------------------
				&\overset{(vii)}{\lesssim} \left(\Delta_{\mathcal{P}_t} + \left(\Delta^m_{\mathcal{Q}_t}\right)^{2}\right) \left(\Delta^m_{\mathcal{Q}_t}\right)^{-2\beta^{(i)}}, \text{ for $t\in\mathcal{T}$ sufficiently large,} \label{eqn:P(mart>nu)}
			\end{align}
			where we use in $(i)$ triangle inequality, in $(ii)$ Markov's inequality and Equation \eqref{eqn:DY=W-A}, in $(iii)$ Equation \eqref{eqn:incr_difference_bound} and triangle inequality, in $(iv)$ [Klenke, 2008] Lemma 22.2 [add reference] and Markov's inequality, in $(v)$ Jensen's inequality and stationarity, in $(vi)$ $\mathbb{E}_{\mathbf{A}^*}[\mathbf{X}_\infty \mathbf{X}^\mathrm{T}_\infty ] < \infty$ and in $(vii)$ Assumption \ref{ass:finite_thresholding}$.(i)$, $\exp\{-x/32\} \lesssim x^{-3}$ for $x=(\Delta^m_{\mathcal{Q}_t})^{\beta^{(i)}- 1/2}$ sufficiently large and $\Delta_{\mathcal{Q}_t}^m\lesssim 1$ for $t$ sufficiently large. Plugging this into the previous inequality and using $\Delta^m_{\mathcal{Q}_t} \geq c_{\mathcal{Q}_t} \Delta_{\mathcal{Q}_t}$, for $t\in\mathcal{T}$ sufficiently large we have 
			\begin{align*}
				\mathbb{P}_{\mathbf{A}^*}\left(A_t^{c}\right) & \lesssim \sum_{i=1}^{d} \left[c_{\mathcal{Q}_t}^{-2\beta^{(i)}} M_t \Delta_{\mathcal{P}_t}\Delta_{\mathcal{Q}_t}^{-2\beta^{(i)}} + t \Delta_{\mathcal{Q}_t}^{1-2\beta^{(i)}} + t \tilde{F}^{(i)}\left(\left(-2(\Delta^m_{\mathcal{Q}_t})^{\beta^{(i)}},\ 2(\Delta^m_{\mathcal{Q}_t})^{\beta^{(i)}}\right)\right) + t \Delta_{\mathcal{Q}_t}\right] \\
				&\lesssim \sum_{i=1}^{d}\left[ c_{\mathcal{Q}_t}^{-2\beta^{(i)}} (t^{-1}N'_t\Delta_{\mathcal{Q}_t}) (t\Delta_{\mathcal{Q}_t}^{-2}\Delta_{\mathcal{P}_t}) \Delta_{\mathcal{Q}_t}^{1-2\beta^{(i)}} + t \Delta_{\mathcal{Q}_t}^{1-2\beta^{(i)}} + t \tilde{F}^{(i)}\left(\left(-2\Delta_{\mathcal{Q}_t}^{\beta^{(i)}},\ 2\Delta_{\mathcal{Q}_t}^{\beta^{(i)}}\right)\right) \right]\rightarrow 0,
			\end{align*}
			as $t\rightarrow \infty$, by Assumption \ref{ass:controlled_sampling}, Assumption \ref{ass:joint_mesh}, Assumption \ref{ass:finite_thresholding}$.(ii)$ and Assumption \ref{ass:finite_thresholding}$.(iii)$. We note that it now suffices to show the $L^1(\Omega, \mathcal{F},\mathbb{P}_{\mathbf{A}^*})$ convergence 
			\begin{equation} \label{eqn:restricted_L1_conv}
				\mathbb{E}_{\mathbf{A}^*}\left[\left\|t^{-1/2}\left( \mathbf{H}_{\mathcal{P}_t, \mathcal{Q}_t, \boldsymbol{\nu}_t} - \mathbf{H}_{\mathcal{P}_t,\mathcal{Q}_t}\right) \mathds{1}_{A_t} \right\|\right] \rightarrow 0, \quad t\rightarrow\infty,
			\end{equation}
			as this immediately implies convergence in probability
			\begin{equation} \label{eqn:restricted_p_conv}
				t^{-1/2}\left( \mathbf{H}_{\mathcal{P}_t, \mathcal{Q}_t, \boldsymbol{\nu}_t} - \mathbf{H}_{\mathcal{P}_t,\mathcal{Q}_t}\right) \mathds{1}_{A_t} \overset{\mathbb{P}_{\mathbf{A}^*}}{\rightarrow} 0, \quad t\rightarrow\infty, \end{equation}
			and hence $\forall\epsilon>0,$
			\begin{align*}
				\mathbb{P}_{\mathbf{A}^*}\left(\left\|t^{-1/2}\left( \mathbf{H}_{\mathcal{P}_t, \mathcal{Q}_t, \boldsymbol{\nu}_t} - \mathbf{H}_{\mathcal{P}_t,\mathcal{Q}_t}\right)\right\| > \epsilon \right) 
				&= \mathbb{P}_{\mathbf{A}^*}\left(\left\|t^{-1/2}\left( \mathbf{H}_{\mathcal{P}_t, \mathcal{Q}_t, \boldsymbol{\nu}_t} - \mathbf{H}_{\mathcal{P}_t,\mathcal{Q}_t}\right)\right\| > \epsilon \cap A_t \right) + \\
				&\quad\quad\quad\quad\quad\quad\quad\quad\quad \mathbb{P}_{\mathbf{A}^*}\left(\left\|t^{-1/2}\left( \mathbf{H}_{\mathcal{P}_t, \mathcal{Q}_t, \boldsymbol{\nu}_t} - \mathbf{H}_{\mathcal{P}_t,\mathcal{Q}_t}\right)\right\| > \epsilon \cap A_t^c \right) \\
				&\leq \mathbb{P}_{\mathbf{A}^*}\left(\left\|t^{-1/2}\left( \mathbf{H}_{\mathcal{P}_t, \mathcal{Q}_t, \boldsymbol{\nu}_t} - \mathbf{H}_{\mathcal{P}_t,\mathcal{Q}_t}\right)\mathds{1}_{A_t}\right\| > \epsilon \right) + \mathbb{P}_{\mathbf{A}^*}\left(A_t^c \right) \\
				&\rightarrow 0, \quad t\rightarrow\infty,
			\end{align*}
			by combining \eqref{eqn:restricted_p_conv} and \eqref{eqn:P(A)_conv_one}, i.e.\ this implies the sufficient condition \eqref{eqn:H_limit} for the required asymptotic properties of $\hat{\mathbf{A}}(\mathbb{Y}_{\mathcal{P}_t}; \mathcal{Q}_t, \boldsymbol{\nu}^m_t)$.
			
			To show Equation \eqref{eqn:restricted_L1_conv} note that for each $d$-entry batch
			
			\begin{align*}
				t^{-1/2}&\mathbb{E}_{\mathbf{A}^*}\left[\left\| \left(\sum_{m = 0}^{M_t-1}
				\hat{D}^{l}Y^{(i)}_{u_m} \Sigma^{-1} \Delta^m_{\mathcal{Q}_t,\boldsymbol{\nu}_t} \hat{D}^{p-1}\mathbf{Y}^c_{\mathbf{A}^{(0)}} - \sum_{m = 0}^{M_t-1}
				\hat{D}^{l}Y^{(i)}_{u_m} \Sigma^{-1} \Delta^m_{\mathcal{Q}_t} \hat{D}^{p-1}\mathbf{Y}^c_{\mathbf{A}^{(0)}} \right)\mathds{1}_{A_t}\right\|\right] \\
				%---------------------------------
				&\overset{(i)}{\leq} t^{-1/2}\|\Sigma^{-1}\|\sum_{m=0}^{M_t-1} \mathbb{E}_{\mathbf{A}^*}\left[\left\|
				\hat{D}^{l}Y^{(i)}_{u_m}\left( \Delta^m_{\mathcal{Q}_t,\boldsymbol{\nu}_t} \hat{D}^{p-1}\mathbf{Y}^c_{\mathbf{A}^{(0)}} -  \Delta^m_{\mathcal{Q}_t} \hat{D}^{p-1}\mathbf{Y}^c_{\mathbf{A}^{(0)}}\right)\mathds{1}_{A_t}\right\|\right] \\
				%---------------------------------
				&\overset{(ii)}{\leq} t^{-1/2}\|\Sigma^{-1}\|\sum_{m=0}^{M_t-1} \sum_{i=1}^d \mathbb{E}_{\mathbf{A}^*}\left[ \left|
				\hat{D}^{l}Y^{(i)}_{u_m} \ \Delta^m_{\mathcal{Q}_t} \hat{D}^{p-1}Y^{c,(i)}_{\mathbf{A}^{(0)}}\ \mathds{1}_{A_t\cap M^{m,i,c}_t}\right|\right] \\
				%---------------------------------
				&\overset{(iii)}{\leq} t^{-1/2}\|\Sigma^{-1}\|\sum_{m=0}^{M_t-1} \sum_{i=1}^d\bigg\{ \mathbb{E}_{\mathbf{A}^*}\left[ \left|
				\hat{D}^{l}Y^{(i)}_{u_m} \ \Delta^m_{\mathcal{Q}_t} \left( \hat{D}^{p-1}Y^{c,(i)}_{\mathbf{A}^{(0)}} - D^{p-1}Y^{c,(i)}_{\mathbf{A}^{(0)}} \right) \right|\right]\\
				&\quad\quad\quad\quad\quad\quad\quad\quad\quad\quad\quad\quad\quad\quad\quad + \mathbb{E}_{\mathbf{A}^*}\left[ \left|
				\left(\hat{D}^{l}Y^{(i)}_{u_m} - D^{l}Y^{(i)}_{u_m}\right) \Delta^m_{\mathcal{Q}_t} D^{p-1}Y^{c,(i)}_{\mathbf{A}^{(0)}}\right|\right] \\
				&\quad\quad\quad\quad\quad\quad\quad\quad\quad\quad\quad\quad\quad\quad\quad\quad\quad\quad + \mathbb{E}_{\mathbf{A}^*}\left[ \left|
				D^{l}Y^{(i)}_{u_m} \ \Delta^m_{\mathcal{Q}_t} D^{p-1}Y^{c,(i)}_{\mathbf{A}^{(0)}}\ \mathds{1}_{M^{m,i,c}_t}\right|\right] \bigg\}\\
				%---------------------------------
				&\overset{(iv)}{\leq} t^{-1/2}\|\Sigma^{-1}\|\sum_{m=0}^{M_t-1} \sum_{i=1}^d\bigg\{ \mathbb{E}_{\mathbf{A}^*}\left[ \left|
				\hat{D}^{l}Y^{(i)}_{\infty}\right|^2\right]^{1/2} \mathbb{E}_{\mathbf{A}^*}\left[ \left| \Delta^m_{\mathcal{Q}_t} \left( \hat{D}^{p-1}Y^{c,(i)}_{\mathbf{A}^{(0)}} - D^{p-1}Y^{c,(i)}_{\mathbf{A}^{(0)}} \right) \right|^2\right]^{1/2}\\
				&\quad\quad\quad\quad\quad\quad\quad\quad\quad\quad\quad\quad + \mathbb{E}_{\mathbf{A}^*}\left[ \left|
				\hat{D}^{l}Y^{(i)}_{u_m} - D^{l}Y^{(i)}_{u_m}\right|^2\right]^{1/2} \mathbb{E}_{\mathbf{A}^*}\left[ \left| \Delta^m_{\mathcal{Q}_t} D^{p-1}Y^{c,(i)}_{\mathbf{A}^{(0)}}\right|^2\right]^{1/2} \\
				&\quad\quad\quad\quad\quad\quad\quad\quad\quad\quad\quad\quad\quad\quad\quad\quad\quad + \mathbb{E}_{\mathbf{A}^*}\left[ \left|
				D^{l}Y^{(i)}_{u_m}\  \Delta_{\mathcal{Q}_t}^m (\Sigma^{1/2}\mathbf{W})^{(i)}\  \mathds{1}_{M^{m,i,c}_t}\right|\right] \\
				&\quad\quad\quad\quad\quad\quad\quad\quad\quad\quad\quad\quad\quad\quad\quad\quad\quad\quad\quad + \mathbb{E}_{\mathbf{A}^*}\Big[ \Big|
				D^{l}Y^{(i)}_{u_m}\  \int_{u_m}^{u_{m+1}} \sum_{j=1}^p  (A^*_jD^{p-j}\mathbf{Y}_u)^{(i)}\,du \  \mathds{1}_{M^{m,i,c}_t}\Big|\Big] \bigg\}\\
				%---------------------------------
				&\overset{(v)}{\lesssim} t^{-1/2} \sum_{m=0}^{M_t-1} \sum_{i=1}^d\bigg\{ \Delta_{\mathcal{P}_t}^{1/2} + \Delta_{\mathcal{P}_t}^{1/2} \Delta_{\mathcal{Q}_t}^{1/2} + \mathbb{E}_{\mathbf{A}^*}\left[ \left|
				D^{l}Y^{(i)}_{\infty}\ \right|\right] \mathbb{E}_{\mathbf{A}^*}\left[ \left| \Delta_{\mathcal{Q}_t}^m (\Sigma^{1/2}\mathbf{W})^{(i)}\right|\right] \mathbb{P}_{\mathbf{A}^*}\left( M^{m,i,c}_t\right) \\
				&\quad\quad\quad\quad\quad\quad\quad\quad\quad\quad\quad\quad\quad + \sum_{j=1}^p \mathbb{E}_{\mathbf{A}^*}\Big[ \Big|
				D^{l}Y^{(i)}_{u_m}\ \mathds{1}_{M^{m,i,c}_t}\Big|^2\Big]^{1/2}  \mathbb{E}_{\mathbf{A}^*}\Big[ \Big| \int_{u_m}^{u_{m+1}} (A^*_jD^{p-j}\mathbf{Y}_u)^{(i)}\,du \Big|^2\Big]^{1/2} \bigg\}\\
				%---------------------------------
				&\overset{(vi)}{\lesssim} t^{-1/2} \sum_{m=0}^{M_t-1} \sum_{i=1}^d\bigg\{ \Delta_{\mathcal{P}_t}^{1/2} + \Delta_{\mathcal{P}_t}^{1/2} \Delta_{\mathcal{Q}_t}^{1/2} +   (\Delta_{\mathcal{Q}_t}^m)^{1/2} \Delta_{\mathcal{Q}_t}^m + \sum_{j=1}^p \mathbb{E}_{\mathbf{A}^*}\Big[ \Big|
				D^{l}Y^{(i)}_{\infty}\Big|^2\Big]^{1/2}  \mathbb{P}_{\mathbf{A}^*}\left( M^{m,i,c}_t\right)^{1/2} \Delta^m_{\mathcal{Q}_t}\bigg\}\\
				%---------------------------------
				&\overset{(vii)}{\lesssim} t^{-1/2} \sum_{m=0}^{M_t-1} \sum_{i=1}^d\bigg\{ \Delta_{\mathcal{P}_t}^{1/2} + \Delta_{\mathcal{P}_t}^{1/2} \Delta_{\mathcal{Q}_t}^{1/2} +   (\Delta_{\mathcal{Q}_t}^m)^{3/2} + (\Delta_{\mathcal{Q}_t}^m)^{3/2} \bigg\}\\
				%---------------------------------
				&\overset{(viii)}{\lesssim} \left(t^{-1} N'_t \Delta_{\mathcal{Q}_t}\right) \left(t\Delta_{\mathcal{Q}_t}^{-2} \Delta_{\mathcal{P}_t}\right)^{1/2}\Delta_{\mathcal{P}_t}^{1/2} + t^{1/2} \Delta^{1/2}_{\mathcal{Q}_t} \rightarrow 0,\quad t\rightarrow\infty,
			\end{align*}
			where we use in $(i)$ the triangle inequality, in $(ii)$ the domination $\|\cdot\|_2\leq \|\cdot\|_1$ and the fact that on $A_t$ either the increment is perfectly recovered, on $M_t^{m,i}$, or set to zero by the thresholding, on $M_t^{m,i,c}$, in $(iii)$ the triangle inequality, in $(iv)$ Cauchy-Schwartz inequality (and stationarity) on the first two terms, Equation \eqref{eqn:DY=W-A} and triangle inequality on the third term, in $(v)$ Equation \eqref{eqn:incr_difference_bound} and boundedness of the second moments on the first term, the bounds \eqref{eqn:D-D_L2_bound} and \eqref{eqn:bound_cont_mart_incr} on the second term, independence of $D^l \mathbf{Y}_{u_m}$, $\Delta_{\mathcal{Q}_t}^m \mathbf{W}$ and $\Delta^m_{\mathcal{Q}_t}\mathbf{N}$ on the third term, triangle inequality and Cauchy-Schwartz on the fourth term, in $(vi)$ boundedness of second moments of $\mathbf{X}_t$, second moment of Brownian increments and $\mathbb{P}_{\mathbf{A}^*}(M_t^{m,i,c}) \leq \lambda^{(i)} \Delta^m_{\mathcal{Q}_t}$ on the third term, independence of $D^l \mathbf{Y}_{u_m}$ and $\Delta^m_{\mathcal{Q}_t}\mathbf{N}$, Jensen's inequality and uniform  boundedness of second moments in the fourth term, and finally $(viii)$ Assumption \ref{ass:controlled_sampling}, Assumption \ref{ass:joint_mesh} and Assumption \ref{ass:finite_thresholding}$.(ii)$.
			
			We can bound the squared increment of the continuous martingale part 
			\begin{align}
				\mathbb{E}_{\mathbf{A}^*}&\left[ \left\| \Delta^m_{\mathcal{Q}_t} D^{p-1}\mathbf{Y}^{c}_{\mathbf{A}^{(0)}}\right\|^2\right]^{1/2} \notag \\
				% 		&= \mathbb{E}_{\mathbf{A}^*}\left[ \left\| \Delta^m_{\mathcal{Q}_t} \Sigma^{1/2}\mathbf{W} - \int_{u_m}^{u_{m+1}} \sum_{j=1}^p A_j^* D^{p-j} \mathbf{Y}_u\, du\right\|^2\right]^{1/2} \notag \\
				&\overset{(i)}{\leq} \mathbb{E}_{\mathbf{A}^*}\left[ \left\| \Delta^m_{\mathcal{Q}_t} \Sigma^{1/2}\mathbf{W}\right\|^2\right]^{1/2} +  \sum_{j=1}^p \|A_j^*\| \, \mathbb{E}_{\mathbf{A}^*}\left[ \left\| \int_{u_m}^{u_{m+1}} D^{p-j} \mathbf{Y}_u\, du\right\|^2\right]^{1/2} \notag \\
				&\overset{(ii)}{\leq} \mathbb{E}_{\mathbf{A}^*}\left[ \left\| \Delta^m_{\mathcal{Q}_t} \Sigma^{1/2}\mathbf{W}\right\|^2\right]^{1/2} +  \sum_{j=1}^p \|A_j^*\| \, \left(\Delta^m _{\mathcal{Q}_t} \int_{u_m}^{u_{m+1}}  \mathbb{E}_{\mathbf{A}^*}\left[ \left\| D^{p-j} \mathbf{Y}_u\right\|^2\right]\, du \right)^{1/2} \notag\\
				&\overset{(iii)}{\lesssim} (\Delta^m_{\mathcal{Q}_t})^{1/2}, \label{eqn:bound_cont_mart_incr}
			\end{align}
			by using $(i)$ the representation \eqref{eqn:DY=W-A} and triangle inequality, $(ii)$ Jensen's inequality and $(iii)$ second moment of Brownian increments and stationarity and boundedness of second moments of $\mathbf{X}_t = (\mathbf{Y}^{\mathrm{T}}_t,\ldots, D^{p-1}\mathbf{Y}^{\mathrm{T}}_t)^{\mathrm{T}}$.
		\end{proof}
		
		\subsection{Proof of Theorem \ref{thm:cons_asymp_third_approx_infinite_activity}} \label{app:proof_thm_cons_asymp_third_approx_infinite_activity}
		
		\begin{proof}
			Recall that by Lemma \ref{lemma:approx} it suffices to show Equation \eqref{eqn:H_limit} with $\hat{\mathbf{A}}_{1,t} = \hat{\mathbf{A}}(\mathbb{Y}_{\mathcal{P}_t}, \mathbb{J}_{\mathcal{Q}_t}, \mathbb{M}_{\mathcal{Q}_t})$ and $\hat{\mathbf{A}}_{2,t} = \hat{\mathbf{A}}(\mathbb{Y}_{\mathcal{P}_t}; \mathcal{Q}_t, \boldsymbol{\nu}_t)$. We start by splitting the limiting random variable in \eqref{eqn:H_limit} into three components. For $i\in\{1,\ldots, d\}$ and $l\in\{0,\ldots, p-1\}$, after factoring out the $\Sigma^{-1}$ term, we decompose each $d$-batch entry as
			\begin{align*}
				t^{-1/2}&\Big(\sum_{m = 0}^{M_t-1}
				\hat{D}^{l}Y^{(i)}_{u_m} \Delta^m_{\mathcal{Q}_t,\boldsymbol{\nu}_t} \hat{D}^{p-1}\mathbf{Y}^c_{\mathbf{A}^{(0)}} - \sum_{m = 0}^{M_t-1}
				\hat{D}^{l}Y^{(i)}_{u_m} \Delta^m_{\mathcal{Q}_t} \hat{D}^{p-1}\mathbf{Y}^c_{\mathbf{A}^{(0)}} \Big) \\
				&= t^{-1/2}\left(\sum_{m = 0}^{M_t-1}
				\hat{D}^{l}Y^{(i)}_{u_m} \Sigma^{-1} \left( \Delta^m_{\mathcal{Q}_t,\boldsymbol{\nu}_t} \hat{D}^{p-1}\mathbf{Y}^c_{\mathbf{A}^{(0)}} - \Delta^m_{\mathcal{Q}_t} \hat{D}^{p-1}\mathbf{Y}^c_{\mathbf{A}^{(0)}} \right)\right) \\
				&= t^{-1/2}\left(\sum_{m = 0}^{M_t-1}
				\hat{D}^{l}Y^{(i)}_{u_m} \left( \left( \Delta_{\mathcal{Q}_t}^m \hat{D}^{p-1}\mathbf{Y} - \mathbf{b} \Delta_{\mathcal{Q}_t}^m \right) \odot \mathds{1}_{\left\{\left|\Delta_{\mathcal{Q}_t}^m \hat{D}^{p-1}\mathbf{Y} - \mathbf{b} \Delta_{\mathcal{Q}_t}^m\right| \leq \boldsymbol{\nu}^m_t\right\}} - \Delta^m_{\mathcal{Q}_t} \hat{D}^{p-1}\mathbf{Y}^c_{\mathbf{A}^{(0)}} \right)\right) \\
				&\overset{(*)}{\approx} t^{-1/2}\left(\sum_{m = 0}^{M_t-1}
				D^{l}Y^{(i)}_{u_m} \left( \left( \Delta_{\mathcal{Q}_t}^m D^{p-1}\mathbf{Y} - \mathbf{b} \Delta_{\mathcal{Q}_t}^m \right) \odot \mathds{1}_{\left\{\left|\Delta_{\mathcal{Q}_t}^m \hat{D}^{p-1}\mathbf{Y} - \mathbf{b} \Delta_{\mathcal{Q}_t}^m\right| \leq \boldsymbol{\nu}^m_t\right\}} - \Delta^m_{\mathcal{Q}_t} D^{p-1}\mathbf{Y}^c_{\mathbf{A}^{(0)}} \right)\right) \\
				%----------------------------
				&= t^{-1/2}\left(\sum_{m = 0}^{M_t-1}
				D^{l}Y^{(i)}_{u_m} \left( \Delta_{\mathcal{Q}_t}^m D^{p-1}\tilde{\mathbf{Y}} \odot \mathds{1}_{\left\{\left|\Delta_{\mathcal{Q}_t}^m \hat{D}^{p-1}\mathbf{Y} - \mathbf{b} \Delta_{\mathcal{Q}_t}^m\right| \leq \boldsymbol{\nu}^m_t\right\}} - \Delta^m_{\mathcal{Q}_t} D^{p-1}\mathbf{Y}^c_{\mathbf{A}^{(0)}} \right)\right) \\
				&\quad\quad\quad\quad\quad\quad\quad\quad\quad\quad\quad\quad\quad + t^{-1/2}\left(\sum_{m = 0}^{M_t-1}
				D^{l}Y^{(i)}_{u_m} \Delta_{\mathcal{Q}_t}^m \mathbf{M} \odot \mathds{1}_{\left\{\left|\Delta_{\mathcal{Q}_t}^m \hat{D}^{p-1}\mathbf{Y} - \mathbf{b} \Delta_{\mathcal{Q}_t}^m\right| \leq \boldsymbol{\nu}^m_t\right\}} \right) \\
				%----------------------------
				&= t^{-1/2}\left(\sum_{m = 0}^{M_t-1}
				D^{l}Y^{(i)}_{u_m} \left( \Delta_{\mathcal{Q}_t}^m D^{p-1}\tilde{\mathbf{Y}} \odot \mathds{1}_{\left\{\left|\Delta_{\mathcal{Q}_t}^m D^{p-1}\tilde{\mathbf{Y}}\right| \leq 2\boldsymbol{\nu}^m_t\right\}} - \Delta^m_{\mathcal{Q}_t} D^{p-1}\mathbf{Y}^c_{\mathbf{A}^{(0)}} \right)\right) \\
				& \quad\quad\quad + t^{-1/2}\left(\sum_{m = 0}^{M_t-1}
				D^{l}Y^{(i)}_{u_m} \Delta_{\mathcal{Q}_t}^m D^{p-1}\tilde{\mathbf{Y}} \odot \left(\mathds{1}_{\left\{\left|\Delta_{\mathcal{Q}_t}^m \hat{D}^{p-1}\mathbf{Y} - \mathbf{b} \Delta_{\mathcal{Q}_t}^m\right| \leq \boldsymbol{\nu}^m_t\right\}} - \mathds{1}_{\left\{\left|\Delta_{\mathcal{Q}_t}^m D^{p-1}\tilde{\mathbf{Y}}\right| \leq 2\boldsymbol{\nu}^m_t\right\}} \right)\right) \\
				&\quad\quad\quad\quad\quad\quad\quad\quad\quad\quad\quad\quad\quad\quad\quad\quad\quad + t^{-1/2}\left(\sum_{m = 0}^{M_t-1}
				D^{l}Y^{(i)}_{u_m} \Delta_{\mathcal{Q}_t}^m \mathbf{M} \odot \mathds{1}_{\left\{\left|\Delta_{\mathcal{Q}_t}^m \hat{D}^{p-1}\mathbf{Y} - \mathbf{b} \Delta_{\mathcal{Q}_t}^m\right| \leq \boldsymbol{\nu}^m_t\right\}} \right) \\
				& =: \mathbf{S}_t^1(i, l) + \mathbf{S}_t^2(i, l) + \mathbf{S}_t^3(i, l),
			\end{align*}
			setting $D^{p-1}\tilde{\mathbb{Y}} = \{D^{p-1}\tilde{\mathbf{Y}}_s,\ s\geq 0\}$ with 
			\[D^{p-1}\tilde{\mathbf{Y}}_s := D^{p-1}\mathbf{Y}_0 + D^{p-1}\mathbf{Y}^c_{\mathbf{A}^{(0)}, s} + \mathbf{J}_s = D^{p-1}\mathbf{Y}_s - \mathbf{b} s - \mathbf{M}_s,\]
			and understanding the approximation step $(*)$ in the $L^1(\Omega, \mathcal{F},\mathbb{P}_{\mathbf{A}^*})$ limit (and hence in the $\mathbb{P}_{\mathbf{A}^*}$ limit) as $t\rightarrow\infty$. The validity of the approximation $(*)$ can be shown as follows for $t\in\mathcal{T}$ sufficiently large:
			\begin{align*}
				t^{-1/2}&\mathbb{E}_{\mathbf{A}^*}\Bigg[\Bigg\|\sum_{m = 0}^{M_t-1}
				\hat{D}^{l}Y^{(i)}_{u_m} \left( \left( \Delta_{\mathcal{Q}_t}^m \hat{D}^{p-1}\mathbf{Y} - \mathbf{b} \Delta_{\mathcal{Q}_t}^m \right) \odot \mathds{1}_{\left\{\left|\Delta_{\mathcal{Q}_t}^m \hat{D}^{p-1}\mathbf{Y} - \mathbf{b} \Delta_{\mathcal{Q}_t}^m\right| \leq \boldsymbol{\nu}^m_t\right\}} - \Delta^m_{\mathcal{Q}_t} \hat{D}^{p-1}\mathbf{Y}^c_{\mathbf{A}^{(0)}} \right) \\
				&\quad\quad\quad\quad - \sum_{m = 0}^{M_t-1}
				D^{l}Y^{(i)}_{u_m} \left( \left( \Delta_{\mathcal{Q}_t}^m D^{p-1}\mathbf{Y} - \mathbf{b} \Delta_{\mathcal{Q}_t}^m \right) \odot \mathds{1}_{\left\{\left|\Delta_{\mathcal{Q}_t}^m \hat{D}^{p-1}\mathbf{Y} - \mathbf{b} \Delta_{\mathcal{Q}_t}^m\right| \leq \boldsymbol{\nu}^m_t\right\}} - \Delta^m_{\mathcal{Q}_t} D^{p-1}\mathbf{Y}^c_{\mathbf{A}^{(0)}} \right)
				\Bigg\|\Bigg] \\
				%----------------------------------------------
				&\overset{(i)}{\leq} t^{-1/2}\sum_{m = 0}^{M_t-1}\Bigg\{ \mathbb{E}_{\mathbf{A}^*}\Bigg[\Bigg\|
				\hat{D}^{l}Y^{(i)}_{u_m} \left( \Delta_{\mathcal{Q}_t}^m \hat{D}^{p-1}\mathbf{Y} - \mathbf{b} \Delta_{\mathcal{Q}_t}^m \right) - D^{l}Y^{(i)}_{u_m} \left( \Delta_{\mathcal{Q}_t}^m D^{p-1}\mathbf{Y} - \mathbf{b} \Delta_{\mathcal{Q}_t}^m \right)
				\Bigg\|\Bigg]  \\
				&\quad\quad\quad\quad\quad\quad\quad\quad\quad\quad\quad\quad\quad\quad + \mathbb{E}_{\mathbf{A}^*}\Bigg[\Bigg\|
				\hat{D}^{l}Y^{(i)}_{u_m} \Delta^m_{\mathcal{Q}_t} \hat{D}^{p-1}\mathbf{Y}^c_{\mathbf{A}^{(0)}} - D^{l}Y^{(i)}_{u_m} \Delta^m_{\mathcal{Q}_t} D^{p-1}\mathbf{Y}^c_{\mathbf{A}^{(0)}} 
				\Bigg\|\Bigg]\Bigg\} \\
				%----------------------------------------------
				&\overset{(ii)}{\leq} t^{-1/2}\sum_{m = 0}^{M_t-1}\Bigg\{ \mathbb{E}_{\mathbf{A}^*}\Bigg[\Bigg\|
				\hat{D}^{l}Y^{(i)}_{u_m} \left( \Delta_{\mathcal{Q}_t}^m \hat{D}^{p-1}\mathbf{Y} - \Delta_{\mathcal{Q}_t}^m D^{p-1}\mathbf{Y} \right)
				\Bigg\|\Bigg]  \\
				& \quad\quad\quad\quad\quad\quad\quad\quad\quad + \mathbb{E}_{\mathbf{A}^*}\Bigg[\Bigg\|
				\left(\hat{D}^{l}Y^{(i)}_{u_m}  - D^{l}Y^{(i)}_{u_m} \right)\left( \Delta_{\mathcal{Q}_t}^m D^{p-1}\mathbf{Y} - \mathbf{b} \Delta_{\mathcal{Q}_t}^m \right) \Bigg\|\Bigg] \\
				&\quad\quad\quad\quad\quad\quad\quad\quad\quad\quad\quad\quad + \mathbb{E}_{\mathbf{A}^*}\Bigg[\Bigg\|
				\hat{D}^{l}Y^{(i)}_{u_m} \left(\Delta^m_{\mathcal{Q}_t} \hat{D}^{p-1}\mathbf{Y}^c_{\mathbf{A}^{(0)}} -\Delta^m_{\mathcal{Q}_t} D^{p-1}\mathbf{Y}^c_{\mathbf{A}^{(0)}} \right)
				\Bigg\|\Bigg] \\
				&\quad\quad\quad\quad\quad\quad\quad\quad\quad\quad\quad\quad\quad\quad\quad\quad\quad + \mathbb{E}_{\mathbf{A}^*}\Bigg[\Bigg\|
				\left(\hat{D}^{l}Y^{(i)}_{u_m} - D^{l}Y^{(i)}_{u_m} \right) \Delta^m_{\mathcal{Q}_t} D^{p-1}\mathbf{Y}^c_{\mathbf{A}^{(0)}} 
				\Bigg\|\Bigg]\Bigg\} \\
				%----------------------------------------------
				&\overset{(iii)}{\lesssim} t^{-1/2}\sum_{m = 0}^{M_t-1}\Bigg\{ \mathbb{E}_{\mathbf{A}^*}\Big[\Big\|
				\hat{D}^{l}Y^{(i)}_{u_m}\Big\|^2\Big]^{1/2} \mathbb{E}_{\mathbf{A}^*}\Big[\Big\| \Delta_{\mathcal{Q}_t}^m \hat{D}^{p-1}\mathbf{Y} - \Delta_{\mathcal{Q}_t}^m D^{p-1}\mathbf{Y}\Big\|^2\Big]^{1/2} \\
				& \quad\quad\quad\quad\quad\quad\quad\quad\quad + \mathbb{E}_{\mathbf{A}^*}\Big[\Big\|
				\hat{D}^{l}Y^{(i)}_{u_m}  - D^{l}Y^{(i)}_{u_m}\Big\|^2\Big]^{1/2} \mathbb{E}_{\mathbf{A}^*}\Big[\Big\| \Delta_{\mathcal{Q}_t}^m D^{p-1}\mathbf{Y} - \mathbf{b} \Delta_{\mathcal{Q}_t}^m \Big\|^2\Big]^{1/2} \\
				&\quad\quad\quad\quad\quad\quad\quad\quad\quad\quad\quad\quad + \mathbb{E}_{\mathbf{A}^*}\Big[\Big\|
				\hat{D}^{l}Y^{(i)}_{u_m} \Big\|^2\Big]^{1/2} \mathbb{E}_{\mathbf{A}^*}\Big[\Big\|\Delta^m_{\mathcal{Q}_t} \hat{D}^{p-1}\mathbf{Y}^c_{\mathbf{A}^{(0)}} -\Delta^m_{\mathcal{Q}_t} D^{p-1}\mathbf{Y}^c_{\mathbf{A}^{(0)}} \Big\|^2\Big]^{1/2}\\
				&\quad\quad\quad\quad\quad\quad\quad\quad\quad\quad\quad\quad\quad\quad\quad\quad\quad + \mathbb{E}_{\mathbf{A}^*}\Big[\Big\|
				\hat{D}^{l}Y^{(i)}_{u_m} - D^{l}Y^{(i)}_{u_m}\Big\|^2\Big]^{1/2} \mathbb{E}_{\mathbf{A}^*}\Big[\Big\| \Delta^m_{\mathcal{Q}_t} D^{p-1}\mathbf{Y}^c_{\mathbf{A}^{(0)}} 
				\Big\|^2\Big]^{1/2} \Bigg\} \\
				%---------------------------------------------------------------
				&\overset{(iv)}{\lesssim} t^{-1/2}\sum_{m = 0}^{M_t-1} \{\Delta_{\mathcal{P}_t}^{1/2} + \Delta_{\mathcal{P}_t}^{1/2} \Delta_{\mathcal{Q}_t}^{1/2} + \Delta_{\mathcal{P}_t}^{1/2} + \Delta_{\mathcal{P}_t}^{1/2} + \Delta_{\mathcal{P}_t}^{1/2}\Delta_{\mathcal{Q}_t}^{1/2} \} \\
				%---------------------------------------------------------------
				&\overset{(v)}{\lesssim} t^{-1/2} M_t \Delta_{\mathcal{P}_t}^{1/2} \lesssim (t^{-1}N'_t\Delta_{\mathcal{Q}_t})(t\Delta^{-2}_{\mathcal{Q}_t}\Delta_{\mathcal{P}_t})^{1/2} \rightarrow 0, \quad t\rightarrow\infty,
			\end{align*}
			where we use in $(i)$ triangle inequality and $\|\mathbf{x}\,\odot\, \mathds{1}_{\{\mathbf{y}\in A\}}\| \leq \|\mathbf{x}\|$ for $\mathbf{x},\mathbf{y}\in\mathbb{R}^d$, in $(ii)$ adding and subtracting the same quantity and applying the triangle inequality multiple times, in $(iii)$ Cauchy-Schwartz inequality on all four terms, in $(iv)$ stationarity and boundedness of second moments of $\mathbf{X}_t = (\mathbf{Y}^{\mathrm{T}}_t,\ldots, D^{p-1}\mathbf{Y}^{\mathrm{T}}_t)^{\mathrm{T}}$, Equation \eqref{eqn:incr_difference_bound}, Equation \eqref{eqn:D-D_L2_bound}, Equation \eqref{eqn:big_O_X2} and Equation \eqref{eqn:bound_cont_mart_incr}, and finally in $(v)$ Assumption \ref{ass:controlled_sampling} and Assumption \ref{ass:joint_mesh}. It now remains to prove $\mathbf{S}_t^1(i, l), \mathbf{S}_t^2(i, l), \mathbf{S}_t^3(i, l)\overset{\mathbb{P}_{\mathbf{A}^*}}{\rightarrow} 0$ as $t\rightarrow\infty$ for $i\in\{1,\ldots,d\}, l\in\{0,\ldots, p-1\}$. 
			
			To show $\mathbf{S}_t^1(i, l) \overset{\mathbb{P}_{\mathbf{A}^*}}{\rightarrow} 0$ as $t\rightarrow\infty$ note that $D^{p-1}\tilde{\mathbb{Y}}$ is a finite jump activity process with $\tilde{\mathbf{b}}=0$ and $\tilde{F}^{(i)}$ for $i\in\{1,\ldots,d\}$ satisfies Assumption \ref{ass:infinite_thresholding}$.(iii)$, i.e.\ a similar condition to Assumption \ref{ass:finite_thresholding}$.(iii)$. We can thus proceed with the same strategy as in the previous paragraph. We do not repeat the full argument but summarize the main steps. We consider the events 
			\begin{align*}
				\tilde{A}_t:= \bigcap_{m=0}^{M_t-1} \bigcap_{i=1}^{d} \tilde{A}_t^{m, i},\quad\mathrm{where } \tilde{A}_t^{m, i} := \left\{\omega\in\Omega : \mathds{1}_{\left\{\left|\Delta_{\mathcal{Q}_t}^m D^{p-1}\tilde{Y}^{(i)}\right| \leq 2\nu^{(i), m}_t \right\}}(\omega) = \mathds{1}_{\left\{ \Delta^m_{\mathcal{Q}_t}N^{(i)} = 0 \right\}}(\omega) \right\},
			\end{align*}
			for $t\in\mathcal{T}$, $m\in\{0,\ldots, M_t-1\}$ and $i\in\{1,\ldots,d\}$. We note $\tilde{A}_t$ has probability approaching one:
			\begin{align*}
				\mathbb{P}_{\mathbf{A}^*}\left(\tilde{A}_t^{c}\right) &\leq \sum_{i=1}^{d} \sum_{m=0}^{M_t-1} \bigg[ 2\mathbb{P}_{\mathbf{A}^*}\left( \left|\Delta_{\mathcal{Q}_t}^m D^{p-1}Y^{c, (i)}_{\mathbf{A}^{(0)}} \right| > 2 \nu^{(i), m}_t \right) + \tilde{F}^{(i)}\left(\left(-4\nu_t^{(i)}, 4\nu_t^{(i)}\right)\right) \lambda^{(i)} \Delta^m_{\mathcal{Q}_t} + \left(\lambda^{(i)} \Delta^m_{\mathcal{Q}_t}\right)^2  \bigg] \\
				%-------------------------------
				&\lesssim \sum_{i=1}^{d}\left[t \Delta_{\mathcal{Q}_t}^{1-2\beta^{(i)}} + t \tilde{F}^{(i)}\left(\left(-4(\Delta^m_{\mathcal{Q}_t})^{\beta^{(i)}},\ 4(\Delta^m_{\mathcal{Q}_t})^{\beta^{(i)}}\right)\right) \right]\rightarrow 0, \quad t\rightarrow \infty,
			\end{align*}
			since
			\begin{align}
				\mathbb{P}_{\mathbf{A}^*}&\left( \left|\Delta_{\mathcal{Q}_t}^m D^{p-1}Y^{c, (i)}_{\mathbf{A}^{(0)}} \right| > 2\nu^{(i), m}_t \right) \notag \\
				%------------------------
				&\leq \mathbb{P}_{\mathbf{A}^*}\Bigg( \Big|\Delta_{\mathcal{Q}_t}^m(\Sigma^{1/2} \mathbf{W})^{(i)}\Big| > \nu^{(i), m}_t \Bigg) + \mathbb{P}_{\mathbf{A}^*}\Bigg( \Big| \int_{u_m}^{u_{m+1}} \sum_{j=1}^p (A^*_j D^{p-j}\mathbf{Y}_u)^{(i)} \,du \Big| > \nu^{(i), m}_t \Bigg) \notag \\
				%----------------------------
				&\lesssim  \left(\Delta^m_{\mathcal{Q}_t}\right)^{2-2\beta^{(i)}}, \text{ for $t\in\mathcal{T}$ sufficiently large.} \label{eqn:P(mart>2nu)}
			\end{align}
			Moreover, on $\tilde{A}_t$, the term $\mathbf{S}^1_t(i, l)$ vanishes in $L^1(\Omega, \mathcal{F}, \mathbb{P}_{\mathbf{A}^*})$
			\begin{align*}
				t^{-1/2}&\mathbb{E}_{\mathbf{A}^*}\left[\left\| \left(\sum_{m = 0}^{M_t-1}
				D^{l}Y^{(i)}_{u_m} \Delta_{\mathcal{Q}_t}^m D^{p-1}\tilde{\mathbf{Y}} \odot \mathds{1}_{\left\{\left|\Delta_{\mathcal{Q}_t}^m D^{p-1}\tilde{\mathbf{Y}}\right| \leq 2\boldsymbol{\nu}^m_t\right\}} - \sum_{m = 0}^{M_t-1}
				D^{l}Y^{(i)}_{u_m} \Delta^m_{\mathcal{Q}_t} D^{p-1}\mathbf{Y}^c_{\mathbf{A}^{(0)}} \right)\mathds{1}_{\tilde{A}_t}\right\|\right] \\
				%---------------------------------
				&\lesssim t^{1/2} \Delta^{1/2}_{\mathcal{Q}_t} \rightarrow 0,\quad t\rightarrow\infty,
			\end{align*}
			thus $\mathbf{S}_t^1(i, l) \overset{\mathbb{P}_{\mathbf{A}^*}}{\rightarrow} 0$ as $t\rightarrow\infty$.
			
			Next, to show $\mathbf{S}_t^2(i, l) \overset{\mathbb{P}_{\mathbf{A}^*}}{\rightarrow} 0$ as $t\rightarrow\infty$ let us write
			\begin{align*}
				\mathbf{S}_t^2(i, l) &= t^{-1/2}\left(\sum_{m = 0}^{M_t-1}
				D^{l}Y^{(i)}_{u_m} \Delta_{\mathcal{Q}_t}^m D^{p-1}\tilde{\mathbf{Y}} \odot \left(\mathds{1}_{\left\{\left|\Delta_{\mathcal{Q}_t}^m \hat{D}^{p-1}{\mathbf{Y}} - \mathbf{b} \Delta_{\mathcal{Q}_t}^m\right| \leq \boldsymbol{\nu}^m_t\right\}} - \mathds{1}_{\left\{\left|\Delta_{\mathcal{Q}_t}^m D^{p-1}\tilde{\mathbf{Y}}\right| \leq 2\boldsymbol{\nu}^m_t\right\}} \right)\right) \\
				&= t^{-1/2}\left(\sum_{m = 0}^{M_t-1}
				D^{l}Y^{(i)}_{u_m} \Delta_{\mathcal{Q}_t}^m D^{p-1}\tilde{\mathbf{Y}}\odot \mathds{1}_{\left\{\left|\Delta_{\mathcal{Q}_t}^m \hat{D}^{p-1}\mathbf{Y} - \mathbf{b} \Delta_{\mathcal{Q}_t}^m\right| \leq \boldsymbol{\nu}^m_t, \left|\Delta_{\mathcal{Q}_t}^m D^{p-1}\tilde{\mathbf{Y}}\right| > 2\boldsymbol{\nu}^m_t \right\}}\right) \\
				&\quad\quad\quad- t^{-1/2}\left(\sum_{m = 0}^{M_t-1}
				D^{l}Y^{(i)}_{u_m} \Delta_{\mathcal{Q}_t}^m D^{p-1}\tilde{\mathbf{Y}} \odot \mathds{1}_{\left\{\left|\Delta_{\mathcal{Q}_t}^m \hat{D}^{p-1}\mathbf{Y} - \mathbf{b} \Delta_{\mathcal{Q}_t}^m\right| > \boldsymbol{\nu}^m_t, \left|\Delta_{\mathcal{Q}_t}^m D^{p-1}\tilde{\mathbf{Y}}\right| \leq 2\boldsymbol{\nu}^m_t \right\}}\right) \\
				& =: \mathbf{S}_t^{2, 1}(i, l) - \mathbf{S}_t^{2, 2}(i, l).
			\end{align*}
			We have $\mathbf{S}_t^{2, 1}(i, l) \overset{\mathbb{P}_{\mathbf{A}^*}}{\rightarrow} 0$ as $t\rightarrow\infty$ since
			\begin{align*}
				\mathbb{P}_{\mathbf{A}^*}\left(\|\mathbf{S}_t^{2, 1}(i, l)\|>0\right) &\overset{(i)}{\leq} \mathbb{P}_{\mathbf{A}^*}\left(\bigcup_{m=0}^{M_t-1}\bigcup_{j=1}^{d} {\left\{\left|\Delta_{\mathcal{Q}_t}^m \hat{D}^{p-1}{Y}^{(j)} - b^{(j)} \Delta_{\mathcal{Q}_t}^m\right| \leq \nu^{(j), m}_t, \left|\Delta_{\mathcal{Q}_t}^m D^{p-1}\tilde{Y}^{(j)}\right| > 2\nu^{(j), m}_t \right\}} \right) \\
				&\overset{(ii)}{\leq} \sum_{m=0}^{M_t-1} \sum_{j=1}^d \mathbb{P}_{\mathbf{A}^*}\left( \left|\Delta_{\mathcal{Q}_t}^m \hat{D}^{p-1}{Y}^{(j)} - b^{(j)} \Delta_{\mathcal{Q}_t}^m\right| \leq \nu^{(j), m}_t, \left|\Delta_{\mathcal{Q}_t}^m D^{p-1}\tilde{Y}^{(j)}\right| > 2\nu^{(j), m}_t \right) \\
				%-----------------------------
				&\overset{(iii)}{\leq} \sum_{m=0}^{M_t-1} \sum_{j=1}^d\Bigg\{
				\mathbb{P}_{\mathbf{A}^*}\left( \left|\Delta_{\mathcal{Q}_t}^m \hat{D}^{p-1}{Y}^{(j)} - \Delta_{\mathcal{Q}_t}^m D^{p-1}{Y}^{(j)} \right| > \frac{\nu^{(j), m}_t}{2} \right)\\
				&\quad\quad\quad\quad + \mathbb{P}_{\mathbf{A}^*}\left( \left|\Delta_{\mathcal{Q}_t}^m D^{p-1}{Y}^{(j)} - b^{(j)} \Delta_{\mathcal{Q}_t}^m\right| \leq \frac{3\nu^{(j), m}_t}{2}, \left|\Delta_{\mathcal{Q}_t}^m D^{p-1}\tilde{Y}^{(j)}\right| > 2\nu^{(j), m}_t \right)\Bigg\} \\
				%-----------------------------
				&\overset{(iv)}{\lesssim} \sum_{m=0}^{M_t-1} \sum_{j=1}^d\Bigg\{
				\Delta_{\mathcal{P}_t} \left(\nu^{(j), m}_t\right)^{-2} + \mathbb{P}_{\mathbf{A}^*}\left( \left|\Delta_{\mathcal{Q}_t}^m M^{(j)}\right| > \frac{\nu^{(j), m}_t}{2}, \Delta_{\mathcal{Q}_t}^m N^{(j)} > 0 \right) \\
				&\quad\quad\quad\quad\quad\quad\quad\quad\quad\quad\quad\quad\quad +  \mathbb{P}_{\mathbf{A}^*}\left( \left|\Delta_{\mathcal{Q}_t}^m D^{p-1}\tilde{Y}^{(j)}\right| > 2\nu^{(j), m}_t, \Delta_{\mathcal{Q}_t}^m J^{(j)}  = 0\right) \Bigg\} \\
				%-----------------------------
				&\overset{(v)}{\lesssim} \sum_{m=0}^{M_t-1} \sum_{j=1}^d\Bigg\{
				\Delta_{\mathcal{P}_t} \left(\nu^{(j), m}_t\right)^{-2} + \mathbb{P}_{\mathbf{A}^*}\left( \left|\Delta_{\mathcal{Q}_t}^m M^{(j)}\right| > \frac{\nu^{(j), m}_t}{2}\right) \mathbb{P}_{\mathbf{A}^*}\left( \Delta_{\mathcal{Q}_t}^m N^{(j)} > 0 \right) \\
				&\quad\quad\quad\quad\quad\quad\quad\quad\quad\quad\quad\quad\quad\quad\quad\quad\quad\quad\quad +  \mathbb{P}_{\mathbf{A}^*}\Bigg( \left|\Delta_{\mathcal{Q}_t}^m D^{p-1}Y^{c, (j)}_{\mathbf{A}^{(0)}} \right| > 2\nu^{(j), m}_t \Bigg) \Bigg\} \\
				%-----------------------------
				&\overset{(vi)}{\lesssim} \sum_{m=0}^{M_t-1} \sum_{j=1}^d\Bigg\{
				\Delta_{\mathcal{P}_t} \left(\nu^{(j), m}_t\right)^{-2} + \Delta^m_{\mathcal{Q}_t} \left(\nu^{(j), m}_t\right)^{-2} \Delta^m_{\mathcal{Q}_t} + \left(\Delta^m_{\mathcal{Q}_t}\right)^{2}\left(\nu^{(j), m}_t\right)^{-2} \Bigg\} \\
				%-----------------------------
				&\overset{(vii)}{\lesssim} \sum_{j=1}^d\Bigg\{ 
				c^{-2\beta^{(j)}}_{\mathcal{Q}_t} M_t \Delta_{\mathcal{P}_t} \Delta^{-2\beta^{(j)}}_{\mathcal{Q}_t} + t \Delta^{1-2\beta^{(j)}}_{\mathcal{Q}_t} \Bigg\} \\
				%-----------------------------
				&\overset{(viii)}{\lesssim} \sum_{j=1}^d\Bigg\{
				c^{-2\beta^{(j)}}_{\mathcal{Q}_t}(t^{-1}N'_t \Delta_{\mathcal{Q}_t}) (t\Delta_{\mathcal{Q}_t}^{-2}\Delta_{\mathcal{P}_t}) \Delta^{1-2\beta^{(j)}}_{\mathcal{Q}_t} + t \Delta^{1-2\beta^{(j)}}_{\mathcal{Q}_t} \Bigg\} \rightarrow 0, \quad t\rightarrow\infty,
			\end{align*}
			by using in $(i)$ the fact that if the norm is non-zero then at least one component of one of the summands is non-zero and thus the corresponding indicator is non-zero, in $(ii)$ sub-additivity of probability, in $(iii)$ theorem of total probability and $|\hat{y}-y|\leq \nu /2, |\hat{y}|\leq \nu \implies |y|\leq 3/2\nu$ for $y,\hat{y},\nu \in\mathbb{R}$, in $(iv)$ Markov's inequality and Equation \eqref{eqn:incr_difference_bound} on the first term, theorem of total probability and $|y|\leq 3/2\nu, |y-m|> 2\nu \implies |m|>\nu/2$ for $y,m,\nu\in\mathbb{R}$ on the second term, in $(v)$ independence of $\mathbb{J}$ and $\mathbb{M}$ on the second term, in $(vi)$ Markov's inequality,
			\begin{align}
				\mathbb{E}_{\mathbf{A}^*}\left[\|\Delta_{\mathcal{Q}_t}^m \mathbf{M}\|^2\right], \mathbb{E}_{\mathbf{A}^*}\left[\|\Delta_{\mathcal{Q}_t}^m \mathbf{J}\|^2\right] \lesssim \Delta^m_{Q_t} \label{eqn:bound_M2,J2}
			\end{align}
			since $\mathbb{M}$ and $\mathbb{J}$ are \Levy processes with bounded second moments and $\Delta_{\mathcal{Q}_t}^m N^{(j)} \overset{\mathbb{P}_{\mathbf{A}^*}}{\sim}\mathrm{Possion}(\lambda^{(j)} \Delta_{\mathcal{Q}_t}^m)$ on the second term, Equation \eqref{eqn:P(mart>2nu)} on the third term, in $(vii)$ Assumption \ref{ass:infinite_thresholding}.$(i)$ and $\Delta^m_{\mathcal{Q}_t} \geq c_{\mathcal{Q}_t} \Delta_{\mathcal{Q}_t}$, in $(viii)$ Assumption \ref{ass:controlled_sampling}, Assumption \ref{ass:joint_mesh} and Assumption \ref{ass:infinite_thresholding}$.(ii)$.
			
			To show $\mathbf{S}_t^{2, 2}(i, l) \overset{\mathbb{P}_{\mathbf{A}^*}}{\rightarrow} 0$ as $t\rightarrow\infty$ we write using Equation \eqref{eqn:DY=W-A}
			\[D^{p-1}\tilde{\mathbf{Y}}_s - D^{p-1}\tilde{\mathbf{Y}}_0 = D^{p-1}\mathbf{Y}^c_{\mathbf{A}^{(0)}, s} + \mathbf{J}_s = \Sigma^{1/2} \mathbf{W}_s - \sum_{j=1}^p\int_0^s A^*_j D^{p-j}\mathbf{Y}_u\ du + \mathbf{J}_s, \]
			and we can hence further decompose $\mathbf{S}_t^{2, 2}(i, l)$ as
			\begin{align*}
				\mathbf{S}_t^{2, 2}&(i, l) = t^{-1/2}\left(\sum_{m = 0}^{M_t-1}
				D^{l}Y^{(i)}_{u_m} \Delta_{\mathcal{Q}_t}^m D^{p-1}\tilde{\mathbf{Y}} \odot \mathds{1}_{\left\{\left|\Delta_{\mathcal{Q}_t}^m \hat{D}^{p-1}\mathbf{Y} - \mathbf{b} \Delta_{\mathcal{Q}_t}^m\right| > \boldsymbol{\nu}^m_t, \left|\Delta_{\mathcal{Q}_t}^m D^{p-1}\tilde{\mathbf{Y}}\right| \leq 2\boldsymbol{\nu}^m_t \right\}}\right) \\
				& = t^{-1/2}\left(\sum_{m = 0}^{M_t-1}
				D^{l}Y^{(i)}_{u_m} \Delta_{\mathcal{Q}_t}^m D^{p-1}\mathbf{Y}^c_{\mathbf{A}^{(0)}} \odot \mathds{1}_{\left\{\left|\Delta_{\mathcal{Q}_t}^m \hat{D}^{p-1}\mathbf{Y} - \mathbf{b} \Delta_{\mathcal{Q}_t}^m\right| > \boldsymbol{\nu}^m_t, \left|\Delta_{\mathcal{Q}_t}^m D^{p-1}\tilde{\mathbf{Y}}\right| \leq 2\boldsymbol{\nu}^m_t \right\}}\right) \\
				% 	&\quad - t^{-1/2}\left(\sum_{m = 0}^{M_t-1}
				% 	D^{l}Y^{(i)}_{u_m} \left(\sum_{j=1}^p \int_{u_m}^{u_{m+1}} A^*_j D^{p-j}\mathbf{Y}_u\ du \right) \odot \mathds{1}_{\left\{\left|\Delta_{\mathcal{Q}_t}^m \hat{D}^{p-1}\mathbf{Y} - \mathbf{b} \Delta_{\mathcal{Q}_t}^m\right| > \boldsymbol{\nu}^m_t, \left|\Delta_{\mathcal{Q}_t}^m D^{p-1}\tilde{\mathbf{Y}}\right| \leq 2\boldsymbol{\nu}^m_t \right\}}\right) \\
				&\quad\quad\quad\quad\quad + t^{-1/2}\left(\sum_{m = 0}^{M_t-1}
				D^{l}Y^{(i)}_{u_m} \Delta_{\mathcal{Q}_t}^m \mathbf{J} \odot \mathds{1}_{\left\{\left|\Delta_{\mathcal{Q}_t}^m \hat{D}^{p-1}\mathbf{Y} - \mathbf{b} \Delta_{\mathcal{Q}_t}^m\right| > \boldsymbol{\nu}^m_t, \left|\Delta_{\mathcal{Q}_t}^m D^{p-1}\tilde{\mathbf{Y}}\right| \leq 2\boldsymbol{\nu}^m_t \right\}}\right) \\
				& =: \mathbf{S}_t^{2, 2, 1}(i, l) + \mathbf{S}_t^{2, 2, 2}(i, l).
			\end{align*}
			We start by showing $\mathbf{S}_t^{2, 2, 1}(i, l)\overset{\mathbb{P}_{\mathbf{A}^*}}{\rightarrow} 0$ as $t\rightarrow\infty$. We work in $L^1(\Omega, \mathcal{F}, \mathbb{P}_{\mathbf{A}^*})$:
			\begin{align*}
				&\mathbb{E}_{\mathbf{A}^*}\left[\|\mathbf{S}_t^{2, 2, 1}(i, l)\|\right]
				\\
				%-------------------------------
				&\overset{(i)}{\leq}  t^{-1/2} \sum_{m = 0}^{M_t-1} \sum_{j=1}^d \mathbb{E}_{\mathbf{A}^*}\left[\left|
				D^{l}Y^{(i)}_{u_m} \Delta_{\mathcal{Q}_t}^m D^{p-1}Y^{c,(j)}_{\mathbf{A}^{(0)}} \right| \mathds{1}_{\left\{\left|\Delta_{\mathcal{Q}_t}^m \hat{D}^{p-1}Y^{(j)} - b^{(j)} \Delta_{\mathcal{Q}_t}^m\right| > \nu^{(j), m}_t \right\}}\right]\\
				%-------------------------------
				&\overset{(ii)}{\leq}  t^{-1/2} \sum_{m = 0}^{M_t-1} \sum_{j=1}^d \mathbb{E}_{\mathbf{A}^*}\bigg[\left|
				D^{l}Y^{(i)}_{u_m} \Delta_{\mathcal{Q}_t}^m D^{p-1}Y^{c,(j)}_{\mathbf{A}^{(0)}} \right| \bigg(\mathds{1}_{\left\{\left|\Delta_{\mathcal{Q}_t}^m \hat{D}^{p-1}Y^{(j)} - \Delta_{\mathcal{Q}_t}^m D^{p-1}Y^{(j)} \right| > \nu^{(j), m}_t/2\right\}} \\
				&\quad\quad\quad\quad\quad\quad\quad\quad\quad\quad\quad\quad\quad\quad\quad\quad\quad\quad\quad\quad\quad\quad\quad\quad\quad + \mathds{1}_{\left\{\left|\Delta_{\mathcal{Q}_t}^m D^{p-1}Y^{(j)} - b^{(j)} \Delta_{\mathcal{Q}_t}^m\right| > \nu^{(j), m}_t/2 \right\}}\bigg)\bigg]\\
				%-------------------------------
				&\overset{(iii)}{\leq}  t^{-1/2} \sum_{m = 0}^{M_t-1} \sum_{j=1}^d \mathbb{E}_{\mathbf{A}^*}\bigg[\left|
				D^{l}Y^{(i)}_{u_m} \Delta_{\mathcal{Q}_t}^m D^{p-1}Y^{c,(j)}_{\mathbf{A}^{(0)}} \right| \bigg(\mathds{1}_{\left\{\left|\Delta_{\mathcal{Q}_t}^m \hat{D}^{p-1}Y^{(j)} - \Delta_{\mathcal{Q}_t}^m D^{p-1}Y^{(j)} \right| > \nu^{(j), m}_t/2\right\}} \\
				&\quad\quad\quad\quad\quad\quad\quad\quad\quad\quad\quad\quad\quad\quad\quad\quad\quad\quad + \mathds{1}_{\left\{\left|\Delta_{\mathcal{Q}_t}^m D^{p-1}Y^{c, (j)}_{\mathbf{A}^{(0)}} + \Delta_{\mathcal{Q}_t}^m M^{(j)} \right| > \nu^{(j), m}_t/2\right\}} + \mathds{1}_{\left\{\Delta_{\mathcal{Q}_t}^m N^{(j)} > 0 \right\}} \bigg)\bigg]\\
				%-------------------------------
				&\overset{(iv)}{\leq}  t^{-1/2} \sum_{m = 0}^{M_t-1} \sum_{j=1}^d \mathbb{E}_{\mathbf{A}^*}\bigg[\left|
				D^{l}Y^{(i)}_{u_m} \Delta_{\mathcal{Q}_t}^m D^{p-1}Y^{c,(j)}_{\mathbf{A}^{(0)}} \right| \bigg(\mathds{1}_{\left\{\left|\Delta_{\mathcal{Q}_t}^m \hat{D}^{p-1}Y^{(j)} - \Delta_{\mathcal{Q}_t}^m D^{p-1}Y^{(j)} \right| > \nu^{(j), m}_t/2\right\}} \\
				&\quad\quad\quad\quad\quad\quad\quad\quad\quad\quad\quad\quad + \mathds{1}_{\left\{\left|\Delta_{\mathcal{Q}_t}^m D^{p-1}Y^{c, (j)}_{\mathbf{A}^{(0)}} \right| > \nu^{(j), m}_t/4\right\}} + \mathds{1}_{\left\{\left| \Delta_{\mathcal{Q}_t}^m M^{(j)} \right| > \nu^{(j), m}_t/4\right\}} + \mathds{1}_{\left\{\Delta_{\mathcal{Q}_t}^m N^{(j)} > 0 \right\}} \bigg)\bigg]\\
				%----------------------
				&\overset{(v)}{\lesssim} t^{-1/2} \sum_{j=1}^d \sum_{m=0}^{M_t-1} \Bigg[\mathbb{E}_{\mathbf{A}^*} \Big[ \Big|D^{l}Y^{(i)}_{u_m}\Delta_{\mathcal{Q}_t}^m (\Sigma^{1/2}\mathbf{W})^{(j)}\Big|^2\Big]^{1/2} + \sum_{k=1}^{p} \mathbb{E}_{\mathbf{A}^*} \left[ \left| D^{l}Y^{(i)}_{u_m} \int_{u_m}^{u_{m+1}} (A_k^*D^{p-k}\mathbf{Y}_u)^{(j)}\, du \right|^2\right]^{1/2}\Bigg] \\
				&\quad \times \Bigg[ \mathbb{P}_{\mathbf{A}^*} \left(\left|\Delta_{\mathcal{Q}_t}^m \hat{D}^{p-1}Y^{(j)} - \Delta_{\mathcal{Q}_t}^m D^{p-1}Y^{(j)} \right| > \nu^{(j), m}_t/2\right)^{1/2} + \mathbb{P}_{\mathbf{A}^*} \left(\left|\Delta_{\mathcal{Q}_t}^m D^{p-1}Y^{c, (j)}_{\mathbf{A}^{(0)}} \right| > \nu^{(j), m}_t/4\right)^{1/2} \\
				&\quad\quad\quad\quad\quad\quad\quad\quad\quad\quad\quad\quad\quad\quad\quad\quad\quad\quad\quad + \mathbb{P}_{\mathbf{A}^*} \left(\left|\Delta_{\mathcal{Q}_t}^m M^{(j)} \right| > \nu^{(j), m}_t/4 \right) + \mathbb{P}_{\mathbf{A}^*} \left(\Delta_{\mathcal{Q}_t}^m N^{(j)} > 0 \right)\Bigg]\\
				%----------------------
				&\overset{(vi)}{\lesssim} t^{-1/2} \sum_{j=1}^d \sum_{m=0}^{M_t-1} \Bigg[\mathbb{E}_{\mathbf{A}^*} \left[ \left|D^{l}Y^{(i)}_{u_m}\right|^2\right]^{1/2} \mathbb{E}_{\mathbf{A}^*} \left[\left| \Delta_{\mathcal{Q}_t}^m(\Sigma^{1/2}\mathbf{W})^{(j)}\right|^2\right]^{1/2} \\
				&\quad\quad\quad\quad\quad\quad\quad\quad\quad\quad\quad + \sum_{k=1}^{p} \left( \Delta_{\mathcal{Q}_t}^m \int_{u_m}^{u_{m+1}} \mathbb{E}_{\mathbf{A}^*} \left[ \left| D^{l}Y^{(i)}_{u_m} (A_k^*D^{p-k}\mathbf{Y}_u)^{(j)} \right|^2\right]\, du \right)^{1/2}\Bigg] \\
				&\quad\quad\quad\quad\quad\quad\quad\quad\quad\quad\quad\quad\quad\quad\quad\quad\quad \times \Big[\Delta_{\mathcal{P}_t}^{1/2}\left(\nu^{(j), m}_t\right)^{-1} + \left(\Delta^m_{\mathcal{Q}_t}\right)^{1 - \beta^{(j)}} + \Delta^m_{\mathcal{Q}_t}\left(\nu^{(j), m}_t\right)^{-2} + \Delta^m_{\mathcal{Q}_t}\Big]\\
				%----------------------	
				&\overset{(vii)}{\lesssim} \sum_{j=1}^d t^{-1/2} M_t \left(\Delta^m_{\mathcal{Q}_t}\right)^{1/2}\Big[\Delta_{\mathcal{P}_t}^{1/2}\left(\Delta^m_{\mathcal{Q}_t}\right)^{-\beta^{(j)}} + \left(\Delta^m_{\mathcal{Q}_t}\right)^{1 - \beta^{(j)}} +\left(\Delta^m_{\mathcal{Q}_t}\right)^{1 -2\beta^{(j)}} + \Delta^m_{\mathcal{Q}_t}\Big] \\
				%----------------------
				&\overset{(viii)}{\lesssim} \sum_{j=1}^d t^{-1/2} M_t \left[ \Delta_{\mathcal{P}_t}^{1/2} \Delta_{\mathcal{Q}_t}^{1/2-\beta^{(j)}} +  \Delta_{\mathcal{Q}_t}^{3/2-2\beta^{(j)}}\right] \\
				%----------------------
				&\overset{(ix)}{\lesssim} \sum_{j=1}^d  \left[ (t^{-1} N'_t \Delta_{\mathcal{Q}_t}) (t\Delta^{-2}_{\mathcal{Q}_t}\Delta_{\mathcal{P}_t})^{1/2} \Delta_{\mathcal{Q}_t}^{1/2-\beta^{(j)}} + (t^{-1} N'_t \Delta_{\mathcal{Q}_t}) \left(t\Delta_{\mathcal{Q}_t}^{1-4\beta^{(j)}}\right)^{1/2} \right] \rightarrow 0, \quad t\rightarrow\infty,
			\end{align*}
			where in $(i)$ we use the triangle inequality, in $(ii)$ $|\hat{y}|>\nu\implies |\hat{y}- y|>\nu/2 \vee |y|>\nu/2$ for $y, \hat{y},\nu\in\mathbb{R}$ and $\mathds{1}_{A}\leq \mathds{1}_{B\cup C}\leq \mathds{1}_{B} + \mathds{1}_{C}$ for sets $A,B,C\subset \Omega$ such that $A\subset B\cup C$, in $(iii)$ $1 = \mathds{1}_{A} + \mathds{1}_{A^c}$ and $\mathds{1}_{A\cap B}\leq \mathds{1}_{A}$ for sets $A,B\subset \Omega$, in $(iv)$ $|y + m|>\nu/2\implies |y|>\nu/4 \vee |m|>\nu/4$ for $y, m,\nu\in\mathbb{R}$ $\mathds{1}_{A}\leq \mathds{1}_{B\cup C}\leq \mathds{1}_{B} + \mathds{1}_{C}$ for sets $A,B,C\subset \Omega$ such that $A\subset B\cup C$, in $(v)$ Cauchy-Schwartz inequality and triangle inequality in the first two terms, independence of $D^l\mathbf{Y}_{u_m}, \Delta_{\mathcal{Q}_t}^m D^{p-1}\mathbf{Y}^c_{\mathbf{A}^{(0)}}$ and $\Delta_{\mathcal{Q}_t}^m \mathbf{M}, \Delta_{\mathcal{Q}_t}^m \mathbf{J}$ in the last two terms, in $(vi)$ independence of $D^l\mathbf{Y}_{u_m}$ and $\Delta_{\mathcal{Q}_t}^m\mathbf{W}$ on the first term, Jensen's inequality on the second term, Markov's inequality and Equation \eqref{eqn:incr_difference_bound} on the third term, a slight modification of Equation \eqref{eqn:P(mart>2nu)} on the fourth term, Markov inequality and Equation \eqref{eqn:bound_M2,J2} on the fifth term, $\Delta_{\mathcal{Q}_t}^m N^{(j)} \overset{\mathbb{P}_{\mathbf{A}^*}}{\sim}\mathrm{Possion}(\lambda^{(j)} \Delta_{\mathcal{Q}_t}^m)$ on the sixth term, in $(vii)$ boundedness of fourth moments of $\mathbf{X}_t = (\mathbf{Y}^{\mathrm{T}}_t,\ldots, D^{p-1}\mathbf{Y}^{\mathrm{T}}_t)^{\mathrm{T}}$, standard deviation of Brownian increments and Assumption \ref{ass:infinite_thresholding}.$(i)$, in $(viii)$ $\Delta^m_{\mathcal{Q}_t} \leq \Delta_{\mathcal{Q}_t}$, in $(ix)$ Assumption \ref{ass:controlled_sampling}, Assumption \ref{ass:joint_mesh} and Assumption \ref{ass:infinite_thresholding}.$(ii)$. 
			
			To show $\mathbf{S}_t^{2, 2, 2}(i, l)\overset{\mathbb{P}_{\mathbf{A}^*}}{\rightarrow} 0$ as $t\rightarrow\infty$ we note
			\begin{align*}				\mathbb{P}_{\mathbf{A}^*}&\left(\|\mathbf{S}_t^{2,2,2}(i, l)\|>0\right) \\
				&\overset{(i)}{\leq} \mathbb{P}_{\mathbf{A}^*}\left(\bigcup_{m=0}^{M_t-1}\bigcup_{j=1}^{d} {\left\{\Delta_{\mathcal{Q}_t}^m J^{(j)} \neq 0, \left|\Delta_{\mathcal{Q}_t}^m D^{p-1}\tilde{Y}^{(j)}\right| \leq 2\nu^{(j), m}_t \right\}} \right) \\
				&\overset{(ii)}{\leq} \sum_{m=0}^{M_t-1}\sum_{j=1}^d \mathbb{P}_{\mathbf{A}^*}\left(\Delta_{\mathcal{Q}_t}^m J^{(j)} \neq 0, \left|\Delta_{\mathcal{Q}_t}^m D^{p-1}\tilde{Y}^{(j)}\right| \leq 2\nu^{(j), m}_t\right) \\
				&\overset{(iii)}{\leq} \sum_{m=0}^{M_t-1}\sum_{j=1}^d \left[ \mathbb{P}_{\mathbf{A}^*}\left(\Delta_{\mathcal{Q}_t}^m N^{(j)} = 1, \left|\Delta_{\mathcal{Q}_t}^m D^{p-1}\tilde{Y}^{(j)}\right| \leq 2\nu^{(j), m}_t\right) + \mathbb{P}_{\mathbf{A}^*}\left(\Delta_{\mathcal{Q}_t}^m N^{(j)} > 1\right)\right] \\
				&\overset{(iv)}{\lesssim} \sum_{m=0}^{M_t-1}\sum_{j=1}^d \left[ \mathbb{P}_{\mathbf{A}^*}\left(\Delta_{\mathcal{Q}_t}^m N^{(j)} = 1, \left|\Delta_{\mathcal{Q}_t}^m D^{p-1}\tilde{Y}^{c, (j)}_{\mathbf{A}^{(0)}} + \Delta_{\mathcal{Q}_t}^m J^{(j)} \right| \leq 2\nu^{(j), m}_t\right) + (\Delta_{\mathcal{Q}_t}^m)^2	\right] \\
				&\overset{(v)}{\lesssim} \sum_{m=0}^{M_t-1}\sum_{j=1}^d \bigg[ \mathbb{P}_{\mathbf{A}^*}\left(\left|\Delta_{\mathcal{Q}_t}^m D^{p-1}\tilde{Y}^{c, (j)}_{\mathbf{A}^{(0)}} + \Delta_{\mathcal{Q}_t}^m J^{(j)} \right| \leq 2\nu^{(j), m}_t, \left|\Delta_{\mathcal{Q}_t}^m J^{(j)} \right| \geq 4\nu^{(j), m}_t \right) \\
				&\quad\quad\quad\quad\quad\quad\quad\quad\quad\quad\quad\quad\quad\quad + \mathbb{P}_{\mathbf{A}^*}\left(\Delta_{\mathcal{Q}_t}^m N^{(j)} = 1, \left|\Delta_{\mathcal{Q}_t}^m J^{(j)} \right| < 4\nu^{(j), m}_t \right) + (\Delta_{\mathcal{Q}_t}^m)^2	\bigg] \\
				&\overset{(vi)}{\lesssim} \sum_{m=0}^{M_t-1}\sum_{j=1}^d \bigg[ \mathbb{P}_{\mathbf{A}^*}\left(\left|\Delta_{\mathcal{Q}_t}^m D^{p-1}\tilde{Y}^{c, (j)}_{\mathbf{A}^{(0)}}  \right| > 2\nu^{(j), m}_t \right) + \Delta_{\mathcal{Q}_t}^m \tilde{F}^{(j)}\left(\left(-4\Delta_{\mathcal{Q}_t}^{\beta^{(j)}}, 4\Delta_{\mathcal{Q}_t}^{\beta^{(j)}}\right)\right) + (\Delta_{\mathcal{Q}_t}^m)^2 \bigg] \\
				&\overset{(vii)}{\lesssim} \sum_{m=0}^{M_t-1}\sum_{j=1}^d \bigg[ (\Delta_{\mathcal{Q}_t}^m)^{2-2\beta^{(j)}} + \Delta_{\mathcal{Q}_t}^m \tilde{F}^{(j)}\left(\left(-4\Delta_{\mathcal{Q}_t}^{\beta^{(j)}}, 4\Delta_{\mathcal{Q}_t}^{\beta^{(j)}}\right)\right) + (\Delta_{\mathcal{Q}_t}^m)^2 \bigg] \\
				&\overset{(viii)}{\lesssim} \sum_{j=1}^d \bigg[ t\Delta_{\mathcal{Q}_t}^{1-2\beta^{(j)}} + t \tilde{F}^{(j)}\left(\left(-4\Delta_{\mathcal{Q}_t}^{\beta^{(j)}}, 4\Delta_{\mathcal{Q}_t}^{\beta^{(j)}}\right)\right) \bigg] \rightarrow 0, \quad t\rightarrow\infty,
			\end{align*}
			where we use in $(i)$ the fact that if the norm is non-zero then at least one component of one of the summands is non-zero and thus the corresponding indicator is non-zero, in $(ii)$ sub-additivity of probability, in $(iii)$ theorem of total probability, in $(iv)$ $\Delta_{\mathcal{Q}_t}^m N^{(j)} \overset{\mathbb{P}_{\mathbf{A}^*}}{\sim}\mathrm{Possion}(\lambda^{(j)} \Delta_{\mathcal{Q}_t}^m)$, in $(v)$ theorem of total probability, in $(vi)$ $|y+j|<2\nu, |j|\geq 4\nu\implies|y|>2\nu$ for $y, j,\nu\in\mathbb{R}$ on the first term, $\Delta_{\mathcal{Q}_t}^m N^{(j)} \overset{\mathbb{P}_{\mathbf{A}^*}}{\sim}\mathrm{Possion}(\lambda^{(j)} \Delta_{\mathcal{Q}_t}^m)$ and $\Delta_{\mathcal{Q}_t}^m \tilde{J}^{(i)} \,\big|\, \Delta^m_{\mathcal{Q}_t} N^{(i)} =1 \overset{\mathbb{P}_{\mathbf{A}^*}}{\sim} \tilde{F}^{(i)}(\cdot)$ on the second term, in $(vii)$ Equation \eqref{eqn:P(mart>2nu)}, in $(viii)$ Assumption \ref{ass:infinite_thresholding}$.(ii)$ and Assumption \ref{ass:infinite_thresholding}$.(iii)$.
			
			It remains to show $\mathbf{S}_t^3(i, l) \overset{\mathbb{P}_{\mathbf{A}^*}}{\rightarrow} 0$ as $t\rightarrow\infty$. We work in $L^1(\Omega, \mathcal{F}, \mathbb{P}_{\mathbf{A}^*})$:
			\begin{align*}
				\mathbb{E}_{\mathbf{A}^*}&\left[\left\| \mathbf{S}_t^3(i, l)\right\|\right] \\
				&= t^{-1/2} \mathbb{E}_{\mathbf{A}^*}\left[\left\|\sum_{m = 0}^{M_t-1}
				D^{l}Y^{(i)}_{u_m} \Delta_{\mathcal{Q}_t}^m \mathbf{M} \odot \mathds{1}_{\left\{\left|\Delta_{\mathcal{Q}_t}^m \hat{D}^{p-1}\mathbf{Y} - \mathbf{b} \Delta_{\mathcal{Q}_t}^m\right| \leq \boldsymbol{\nu}^m_t\right\}}  \right\|\right] \\
				%-----------------------------------
				&\overset{(i)}{\leq} t^{-1/2} \sum_{m = 0}^{M_t-1} \sum_{j=1}^d \mathbb{E}_{\mathbf{A}^*}\left[\left|
				D^{l}Y^{(i)}_{u_m} \Delta_{\mathcal{Q}_t}^m M^{(j)}  \right| \mathds{1}_{\left\{\left|\Delta_{\mathcal{Q}_t}^m \hat{D}^{p-1}Y^{(j)} - b^{(j)} \Delta_{\mathcal{Q}_t}^m\right| \leq \nu^{(j), m}_t\right\}} \right] \\
				%-----------------------------------
				&\overset{(ii)}{\leq} t^{-1/2} \sum_{m = 0}^{M_t-1} \sum_{j=1}^d \mathbb{E}_{\mathbf{A}^*}\bigg[\left|
				D^{l}Y^{(i)}_{u_m}\Delta_{\mathcal{Q}_t}^m M^{(j)} \right|\bigg( \mathds{1}_{\left\{\left|\Delta_{\mathcal{Q}_t}^m \hat{D}^{p-1}Y^{(j)} - \Delta_{\mathcal{Q}_t}^m D^{p-1}Y^{(j)} \right| > \nu^{(j), m}_t/2 \right\}} \\
				&\quad\quad\quad\quad\quad\quad\quad\quad\quad\quad\quad\quad\quad\quad\quad\quad\quad\quad\quad\quad\quad\quad\quad\quad + \mathds{1}_{\left\{\left|\Delta_{\mathcal{Q}_t}^m D^{p-1}Y^{(j)} - b^{(j)} \Delta_{\mathcal{Q}_t}^m\right| \leq 3\nu^{(j), m}_t/2\right\}}	\bigg)\bigg]\\
				%-----------------------------------
				&\overset{(iii)}{\leq} t^{-1/2} \sum_{m = 0}^{M_t-1} \sum_{j=1}^d \mathbb{E}_{\mathbf{A}^*}\bigg[\left|
				D^{l}Y^{(i)}_{u_m} \Delta_{\mathcal{Q}_t}^m M^{(j)}  \right|\bigg( \mathds{1}_{\left\{\left|\Delta_{\mathcal{Q}_t}^m \hat{D}^{p-1}Y^{(j)} - \Delta_{\mathcal{Q}_t}^m D^{p-1}Y^{(j)} \right| > \nu^{(j), m}_t/2 \right\}} \\
				&\quad\quad\quad\quad\quad\quad\quad\quad\quad\quad\quad\quad\quad\quad\quad + \mathds{1}_{\left\{\left|\Delta_{\mathcal{Q}_t}^m D^{p-1}Y^{c, (j)}_{\mathbf{A}^{(0)}} + \Delta^m_{\mathcal{Q}_t} M^{(j)} \right| \leq 3\nu^{(j), m}_t/2\right\}}	+ \mathds{1}_{\left\{\Delta_{\mathcal{Q}_t}^m N^{(j)} >  0 \right\}}	\bigg)\bigg] \\
				%-----------------------------------
				% 	&\overset{(iv)}{\leq} t^{-1/2} \sum_{m = 0}^{M_t-1} \sum_{j=1}^d \mathbb{E}_{\mathbf{A}^*}\bigg[\left|
				% 	D^{l}Y^{(i)}_{u_m} \Delta_{\mathcal{Q}_t}^m M^{(j)}  \right|\bigg( \mathds{1}_{\left\{\left|\Delta_{\mathcal{Q}_t}^m \hat{D}^{p-1}Y^{(j)} - \Delta_{\mathcal{Q}_t}^m D^{p-1}Y^{(j)} \right| > \nu^{(j), m}_t/2 \right\}} \\
				% 	&\quad\quad\quad\quad\quad\quad\quad\quad\quad\quad\quad\quad\quad\quad\quad + \mathds{1}_{\left\{\left|\Delta_{\mathcal{Q}_t}^m D^{p-1}Y^{c, (j)}_{\mathbf{A}^{(0)}} + \Delta^m_{\mathcal{Q}_t} M^{(j)} \right| \leq 3\nu^{(j), m}_t/2, \left|\Delta^m_{\mathcal{Q}_t} M^{(j)} \right|> 2\nu^{(j), m}_t \right\}}	\\
				% 	&\quad\quad\quad\quad\quad\quad\quad\quad\quad\quad\quad\quad\quad\quad\quad\quad\quad\quad\quad\quad\quad\quad\quad + \mathds{1}_{\left\{\left|\Delta^m_{\mathcal{Q}_t} M^{(j)} \right| \leq 2\nu^{(j), m}_t\right\}} + \mathds{1}_{\left\{\Delta_{\mathcal{Q}_t}^m N^{(j)} >  0 \right\}}	\bigg)\bigg] \\
				%-----------------------------------
				&\overset{(iv)}{\leq} t^{-1/2} \sum_{m = 0}^{M_t-1} \sum_{j=1}^d \mathbb{E}_{\mathbf{A}^*}\bigg[\left|
				D^{l}Y^{(i)}_{u_m} \Delta_{\mathcal{Q}_t}^m M^{(j)}  \right|\bigg( \mathds{1}_{\left\{\left|\Delta_{\mathcal{Q}_t}^m \hat{D}^{p-1}Y^{(j)} - \Delta_{\mathcal{Q}_t}^m D^{p-1}Y^{(j)} \right| > \nu^{(j), m}_t/2 \right\}} \\
				&\quad\quad\quad\quad\quad\quad\quad\quad\quad\quad\quad + \mathds{1}_{\left\{\left|\Delta_{\mathcal{Q}_t}^m D^{p-1}Y^{c, (j)}_{\mathbf{A}^{(0)}} \right|> \nu^{(j), m}_t/2 \right\}} + \mathds{1}_{\left\{\left|\Delta^m_{\mathcal{Q}_t} M^{(j)} \right| \leq 2\nu^{(j), m}_t\right\}} + \mathds{1}_{\left\{\Delta_{\mathcal{Q}_t}^m N^{(j)} >  0 \right\}}	\bigg)\bigg] \\
				%-----------------------------------
				&\overset{(v)}{\leq} t^{-1/2} \sum_{m = 0}^{M_t-1} \sum_{j=1}^d\Bigg\{ \mathbb{E}_{\mathbf{A}^*}\left[\left|
				D^{l}Y^{(i)}_{u_m}\right|^2\right]^{1/2} \mathbb{E}_{\mathbf{A}^*}\left[\left| \Delta_{\mathcal{Q}_t}^m M^{(j)}  \right|^2\right]^{1/2} \\
				&\quad\quad\quad\quad\quad\quad\quad\quad\quad\quad \times \bigg[\mathbb{P}_{\mathbf{A}^*} \left(\left|\Delta_{\mathcal{Q}_t}^m \hat{D}^{p-1}Y^{(j)} - \Delta_{\mathcal{Q}_t}^m D^{p-1}Y^{(j)} \right| > \nu^{(j), m}_t/2\right)^{1/2} \\
				&\quad\quad\quad\quad\quad\quad\quad\quad\quad\quad\quad\quad\quad + \mathbb{P}_{\mathbf{A}^*} \left(\left|\Delta_{\mathcal{Q}_t}^m D^{p-1}Y^{c, (j)}_{\mathbf{A}^{(0)}} \right| > \nu^{(j), m}_t/2\right)^{1/2} + \mathbb{P}_{\mathbf{A}^*} \left(\Delta_{\mathcal{Q}_t}^m N^{(j)} > 0 \right)\bigg]\\
				&\quad\quad\quad\quad\quad\quad\quad\quad\quad + \mathbb{E}_{\mathbf{A}^*}\left[\left|
				D^{l}Y^{(i)}_{u_m}\right|\right] \mathbb{E}_{\mathbf{A}^*}\left[\left| \Delta_{\mathcal{Q}_t}^m M^{(j)}  \right|  \mathds{1}_{\left\{\left|\Delta^m_{\mathcal{Q}_t} M^{(j)} \right| \leq 2\nu^{(j), m}_t\right\}} \right] \Bigg\} \\
				%------------------------------
				&\overset{(vi)}{\lesssim}  \sum_{j=1}^d t^{-1/2} M_t \left\{\left(\Delta^m_{\mathcal{Q}_t}\right)^{1/2} \left[\left(\nu^{(j), m}_t\right)^{-1} \Delta_{\mathcal{P}_t}^{1/2} + \left(\Delta^m_{\mathcal{Q}_t}\right)^{1-\beta^{(j)}} + \Delta^m_{\mathcal{Q}_t}\right] +  \Delta_{\mathcal{Q}_t}^{1+\epsilon^{(j)}} \right\}\\
				%------------------------------
				&\overset{(vii)}{\lesssim} \sum_{j=1}^d t^{-1/2} M_t \left\{  \left(\Delta^m_{\mathcal{Q}_t}\right)^{1/2-\beta^{(j)}} \Delta_{\mathcal{P}_t}^{1/2} + \Delta_{\mathcal{Q}_t}^{3/2-\beta^{(j)}} 
				+ \Delta_{\mathcal{Q}_t}^{1+\epsilon^{(j)}} \right\}\\
				%------------------------------
				&\overset{(viii)}{\lesssim} \sum_{j=1}^d (t^{-1} N'_t \Delta_{\mathcal{Q}_t}) \left[ \left(t\Delta_{\mathcal{Q}_t}^{-2}\Delta_{\mathcal{P}_t}\right)^{1/2}\Delta_{\mathcal{Q}_t}^{1/2-\beta^{(j)}} + \left(t \Delta_{\mathcal{Q}_t}^{1-2\beta^{(j)}}\right)^{1/2} + \left(t \Delta_{\mathcal{Q}_t}^{2\epsilon^{(j)}}\right)^{1/2} \right] \rightarrow 0, \quad t\rightarrow\infty,
			\end{align*}
			where we use in $(i)$ triangle inequality, in $(ii)$ $1 = \mathds{1}_{A} + \mathds{1}_{A^c}$, $|\hat{y}|\leq \nu \wedge |\hat{y}-y|\leq\nu/2 \implies |y|\leq 3\nu/2$, $\mathds{1}_{A}\leq \mathds{1}_{B}$ and $\mathds{1}_{A\cap C}\leq \mathds{1}_{C}$ for $y,\hat{y},\nu\in\mathbb{R}$ and sets $A, B, C\subset\Omega$ with $A\subset B$, in $(iii)$ $1 = \mathds{1}_{A} + \mathds{1}_{A^c}$ and $\mathds{1}_{A\cap B}\leq \mathds{1}_{A}$ for sets $A,B\subset\Omega$, in $(iv)$ $1 = \mathds{1}_{A} + \mathds{1}_{A^c}$, $|y+m|\leq 3\nu/2 \wedge |m|>2\nu \implies |y|>\nu/2$, $\mathds{1}_{A}\leq \mathds{1}_{B}$ and $\mathds{1}_{A\cap C}\leq \mathds{1}_{C}$ for $y,m,\nu\in\mathbb{R}$ and sets $A, B, C\subset\Omega$ with $A\subset B$, in $(v)$ Cauchy-Schwartz inequality and independence of $D^l\mathbf{Y}_{u_m}$ and $\Delta_{\mathcal{Q}_t}^m \mathbf{M}$ on the first two terms, independence of $D^l\mathbf{Y}_{u_m}$, $\Delta_{\mathcal{Q}_t}^m \mathbf{M}$ and $\Delta_{\mathcal{Q}_t}^m \mathbf{N}$ on the third term, independence of $D^l\mathbf{Y}_{u_m}$ and $\Delta_{\mathcal{Q}_t}^m \mathbf{M}$ on the fourth term, in $(vi)$ boundedness of second moments of $\mathbf{X}_t = (\mathbf{Y}^{\mathrm{T}}_t,\ldots, D^{p-1}\mathbf{Y}^{\mathrm{T}}_t)^{\mathrm{T}}$, Equation \eqref{eqn:bound_M2,J2}, Markov inequality, Equation \eqref{eqn:D-D_L2_bound}, Equation \eqref{eqn:P(mart>2nu)}, $\Delta_{\mathcal{Q}_t}^m N^{(j)} \overset{\mathbb{P}_{\mathbf{A}^*}}{\sim}\mathrm{Possion}(\lambda^{(j)} \Delta_{\mathcal{Q}_t}^m)$ and Assumption \ref{ass:infinite_thresholding}$.(iii)$, in $(vii)$ Assumption \ref{ass:infinite_thresholding}$.(i)$, in  $(viii)$ Assumption \ref{ass:controlled_sampling}, Assumption \ref{ass:joint_mesh} and Assumption \ref{ass:infinite_thresholding}$.(ii)$ and Assumption \ref{ass:infinite_thresholding}$.(iii)$.
		\end{proof}
		
		\subsubsection{Infinite jump activity: the symmetric driver case} \label{app:infinite_activity_symmetric}
		
		In this section we give a slightly different version of Theorem \ref{thm:cons_asymp_third_approx_infinite_activity} with Assumption \ref{ass:infinite_thresholding} replaced by the following join condition on $\{\mathcal{Q}_t,\ t\in\mathcal{T}\}$ and $\{(\boldsymbol{\nu}^m_t)_{m=0}^{M_t-1},\ t\in\mathcal{T}\}$.
		
		\begin{assumption} \label{ass:infinite_thresholding_symmetric}
			The sequence of partitions $\{\mathcal{Q}_t,\ t\in\mathcal{T}\}$ and the sequence of jump thresholds $\{(\boldsymbol{\nu}^m_t)_{m=0}^{M_t-1},\ t\in\mathcal{T}\}$ satisfy for $i\in\{1,\ldots, d\}$ 
			\begin{enumerate}[label=(\roman*)]
				\item $\nu^{(i), m}_t = (\Delta^m_{\mathcal{Q}_t})^{\beta^{(i)}}$ with $\beta^{(i)}\in(0,\, 1/4)$;
				\item $\displaystyle t \Delta^{1 - 4\beta^{(i)}}_{\mathcal{Q}_t} \rightarrow 0, \ t\rightarrow \infty$;
				\item  			$t\tilde{F}^{(i)}\left(\left(-4\Delta_{\mathcal{Q}_t}^{\beta^{(i)}}, 4\Delta_{\mathcal{Q}_t}^{\beta^{(i)}}\right)\right)\rightarrow 0,\ t\rightarrow\infty$ and $\exists \epsilon^{(i)} > 2\beta^{(i)}$ such that $\displaystyle t \Delta^{\epsilon^{(i)} - 2\beta^{(i)}}_{\mathcal{Q}_t} \rightarrow 0, \ t\rightarrow \infty$ and \[\mathbb{E}_{\mathbf{A}^*}\left[ |{M}^{(i)}_s|^2 \mathds{1}_{\left\{|{M}^{(i)}_s| \leq s^{\beta^{(i)}} \right\}}\right] = O\left(s^{1+\epsilon^{(i)}}\right)\ \mathrm{as}\ s\downarrow 0;\]
				\item $\exists \epsilon_0 > 0, t_0 \in \mathcal{T}$ s.t. $\forall \epsilon \leq \epsilon_0, t \geq t_0$, 
				\[\mathbb{E}_{\mathbf{A}^*} \left[ \Delta^m_{\mathcal{Q}_t} {M}^{(i)} \mathds{1}_{\{|\Delta^m_{\mathcal{Q}_t} {M}^{(i)}| \leq \epsilon\}}\right]= 0\quad \mathrm{for}\ i\in\{1,\ldots, d\},\  m\in\{0,\ldots, M_t - 1\}.\]
			\end{enumerate}
		\end{assumption}
		\begin{remark}
			Assumption \ref{ass:infinite_thresholding_symmetric}$.(iv)$ was introduced in \citet{mai_OU} and is a symmetry condition on $\mathbb{M}$. In \citet[Lemma~4.4]{mai_OU} shows that a sufficient condition for Assumption \ref{ass:infinite_thresholding_symmetric}$.(iv)$ is that $F|_{[-1,1]}$ is symmetric. We note that introducing this condition allows us to modify \ref{ass:infinite_thresholding}$.(iii)$: we now control the $\ell_2$ norm of small increments of $\mathbb{M}$, a ``weaker'' condition.
		\end{remark}
		\begin{theorem}
			Theorem \ref{thm:cons_asymp_third_approx_infinite_activity} holds ceteris paribus with Assumption \ref{ass:infinite_thresholding} replaced by Assumption \ref{ass:infinite_thresholding_symmetric}.
		\end{theorem}
		
		\begin{proof}
			The proof proceeds exactly as in the proof of Theorem \ref{thm:cons_asymp_third_approx_infinite_activity}, cf.\ \ref{app:proof_thm_cons_asymp_third_approx_infinite_activity} for details, with the only difference being in the arguments used to show $\mathbf{S}_t^3(i, l) \overset{\mathbb{P}_{\mathbf{A}^*}}{\rightarrow} 0$ as $t\rightarrow\infty$. In this setting, we rewrite $\mathbf{S}_t^3(i, l)$ as 
			\begin{align*}
				\mathbf{S}_t^3(i, l) &= t^{-1/2}\left(\sum_{m = 0}^{M_t-1} D^{l}Y^{(i)}_{u_m} \Delta_{\mathcal{Q}_t}^m \mathbf{M} \odot \mathds{1}_{\left\{\left|\Delta_{\mathcal{Q}_t}^m \hat{D}^{p-1}\mathbf{Y} - \mathbf{b} \Delta_{\mathcal{Q}_t}^m\right| \leq \boldsymbol{\nu}^m_t\right\}} \right) \\
				&= t^{-1/2}\left(\sum_{m = 0}^{M_t-1} D^{l}Y^{(i)}_{u_m} \Delta_{\mathcal{Q}_t}^m \mathbf{M} \odot \mathds{1}_{\left\{\left|\Delta_{\mathcal{Q}_t}^m \hat{D}^{p-1}\mathbf{Y} - \mathbf{b} \Delta_{\mathcal{Q}_t}^m\right| \leq \boldsymbol{\nu}^m_t, |\Delta_{\mathcal{Q}_t}^m\mathbf{M}|> 2\boldsymbol{\nu}^m_t \right\}} \right) \\
				&\quad\quad + t^{-1/2}\left(\sum_{m = 0}^{M_t-1} D^{l}Y^{(i)}_{u_m} \Delta_{\mathcal{Q}_t}^m \mathbf{M} \odot\left( \mathds{1}_{\left\{\left|\Delta_{\mathcal{Q}_t}^m \hat{D}^{p-1}\mathbf{Y} - \mathbf{b} \Delta_{\mathcal{Q}_t}^m\right| \leq \boldsymbol{\nu}^m_t, |\Delta_{\mathcal{Q}_t}^m\mathbf{M}|\leq 2\boldsymbol{\nu}^m_t \right\}} -  \mathds{1}_{\left\{|\Delta_{\mathcal{Q}_t}^m\mathbf{M}|\leq 2\boldsymbol{\nu}^m_t \right\}} \right) \right) \\
				&\quad\quad\quad\quad + t^{-1/2}\left(\sum_{m = 0}^{M_t-1} D^{l}Y^{(i)}_{u_m} \Delta_{\mathcal{Q}_t}^m \mathbf{M} \odot \mathds{1}_{\left\{|\Delta_{\mathcal{Q}_t}^m\mathbf{M}|\leq 2\boldsymbol{\nu}^m_t \right\}}\right) \\
				&=: \mathbf{S}_t^{3, 1}(i, l) + \mathbf{S}_t^{3, 2}(i, l) + \mathbf{S}_t^{3, 3}(i, l),
			\end{align*}
			For $\mathbf{S}_t^{3, 1}(i, l)$ note that in $L^1(\Omega, \mathcal{F}, \mathbb{P}_{\mathbf{A}^*})$ we have
			\begin{align*}
				\mathbb{E}_{\mathbf{A}^*}&\left[\left\| \mathbf{S}_t^{3, 1}(i, l)\right\|\right] \\
				&\overset{(i)}{\leq} t^{-1/2} \sum_{m = 0}^{M_t-1} \sum_{j=1}^d \mathbb{E}_{\mathbf{A}^*}\left[\left|
				D^{l}Y^{(i)}_{u_m} \Delta_{\mathcal{Q}_t}^m M^{(j)}  \right| \mathds{1}_{\left\{\left|\Delta_{\mathcal{Q}_t}^m \hat{D}^{p-1}Y^{(j)} - b^{(j)} \Delta_{\mathcal{Q}_t}^m\right| \leq \nu^{(j), m}_t, |\Delta_{\mathcal{Q}_t}^m M^{(j)}|> 2\nu^{(j), m}_t \right\}} \right] \\
				%-----------------------
				&\overset{(ii)}{\leq} t^{-1/2} \sum_{m = 0}^{M_t-1} \sum_{j=1}^d \mathbb{E}_{\mathbf{A}^*}\left[\left|
				D^{l}Y^{(i)}_{u_m} \Delta_{\mathcal{Q}_t}^m M^{(j)}  \right| \mathds{1}_{\left\{\left|\Delta_{\mathcal{Q}_t}^m \hat{D}^{p-1}Y^{c, (j)} + \Delta_{\mathcal{Q}_t}^m J^{(j)} \right| > \nu^{(j), m}_t \right\}} \right] \\
				%----------------------
				&\overset{(iii)}{\leq} t^{-1/2} \sum_{m = 0}^{M_t-1} \sum_{j=1}^d \mathbb{E}_{\mathbf{A}^*}\left[\left|
				D^{l}Y^{(i)}_{u_m} \Delta_{\mathcal{Q}_t}^m M^{(j)}  \right|\left( \mathds{1}_{\left\{\left|\Delta_{\mathcal{Q}_t}^m \hat{D}^{p-1}Y^{c, (j)}_{\mathbf{A}^{(0)}} \right| > \nu^{(j), m}_t \right\}} + \mathds{1}_{\left\{\Delta_{\mathcal{Q}_t}^m N^{(j)} > 0 \right\}} \right) \right] \\
				%----------------------
				&\overset{(iv)}{\leq} t^{-1/2} \sum_{m = 0}^{M_t-1} \sum_{j=1}^d \mathbb{E}_{\mathbf{A}^*}\Bigg[\left|
				D^{l}Y^{(i)}_{u_m} \Delta_{\mathcal{Q}_t}^m M^{(j)}  \right|\Bigg( \mathds{1}_{\left\{\left|\Delta_{\mathcal{Q}_t}^m \hat{D}^{p-1}Y^{c, (j)}_{\mathbf{A}^{(0)}} - \Delta_{\mathcal{Q}_t}^m D^{p-1}Y^{c, (j)}_{\mathbf{A}^{(0)}} \right| > \nu^{(j), m}_t/2 \right\}} \\
				&\quad\quad\quad\quad\quad\quad\quad\quad\quad\quad\quad\quad\quad\quad\quad\quad + \mathds{1}_{\left\{\left|\Delta_{\mathcal{Q}_t}^m D^{p-1}Y^{c, (j)}_{\mathbf{A}^{(0)}} \right| > \nu^{(j), m}_t/2 \right\}} + \mathds{1}_{\left\{\Delta_{\mathcal{Q}_t}^m N^{(j)} > 0 \right\}} \Bigg) \Bigg] \\
				%-----------------------------------------
				&\overset{(v)}{\leq} t^{-1/2} \sum_{m = 0}^{M_t-1} \sum_{j=1}^d \mathbb{E}_{\mathbf{A}^*}\left[\left|
				D^{l}Y^{(i)}_{u_m}\right|^2\right]^{1/2} \mathbb{E}_{\mathbf{A}^*}\left[\left| \Delta_{\mathcal{Q}_t}^m M^{(j)}  \right|^2\right]^{1/2} \\
				&\quad\quad\quad\quad\quad\quad\quad \times \Bigg[ \mathbb{P}\left(\left|\Delta_{\mathcal{Q}_t}^m \hat{D}^{p-1}Y^{c, (j)}_{\mathbf{A}^{(0)}} - \Delta_{\mathcal{Q}_t}^m D^{p-1}Y^{c, (j)}_{\mathbf{A}^{(0)}} \right| > \nu^{(j), m}_t/2 \right)^{1/2} \\
				&\quad\quad\quad\quad\quad\quad\quad\quad\quad\quad\quad\quad + \mathbb{P}\left(\left|\Delta_{\mathcal{Q}_t}^m D^{p-1}Y^{c, (j)}_{\mathbf{A}^{(0)}} \right| > \nu^{(j), m}_t/2 \right)^{1/2} + \mathbb{P}\left(\Delta_{\mathcal{Q}_t}^m N^{(j)} > 0 \right) \Bigg] \\
				%-----------------------------------------
				&\overset{(vi)}{\lesssim} \sum_{j=1}^d t^{-1/2}M_t \left(\Delta^m_{\mathcal{Q}_t}\right)^{1/2} \left[ \left(\nu^{(j), m}_t\right)^{-1} \Delta_{\mathcal{P}_t}^{1/2} + \Delta_{\mathcal{Q}_t}^{1-\beta^{(j)}} + \Delta_{\mathcal{Q}_t}\right] \\
				%-----------------------------------------
				&\overset{(vii)}{\lesssim} \sum_{j=1}^d (t^{-1}N'_t \Delta_{\mathcal{Q}_t}) \left[  \left(t\Delta_{\mathcal{P}_t}\Delta^{-2}_{\mathcal{Q}_t}\right)^{1/2}\Delta_{\mathcal{Q}_t}^{1/2-\beta^{(j)}} + \left(t\Delta_{\mathcal{Q}_t}^{1-2\beta^{(j)}}\right)^{1/2} \right] \rightarrow 0,\quad  t\rightarrow\infty,
			\end{align*}
			where we use in $(i)$ triangle inequality, in $(ii)$ $|\hat{y}|\leq \nu, |m|>2\nu \implies |\hat{y}-m|>\nu$ for $\hat{y},m,\nu\in\mathbb{R}$ and $\mathds{1}_{A}\leq \mathds{1}_{B}$ for sets $A\subset B \subset\Omega$, in $(iii)$ 
			$1 = \mathds{1}_{A} + \mathds{1}_{A^c}$ and $\mathds{1}_{A\cap B}\leq \mathds{1}_{A}$ for sets $A,B\subset\Omega$, in $(iv)$ $|\hat{y}|>\nu\implies |y|>\nu/2 \vee |y-\hat{y}|>\nu/2$ for $y,\hat{y},\nu\in\mathbb{R}$ and $\mathds{1}_{A}\leq \mathds{1}_{B} + \mathds{1}_{C}$ for sets $A,B, C\subset\Omega$ such that $A\subset B\cup C$, in $(v)$ Cauchy-Schwartz inequality and independence of $D^l\mathbf{Y}_{u_m}$ and $\Delta_{\mathcal{Q}_t}^m \mathbf{M}$ on the first two terms, independence of $D^l\mathbf{Y}_{u_m}$, $\Delta_{\mathcal{Q}_t}^m \mathbf{M}$ and $\Delta_{\mathcal{Q}_t}^m \mathbf{N}$ on the third term, in $(vi)$ boundedness of second moments of $\mathbf{X}_t = (\mathbf{Y}^{\mathrm{T}}_t,\ldots, D^{p-1}\mathbf{Y}^{\mathrm{T}}_t)^{\mathrm{T}}$, Equation \eqref{eqn:bound_M2,J2}, Markov inequality, Equation \eqref{eqn:D-D_L2_bound}, Equation \eqref{eqn:P(mart>2nu)}, $\Delta_{\mathcal{Q}_t}^m N^{(j)} \overset{\mathbb{P}_{\mathbf{A}^*}}{\sim}\mathrm{Possion}(\lambda^{(j)} \Delta_{\mathcal{Q}_t}^m)$, in $(vii)$ Assumption \ref{ass:controlled_sampling},  Assumption \ref{ass:infinite_thresholding_symmetric}$.(i)$ and $\Delta^m_{\mathcal{Q}_t} \leq \Delta_{\mathcal{Q}_t}$, Assumption \ref{ass:joint_mesh} and Assumption \ref{ass:infinite_thresholding_symmetric}$.(ii)$.
			
			Next, we consider $\mathbf{S}_t^{3, 2}(i, l)$. Again we work in $L^1(\Omega, \mathcal{F}, \mathbb{P}_{\mathbf{A}^*})$.
			\begin{align*}
				\mathbb{E}_{\mathbf{A}^*}&\left[\left\| \mathbf{S}_t^{3, 2}(i, l)\right\|\right] \\ 
				&\overset{(i)}{\leq} t^{-1/2} \sum_{m = 0}^{M_t-1} \sum_{j=1}^d \mathbb{E}_{\mathbf{A}^*}\left[\left|
				D^{l}Y^{(i)}_{u_m} \Delta_{\mathcal{Q}_t}^m M^{(j)}  \right| \mathds{1}_{\left\{\left|\Delta_{\mathcal{Q}_t}^m \hat{D}^{p-1}Y^{(j)} - b^{(j)} \Delta_{\mathcal{Q}_t}^m\right| > \nu^{(j), m}_t, |\Delta_{\mathcal{Q}_t}^m M^{(j)}|\leq 2\nu^{(j), m}_t \right\}} \right] \\
				%------------------------------------------
				%	&\overset{(ii)}{\leq} t^{-1/2} \sum_{m = 0}^{M_t-1} \sum_{j=1}^d \mathbb{E}_{\mathbf{A}^*}\left[\left|
				%	D^{l}Y^{(i)}_{u_m} \Delta_{\mathcal{Q}_t}^m M^{(j)}  \right| \mathds{1}_{\left\{\left|\Delta_{\mathcal{Q}_t}^m \hat{D}^{p-1}Y^{c, (j)}_{\mathbf{A}^{(0)}} + \Delta_{\mathcal{Q}_t}^m M^{(j)} + \Delta_{\mathcal{Q}_t}^m J^{(j)} \right| > \nu^{(j), m}_t, |\Delta_{\mathcal{Q}_t}^m M^{(j)}|\leq 2\nu^{(j), m}_t \right\}} \right] \\
				%-----------------------------------------------
				&\overset{(ii)}{\leq} t^{-1/2} \sum_{m = 0}^{M_t-1} \sum_{j=1}^d \mathbb{E}_{\mathbf{A}^*}\bigg[\left|
				D^{l}Y^{(i)}_{u_m} \Delta_{\mathcal{Q}_t}^m M^{(j)}  \right| \bigg( \mathds{1}_{\left\{\left|\Delta_{\mathcal{Q}_t}^m \hat{D}^{p-1}Y^{c, (j)}_{\mathbf{A}^{(0)}} + \Delta_{\mathcal{Q}_t}^m M^{(j)} \right| > \nu^{(j), m}_t, |\Delta_{\mathcal{Q}_t}^m M^{(j)}|\leq 2\nu^{(j), m}_t \right\}} \\
				&\quad\quad\quad\quad\quad\quad\quad\quad\quad\quad\quad\quad\quad\quad\quad\quad\quad\quad\quad\quad\quad\quad\quad\quad\quad\quad\quad\quad\quad\quad\quad\quad\quad + \mathds{1}_{\left\{\Delta_{\mathcal{Q}_t}^m N^{(j)} > 0 \right\}} \bigg) \bigg] \\
				%-----------------------------------------------
				&\overset{(iii)}{\leq} t^{-1/2} \sum_{m = 0}^{M_t-1} \sum_{j=1}^d \mathbb{E}_{\mathbf{A}^*}\Bigg[\left|
				D^{l}Y^{(i)}_{u_m} \Delta_{\mathcal{Q}_t}^m M^{(j)}  \right| \Bigg( \mathds{1}_{\left\{\left|\Delta_{\mathcal{Q}_t}^m \hat{D}^{p-1}Y^{c, (j)}_{\mathbf{A}^{(0)}} \right| > \nu^{(j), m}_t/2 \right\}} \\
				&\quad\quad\quad\quad\quad\quad\quad\quad\quad\quad\quad\quad\quad\quad\quad\quad\quad\quad\quad\quad + \mathds{1}_{\left\{\nu^{(j), m}_t/2 < \left| \Delta_{\mathcal{Q}_t}^m M^{(j)} \right| < 2\nu^{(j), m}_t \right\}} + \mathds{1}_{\left\{\Delta_{\mathcal{Q}_t}^m N^{(j)} > 0 \right\}} \Bigg) \Bigg] \\
				%-----------------------------------------------
				&\overset{(iv)}{\leq} t^{-1/2} \sum_{m = 0}^{M_t-1} \sum_{j=1}^d \mathbb{E}_{\mathbf{A}^*}\Bigg[\left|
				D^{l}Y^{(i)}_{u_m} \Delta_{\mathcal{Q}_t}^m M^{(j)}  \right| \Bigg( \mathds{1}_{\left\{\left|\Delta_{\mathcal{Q}_t}^m \hat{D}^{p-1}Y^{c, (j)}_{\mathbf{A}^{(0)}} - \Delta_{\mathcal{Q}_t}^m D^{p-1}Y^{c, (j)}_{\mathbf{A}^{(0)}} \right| > \nu^{(j), m}_t/4 \right\}} \\
				&\quad\quad\quad\quad\quad\quad\quad\quad +\mathds{1}_{\left\{\left|\Delta_{\mathcal{Q}_t}^m D^{p-1}Y^{c, (j)}_{\mathbf{A}^{(0)}} \right| > \nu^{(j), m}_t/4 \right\}}  + \mathds{1}_{\left\{\nu^{(j), m}_t/2 < \left| \Delta_{\mathcal{Q}_t}^m M^{(j)} \right| < 2\nu^{(j), m}_t \right\}} + \mathds{1}_{\left\{\Delta_{\mathcal{Q}_t}^m N^{(j)} > 0 \right\}} \Bigg) \Bigg]	\\
				%-----------------------------------------
				&\overset{(v)}{\leq} t^{-1/2} \sum_{m = 0}^{M_t-1} \Bigg\{\sum_{j=1}^d \mathbb{E}_{\mathbf{A}^*}\left[\left|
				D^{l}Y^{(i)}_{u_m}\right|^2\right]^{1/2} \mathbb{E}_{\mathbf{A}^*}\left[\left| \Delta_{\mathcal{Q}_t}^m M^{(j)}  \right|^2\right]^{1/2} \\
				&\quad\quad\quad\quad\quad\quad\quad \times \Bigg[ \mathbb{P}\left(\left|\Delta_{\mathcal{Q}_t}^m \hat{D}^{p-1}Y^{c, (j)}_{\mathbf{A}^{(0)}} - \Delta_{\mathcal{Q}_t}^m D^{p-1}Y^{c, (j)}_{\mathbf{A}^{(0)}} \right| > \nu^{(j), m}_t/4 \right)^{1/2} \\
				&\quad\quad\quad\quad\quad\quad\quad\quad\quad\quad\quad\quad + \mathbb{P}\left(\left|\Delta_{\mathcal{Q}_t}^m D^{p-1}Y^{c, (j)}_{\mathbf{A}^{(0)}} \right| > \nu^{(j), m}_t/4 \right)^{1/2} + \mathbb{P}\left(\Delta_{\mathcal{Q}_t}^m N^{(j)} > 0 \right) \Bigg] \\
				&\quad\quad\quad\quad\quad\quad\quad + \mathbb{E}_{\mathbf{A}^*}\left[\left|
				D^{l}Y^{(i)}_{u_m}\right|\right] \mathbb{E}_{\mathbf{A}^*}\left[\left| \Delta_{\mathcal{Q}_t}^m M^{(j)}\right|  \mathds{1}_{\left\{\nu^{(j), m}_t/2 < \left| \Delta_{\mathcal{Q}_t}^m M^{(j)} \right| < 2\nu^{(j), m}_t \right\}}\right] \Bigg\} \\
				%------------------------------------
				&\overset{(vi)}{\lesssim} \sum_{j=1}^d \left\{ t^{-1/2}M_t \left(\Delta^m_{\mathcal{Q}_t}\right)^{1/2} \left[ \left(\nu^{(j), m}_t\right)^{-1} \Delta_{\mathcal{P}_t}^{1/2} + \left(\Delta^m_{\mathcal{Q}_t}\right)^{1-\beta^{(j)}} + \Delta^m_{\mathcal{Q}_t} \right] +  t^{-1/2}M_t \Delta_{\mathcal{Q}_t}^{1-\beta^{(j)} +\epsilon^{(j)}/2}\right\}\\
				%--------------------------------
				&\overset{(vii)}{\lesssim} \sum_{j=1}^d t^{-1/2} M_t \left\{ \left(\Delta^m_{\mathcal{Q}_t}\right)^{1/2-\beta^{(j)}} \Delta_{\mathcal{P}_t}^{1/2} + \Delta_{\mathcal{Q}_t}^{3/2-\beta^{(j)}} +\Delta_{\mathcal{Q}_t}^{1-\beta^{(j)} +\epsilon^{(j)}/2} \right\}\\
				%--------------------------------
				&\overset{(viii)}{\lesssim} \sum_{j=1}^d (t^{-1} N'_t \Delta_{\mathcal{Q}_t}) \left[\left(t\Delta_{\mathcal{Q}_t}^{-2}\Delta_{\mathcal{P}_t}\right)^{1/2}\Delta_{\mathcal{Q}_t}^{1/2-\beta^{(j)}} + \left(t \Delta_{\mathcal{Q}_t}^{1-2\beta^{(j)}}\right)^{1/2} + \left(t \Delta_{\mathcal{Q}_t}^{\epsilon^{(j)}-2\beta^{(j)} }\right)^{1/2} \right]\rightarrow 0, \quad t\rightarrow\infty,
			\end{align*}
			where we use in $(i)$ triangle inequality, in $(ii)$ $1 = \mathds{1}_{A} + \mathds{1}_{A^c}$ and $\mathds{1}_{A\cap B}\leq \mathds{1}_{A}$ for sets $A,B\subset\Omega$, in $(iii)$ $|\hat{y}+m|>\nu \implies |\hat{y}|>\nu/2\vee |m|>\nu/2$ for $\hat{y},m,\nu,\in\mathbb{R}$ and $\mathds{1}_{A}\leq \mathds{1}_{B} + \mathds{1}_{C}$ for sets $A,B, C\subset\Omega$ such that $A\subset B\cup C$, in $(iv)$ $|\hat{y}|>\nu/2\implies |y|>\nu/4 \vee |y-\hat{y}|>\nu/4$ for $y,\hat{y},\nu\in\mathbb{R}$ and $\mathds{1}_{A}\leq \mathds{1}_{B} + \mathds{1}_{C}$ for sets $A,B, C\subset\Omega$ such that $A\subset B\cup C$, in $(v)$ Cauchy-Schwartz inequality and independence of $D^l\mathbf{Y}_{u_m}$ and $\Delta_{\mathcal{Q}_t}^m \mathbf{M}$ on the first two terms, independence of $D^l\mathbf{Y}_{u_m}$, $\Delta_{\mathcal{Q}_t}^m \mathbf{M}$ and $\Delta_{\mathcal{Q}_t}^m \mathbf{N}$ on the third term, independence of $D^l\mathbf{Y}_{u_m}$ and $\Delta_{\mathcal{Q}_t}^m \mathbf{M}$ on the fourth term, in $(vi)$ boundedness of second moments of $\mathbf{X}_t = (\mathbf{Y}^{\mathrm{T}}_t,\ldots, D^{p-1}\mathbf{Y}^{\mathrm{T}}_t)^{\mathrm{T}}$, Equation \eqref{eqn:bound_M2,J2}, Markov's inequality, Equation \eqref{eqn:D-D_L2_bound}, Equation \eqref{eqn:P(mart>2nu)}, $\Delta_{\mathcal{Q}_t}^m N^{(j)} \overset{\mathbb{P}_{\mathbf{A}^*}}{\sim}\mathrm{Possion}(\lambda^{(j)} \Delta_{\mathcal{Q}_t}^m)$ and
			\begin{align}
				\mathbb{E}_{\mathbf{A}^*}&\left[\left| \Delta_{\mathcal{Q}_t}^m M^{(j)}\right|  \mathds{1}_{\left\{\nu^{(j), m}_t/2 < \left| \Delta_{\mathcal{Q}_t}^m M^{(j)} \right| < 2\nu^{(j), m}_t \right\}}\right] \notag \\
				&\leq \mathbb{P}_{\mathbf{A}^*}\left(\nu^{(j), m}_t/2 < \left| \Delta_{\mathcal{Q}_t}^m M^{(j)} \right| \right)^{1/2} \mathbb{E}_{\mathbf{A}^*}\left[\left| \Delta_{\mathcal{Q}_t}^m M^{(j)}\right|^2 \mathds{1}_{\left\{\left| \Delta_{\mathcal{Q}_t}^m M^{(j)} \right| < 2\nu^{(j), m}_t \right\}}\right]^{1/2} \notag \\
				&\leq \left(\nu^{(j), m}_t\right)^{-1} \mathbb{E}_{\mathbf{A}^*}\left[\left| \Delta_{\mathcal{Q}_t}^m M^{(j)} \right|^2 \right]^{1/2} \mathbb{E}_{\mathbf{A}^*}\left[\left| \Delta_{\mathcal{Q}_t}^m M^{(j)}\right|^2 \mathds{1}_{\left\{\left| \Delta_{\mathcal{Q}_t}^m M^{(j)} \right| < 2\nu^{(j), m}_t \right\}}\right]^{1/2} \notag\\
				&\leq (\Delta^m_{\mathcal{Q}_t})^{-\beta^{(j)}} (\Delta^m_{\mathcal{Q}_t})^{1/2} (\Delta^m_{\mathcal{Q}_t})^{1/2 + \epsilon^{(j)}/2} = (\Delta^m_{\mathcal{Q}_t})^{1-\beta^{(j)} +\epsilon^{(j)}/2}, \label{eqn:M_{v/2<M<2v}}
			\end{align}
			in $(vii)$ Assumption \ref{ass:infinite_thresholding_symmetric}$.(i)$, in  $(viii)$ Assumption \ref{ass:controlled_sampling}, Assumption \ref{ass:joint_mesh} and Assumption \ref{ass:infinite_thresholding_symmetric}$.(iii)$.
			
			Finally we show $\mathbf{S}_t^{3,3}(i, l) \overset{\mathbb{P}_{\mathbf{A}^*}}{\rightarrow} 0$ as $t\rightarrow\infty$. In this case we work in $L^2(\Omega, \mathcal{F},\mathbb{P}_{\mathbf{A}*})$ and note that, for $t$ sufficiently large, $\mathbf{S}_t^{3,3}(i, l)$ is mean-zero by Assumption \ref{ass:infinite_thresholding_symmetric}$.(iv)$:
			
			\begin{align*}
				\mathbb{E}_{\mathbf{A}^*}&\left[\left\| \mathbf{S}_t^{3, 3}(i, l)\right\|^2\right] \\
				&= t^{-1} \sum_{j=1}^d \mathbb{E}_{\mathbf{A}^*}\left[\left| \sum_{m = 0}^{M_t-1} D^{l}Y^{(i)}_{u_m} \Delta_{\mathcal{Q}_t}^m M^{(j)}  \mathds{1}_{\left\{|\Delta_{\mathcal{Q}_t}^m M^{(j)}|\leq 2\nu^{(j), m}_t \right\}}\right|^2 \right] \\
				&=t^{-1} \sum_{j=1}^d \sum_{m = 0}^{M_t-1}\sum_{m' = 0}^{M_t-1} \mathbb{E}_{\mathbf{A}^*}\left[ D^{l}Y^{(i)}_{u_m} D^{l}Y^{(i)}_{u_{m'}} \Delta_{\mathcal{Q}_t}^m M^{(j)} \Delta_{\mathcal{Q}_t}^{m'} M^{(j)}  \mathds{1}_{\left\{|\Delta_{\mathcal{Q}_t}^m M^{(j)}|\leq 2\nu^{(j), m}_t \right\}} \mathds{1}_{\left\{|\Delta_{\mathcal{Q}_t}^{m'} M^{(j)}|\leq 2\nu^{(j), m}_t \right\}} \right] \\
				&=t^{-1} \sum_{j=1}^d \sum_{m = 0}^{M_t-1} \mathbb{E}_{\mathbf{A}^*}\left[ \left|D^{l}Y^{(i)}_{u_m}\right|^2\right] \mathbb{E}_{\mathbf{A}^*}\left[ \left| \Delta_{\mathcal{Q}_t}^m M^{(j)}\right|^2 \mathds{1}_{\left\{|\Delta_{\mathcal{Q}_t}^m M^{(j)}|\leq 2\nu^{(j), m}_t \right\}} \right] \\
				&\lesssim \sum_{j=1}^d t^{-1} \sum_{m = 0}^{M_t-1} (\Delta^m_{\mathcal{Q}_t})^{1+\epsilon^{(j)}} \lesssim \sum_{j=1}^d \Delta_{\mathcal{Q}_t}^{\epsilon^{(j)}} \rightarrow 0, \quad t\rightarrow\infty,
			\end{align*}
			using independence of $\Delta_{\mathcal{Q}_t}^m M^{(j)}$, $\Delta_{\mathcal{Q}_t}^{m'} M^{(j)}$, $D^{l}Y^{(i)}_{u_m}$, $D^{l}Y^{(i)}_{u_{m'}}$ for $m\neq m'$ and Assumption \ref{ass:infinite_thresholding_symmetric}$.(iv)$ to show that off-diagonal elements are zero and Assumption \ref{ass:infinite_thresholding_symmetric}$.(iii)$.
		\end{proof}
		
		\section{GrCAR: proofs}
		\subsection{Proof of Theorem \ref{thm:GrCAR}} \label{app:thm_GrCAR_proof}
		\begin{proof}
			We proceed exactly as in Section \ref{sec:MCAR_cont} replacing the dependencies on the parameter $\mathbf{A}\in(\mathcal{M}(\mathbb{R}^d))^p$ with dependencies on $\theta\in\mathbb{R}^{p\times 2}$. Expanding the form of the likelihood in Proposition \ref{prop:exist_likelihood} and grouping in terms of $\theta^{(i,j)}$ for $i=1,\ldots,p$ and $j=1,2$ we can rewrite the likelihood of the GrCAR($p$) process as 
			\[\mathcal{L}(\theta; \mathbb{Y}_{[0,t]}) = \exp\left\{\mathrm{vec} (\theta)^\mathrm{T} \mathbf{K}_t - \frac{1}{2} \mathrm{vec}(\theta)^\mathrm{T} [\mathbf{K}]_t \mathrm{vec}(\theta)\right\}, \quad 
			\theta \in\mathbb{R}^{p\times 2},\]
			with $\mathbb{K} = \{\mathbf{K}_t,\ t\geq 0\}$ as above. One can then proceed as in Lemma \ref{lemma:loc_mart_MCAR} to check that for any $\theta\in\mathbb{R}^{p\times 2}$ the score 
			\[\nabla_{\mathrm{vec}(\theta)} \log \mathcal{L}(\theta; \mathbb{Y}_{[0,t]}) = \mathbf{K}_t - [\mathbf{K}]_t \mathrm{vec}(\theta), \]
			defines a continuous square-integrable martingale under $\mathbb{P}_{\theta, \mathbf{x}_0}$. This is ensured by showing the quadratic variation is integrable, i.e.\
			\begin{align*}
				\mathbb{E}_{\theta, \mathbf{x}_0}\left[ \|[\mathbf{K}]_t\| \right] &\leq (\| I\| + \|\bar{A}^\mathrm{T}\|) \mathbb{E}_{\theta, \mathbf{x}_0}\left[\int_0^t \sum_{i=0}^{p-1} \| D^i\mathbf{Y}_s\|^2_\Sigma \, ds \right] \\
				& \leq (\| I\| + \|\bar{A}^\mathrm{T}\|) \|\Sigma^{-1} \| \int_0^t \mathbb{E}_{\theta, \mathbf{x}_0} [\|\mathbf{X}_s\|^2] \, ds 
				<\infty.
			\end{align*}
			Finally, as in Theorem \ref{thm:cons_asymp}, we check that $\forall t>0$ the matrix $[\mathbf{K}]_t$ is strictly positive definite $\mathbb{P}_{\theta^*, \mathbf{x}_0}$-a.s. By carrying out the relevant computation one can show this boils down to proving \[\{D^{l}Y_s^{(i)},\ i=1,\ldots,d, l=0,\ldots, p-1 \}\] is an affinely independent collection of random variables for any $s\in(0,t]$ under $\mathbb{P}_{\theta^{(0)}, \mathbf{x}_0}$, which follows by the same arguments as before. The rest of the results carry over ceteris paribus for the MLE $\hat{\theta}(\mathbb{Y}_{[0,t]})$
			with 
			\begin{equation} \label{eqn:K_infty}
				\mathcal{K}_\infty =
				\begin{pmatrix}
					\mathbb{E}_{\theta^*} [\langle D^{p-1}\mathbf{Y}_{\infty}, D^{p-1}\mathbf{Y}_{\infty}\rangle_{\Sigma}] & \cdots & \mathbb{E}_{\theta^*} [\langle D^{p-1}\mathbf{Y}_{\infty}, \bar{A}^\mathrm{T} \mathbf{Y}_{\infty}\rangle_{\Sigma}]\\
					\mathbb{E}_{\theta^*} [\langle \bar{A}^\mathrm{T} D^{p-1}\mathbf{Y}_{\infty}, D^{p-1}\mathbf{Y}_{\infty}\rangle_{\Sigma}] & \cdots & \mathbb{E}_{\theta^*} [\langle \bar{A}^\mathrm{T} D^{p-1}\mathbf{Y}_{\infty}, \bar{A}^\mathrm{T} \mathbf{Y}_{\infty}\rangle_{\Sigma}] \\
					\vdots & \ddots & \vdots \\
					\mathbb{E}_{\theta^*}[\langle \mathbf{Y}_{\infty}, D^{p-1}\mathbf{Y}_{\infty}\rangle_{\Sigma}] & \cdots & \mathbb{E}_{\theta^*} [\langle \mathbf{Y}_{\infty}, \bar{A}^\mathrm{T} \mathbf{Y}_{\infty}\rangle_{\Sigma}]\\
					\mathbb{E}_{\theta^*} [\langle \bar{A}^\mathrm{T} \mathbf{Y}_{\infty}, D^{p-1}\mathbf{Y}_{\infty}\rangle_{\Sigma}] & \cdots & \mathbb{E}_{\theta^*} [\langle \bar{A}^\mathrm{T} \mathbf{Y}_{\infty}, \bar{A}^\mathrm{T} \mathbf{Y}_{\infty}\rangle_{\Sigma}] \\
				\end{pmatrix}.
			\end{equation} 
			The inference results for the discretized estimator follow by the same subsequent approximations as in Section \ref{sec:MCAR_discr} (and the corresponding proofs in Appendix \ref{app:proof_discr}), noting that one can simply factor out the $\bar{A}^\mathrm{T}$ term.
		\end{proof}
		
		\section{Finite difference approximations of derivatives} 
		For $k\in\{1,\ldots, p-1\}$ let $D^k\mathbb{Y}_{[0,t]} = \{D^k\mathbf{Y}_s,\ s\in [0,t]\}$ denote the $k$-th time derivative of the process $\mathbb{Y}_{[0,t]} = \{\mathbf{Y}_s,\ s\in [0,t]\}$. We define the forward finite difference approximation of the derivatives over the partition
		\[\mathcal{P}_t = \{0=s_0<s_1<\ldots <s_{N_t} = t\},\]
		iteratively by
		\[ \hat{D}^k\mathbb{Y}_{\mathcal{P}_t} := \{ \hat{D}^k\mathbf{Y}_{s_n}, \ n=0,\ldots, N_t - k\}\ \mathrm{s.t.}\ \hat{D}^j\mathbf{Y}_{s_n} = 
		\frac{\hat{D}^{k-1}\mathbf{Y}_{s_{n+1}} - \hat{D}^{k-1}\mathbf{Y}_{s_{n}}}{s_{n+1} - s_{n}},\]
		for $j=k,\ldots, p-1$ and $\hat{D}^0\mathbb{Y}_{\mathcal{P}_t} = \mathbb{Y}_{\mathcal{P}_t}$. We make use of the following results.
		
		\begin{lemma} \label{lemma:FD_Y} Let $k\in\{1,\ldots, p-1\}$. The forward differences can be written as follows for $n\in\{0,\ldots, N_t-k\}$
			\begin{equation*} \hat{D}^k\mathbf{Y}_{s_n} = \sum_{i=1}^k (-1)^{k+i} \sum_{\substack{n_1,\ldots, n_i\geq 1\\ n_1 +\cdots + n_i = k}} (s_{n+1} - s_n)^{-n_1} \cdots (s_{n+i} - s_{n+i-1})^{-n_i} (\mathbf{Y}_{s_{n+i}} - \mathbf{Y}_{s_{n+i-1}}). \end{equation*}
		\end{lemma}
		\begin{proof}
			For $m=1,\ldots, k$,
			\begin{align*} 
				\hat{D}^k\mathbf{Y}_{s_n} \\=\sum_{i=1}^m (-1)^{m+i} \sum_{\substack{n_1,\ldots, n_i\geq 1\\ n_1 +\cdots + n_i = m}} (s_{n+1} - s_n)^{-n_1} \cdots (s_{n+i} - s_{n+i-1})^{-n_i} (\hat{D}^{k-m}\mathbf{Y}_{s_{n+i}} - \hat{D}^{k-m}\mathbf{Y}_{s_{n+i-1}}).
			\end{align*}
			This can be proved inductively:
			\begin{itemize}
				\item for the case $m=1$ this is the iterative definition of the forward finite difference;
				\item given the case $m<k$
				\begin{align*} 
					\hat{D}^k\mathbf{Y}_{s_n} &=\sum_{i=1}^m (-1)^{m+i} \sum_{\substack{n_1,\ldots, n_i\geq 1\\ n_1 +\cdots + n_i = m}} (s_{n+1} - s_n)^{-n_1} \cdots (s_{n+i} - s_{n+i-1})^{-n_i} \\
					&\quad\quad\quad\quad\quad\quad\quad\quad\quad\quad\quad\quad\quad\quad\quad\quad\quad\quad\quad\quad\quad\quad\quad\quad \times (\hat{D}^{k-m}\mathbf{Y}_{s_{n+i}} - \hat{D}^{k-m}\mathbf{Y}_{s_{n+i-1}}) \\
					&=\sum_{i=1}^m (-1)^{m+i} \sum_{\substack{n_1,\ldots, n_i\geq 1\\ n_1 +\cdots + n_i = m}} (s_{n+1} - s_n)^{-n_1} \cdots (s_{n+i} - s_{n+i-1})^{-n_i} \\
					&\quad\quad\quad\quad\quad\quad\quad\quad\quad\quad\quad\quad \times \Big[(s_{n+i+1} - s_{n+i})^{-1}(\hat{D}^{k-m-1}\mathbf{Y}_{s_{n+i+1}} - \hat{D}^{k-m-1}\mathbf{Y}_{s_{n+i}})  \\
					&\quad\quad\quad\quad\quad\quad\quad\quad\quad\quad\quad\quad\quad\quad\quad\quad - (s_{n+i} - s_{n+i-1})^{-1}(\hat{D}^{k-m-1}\mathbf{Y}_{s_{n+i}} - \hat{D}^{k-m-1}\mathbf{Y}_{s_{n+i-1}})\Big] \\
					&=\sum_{i=1}^{m+1} (-1)^{m+1+i} \left(\sum_{\substack{n_1,\ldots, n_{i-1}\geq 1, n_i = 1\\ n_1 +\cdots + n_i = m+1}} + \sum_{\substack{n_1,\ldots, n_{i-1}\geq 1, n_i \geq 2 \\ n_1 +\cdots + n_i = m+1}} \right) (s_{n+1} - s_n)^{-n_1} \cdots (s_{n+i} - s_{n+i-1})^{-n_i}  \\
					&\quad\quad\quad\quad\quad\quad\quad\quad\quad\quad\quad\quad\quad\quad\quad\quad\quad\quad\quad\quad\quad\quad\quad\quad
					\times (\hat{D}^{k-m-1}\mathbf{Y}_{s_{n+i}} - \hat{D}^{k-m-1}\mathbf{Y}_{s_{n+i-1}}) \\
					&=\sum_{i=1}^{m+1} (-1)^{m+1+i} \sum_{\substack{n_1,\ldots, n_{i}\geq 1\\ n_1 +\cdots + n_i = m+1}}  (s_{n+1} - s_n)^{-n_1} \cdots (s_{n+i} - s_{n+i-1})^{-n_i} \\
					&\quad\quad\quad\quad\quad\quad\quad\quad\quad\quad\quad\quad\quad\quad\quad\quad\quad\quad\quad\quad\quad\quad\quad\quad \times (\hat{D}^{k-m-1}\mathbf{Y}_{s_{n+i}} - \hat{D}^{k-m-1}\mathbf{Y}_{s_{n+i-1}}), \\
				\end{align*}
				where we use that
				\begin{multline*}\{(n_1,\ldots, n_i): n_1,\ldots, n_i\geq 1, n_1 +\cdots + n_i = m+1\} = \\ \{(n_1,\ldots, n_{i-1}, 1): n_1,\ldots, n_{i-1} \geq 1, n_1 +\cdots + n_{i-1} = m\} + \\ \{(n_1,\ldots, n_i): n_1,\ldots,n_{i-1}\geq 1, n_i\geq 2, n_1 +\cdots + n_i = m+1\}. \end{multline*}
			\end{itemize}
		\end{proof}
		
		\begin{lemma} \label{lemma:Y_iterated_int} Let $k\in\{1,\ldots, p-1\}$. We can write for $n\in\{0,\ldots, N_t-k\}$ and $i\in\{1,\ldots, k\}$
			\begin{equation*}
				\begin{split} \mathbf{Y}_{s_{n+i}} - \mathbf{Y}_{s_{n+i-1}} = \sum_{j=1}^{k-1} D^j \mathbf{Y}_{s_n} &\frac{(s_{n+i} - s_{n})^j - (s_{n+i-1} - s_{n})^j}{j!} \\ &+ \int_{s_{n+i-1}}^{s_{n+i}} \int_{s_n}^{u_1}\cdots\int_{s_n}^{u_{k-1}} D^k \mathbf{Y}_{u_k}\ du_k\ldots du_1. \end{split}
			\end{equation*}
		\end{lemma}
		\begin{proof}
			This follows by writing 
			\[D^j\mathbf{Y}_s = D^j\mathbf{Y}_{s_n} + \int_{s_n}^s D^{j+1}\mathbf{Y}_{u}\ du,\] 
			and noting that for $j=1,\ldots, k-1$
			\[\int_{s_{n+i-1}}^{s_{n+i}} \int_{s_n}^{u_1}\cdots\int_{s_n}^{u_{j-1}}du_{j}\ldots du_1 = \frac{(s_{n+i} - s_{n})^j - (s_{n+i-1} - s_{n})^j}{j!}.\]
		\end{proof}
		
		\begin{lemma} \label{lemma:F(j,k)} The following holds
			\begin{equation} \label{eqn:F(j,k)}
				F(j, k) := (-1)^k \frac{\Delta_{\mathcal{P}_t}^{j-k}}{j!} \sum_{i=1}^k (-1)^i \binom{k-1}{i-1} [i^j - (i-1)^j] = \begin{cases} 0, &j=0,\ldots, k-1,\\
					1, &j=k.\end{cases}.
			\end{equation}
		\end{lemma}
		\begin{proof}
			Note that we can write 
			\[F(j, k) = \frac{\Delta_{\mathcal{P}_t}^{j-k}}{j!}  f(j, k)\]
			where
			\begin{align*}
				f(j, k) &= (-1)^k \sum_{i=1}^k (-1)^i \binom{k-1}{i-1} [i^j - (i-1)^j] \\
				&= (-1)^k\left[\sum_{i=1}^k (-1)^i \binom{k-1}{i-1} i^j - \sum_{i=2}^k (-1)^i \frac{(k-1)!}{(i-1)!(k-i)!} (i-1)^j\right] \\
				&= (-1)^k\left[\sum_{i=1}^k (-1)^i \binom{k-1}{i-1} i^j - \sum_{i=2}^k (-1)^i \frac{(k-1)!}{(i-2)!(k-i)!} (i-1)^{j-1}\right] \\
				&= (-1)^k\left[\sum_{i=1}^k (-1)^i \binom{k-1}{i-1} i^j + \sum_{i=1}^{k-1} (-1)^{i} \frac{(k-1)!}{(i-1)!(k-i-1)!} i^{j-1}\right] \\
				&= (-1)^k\left[(-1)^k k^j + \sum_{i=1}^{k-1} (-1)^i \left( \frac{(k-1)!}{(i-1)!(k-i)!} i^j +\frac{(k-1)!}{(i-1)!(k-i-1)!}i^{j-1} \right)\right] \\
				&= (-1)^k\left[(-1)^k k^j + \sum_{i=1}^{k-1} (-1)^i \frac{(k-1)!}{(i-1)!(k-i-1)!} i^{j-1} \left( \frac{i}{k-i} +1 \right)\right] \\
				&= (-1)^k\left[(-1)^k k^j + \sum_{i=1}^{k-1} (-1)^i \frac{k!}{(i-1)!(k-i)!} i^{j-1} \right] \\
				&= (-1)^k k \sum_{i=1}^{k} (-1)^{i} \binom{k-1}{i-1} i^{j-1} \\
			\end{align*}
			We note that $f(j, k)$ satisfies the relation
			\begin{align*}
				f(j, k) &=(-1)^k k \sum_{i=1}^{k} (-1)^{i} \binom{k-1}{i-1} \sum_{m=0}^{j-1} \binom{j-1}{m} (i-1)^m \\
				&= (-1)^k k \left\{ \sum_{i=0}^{k-1} (-1)^{i+1} \binom{k-1}{i} + \sum_{m=1}^{j-1} \binom{j-1}{m} \left[\sum_{i=1}^{k} (-1)^{i} \frac{(k-1)!}{(i-1)!(k-i)!} (i-1)^m\right]\right\} \\
				&= (-1)^k k \left\{(1 - 1)^{k-1} + \sum_{m=1}^{j-1} \binom{j-1}{m} \left[(k-1) \sum_{i=2}^{k} (-1)^{i} \frac{(k-2)!}{(i-2)!(k-i)!} (i-1)^{m-1}\right]\right\} \\
				&= (-1)^k k \sum_{m=1}^{j-1} \binom{j-1}{m} (k-1) \sum_{i=2}^{k} (-1)^{i} \binom{k-2}{i-2} (i-1)^{m-1} \\
				&= k \sum_{m=1}^{j-1} \binom{j-1}{m} (-1)^{k-1}(k-1) \sum_{i=1}^{k-1} (-1)^{i} \binom{k-2}{i-1} i^{m-1} \\
				&= k \sum_{m=1}^{j-1} f(k-1, m).
			\end{align*}
			Proceeding by induction on $k\geq 1$ one can easily show that 
			\[f(j, k) = \begin{cases} 0, &j=0,\ldots, k-1,\\
				k!, &j=k.\end{cases}\]
		\end{proof}
		
		\section{Simulating MCAR(p) processes} \label{app:simulate_MCAR}
		In this section, we will discuss numerical methods for simulating a discrete-time realization $\mathbb{Y}_{\mathcal{P}_t}$ of the MCAR($p$) processes $\mathbb{Y}$ defined in Definition \ref{def:MCAR}. We note that the state-space representation of the MCAR($p$) process $\mathbb{X}=\{\mathbf{X}_t,\ t\geq 0\}$ satisfies the Ornstein-Uhlenbeck SDE \ref{eqn:SDE_def} with drift matrix $\mathcal{A}_{\mathbf{A}}$ and degenerate \Levy driver $\mathcal{E}\mathbb{L}=\{\mathcal{E}\mathbf{L}_t,\ t\geq 0\}$. The task of simulating a realization of $\mathbb{Y}_{\mathcal{P}_t}$ can thus be tackled by simulating the OU process $\mathbb{X}_{\mathcal{P}_t}$ and then recovering then keeping only the first $d$ entries. We first discuss an approximate scheme and then consider some special cases where exact simulation is possible.
		
		\subsection{An approximate simulation method: Euler-Maruyama scheme} \label{app:simulate_MCAR_approx}
		We note SDE \eqref{eqn:SDE_def} can be approximated over the partition $\mathcal{P}_t:=\{0=t_0<t_1<\ldots<t_{N_t} = t\}$ by
		\begin{equation}
			\mathbf{X}^{\mathcal{P}_t}_0 = \mathbf{X}_0, \quad \mathbf{X}^{\mathcal{P}_t}_{t_{n+1}} = \mathbf{X}^{\mathcal{P}_t}_{t_{n}} + \mathcal{A}_{\mathbf{A}} \mathbf{X}^{\mathcal{P}_t}_{t_n} \Delta^n_{\mathcal{P}_t} + \mathcal{E} \Delta^n_{\mathcal{P}_t} \mathbf{L}, \quad n=0,\ldots, N_t-1,
		\end{equation}
		% \begin{equation}
			%     \mathbf{X}_{t_{n+1}} - \mathbf{X}_{t_n} = \mathcal{A}_{\mathbf{A}} \mathbf{X}_{t_n} (t_{n+1} - t_n) + \mathcal{E}(\mathbf{L}_{t_{n+1}} - \mathbf{L}_{t_{n}}), \quad n=0,\ldots, N_t-1,
			% \end{equation}
		where we simulate \Levy increments $\Delta^n_{\mathcal{P}_t} \mathbf{L}$ using an exact simulation method. The Brownian component and the jump component of $\Delta^n_{\mathcal{P}_t} \mathbf{L} = \Delta^n_{\mathcal{P}_t} \mathbf{W}^{\Sigma} + \Delta^n_{\mathcal{P}_t} \mathbf{J}$ can be simulated independently:
		\begin{itemize}
			\item the continuous increment is given by $\Delta^n_{\mathcal{P}_t} \mathbf{W}^\Sigma \sim N(\mathbf{0}, \Delta^n_{\mathcal{P}_t} \Sigma)$;
			\item the jump component depends on whether:
			\begin{itemize}
				\item $\mathbb{L}$ has finite jump activity, i.e.\ $F = \lambda \tilde{F}$ where $\tilde{F}$ is the jump size distribution and $\lambda$ is the jump rate, in which case we can simulate the number of jumps in $[t_n, t_{n+1}]$ by a Poisson random variable $K_{\Delta^n_{\mathcal{P}_t}} \sim \mathrm{Poisson}(\Delta^n_{\mathcal{P}_t} \lambda)$ and i.i.d.\ jump sizes from $\tilde{F}$;
				\item $\mathbb{L}$ has infinite jump activity, in which case one simulates directly from the prescribed increment distribution, e.g.\ assuming the $d$ jump components follow independent Gamma processes we can sample each increment as $\Delta^n_{\mathcal{P}_t} J^{(j)} \sim \Gamma(\gamma^j \Delta^n_{\mathcal{P}_t}, \lambda^j)$, $j=,\ldots,d$.
			\end{itemize} 
		\end{itemize}
		
		\begin{remark}
			We can vectorize the operations to speed up computations. For example, if $\mathbf{L} = \mathbf{W}^{\Sigma} + \tilde{\mathbf{J}}$ has finite jump activity, one can generate
			\begin{itemize}
				\item $\{\Delta^n_{\mathcal{P}_t} \mathbf{W}^\Sigma, n=0,\ldots, N_t-1\}$ by simulating $N_t$ i.i.d.\ $N(\mathbf{0}, \Sigma)$ variables and scaling them by $\sqrt{\Delta^n_{\mathcal{P}_t}}$;
				\item $\{\Delta^n_{\mathcal{P}_t} \tilde{\mathbf{J}}, n=0,\ldots, N_t-1\}$ by 
				\begin{itemize}
					\item simulating $K_t\sim \mathrm{Poisson}(t \lambda)$ the total number of jumps over $[0,t]$,
					\item generating the jump times $U^{(1)}\leq \ldots\leq U^{(K_t)}$ as uniform order statistics over $[0,t]$ and the corresponding i.i.d.\ jump sizes $\Delta \tilde{\mathbf{J}}_1, \ldots, \Delta \tilde{\mathbf{J}}_{K_t} \sim \tilde{F}$,
					\item aggregate the jumps over the partition bins $[t_n, t_{n+1}]$ by $\Delta^n_{\mathcal{P}_t} \tilde{\mathbf{J}} = \sum_{k: U^{(k)}\in[t_n, t_{n+1})} \Delta \tilde{\mathbf{J}}_{k}$.
				\end{itemize}
			\end{itemize}
		\end{remark}
		The strong convergence rate of this Euler-Maruyama scheme is studied in \citet[Theorem~2.1]{Euler_scheme_Levy}. \citet[Remark~2.3.(i)]{Euler_scheme_Levy} covers the case where $\mathbb{L}$ has a non-vanishing (possibly degenerate) diffusion part, as in SDE \eqref{eqn:SDE_def}. For the OU process $\mathbb{X}$, under sufficient regularity conditions for the transition density of $\mathcal{E}\mathbb{L}$, they prove the strong convergence rate is
		\[\mathbb{E}_{\mathbf{A}}\left[ \sup_{s\in[0,t]} \|\mathbf{X}_s - \mathbf{X}^{\mathcal{P}_t}_s\|^2\right] \lesssim \Delta_{\mathcal{P}_t}.\]
		
		\subsection{An exact simulation method: the finite jump activity case} \label{app:simulate_MCAR_exact}
		When the driving \Levy process has finite jump activity we can simulate the OU process $\mathbb{X}$ exactly. Recall that when $\mathbb{L}$ has finite jump activity we can write the \Levy measure as $F=\lambda \tilde{F}$ where $\lambda$ is the jump arrival rate and $\tilde{F}$ is the distribution of jump sizes. The process $\mathbb{X}$ satisfies the integral equation
		\begin{equation} \label{eqn:evolve_X} \mathbf{X}_t = e^{\mathcal{A}_{\mathbf{A}}(t-s)}\mathbf{X}_s + \int_s^t e^{\mathcal{A}_{\mathbf{A}}(t-u)} \mathcal{E} d\mathbf{L}_u, \quad t\geq 0.\end{equation}
		By plugging-in the \LevyIto decomposition of $\mathbb{L}$, i.e.\
		$\mathbf{L}_t = \mathbf{b} t + \mathbf{W}^{\Sigma}_t + \tilde{\mathbf{J}}_t$ for $t\geq0$, we can write 
		\[ \mathbf{X}_t = e^{\mathcal{A}_{\mathbf{A}}(t-s)}\mathbf{X}_s + \int_s^t e^{\mathcal{A}_{\mathbf{A}}(t-u)} \mathcal{E} \mathbf{b} du + \int_s^t e^{\mathcal{A}_{\mathbf{A}}(t-u)} \mathcal{E} d\mathbf{W}^{\Sigma}_u + \int_s^t e^{\mathcal{A}_{\mathbf{A}}(t-u)} \mathcal{E} d\tilde{\mathbf{J}}_u, \]
		where 
		\begin{itemize}
			\item $\displaystyle \int_s^t e^{\mathcal{A}_{\mathbf{A}}(t-u)} \mathcal{E} \mathbf{b} du = \left(I - e^{\mathcal{A}_{\mathbf{A}}(t-s)}\right) \mathcal{A}_{\mathbf{A}}^{-1} \mathcal{E} \mathbf{b}$;
			\item $\displaystyle \int_s^t e^{\mathcal{A}_{\mathbf{A}}(t-u)} \mathcal{E} d\mathbf{W}^{\Sigma}_u \sim N\left(\mathbf{0}, F(t-s) e^{\mathcal{A}^\mathrm{T}_{\mathbf{A}}(t-s)}\right) $ where we use the \Ito isometry to compute
			\begin{align*}
				\mathbb{E}_{\mathbf{A}}\left[ \left(\int_s^t e^{\mathcal{A}_{\mathbf{A}}(t-u)} \mathcal{E} d\mathbf{W}^{\Sigma}_u \right) \left(\int_s^t e^{\mathcal{A}_{\mathbf{A}}(t-u)} \mathcal{E} d\mathbf{W}^{\Sigma}_u\right)^\mathrm{T} \right] &= \int_s^t e^{\mathcal{A}_{\mathbf{A}}(t-u)} \mathcal{E}\Sigma \mathcal{E}^\mathrm{T} e^{\mathcal{A}^\mathrm{T}_{\mathbf{A}}(t-u)} du \\
				&= \int_0^{t-s} e^{\mathcal{A}_{\mathbf{A}}(t-s-u)} \mathcal{E}\Sigma \mathcal{E}^\mathrm{T} e^{\mathcal{A}^\mathrm{T}_{\mathbf{A}}(t-s-u)} du \, \\
				& = F(t-s) e^{\mathcal{A}^\mathrm{T}_{\mathbf{A}}(t-s)},
			\end{align*}
			and define
			\[M = \begin{pmatrix} \mathcal{A}_{\mathbf{A}} & \mathcal{E}\Sigma \mathcal{E}^\mathrm{T} \\
				0 & -\mathcal{A}_{\mathbf{A}}^\mathrm{T} \end{pmatrix},\]
			to compute
			\[ e^{M(t-s)} = \begin{pmatrix} e^{\mathcal{A}_{\mathbf{A}}(t-s)} & F(t-s) \\
				0 & e^{-\mathcal{A}_{\mathbf{A}}^\mathrm{T} (t-s)} \end{pmatrix}; \]
			\item $\displaystyle \int_s^t e^{\mathcal{A}_{\mathbf{A}}(t-u)} \mathcal{E} d\tilde{\mathbf{J}}_u = \sum_{\tau_k\in[s,t]} e^{\mathcal{A}_{\mathbf{A}}(t-\tau_k)} \mathcal{E} \Delta\tilde{\mathbf{J}}_{\tau_k} $ where
			\begin{itemize}
				\item $\{\tau_k,\ k\geq1\}$ denote the jump times of $\tilde{\mathbb{J}}$, i.e.\ Poisson arrival times with rate $\lambda$, and 
				\item $\{\Delta\tilde{\mathbf{J}}_{\tau_k},\ k\geq1\}$ denote the jump sizes of $\tilde{\mathbb{J}}$, i.e.\ i.i.d.\ $\tilde{F}$-distributed random variables.
			\end{itemize}
		\end{itemize}
		Starting from $\mathbf{X}_0 = \mathbf{x}_0$ we can thus evolve the process $\mathbb{X}$ along $\mathcal{P}_{t'}=\{0=t_0<t_1<\ldots<t_{N_{t'}} = t'\}$ by setting $s=t_{n}$ and $t=t_{n+1}$ in Equation \eqref{eqn:evolve_X} and simulating the increments due to $\mathbf{W}^\Sigma$ and $\tilde{\mathbf{J}}$ independently. Similarly, we can simulate the initial condition from the stationary distribution by writing
		\begin{align*}\mathbf{X}_0 &= \int_0^\infty e^{\mathcal{A}_{\mathbf{A}}u} \mathcal{E} \mathbf{b} du + \int_0^\infty e^{\mathcal{A}_{\mathbf{A}} u} \mathcal{E} d\mathbf{W}^{\Sigma, \prime}_u + \int_0^\infty e^{\mathcal{A}_{\mathbf{A}}u} \mathcal{E} d\tilde{\mathbf{J}}'_u \\
			&\approx  - \mathcal{A}_{\mathbf{A}}^{-1} \mathcal{E} \mathbf{b} + \int_0^T e^{\mathcal{A}_{\mathbf{A}} u} \mathcal{E} d\mathbf{W}^{\Sigma, \prime}_u + \sum_{\tau_k\in[0,T]} e^{\mathcal{A}_{\mathbf{A}}\tau_k} \mathcal{E} \Delta\tilde{\mathbf{J}}'_{\tau_k} , \end{align*}
		where $T>0$ is a big enough cutoff threshold (which can be chosen depending on $\|\mathcal{A}_{\mathbf{A}}\|$) and $\mathbb{W}^{\Sigma, \prime}$ and $\tilde{\mathbb{J}}'$ are independent copies of $\mathbb{W}^{\Sigma}$ and $\tilde{\mathbb{J}}$.
		
		\section{Practical considerations: details and simulation study}
		\subsection{Disentangling Brownian and jump component of a \Levy process} \label{app:disentangling}
		In this section, we focus on how to disentangle the continuous and jump components of a discretely observed $d$-dimensional \Levy process $\mathbb{L}_{\mathcal{Q}_t}=\{\mathbf{L}_u,\ u\in\mathcal{Q}_t\}$ with $\mathcal{Q}_t =\{0=t_0<t_1<\ldots<t_{M_t}=t\}$. We assume throught the \Levy process $\mathbb{L}$ has \Levy triplet $(\mathbf{0}, \Sigma, F)$. In our setting, this is particularly useful when estimating the covariance matrix of the Brownian component $\Sigma$ and choosing the thresholding powers $\beta^{(i)}$, cf.\ Section \ref{sec:practical_considerations}. To do so we follow the approach outlined in \citet{Gegler_2011} based on a critical region $B^\gamma_t$ but generalize it to irregularly spaced data. The idea is quite simple and is essentially a multidimensional version of thresholding. Let us fix the hyperparameter $\gamma > 1$ and define the true critical regions depending on the unknown parameter $\Sigma$
		\begin{equation} \label{eqn:critical_region}
			B^\gamma_{t, m}(\Sigma) = \{\mathbf{x}\in\mathbb{R}^d:\mathbf{x}^\mathrm{T}\Sigma\mathbf{x} \leq 2 \gamma \Delta_{\mathcal{Q}_t}^m \log(M_t)\}.
		\end{equation}
		We can then classify each \Levy increment in $\{\Delta^m_{\mathcal{Q}_t}\mathbf{L} ,\ m=0,\ldots, M_t-1\}$ as
		\begin{itemize}
			\item a Brownian increment if $\Delta^m_{\mathcal{Q}_t}\mathbf{L} \in B^\gamma_{t, m}(\Sigma)$;
			\item a jump increment if $\Delta^m_{\mathcal{Q}_t}\mathbf{L} \notin B^\gamma_{t, m}(\Sigma)$.
		\end{itemize}
		The choice of such critical regions, inherently tied to the elliptical shape of multivariate normal level curves, ensures that asymptotically only Brownian increments are retained. This procedure leads to a consistent (as $\Delta_{\mathcal{Q}_t}\rightarrow 0$ with $t$ fixed) estimator for $\Sigma$
		\begin{equation} \label{eqn:Sigma_hat}
			\hat{\Sigma}(\mathbb{L}_{\mathcal{Q}_t};\gamma, \Sigma) = \frac{1}{M_t} \sum_{m=0}^{M_t-1} \frac{1}{\Delta^m_{\mathcal{Q}_t}} (\Delta^m_{\mathcal{Q}_t}\mathbf{L}) (\Delta^m_{\mathcal{Q}_t}\mathbf{L})^\mathrm{T} \mathds{1}_{\{\Delta^m_{\mathcal{Q}_t}\mathbf{L} \in B^\gamma_{t, m}(\Sigma)\}}.
		\end{equation}
		Clearly, this estimator is infeasible in practice since the true critical regions depend on the unknown covariance matrix $\Sigma$. To use this in practice we apply the following iterative scheme with tolerance level $\epsilon>0$:
		\begin{enumerate}
			\item initialize $\hat{\Sigma}^{(0)} = \frac{1}{M_t} \sum_{m=0}^{M_t-1} \frac{1}{\Delta^m_{\mathcal{Q}_t}} (\Delta^m_{\mathcal{Q}_t}\mathbf{L}) (\Delta^m_{\mathcal{Q}_t}\mathbf{L})^\mathrm{T};$
			\item for $n\geq 0$ set $\hat{\Sigma}^{(n+1)} = \hat{\Sigma}(\mathbb{L}_{\mathcal{Q}_t};\gamma, \hat{\Sigma}^{(n)})$ until $\|\hat{\Sigma}^{(n+1)} - \hat{\Sigma}^{(n)}\| \leq \epsilon$.
		\end{enumerate}
		\citet{Gegler_2011} shows that the limiting critical regions $B^\gamma_{t, m}(\hat{\Sigma}^*)$ where $\hat{\Sigma}^* = \lim_{n\rightarrow\infty}\hat{\Sigma}^{(n)}$ converge in probability to the true critical regions $B^\gamma_{t, m}(\Sigma)$. From the iterative procedure above we thus obtain an estimator $\hat{\Sigma} = \hat{\Sigma}(\mathbb{L}_{\mathcal{Q}_t};\gamma, \epsilon)$ for $\Sigma$ and feasible critical regions $\{B^\gamma_{t, m}(\hat{\Sigma}),\ m=0,\ldots, M_t-1\}$ which can be used to disentangle the Brownian and jump components of $\mathbb{L}_{\mathcal{Q}_t}$.
		
		\subsection{Simulation study with unknown \texorpdfstring{$\Sigma$}{Sigma} and \texorpdfstring{$F$}{F}} \label{app:realistic_simulations}
		In this section, we repeat the simulation study carried out in Section \ref{sec:simulation_study} but now assume the \Levy process parameters $\Sigma$ and $F$ are unknown. We use the data-driven methods described in Section \ref{sec:practical_considerations} to modify the estimation and inference procedure. In summary, we make the following adjustments:
		\begin{itemize}
			\item In the estimation step we select the thresholds $\nu_t^{m, (i)},\ m=0,\ldots, M_t-1$ by treating the increments $\Delta_{\mathcal{Q}_t}^m \hat{D}^{p-1}Y^{(i)},\ m=0,\ldots, M_t-1$ as univariate \Levy increments. We can thus disentangle the Brownian component by using the procedure described in Appendix \ref{app:disentangling}.
			\item Once we obtain the estimator $\hat{\mathbf{A}}(\mathbb{Y}_{\mathcal{P}_t}; \mathcal{Q}_t, \boldsymbol{\nu}_t)$ we recover the \Levy increments with Equation \eqref{eqn:recover_Levy} and estimate the Brownian covariance $\Sigma$ by using again the procedure described in Appendix \ref{app:disentangling}. We can then compute an approximation to the asymptotic z-statistic \eqref{eqn:feasible_CLT} by plugging the estimator $\hat{\Sigma}$ into the formula for $[\mathbf{H}]_{\mathcal{P}_t,\mathcal{Q}_t}$, i.e.\ Equation \eqref{eqn:[H]_hat_hat}.
		\end{itemize}
		The experiment is set up exactly as in Section \ref{sec:simulation_study}, i.e.\ we consider a $d=1$ dimensional CAR process of order $p=2$ such that the driving \Levy process has triplet $(0, \Sigma=1, F)$ and $F$ follows the three jump regimes (BM), (CP) and ($\Gamma$). As before $\mathcal{P}_t$ and $\mathcal{Q}_t$ are chosen to be uniform with mesh sizes $\Delta_{\mathcal{P}_t}=t^{-6}$ and $\Delta_{\mathcal{Q}_t}=t^{-2}$. The results are reported in Figures \ref{fig:kde_contour_estimators_realistic}, \ref{fig:kde_contour_statistic_realistic}, and \ref{fig:histograms_Sigma}, which empirically suggest consistency of $\hat{\mathbf{A}}(\mathbb{Y}_{\mathcal{P}_t}; \mathcal{Q}_t, \boldsymbol{\nu}_t)$, asymptotic normality of the approximate statistic $\mathbf{Z}_t(\hat{\Sigma})$ and consistency of the estimator $\hat{\Sigma}$.
		
		\begin{figure}
			\centering
			\includegraphics[width=\textwidth]{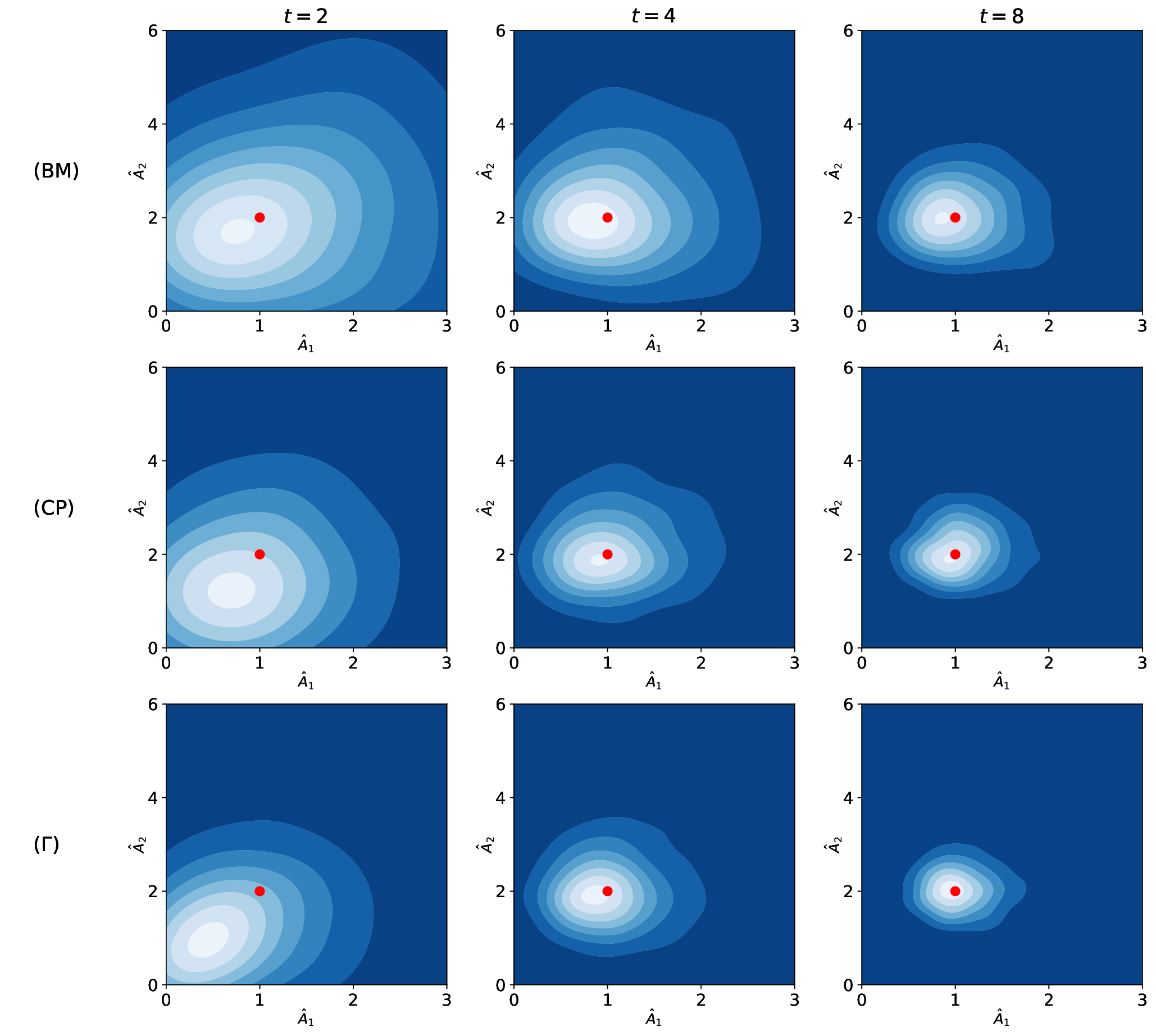}
			\caption{Empirical distribution of the estimator $\hat{\mathbf{A}}(\mathbb{Y}_{\mathcal{P}_t}; \mathcal{Q}_t, \boldsymbol{\nu}_t)$ for 1000 Monte Carlo samples under the three jump regimes for $t=2, 4, 8$ when $\Sigma$ and $F$ are unknown.}
			\label{fig:kde_contour_estimators_realistic}
		\end{figure}
		
		\begin{figure}
			\centering            \includegraphics[width=\textwidth]{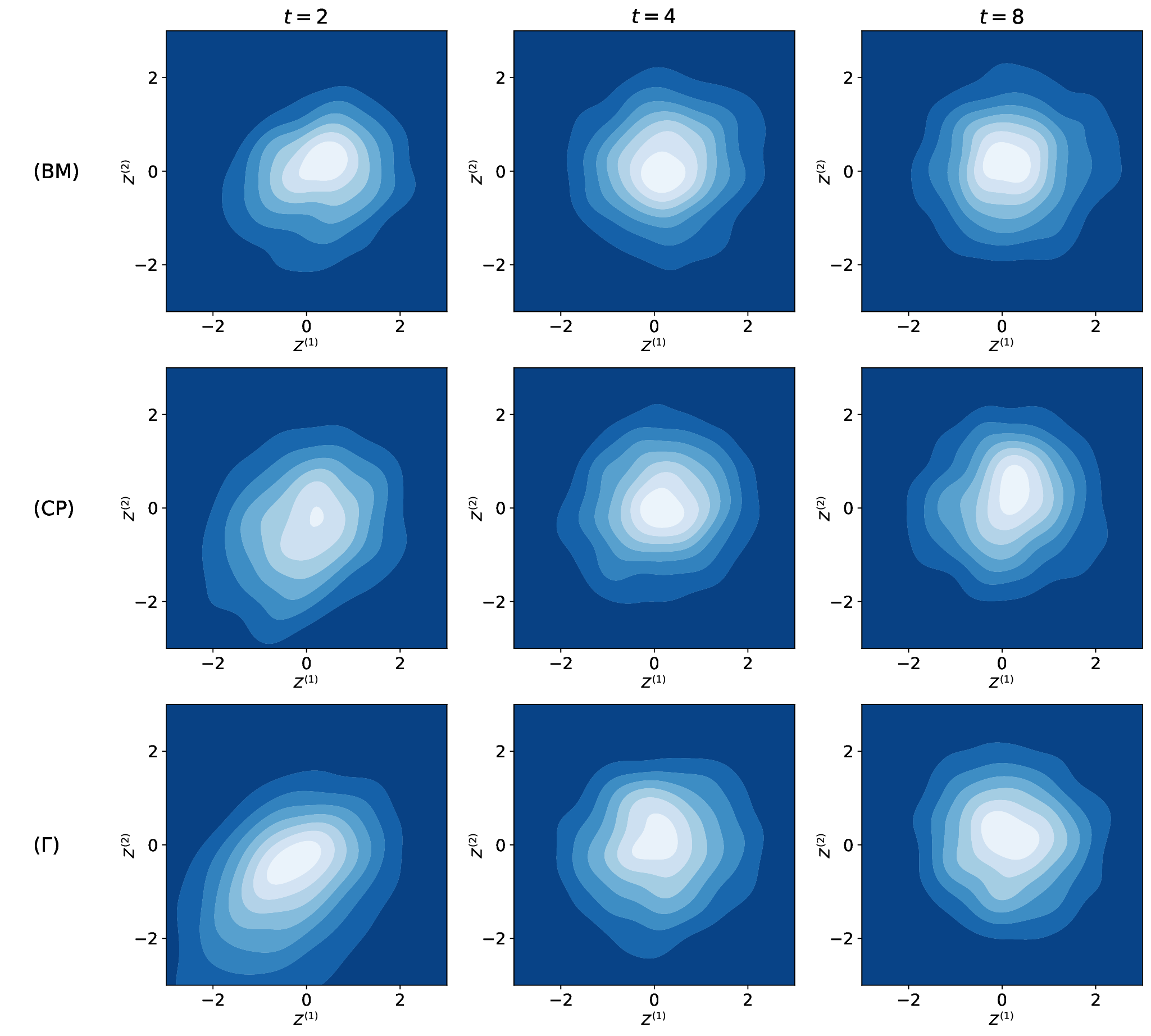}
			\caption{Empirical distribution of the statistic $\mathbf{Z}_t(\hat{\Sigma})$ for 1000 Monte Carlo samples under the three jump regimes for $t=2, 4, 8$ when $\Sigma$ and $F$ are unknown.}
			\label{fig:kde_contour_statistic_realistic}
		\end{figure}
		
		\begin{figure}
			\centering            \includegraphics[width=\textwidth]{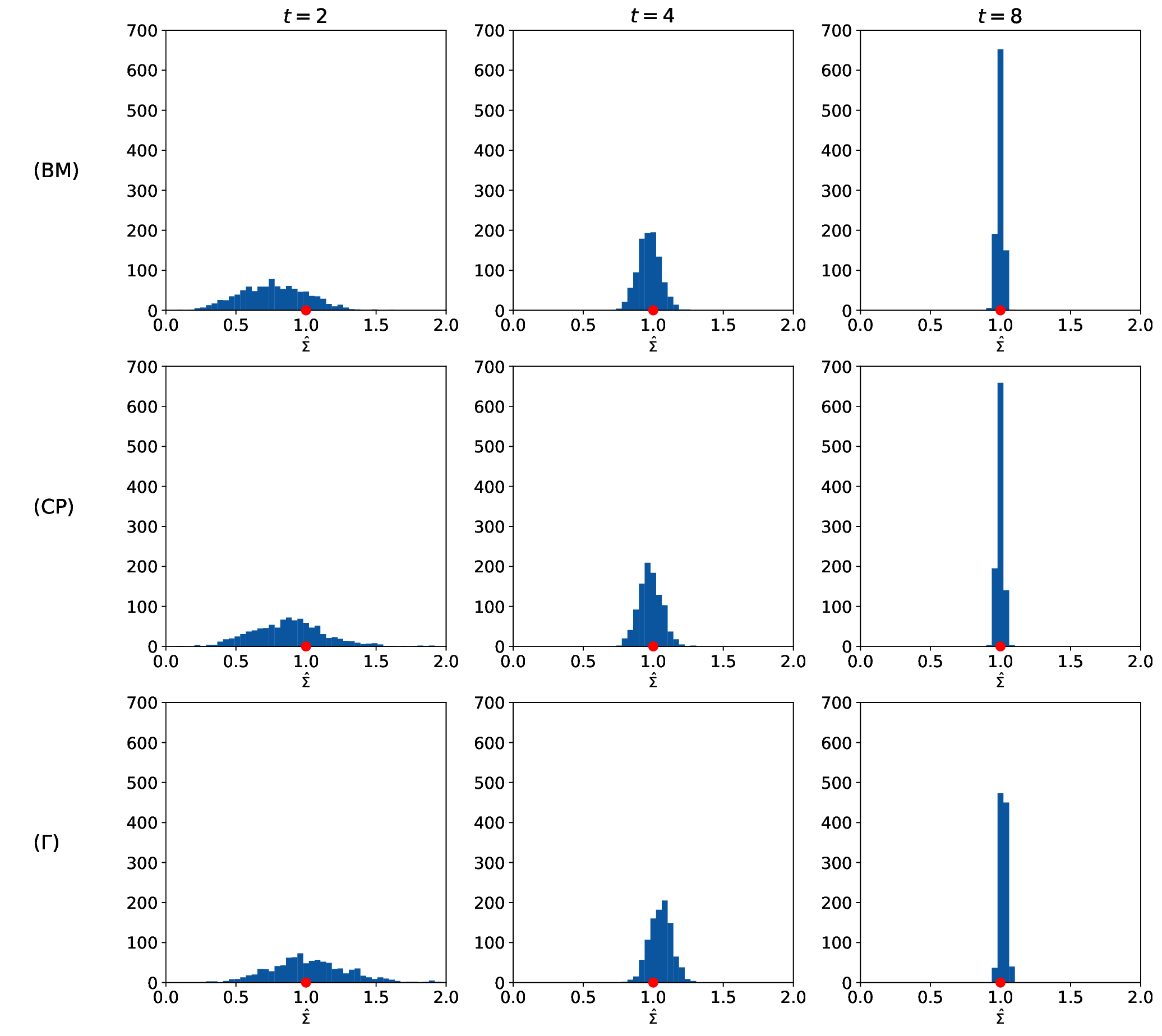}
			\caption{Empirical distribution of the estimator $\hat{\Sigma}$ for 1000 Monte Carlo samples under the three jump regimes for $t=2, 4, 8$ when $\Sigma$ and $F$ are unknown.}
			\label{fig:histograms_Sigma}
		\end{figure}

		\newpage
		%%%%% Nomenclature
		\nomenclature{$(\Omega', \mathcal{F}', \{\mathcal{F}'_t,\ t\geq0\}, \mathbb{P}')$}{is an arbitrary filtered probability space.}
		\nomenclature{$\mathbb{L}=\{\mathbf{L}_t,\ t\geq 0\}$}{is a $d$-dimensional \Levy process on $(\Omega', \mathcal{F}', \{\mathcal{F}'_t,\ t\geq0\}, \mathbb{P}')$.}
		\nomenclature{$(\mathbf{b}, \Sigma, F)$}{is the characteristic triplet of the \Levy process $\mathbb{L}$, i.e.\ $\mathbf{b}\in\mathbb{R}^d$, $\Sigma$ is a $d\times d$  symmetric positive definite matrix and $F$ is a \Levy measure.}
		\nomenclature{$\mathbb{W}=\{\mathbf{W}_t,\ t\geq 0\}$}{is a $d$-dimensional standard Brownian motion on $(\Omega', \mathcal{F}', \{\mathcal{F}'_t,\ t\geq0\}, \mathbb{P}')$.}
		\nomenclature{$\{\mu(B): B\in\mathcal{B}((0,\infty)\times \mathbb{R}^d\setminus\{0\})\}$}{is a Poisson random measure on $(\Omega', \mathcal{F}', \{\mathcal{F}'_t,\ t\geq0\}, \mathbb{P}')$ with compensator $F(dz)dt$.}
		\nomenclature{$\Delta \mathbb{L} = \{\Delta \mathbf{L}_t=  \mathbf{L}_t - \mathbf{L}_{t-},\ t\geq 0\}$}{is the jump process of $\mathbb{L}$.}
		\nomenclature{$\mathcal{B}(E)$}{is the Borel sigma algebra over the space $E$ equipped with its natural topology.}
		% \nomenclature{$\mathbb{X} = \{\mathbf{X}_t,\ t\geq 0\}$}{is an $\mathbb{R}^d$-valued semimartingale on $(\Omega', \mathcal{F}', \{\mathcal{F}'_t,\ t\geq0\}, \mathbb{P}')$.}
		\nomenclature{$\mathbb{A} = \{\mathbf{A}_t,\ t\geq 0\}$}{is a predictable finite variation process on $(\Omega', \mathcal{F}', \{\mathcal{F}'_t,\ t\geq0\}, \mathbb{P}')$.}
		\nomenclature{$\mathbb{M}=\{\mathbf{M}_t,\ t\geq 0\}$}{is a local martingale on $(\Omega', \mathcal{F}', \{\mathcal{F}'_t,\ t\geq0\}, \mathbb{P}')$.}
		\nomenclature{$(\Omega, \mathcal{F}, \{\mathcal{F}_t, \ t\geq 0\})$}{denotes the canonical space of $\mathbb{R}^{pd}$-valued paths.}
		\nomenclature{$\mathbb{L}=\{\mathbf{L}_t,\ t\geq 0\}$}{is a $d$-dimensional \Levy process on $(\Omega', \mathcal{F}', \{\mathcal{F}'_t,\ t\geq0\}, \mathbb{P}')$.}
		\nomenclature{$(\mathbf{b}, \Sigma, F)$}{is the characteristic triplet of the \Levy process $\mathbb{L}$, i.e.\ $\mathbf{b}\in\mathbb{R}^d$, $\Sigma$ is a $d\times d$  symmetric positive definite matrix and $F$ is a \Levy measure.}
		\nomenclature{$\mathcal{M}_{n,d}(E)$}{is the set of $n\times d$ matrices over $E$, e.g.\ $E=\mathbb{R}$.}
		\nomenclature{$\mathcal{M}_{d}(E)$}{is the set of $d\times d$ matrices over $E$, e.g.\ $E=\mathbb{R}$.}
		\nomenclature{$\mathcal{M}^{-}_{d}(\mathbb{R})$}{is the set of $d\times d$ matrices over $\mathbb{R}$ with eigenvalues having strictly negative real part.}
		\nomenclature{$\mathbf{A}=(A_1,\ldots, A_p)\in(\mathcal{M}_d(\mathbb{R}))^p$}{is a set of coefficient matrices for the MCAR process $\mathbb{Y}$.}
		\nomenclature{$\mathcal{A}_\mathbf{A}$}{is the $pd\times pd$ design matrix of the MCAR process $\mathbb{Y}$, given by Equation \eqref{eqn:design_A_E}.}
		\nomenclature{$\mathcal{E}$}{is the $pd\times d$ matrix given in Equation \eqref{eqn:design_A_E}.}
		\nomenclature{$\mathbb{Y} = \{\mathbf{Y}_t,\ t\geq 0\}$}{is a $d$-dimensional MCAR process of order $p$ as defined in Definition \ref{def:MCAR}, this is defined on $(\Omega, \mathcal{F}, \{\mathcal{F}_t, \ t\geq 0\})$ or on $(\Omega', \mathcal{F}', \{\mathcal{F}'_t, \ t\geq 0\})$ depending on the context.}
		\nomenclature{$\xi$}{is an $\mathcal{F}'_0$ measurable random variable.}
		\nomenclature{$\sigma(\cdot)$}{is the spectrum operator.}
		\nomenclature{$A\in\mathcal{M}_d(\{0,1\})$}{is a graph adjacency matrix.}
		\nomenclature{$\bar{A}\in\mathcal{M}_d(\mathbb{R})$}{is a column-normalized graph adjacency matrix.}
		\nomenclature{$\theta\in\mathbb{R}^{p\times 2}$}{is a parametrization of the GrCAR process, as defined in Definition \ref{def:GrCAR}.}
		
		%%%% MCAR continuous likelihood
		\nomenclature{$\mathbb{X}_{\mathbf{A},\mathbf{x}_0} = \{\mathbf{X}_{\mathbf{A},\mathbf{x}_0,t},\ t\geq 0\}$}{denotes the state-space representation of a MCAR($p$) with parameter $\mathbf{A}$ and deterministic initial condition $\xi = \mathbf{x}_0$ on $(\Omega', \mathcal{F}', \{\mathcal{F}'_t, \ t\geq 0\}, \mathbb{P}')$.}
		\nomenclature{$\mathbb{W}=\{\mathbf{W}_t,\ t\geq 0\}$}{is a $d$-dimensional standard Brownian motion on $(\Omega', \mathcal{F}', \{\mathcal{F}'_t,\ t\geq0\}, \mathbb{P}')$.}
		\nomenclature{$\{\mu(B): B\in\mathcal{B}((0,\infty)\times \mathbb{R}^d\setminus\{0\})\}$}{is a Poisson random measure on $(\Omega', \mathcal{F}', \{\mathcal{F}'_t,\ t\geq0\}, \mathbb{P}')$ with compensator $F(dz)dt$.}
		\nomenclature{$\Sigma^{1/2}$}{is the unique symmetric positive definite square root of the strictly positive definite symmetric matrix $\Sigma$.}
		\nomenclature{$(\mathbf{B}'_{\mathbf{A}}, C', \nu')$}{are the local characteristics of $\mathbb{X}_{\mathbf{A},\mathbf{x}_0}$ on $(\Omega', \mathcal{F}', \{\mathcal{F}'_t,\ t\geq0\}, \mathbb{P}')$ as defined in Equation \eqref{eqn:loc_char_P'}.}
		\nomenclature{$(\mathbf{B}_{\mathbf{A}}, C, \nu)$}{are local characteristics on $(\Omega, \mathcal{F}, \{\mathcal{F}_t,\ t\geq0\})$ as defined in Equation \eqref{eqn:loc_char_P}.}
		\nomenclature{$\mathbb{P}_{\mathbf{A}, \mathbf{x}_0}$}{is a probability measure on $(\Omega, \mathcal{F}, \{\mathcal{F}_t,\ t\geq0\})$. This denotes the unique solution measure to SDE \eqref{eqn:SDE} with deterministic initial condition $\mathbf{X}_0=\mathbf{x}_0$ and, equivalently, the unique solution to the martingale problem for the local characteristics $(\mathbf{B}_{\mathbf{A}}, C, \nu)$ with initial condition $\delta_{\mathbf{x}_0}$.}
		\nomenclature{$\mathbb{E}^{\mathbb{P}}[\cdot]$}{is the expectation operator corresponding to the probability measure $\mathbb{P}$.}
		\nomenclature{$\mathbb{X} = \{\mathbf{X}_t,\ t\geq 0\}$}{is the canonical process on $(\Omega, \mathcal{F}, \{\mathcal{F}_t,\ t\geq0\})$, i.e.\ $\mathbf{X}_t(\omega) = \omega(t),\ \forall t\geq0$. When $(\Omega, \mathcal{F}, \{\mathcal{F}_t,\ t\geq0\})$ is endowed with the measure $\mathbb{P}_{\mathbf{A},\mathbf{x}_0}$ this is the state-space representation of the MCAR process $\mathbb{Y}$, i.e.\ $\mathbf{X}_t = (\mathbf{Y}^\mathrm{T}_t, \ldots, D^{p-1}\mathbf{Y}^\mathrm{T}_t)^\mathrm{T}$.}
		\nomenclature{$\mathbb{X}^c_{\mathbf{A},\mathbf{x}_0} = \{\mathbf{X}^c_{\mathbf{A},t},\ t\geq 0\}$}{is the process on $(\Omega, \mathcal{F}, \{\mathcal{F}_t,\ t\geq0\})$ given by Equation \eqref{eqn:cont_mart_part}.}
		\nomenclature{$\Delta\mathbb{X} = \{\Delta\mathbf{X}_t=\mathbf{X}_t - \mathbf{X}_{t-},\ t\geq 0\}$}{is the canonical jump process.}
		\nomenclature{$\{\mu_\mathbb{X}(B): B \in \mathcal{B}((0,\infty)\times \mathbb{R}^{pd}\setminus \{0\})\}$}{is the canonical jump measure given by Equation \eqref{eqn:jump_meas}.}
		\nomenclature{$\mathbf{A}^{(0)} = (0_{d\times d}, \ldots, 0_{d\times d})\in(\mathcal{M}_d(\mathbb{R}))^p$}{is a reference set of parameters.}
		\nomenclature{$\mathbb{P}^t_{\mathbf{A}, \mathbf{x}_0}$}{denotes the restriction of $\mathbb{P}_{\mathbf{A}, \mathbf{x}_0}$ to $\mathcal{F}_t$.}
		\nomenclature{$\tilde{\mathcal{A}}_\mathbf{A}\in\mathcal{M}_{pd}(\mathbb{R})$}{is given by Equation \eqref{eqn:A_tilde}.}
		\nomenclature{$\tilde{\Sigma}^{-1}$}{is a pseudo-inverse for $\tilde{\Sigma}=\mathcal{E}\Sigma \mathcal{E}^\mathrm{T}$, given by Equation \eqref{eqn:Sigma_tilde}.}
		\nomenclature{$(\mathbb{X}, (B, C, \nu), \delta_{\mathbf{x}_0})$}{denotes the set of solutions to the martingale problem on canonical space with local characteristics $(B, C, \nu)$ and initial condition $\delta_{\mathbf{x}_0}$, i.e.\ a set of measure on $(\Omega, \mathcal{F}, \{\mathcal{F}_t,\ t\geq0\})$.}
		\nomenclature{$\mathbb{X}^t_{\mathbf{A},\mathbf{x}_0} = \{\mathbf{X}^t_{\mathbf{A},\mathbf{x}_0,s},\ s\geq 0\}$}{denotes the process $\mathbb{X}_{\mathbf{A},\mathbf{x}_0}$ stopped at $t>0$.}
		\nomenclature{$(\mathbf{B}^{\prime,t}_{\mathbf{A}}, C^{\prime,t}, \nu^{\prime,t})$}{are the local characteristics of $\mathbb{X}^t_{\mathbf{A},\mathbf{x}_0}$ on $(\Omega', \mathcal{F}', \{\mathcal{F}'_s,\ s\geq0\}, \mathbb{P}')$.}
		\nomenclature{$(\mathbf{B}_\mathbf{A}^t, C^t, \nu^t)$}{are local characteristics on $(\Omega, \mathcal{F}, \{\mathcal{F}_t,\ t\geq0\})$ as defined in Equation \eqref{eqn:loc_char_P_stopped}.}
		\nomenclature{$\mathbb{Q}_{\mathbf{A}, \mathbf{x}_0}$}{is a probability measure on $(\Omega, \mathcal{F}, \{\mathcal{F}_t,\ t\geq0\})$. This denotes the unique solution measure to SDE \eqref{eqn:SDE_stopped} with deterministic initial condition $\mathbf{X}_0=\mathbf{x}_0$ and, equivalently, the unique solution to the martingale problem for the local characteristics $(\mathbf{B}_\mathbf{A}^t, C^t, \nu^t)$ with initial condition $\delta_{\mathbf{x}_0}$.}
		\nomenclature{$\mathbb{Q}^t_{\mathbf{A}, \mathbf{x}_0}$}{denotes the restriction of $\mathbb{Q}_{\mathbf{A}, \mathbf{x}_0}$ to $\mathcal{F}_t$.}
		\nomenclature{$\mathbb{Z}_\mathbf{A} = \{\mathbf{Z}_{\mathbf{A},t},\ t\geq 0\}$}{is a predictable process on $(\Omega, \mathcal{F}, \{\mathcal{F}_{t},\ t\geq 0\})$ given by $\mathbf{Z}_{\mathbf{A},t}= \tilde{\Sigma}^{-1} \tilde{\mathcal{A}}_\mathbf{A} \mathbf{X}_{t-}$.}
		\nomenclature{$A^t_{\mathbf{A}, \infty} $}{is an $\mathcal{F}_t$ measurable random variable given by Equation \eqref{eqn:A^t_infty}.}
		\nomenclature{$P_{t, \mathbf{x}_0, \mathbf{A}}(\cdot) = \mathbb{P}'(\mathbf{X}_{\mathbf{A},\mathbf{x}_0, t}\in\cdot)$}{is the infinitely divisible transition probability of the unique (Markovian) solution to SDE \eqref{eqn:SDE} starting at $\mathbf{X}_0 = \mathbf{x}_0$, i.e.\ the law at time $t\geq 0$ of the state-space representation of an MCAR process with parameters $\mathbb{A}$ and initial condition $\mathbf{x}_0$. This law is infinitely divisible with characteristic triplet $(b_{t,\mathbf{x}_0, \mathbf{A}}, \Sigma_{t, \mathbf{A}}, F_{t, \mathbf{A}})$.}
		\nomenclature{$(b_{t,\mathbf{x}_0, \mathbf{A}}, \Sigma_{t, \mathbf{A}}, F_{t, \mathbf{A}})$}{is the characteristic triplet of $P_{t, \mathbf{x}_0, \mathbf{A}}(\cdot)$, given by Equation \eqref{eqn:triplet_X}.}
		\nomenclature{$\mathcal{Q}_{t,x} (d\omega) = \mathbb{Q}_{\mathbf{A}, t, x} (d\omega)$}{is a transition kernel from $(\mathbb{R}_+ \times \mathbb{R}^{pd}, \mathcal{B}(\mathbb{R}_+ \times \mathbb{R}^{pd}))$ to $(\Omega, \mathcal{F})$ associating to each $(t,x)$ a probability measure on canonical space.}
		\nomenclature{$L^p(\Omega', \mathcal{F}', \mathbb{P}')$}{is the space of $p$-integrable random variable on $(\Omega', \mathcal{F}', \mathbb{P}')$.}
		%% MCAR MLE inference
		\nomenclature{$D^j\mathbb{Y}_{[0,T]} = \{D^j\mathbf{Y}_t, \ t\in[0,T]\}$}{is the $j$-th time derivative of $\mathbb{Y}_{[0,T]}$ for $j=0,\ldots, p-1$.}
		\nomenclature{$\{\tilde{\mu}(B): B\in\mathcal{B}((0,\infty)\times \mathbb{R}^d\setminus\{0\})\}$}{is the jump measure of $D^{p-1}\mathbb{Y}_{[0,T]}$.}
		\nomenclature{$\mathcal{L}(\mathbf{A}; \mathbb{Y}_{[0,T]})$}{is the likelihood of the parameters $\mathbb{A}\in(\mathcal{M}_d(\mathbb{R}))^p$ given an observation $\mathbb{Y}_{[0,T]}$ of the MCAR($p$) process.}
		\nomenclature{$\mathbb{H} = \{\mathbf{H}_t, \ t\geq 0\}$}{is the $\mathbb{R}^{pd^2}$-valued process given by Equation \eqref{eqn:H}.}
		\nomenclature{$[\mathbb{H}] = \{[\mathbf{H}]_t, \ t\geq 0\}$}{is the $\mathbb{R}^{pd^2\times pd^2}$-valued  process given by Equation \eqref{eqn:[H]}, the quadratic variation of $\mathbb{H}$.}
		\nomenclature{$\mathfrak{A}\subseteq (\mathcal{M}_d(\mathbb{R}))^p$}{is the set of parameters which ensure, along with conditions on the driving \Levy process and the initial condition, stationarity and ergodicity of the MCAR($p$) process and its state-space representation. This is given by Equation \eqref{eqn:param_set}.}
		\nomenclature{$\hat{\mathbf{A}}(\mathbb{Y}_{[0,t]})\in(\mathcal{M}_d(\mathbb{R}))^p$}{is the maximum likelihood estimator of the MCAR($p$) parameters given a continuous time observation $\mathbb{Y}_{[0,t]}$, Equation \eqref{eqn:MLE}.}
		\nomenclature{$\mathbf{A}^*\in(\mathcal{M}_d(\mathbb{R}))^p$}{is the true parameter of the MCAR($p$) process.}
		\nomenclature{$\mathrm{vec}(\cdot)$}{is the vectorization operator, mapping an input tensor or matrix to a vector.}
		\nomenclature{$\mathcal{H}_\infty$}{is the $pd^2\times pd^2$ matrix given by Equation \eqref{eqn:H_infty}.}
		\nomenclature{$\mathbf{X}_\infty = (\mathbf{Y}_\infty^\mathrm{T},\ldots, D^{p-1}\mathbf{Y}_\infty^\mathrm{T})^\mathrm{T}$}{is a ``representative'' random variable for the stationary distribution of the MCAR($p$) state-space representation.}
		%% GrCAR MLE inference
		\nomenclature{$\mathcal{L}(\theta; \mathbb{Y}_{[0,T]})$}{is the likelihood of the parameters $\theta\in\mathbb{R}^{p\times 2}$ given an observation of the GrCAR($p$) process $\mathbb{Y}_{[0,T]}$.}
		\nomenclature{$\mathbb{K} = \{\mathbf{K}_t, \ t\geq 0\}$}{is the $\mathbb{R}^{2p}$-valued process given by Equation \eqref{eqn:K}.}
		\nomenclature{$[\mathbb{K}] = \{[\mathbf{K}]_t, \ t\geq 0\}$}{is the $\mathbb{R}^{2p\times 2p}$-valued process given by Equation \eqref{eqn:[K]}, the quadratic variation of $\mathbb{K}$.}
		\nomenclature{$\Theta\subseteq \mathbb{R}^{p\times 2}$}{is the set of parameters which ensure, along with conditions on the driving \Levy process and the initial condition, stationarity and ergodicity of the GrCAR($p$) process and its state-space representation. This is given by Equation \eqref{eqn:param_set_Theta}.}
		\nomenclature{$\hat{\theta}_t\in\mathbb{R}^{p\times 2}$}{is the maximum likelihood estimator of the GrCAR($p$) parameters given a continuous time observation $\mathbb{Y}_{[0,t]}$, Equation \eqref{eqn:MLE_GrCAR}.}
		\nomenclature{$\theta^*\in\mathbb{R}^{p\times 2}$}{is the true parameter of the GrCAR($p$) process.}
		\nomenclature{$\mathcal{K}_\infty$}{is the $2p\times 2p$ matrix given by Equation \eqref{eqn:K_infty}.}

		%%% MCAR discrete estimator
		\nomenclature{$\mathbb{Y}_{[0,t]} = \{\mathbf{Y}_s, \ s\in[0,t]\}$}{is a continuous time observation of an MCAR($p$), or GrCAR($p$), process over $[0,t]$.}
		\nomenclature{$D^j\mathbb{Y}_{[0,t]} = \{D^j\mathbf{Y}_s, \ s\in[0,t]\}$}{is the $j$-th time derivative of $\mathbb{Y}_{[0,t]}$ for $j=0,\ldots, p-1$.}
		\nomenclature{$\mathcal{P}_t := \{0=s_0< s_1<\ldots <s_{N_t} = t\}$}{is a partition of $[0,t]$.}
		\nomenclature{$\mathbb{X}_{[0,t]} = \{\mathbf{X}_s,\ s\in[0,t]\}$}{is the state-space representation of the MCAR process $\mathbb{Y}_{[0,t]}$, i.e.\ $\mathbf{X}_s = (\mathbf{Y}^\mathrm{T}_s, \ldots, D^{p-1}\mathbf{Y}^\mathrm{T}_s)^\mathrm{T}$. This is the canonical process on $(\Omega, \mathcal{F}, \{\mathcal{F}_t, \ t\geq0\})$ over $[0,t]$.}
		\nomenclature{$\mathbb{H} = \{\mathbf{H}_t, \ t\geq 0\}$}{is the $\mathbb{R}^{pd^2}$-valued process given by Equation \eqref{eqn:H}.}
		\nomenclature{$[\mathbb{H}] = \{[\mathbf{H}]_t, \ t\geq 0\}$}{is the $\mathbb{R}^{pd^2\times pd^2}$-valued  process given by Equation \eqref{eqn:[H]}, the quadratic variation of $\mathbb{H}$.}
		\nomenclature{$\Delta_{\mathcal{P}_t}$}{is the mesh of the partition $\mathcal{P}_t$.}
		\nomenclature{$\mathbb{Y}_{\mathcal{P}_t} = \{\mathbf{Y}_s, \ s\in\mathcal{P}_t\}$}{is a discrete time observation of an MCAR($p$), or GrCAR($p$), process over $\mathcal{P}_t$.}
		\nomenclature{$\mathbb{L}_{[0,t]} = \{\mathbf{L}_t,\ s\in[0,t]\}$}{is the background driving \Levy process of the MCAR($p$) process $\mathbb{Y}_{[0,t]}$ under $\mathbb{P}_{\mathbf{A}^*}$, i.e.\ \eqref{eqn:L=p(D)Y}.} 
		\nomenclature{$\mathbb{P}_{\mathbf{A}^*}$}{is the probability measure on $(\Omega, \mathcal{F}, \{\mathcal{F}_t, \ t\geq0\})$ such that the canonical process $\mathbb{X}$ is the state-space representation of a stationary and ergodic MCAR($p$) process with parameter $\mathbf{A}^*\in \mathfrak{A}$.}
		\nomenclature{$D^{p-1}\mathbb{Y}^c_{\mathbf{A}^{(0)}, [0,t]} = \{D^{p-1}\mathbf{Y}^c_{\mathbf{A}^{(0)}, s}, \ s\in[0,t]\}$}{is given by Equation \eqref{eqn:cont_mart_part_Dp-1Y}.}
		\nomenclature{$\Delta D^{p-1}\mathbb{Y}_{[0,t]} = \{\Delta D^{p-1}\mathbf{Y}_{s} = D^{p-1}\mathbf{Y}_{s} - D^{p-1}\mathbf{Y}_{s-},\ s\in[0,t]\}$}{is the jump process of $D^{p-1}\mathbb{Y}_{[0,t]}$.}
		\nomenclature{$\mathbb{J}_{[0,t]} = \{\mathbf{J}_s, \ s\in[0,t]\}$}{is a continuous time observation of the jump process \eqref{eqn:jump_J} over $[0,t]$.}
		\nomenclature{$\tilde{\mathbb{J}}_{[0,t]} = \{\tilde{\mathbf{J}}_s, \ s\in[0,t]\}$}{is a continuous time observation of the jump process \eqref{eqn:jump_J_tilde} over $[0,t]$. This is defined when the background \Levy process is assumed to have finite jump activity.}
		\nomenclature{$\mathbb{M}_{[0,t]} = \{\mathbf{M}_s, \ s\in[0,t]\}$}{is a continuous time observation of the jump process \eqref{eqn:jump_M} over $[0,t]$.}
		\nomenclature{$\mathbf{A}^{(0)} = (0_{d\times d}, \ldots, 0_{d\times d})\in(\mathcal{M}_d(\mathbb{R}))^p$}{is a reference set of parameters.}
		% \nomenclature{$\mathbf{A}^*\in\mathfrak{A}$}{is the true parameter of the MCAR($p$) process.}
		\nomenclature{$\mathbb{E}_{\mathbf{A}^*}[\cdot]$}{is the expectation operator corresponding to the probability measure $\mathbb{P}_{\mathbb{A}^*}$.}
		\nomenclature{$\mathbb{J}_{\mathcal{P}_t} = \{\mathbf{J}_s, \ s\in\mathcal{P}_t\}$}{is a discrete time observation of $\mathbb{J}_{[0,t]}$ over $\mathcal{P}_t$.}
		\nomenclature{$\mathbb{M}_{\mathcal{P}_t} = \{\mathbf{M}_s, \ s\in\mathcal{P}_t\}$}{is a discrete time observation of $\mathbb{M}_{[0,t]}$ over $\mathcal{P}_t$.}
		\nomenclature{$D^{p-1}\mathbb{Y}^c_{\mathbf{A}^{(0)}, \mathcal{P}_t} = \{D^{p-1}\mathbf{Y}^c_{\mathbf{A}^{(0)}, s}, \ s\in\mathcal{P}_t\}$}{is a discrete time observation of $D^{p-1}\mathbb{Y}^c_{\mathbf{A}^{(0)}, [0,t]}$ over $\mathcal{P}_t$.}
		\nomenclature{$D^{j}\mathbb{Y}_{\mathcal{P}_t} = \{D^{j}\mathbf{Y}_{s}, \ s\in\mathcal{P}_t\}$}{is a discrete time observation of $D^{i}\mathbb{Y}_{[0,t]}$ over $\mathcal{P}_t$, for $j=0,\ldots, p-1$.}
		\nomenclature{$\mathcal{T}$}{is a countable set of times increasing to $\infty$.}
		\nomenclature{$\{\mathcal{P}_t,\ t\in\mathcal{T}\}$}{is a sequence of partitions $\mathcal{P}_t$ of $[0,t]$.}
		\nomenclature{$\hat{\mathbf{A}}(\mathbb{Y}_{\mathcal{P}_t}, \ldots, D^{p-1}\mathbb{Y}_{\mathcal{P}_t},\mathbb{J}_{\mathcal{P}_t}, \mathbb{M}_{\mathcal{P}_t})$}{is the estimator for $\mathbf{A}^*$ obtained from the discrete observations $\mathbb{Y}_{\mathcal{P}_t}, \ldots, D^{p-1}\mathbb{Y}_{\mathcal{P}_t}, \mathbb{J}_{\mathcal{P}_t}, \mathbb{M}_{\mathcal{P}_t}$ given in Equation \eqref{eqn:discr_estimator_A}.}
		\nomenclature{$\{\mathbf{H}_{\mathcal{P}_t},\ t\in\mathcal{T}\}$}{is the approximator of $\mathbf{H}_t$ obtained from the discrete observations $\mathbb{Y}_{\mathcal{P}_t}, \ldots, D^{p-1}\mathbb{Y}_{\mathcal{P}_t}, \mathbb{J}_{\mathcal{P}_t}, \mathbb{M}_{\mathcal{P}_t}$ given in Equation \eqref{eqn:H_hat}.}
		\nomenclature{$\{[\mathbf{H}]_{\mathcal{P}_t},\ t\in\mathcal{T}\}$}{is the approximator of $[\mathbf{H}]_t$ obtained from the discrete observations $\mathbb{Y}_{\mathcal{P}_t}, \ldots, D^{p-1}\mathbb{Y}_{\mathcal{P}_t}, \mathbb{J}_{\mathcal{P}_t}, \mathbb{M}_{\mathcal{P}_t}$ given in Equation \eqref{eqn:[H]_hat}.}
		\nomenclature{$\mathcal{H}_\infty$}{is the $pd^2\times pd^2$ matrix given by Equation \eqref{eqn:H_infty}.}
		\nomenclature{$\mathbf{X}_\infty = (\mathbf{Y}_\infty^\mathrm{T},\ldots, D^{p-1}\mathbf{Y}_\infty^\mathrm{T})^\mathrm{T}$}{is a ``representative'' random variable for the stationary distribution of the MCAR($p$) state-space representation.}
		\nomenclature{$L^p(\Omega, \mathcal{F}, \mathbb{P}_{\mathbf{A}^*})$}{is the space of $p$-integrable random variable on $(\Omega, \mathcal{F}, \mathbb{P}_{\mathbf{A}^*})$.}
		\nomenclature{$\mathbb{W}_{[0,t]} = \{\mathbf{W}_s, \ s\in[0,t]\}$}{is a standard Brownian motion under $\mathbb{P}_{\mathbf{A}^*}$ implicitly defined via Equation \eqref{eqn:DY=W-A}.}
		\nomenclature{$\mathrm{tr}(\cdot)$}{is the trace operator.}
		\nomenclature{$\hat{D}^{j}\mathbb{Y}_{\mathcal{P}_t} = \{\hat{D}^{j}\mathbf{Y}_{s}, \ s\in\mathcal{P}_t\}$}{is the forward difference approximations of $D^{j}\mathbb{Y}_{\mathcal{P}_t}$ computed from $\mathbb{Y}_{\mathcal{P}_t}$ over $\mathcal{P}_t$, for $j=0,\ldots, p-1$. This is defined in Equation \eqref{eqn:D_hat}.}
		\nomenclature{$c_{\mathcal{P}_t}$}{is the ratio between the shortest and the longest subintervals in the partition $\mathcal{P}_t$.}
		\nomenclature{$\mathcal{Q}_t = \{0=u_0<u_1<\ldots < u_{N'_t} = t\}$}{is a coarsening of $\mathcal{P}_t$.}
		\nomenclature{$\hat{\mathbf{A}}(\mathbb{Y}_{\mathcal{P}_t},\mathbb{J}_{\mathcal{P}_t}, \mathbb{M}_{\mathcal{P}_t};\mathcal{Q}_t)$}{is the estimator for $\mathbf{A}^*$ obtained from the discrete observations $\mathbb{Y}_{\mathcal{P}_t}, \mathbb{J}_{\mathcal{P}_t}, \mathbb{M}_{\mathcal{P}_t}$ and the coarsening partition $\mathcal{Q}_t$ given in Equation \eqref{eqn:discr_estimator_A_2}.}
		\nomenclature{$\{\mathbf{H}_{\mathcal{P}_t, \mathcal{Q}_t},\ t\in\mathcal{T}\}$}{is the approximator of $\mathbf{H}_t$ obtained from the discrete observations $\mathbb{Y}_{\mathcal{P}_t}, \mathbb{J}_{\mathcal{P}_t}, \mathbb{M}_{\mathcal{P}_t}$ and the coarsening partition $\mathcal{Q}_t$ given in Equation \eqref{eqn:H_hat_hat}.}
		\nomenclature{$\{[\mathbf{H}]_{\mathcal{P}_t, \mathcal{Q}_t},\ t\in\mathcal{T}\}$}{is the approximator of $[\mathbf{H}]_t$ obtained from the discrete observations $\mathbb{Y}_{\mathcal{P}_t}, \mathbb{J}_{\mathcal{P}_t}, \mathbb{M}_{\mathcal{P}_t}$ and the coarsening partition $\mathcal{Q}_t$ given in Equation \eqref{eqn:[H]_hat_hat}.}
		\nomenclature{$C(s_n,\ldots, s_{n+k};j)$}{is given by Equation \eqref{eqn:C_partition} for any collection of $k$ consecutive points $\{s_n,\ldots,s_{n+k}\}\in\mathcal{P}_t$ and $j\leq k$.}
		\nomenclature{$F(j,k)$}{for $j\leq k$ is defined in Equation \eqref{eqn:F(j,k)}.}
		\nomenclature{$|\bar{C}(s_n,\ldots, s_{n+k};k)|$}{is given by Equation \eqref{eqn:C_bar_partition} for any collection of $k$ consecutive points $\{s_n,\ldots,s_{n+k}\}\in\mathcal{P}_t$.}
		\nomenclature{$\Delta_{\mathcal{Q}_t}^m=(u_{m+1}-u_m)$}{is the $m$-th increment of the partition $\mathcal{Q}_t$.}
		\nomenclature{$\Delta_{\mathcal{Q}_t}^m\mathbf{Z}=(\mathbf{Z}_{u_{m+1}}-\mathbf{Z}_{u_m})$}{is the $m$-th increment of a process $\mathbb{Z}=\{\mathbf{Z}_s,\ s\geq 0\}$ over the partition $\mathcal{Q}_t$.}
		\nomenclature{$\boldsymbol{\nu}_t=\{\boldsymbol{\nu}_t^m,\ m=0,\ldots, M_t-1\}$}{is a sequence of thresholding vectors $\boldsymbol{\nu}_t^m\in\mathbb{R}^d$ over the partition $\mathcal{Q}_t$.}
		\nomenclature{$\Delta_{\mathcal{Q}_t}^m\hat{D}^{p-1}\mathbf{Y}^c_{\mathbf{A}^{(0)}}$}{is the $m$-th thresholded increment defined in Equation \eqref{eqn:thresholded_incr}.}
		\nomenclature{$\hat{\mathbf{A}}(\mathbb{Y}_{\mathcal{P}_t}; \mathcal{Q}_t, \boldsymbol{\nu}_t)$}{is the estimator for $\mathbf{A}^*$ obtained from the discrete observations $\mathbb{Y}_{\mathcal{P}_t}$, the coarsening partition $\mathcal{Q}_t$ and the thresholding sequence $\boldsymbol{\nu}_t$ given in Equation \eqref{eqn:discr_estimator_A_3}.}
		\nomenclature{$\{\mathbf{H}_{\mathcal{P}_t, \mathcal{Q}_t, \boldsymbol{\nu}_t},\ t\in\mathcal{T}\}$}{is the approximator of $\mathbf{H}_t$ obtained from the discrete observations $\mathbb{Y}_{\mathcal{P}_t}$, the coarsening partition $\mathcal{Q}_t$ and the thresholding sequence $\boldsymbol{\nu}_t$  given in Equation \eqref{eqn:H_hat_hat_hat}.}
		\nomenclature{$\{\mathbf{N}_{t}=(N_t^{(1)},\ldots, N_t^{(d)})^\textrm{T},\ t\geq 0\}$}{is the collection of marginal Poisson counting processes associated to $\tilde{\mathbb{J}}$ in the finite activity case and to $\mathbb{J}$ in the infinite activity case.}
		\nomenclature{$\lambda^{(1)},\ldots, \lambda^{(d)}$}{is the collection of marginal Poisson jump rates associated to $\tilde{\mathbb{J}}$ in the finite activity case and to $\mathbb{J}$ in the infinite activity case.}
		\nomenclature{$\tilde{F}^{(1)},\ldots, \tilde{F}^{(d)}$}{is the collection of marginal jump size distributions associated to $\tilde{\mathbb{J}}$ in the finite activity case and to $\mathbb{J}$ in the infinite activity case.}
		\nomenclature{$\beta^{(1)},\ldots, \beta^{(d)}$}{are the thresholding powers used to construct the collection of thresholding vectors $\boldsymbol{\nu}_t$ from the partition $\mathcal{Q}_t$ in Assumption \ref{ass:finite_thresholding} and Assumption \ref{ass:infinite_thresholding}, i.e.\ $\nu_t^{m, (i)}=(\Delta_{\mathcal{Q}_t})^{\beta^{(i)}}$ for $i=1,\ldots, d$.}
		\nomenclature{$\{\mathbf{Z}_{t},\ t\in\mathcal{T}\}$}{is the asymptotic statistic of the feasible CLT \eqref{eqn:feasible_CLT}.}
		\nomenclature{$B_{t,m}^\gamma(\Sigma)$}{is the critical region defined in Equation \eqref{eqn:critical_region} to disentangle the Brownian and jump components of a \Levy process.}
		\nomenclature{$\hat{\Sigma}(\mathbb{L}_{\mathcal{Q}_t}; \gamma, \Sigma)$}{is the estimator for $\Sigma$ defined in Equation \eqref{eqn:Sigma_hat}.}
		\nomenclature{$\hat{\Sigma}(\mathbb{L}_{\mathcal{Q}_t}; \gamma, \epsilon)$}{is the estimator for $\Sigma$ resulting from the iterative procedure outlined in Appendix \ref{app:disentangling}.}
		
		\printnomenclature
		
	\end{document}